\documentclass[aps,pra,twocolumn,superscriptaddress,longbibliography]{revtex4-1}
\usepackage[normalem]{ulem}
\usepackage{float}
\usepackage{graphicx}  
\usepackage{dcolumn}          
\usepackage{amssymb}
\usepackage{appendix}
\usepackage{physics}   
\usepackage{mathtools}
\usepackage{esvect}
\usepackage{wrapfig}
\usepackage{amsthm}
\usepackage{verbatim}
\usepackage{bbm}
\usepackage[colorlinks=true,linkcolor=blue,citecolor=red,plainpages=false,pdfpagelabels]{hyperref}

\usepackage[mathscr]{euscript}
\usepackage{enumitem}

\def\Tr{\operatorname{Tr}}

\def\FS{\operatorname{FS}}
\def\GE{\operatorname{GE}}
\def\E{\operatorname{E}}
\def\BS{\operatorname{BS}}

\def\supp{\operatorname{supp}}
\def\LOCC{\operatorname{LOCC}}
\def\GHZ{\operatorname{GHZ}}

\def\({\left(}
\def\){\right)}
\def\[{\left[}
\def\]{\right]}

\let\emptyset\varnothing

\newcommand{\mc}[1]{\mathcal{#1}}
\newcommand{\wt}[1]{\widetilde{#1}}

\newcommand{\tf}[1]{\textbf{#1}}

\newcommand{\msc}[1]{\mathscr{#1}}

\newcommand{\bbm}[1]{\mathbbm{#1}}

\newtheorem{theorem}{Theorem}

\newtheorem{corollary}{Corollary}

\newtheorem{definition}{Definition}

\newtheorem{lemma}{Lemma}

\newtheorem{proposition}{Proposition}
\newtheorem{remark}{Remark}

\begin{document}

\widetext

\title{Universal Limitations on Quantum Key Distribution over a Network}

\author{Siddhartha Das}\email{das.seed@gmail.com}
\affiliation{Centre for Quantum Information \& Communication (QuIC), \'{E}cole polytechnique de Bruxelles,   Universit\'{e} libre de Bruxelles, Brussels, B-1050, Belgium}
\author{Stefan B\"{a}uml}\email{stefan.baeuml@icfo.eu}
\affiliation{ICFO-Institut de Ciencies Fotoniques, The Barcelona Institute of Science and Technology, Avinguda Carl Friedrich Gauss 3, 08860 Castelldefels (Barcelona), Spain}
\author{Marek Winczewski}
\affiliation{Institute of Theoretical Physics and Astrophysics, National Quantum Information Centre,
Faculty of Mathematics, Physics and Informatics, University of Gdańsk, Wita Stwosza 57, 80-308 Gdańsk, Poland}
\affiliation{International Centre for Theory of Quantum Technologies (ICTQT), University of Gdańsk,
80-308 Gdańsk, Poland}
\author{Karol Horodecki}
\affiliation{International Centre for Theory of Quantum Technologies (ICTQT), University of Gdańsk,
80-308 Gdańsk, Poland}
\affiliation{Institute of Informatics, National Quantum Information Centre, Faculty of Mathematics,
Physics and Informatics, University of Gdańsk, Wita Stwosza 57, 80-308 Gdańsk, Poland}

\date{\today}
\begin{abstract}
We consider the distribution of secret keys, both in a bipartite and a multipartite (conference) setting, via a quantum network and establish a framework to obtain bounds on the achievable rates. We show that any multipartite private state--the output of a protocol distilling secret key among the trusted parties--has to be genuinely multipartite entangled. In order to describe general network settings, we introduce a multiplex quantum channel, which links an arbitrary number of parties, where each party can take the role of sender only, receiver only, or both sender and receiver. We define asymptotic and non-asymptotic LOCC-assisted secret-key-agreement (SKA) capacities for multiplex quantum channels and provide strong and weak converse bounds. The structure of the protocols we consider, manifested by an adaptive strategy of secret key and entanglement [Greenberger–Horne–Zeilinger (GHZ state)] distillation over an arbitrary multiplex quantum channel, is generic. As a result, our approach also allows us to study the performance of quantum key repeaters and measurement-device-independent quantum key distribution (MDI-QKD) setups. For teleportation-covariant multiplex quantum channels, we get upper bounds on the SKA capacities in terms of the entanglement measures of their Choi states. We also obtain bounds on the rates at which secret key and GHZ states can be distilled from a finite number of copies of an arbitrary multipartite quantum state. We are able to determine the capacities for MDI-QKD setups and rates of GHZ-state distillation for some cases of interest. 

\end{abstract}

\maketitle
\tableofcontents

\section{Introduction}
 Quantum communication over a network is a pertinent issue from both fundamental and application aspects~\cite{BB84,Eke91,DM03,RennerThesis,CLL+09,VBD+15,ZXC+18}. With technological advancement~\cite{HBD+15,Mst19,BRA+19,CZC+21}, and concerns for privacy~\cite{Sho94,ZXC+18}, there is a need for determining protocols and criteria for secret communication among multiple trusted parties in a network. Quantum key distribution (QKD) provides unconditional security for generating secure, random bits among trusted parties against a quantum eavesdropper, i.e., an eavesdropper that is only limited by the laws of quantum mechanics. Secret key agreement (SKA) among multiple allies is called conference key agreement \cite{Chen2005,AH09}. Conference key agreement can be achieved if all parties involved share a Greenberger–Horne–Zeilinger (GHZ) state \cite{GHZ89}. As in the case of bipartite QKD, however, there exists a larger class of states, known as multipartite private states \cite{AH09}, which can provide conference key by means of local measurements by the parties.
 
  \begin{figure}[ht]
        \centering
        \includegraphics[width=0.67\textwidth]{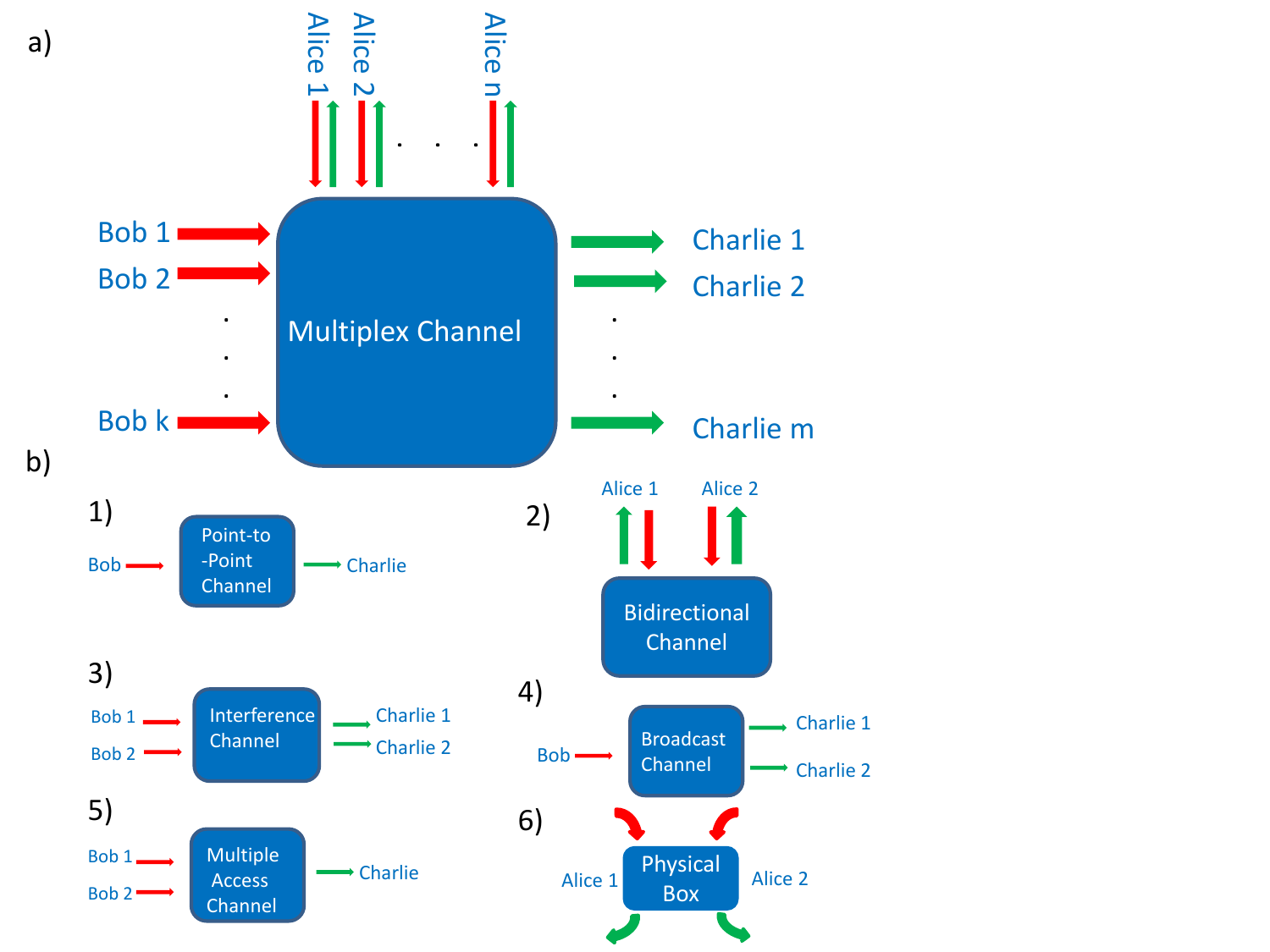}
        \caption{Pictorial illustration of a universal nature of a multiplex quantum channel from which all other network quantum channels arise, where red and green arrows show inputs and outputs to channels, respectively; see Section~\ref{sec:mc} for definitions.}\label{fig:mpc}
    \end{figure}
 
 Given the global efforts towards a so-called quantum internet \cite{DM03,Kim08,WEH18}, as well as quantum key distribution over long distances \cite{MHS+12,LCL+17}, it is thus pertinent to establish security criteria and benchmarks on key distribution and entanglement generation capabilities over a quantum network. A quantum network is a complex structure as it inherits various setups of different quantum channels with particular alignment due to local environmental conditions. One of the biggest obstacles in building this structure is an attenuation of the signal, which cannot be amplified by cloning or broadcasting due to its inherent quantum nature. The signal decays exponentially with distance over an optical fiber~\cite{ATL15}, and also, the interaction with the environment makes it difficult to preserve entanglement for long time~\cite{BRA+19}. Hence, even obtaining a metropolitan scale quantum network remains a challenge. To overcome these problems, there is a global effort in building technology of quantum repeaters~\cite{PhysRevLett.81.5932,PhysRevA.59.169,munro2015inside,CZC+21} that could act as relay stations for long-distance quantum communication~\cite{LCL+17,ZXC+18}. 
 
 \begin{figure}
        \centering
        \includegraphics[trim=1cm 8cm 0cm 0cm,width=11cm]{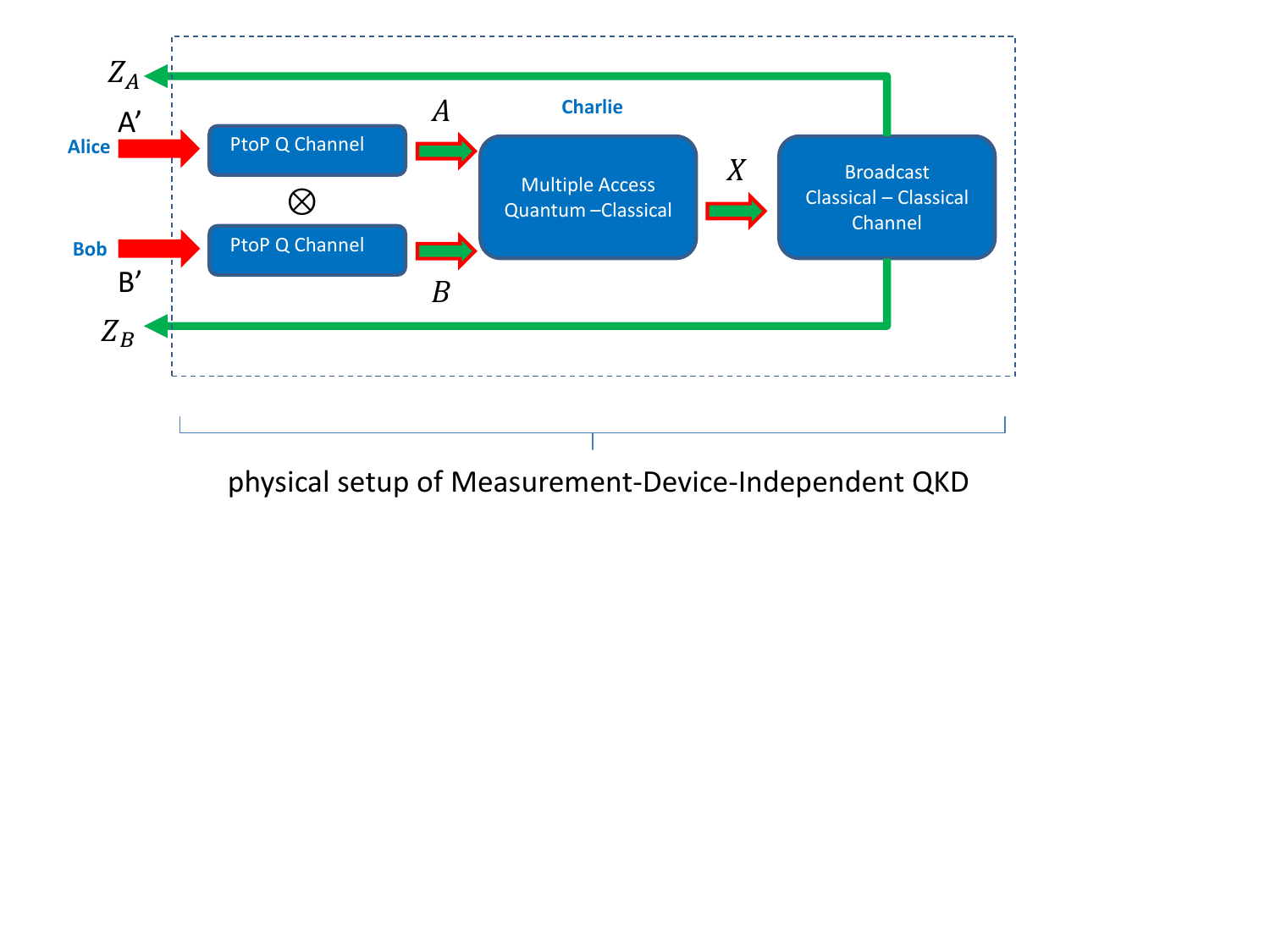}
        \caption{ Graphical depiction of a quantum to classical multiplex channel ${\cal N}_{A'B'\rightarrow Z_AZ_B}^{\text{MDI}}$ as a bidirectional Channel, which is a composition of three elementary multiplex channels. A pair of point-to-point channels from Alice to Charlie, and from Bob to Charlie composed with a multiple access quantum to classical channel (quantum instrument) performed by Charlie, followed by a broadcast classical channel back to Alice and Bob. The green arrows with red boundaries are the outputs of one multiplex channel, which are at the same time inputs to the other,
        hence the coloring.
        }\label{fig:MDI-fig}
    \end{figure}

Some of the first protocols to be performed once a quantum network is available will likely be bipartite as well as multipartite secret key agreement. The main concern to secure the network is a necessity for these QKD protocols to be free of loopholes. A number of spectacular attacks on implementations are based on inaccuracy (inefficiency) of detectors of polarized light~\cite{MAS06,QFLM05,DLQY16}. 
Based on the idea of entanglement swapping, a novel protocol known as
measurement-device-independent QKD (MDI-QKD) \cite{LCQ+12,braunstein2012side} was introduced, which does not require the honest parties to detect an incoming quantum signal, thus avoiding the problem of detector inefficiencies.
This allowed for a QKD protocol that is independent of any measurement device and hence called {\it measurement-device-independent} QKD (MDI-QKD)~\cite{LCQ+12,braunstein2012side}. This idea has drawn enormous theoretical and experimental attention over the last few years in terms of analyzing achievable key rates for such a scheme with various noise models and performing experiments with current technologies \cite{pirandola2015high,FYCC15,LYDS18,MZZ18,tamaki2018information,LL18,cui2019twin,curty2019simple,LWW+19,MPR+19,PLG+19}. 

Given the broad interest in implementing such technologies, understanding the fundamental limitations on the key rates achievable in scenarios such as quantum networks, quantum repeaters as well as setups for MDI-QKD is an important task. Seminal papers~\cite{HHHO05,christandl2004squashed} on upper bounds on secret key distillation from states along with results from Refs.~\cite{BDSW96,VPRK97,VP98,HHH99,Dat09} have led to notable recent progress in the aforementioned direction, for two parties over point-to-point channels assisted by local quantum operations and classical communication (LOCC)~\cite{TGW14,PLOB15,WTB16,CM17}. Building upon these works, further progress has been made in restricted network settings, e.g., between two parties over bidirectional~\cite{DBW17,BDW18,D18thesis}, broadcast~\cite{LP17,seshadreesan2016bounds,TSW17}, multiple access, and interference quantum channels~\cite{LP17}, as well as quantum repeaters~\cite{bauml2015limitations,CM17} and networks consisting of point-to-point~\cite{AML16,RKB+17,pirandola2019capacities} or broadcast channels~\cite{bauml2017fundamental}.

In this work, we aim to provide a unifying framework to derive upper bounds on the key rates, both in bipartite and conference settings, achievable in a broad range of different scenarios, including but not limited to broadcast, multiple access, interference channels, repeaters, some MDI-QKD setups and more general network scenarios. To that purpose, we introduce a \textit{multiplex quantum channel}, i.e., a multipartite quantum process that connects parties, each playing one of three possible roles -- both sender and receiver, only sender, or only receiver. A multiplex quantum channel is the most general form of a
memoryless multipartite quantum channel in a communication network setting. All other network quantum channels can be seen as a special case of this channel, see Fig.~\ref{fig:mpc} for certain common examples. Even the physical setups of MDI-QKD and \textit{key repeaters} can be described as  special cases of multiplex quantum channels, see Fig.~\ref{fig:MDI-fig}. In general, the input and output systems on which such a channel acts can be discrete (finite dimensional) or continuous variable (infinite dimensional) quantum systems.

Next, we introduce secret key distribution protocols over multiplex quantum channels with LOCC-assistance between users, as shown in Fig.~\ref{fig:LOCC}. This provides a unifying framework to evaluate performances of various seemingly different QKD protocols. In particular, we describe a general paradigm of QKD protocols where a fixed number of trusted allies are connected over a multiplex quantum channel $\mathcal{N}$. In these protocols, the allies
are allowed to perform LOCC between each use of $\mc{N}$ to generate in the end a key secure against any eavesdropper that satisfies the laws of quantum mechanics. This so-called quantum eavesdropper can have access to all environment parts, including the isometric extension to the channel $\mc{N}$. 

Our main technical result consists of a meta-converse bound on the one-shot conference key agreement capacity of a multiplex quantum channel, from which we can obtain a number of weak as well as strong converse bounds for many uses of the multiplex quantum channel, including adaptive and non-adaptive strategies. As our results work in the non-asymptotic setting of a finite number of channel uses, we believe them to be of wide practical interest.

In particular, as an important observation, we show that key repeater protocols, as well as commonly used setups for MDI-QKD are special cases of LOCC-assisted secret key agreement via a multiplex quantum channel. Whereas bounds on the key rates in such scenarios can also be obtained from a number of earlier results, e.g. from \cite{pirandola2019capacities,AML16,CM17}, our framework allows for higher level of specificity in the setups, e.g. by taking into consideration the lack of quantum memory or a particular kind of noisy measurement that is performed in the relay station. Thus our framework allows us obtain tighter bounds than those in Refs.~\cite{pirandola2019capacities,AML16,CM17} and even to compute MDI-QKD capacities of certain photon-based practical prototypes that use the so-called \textit{dual-rail} encoding scheme. This approach provides important tools for benchmarking the performance of such experimentally relevant protocols.

When considering conference key agreement, the pivotal observation we arrive at is that multipartite quantum states with directly accessible secret bits, also called (multipartite) private states~\cite{HHHO09,AH09}, are genuinely multipartite entangled. This fact also allows us to derive non-asymptotic upper bounds on the secret key distillation from a finite number of copies of a multipartite quantum state.

Our work showcases the topology-dependent and yet universal nature of entanglement measures based on sandwiched R\'enyi relative entropies~\cite{WWY14,MDSFT13}, of which relative entropy is a special case. These entanglement measures provide upper bounds on the secret key rate over an arbitrary multiplex quantum channel which was first shown for bipartite states in Ref.~\cite{HHHO05}. The entanglement measures are topology-dependent because the upper bound's argument depends (only) on the partition of quantum systems held by trusted allies based on their roles in the network channel. The results are based on the observation that multipartie private states are necessarily genuinely multipartite entangled.

\begin{figure}
\centering 
   \includegraphics[trim={1cm 3cm 1cm 1cm},clip,width=0.45\textwidth]{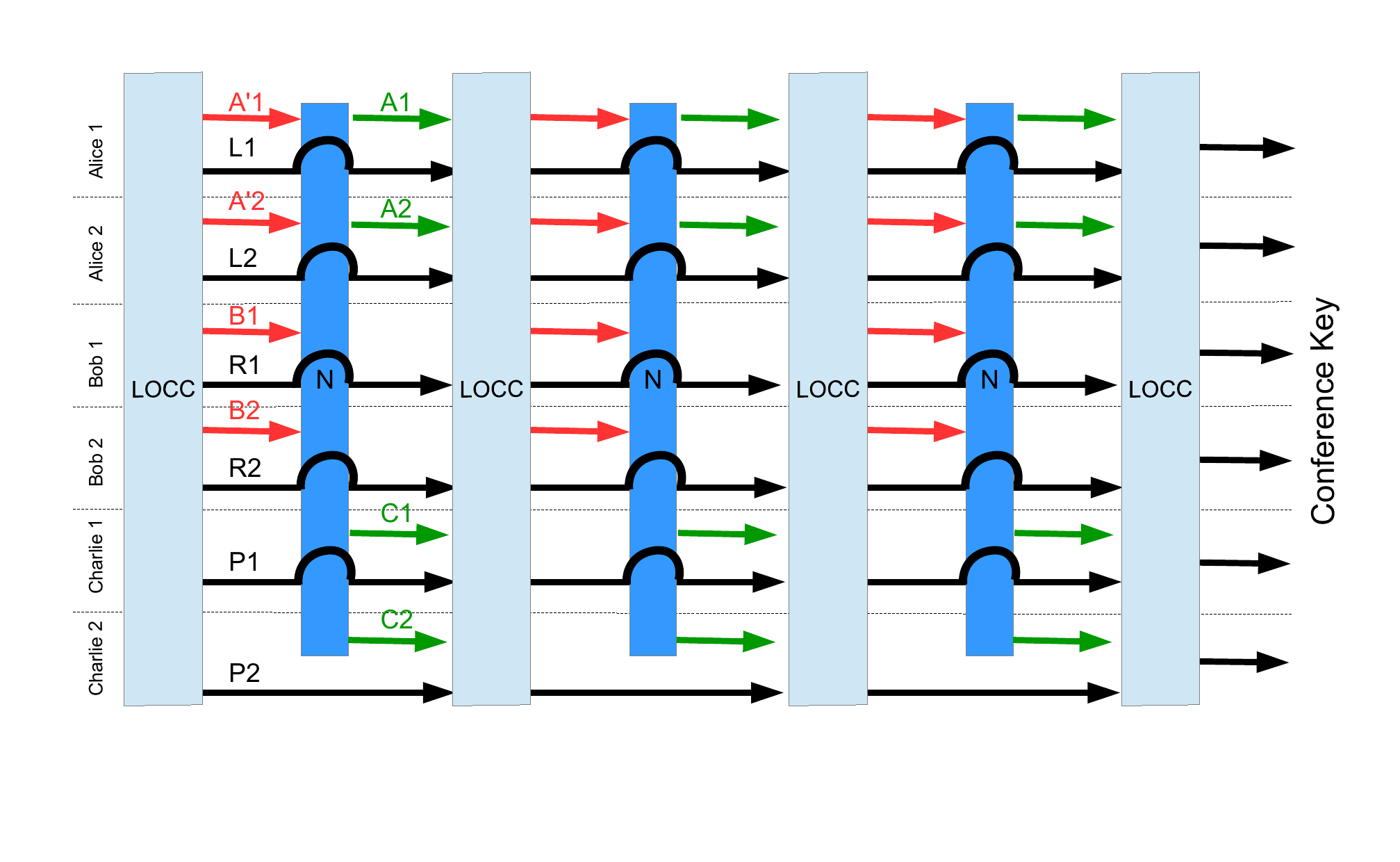}
 \caption{Example of an LOCC-assisted secret key agreement protocol among six parties, Alice 1, Alice 2, Bob 1, Bob 2, Charlie 1 and Charlie 2 using the multiplex channel $\mc{N}$ three times. Inputs into $\mc{N}$ are depicted red, outputs green and reference systems black. Alice 1 and 2 enter systems into and receive systems from $\mc{N}$, Bob 1 and 2 only enter and Charlie 1 and 2 only receive systems. In the end the six parties obtain a six-partite conference key.} \label{fig:LOCC}
 \end{figure}

The structure of this paper is as follows. We begin with a brief overview of the main results and briefly mention some important prior results along the direction of our work in Secs.~\ref{sec:main} and~\ref{sec:survey}, respectively. We introduce notations and review basic definitions and relevant prior results in Section~\ref{sec:pre}. In Section~\ref{sec:MEPT}, we introduce and discuss the properties of entanglement measures for the multiplex quantum channel. We show that genuine multipartite entanglement is a necessary criterion for secrecy. In Section~\ref{sec:cka}, we introduce LOCC-assisted secret key agreement protocols over an arbitrary multiplex quantum channel. We derive upper bounds on the maximum achievable rate for conference key agreement over finite uses of multiplex quantum channels. In Section~\ref{sec:Appl}, we leverage our bounds to provide non-trivial upper bounds on other quantum key distribution schemes such as measurement-device-independent quantum key distribution and quantum key repeaters. In Section~\ref{sec:DevetakWinter}, we derive lower bounds on the secret-key-agreement capacity over an arbitrary multiplex quantum channel. In Section~\ref{sec:dist-state}, we derive upper bounds on the number of secret key bits that can be distilled via LOCC among trusted parties sharing a finite number of copies of multipartite quantum states. We provide concluding remarks and open questions in Section~\ref{sec:dis}.

\section{Summary of the main results}\label{sec:main}
In the following, we will provide a brief overview of our main results. Regarding technique, our focus is on multipartite private states, which are the most general class of states that provide quantum conference key directly (i.e. without distillation) by local measurements. Such states are of the form  \cite{AH09}
\begin{equation}
\gamma_{\vv{SK}}\coloneqq U^{\text{tw}}_{\vv{SK}} (\Phi^{\GHZ}_{\vv{K}}\otimes\omega_{\vv{S}}) (U^{\text{tw}}_{\vv{SK}})^\dag,
\end{equation}
where $\vv{K}=K_1,...,K_N$ denotes the so-called key part, i.e. the systems which the $N$ parties involved have to measure in order to obtain conference and $\vv{S}=S_1,...,S_N$ denotes the so-called shield systems, which the parties have to keep secure from the eavesdropper. Also, $\Phi^{\GHZ}$ is an $N$-partite GHZ state, $\omega$ is some density operator, and $U^{\text{tw}}$ is a specifically constructed bipartite unitary operation known as twisting. 

We show that states of this form are necessarily genuinely multipartite entangled (GME), i.e., they cannot be expressed as a convex sum of product states no matter with respect to which partition the states are products. To show this, we define a multipartite privacy test, i.e. a dichotomic measurement $\{\Pi^\gamma,\bbm{1}-\Pi^\gamma\}$ such that any $\epsilon$-approximate multipartite private state $\rho$ with fidelity $F(\rho,\gamma)\geq1-\epsilon$ passes the test with success probability $\Tr[\Pi^\gamma\rho]\geq 1-\epsilon$. We then show that any biseparable state $\sigma$ cannot pass the privacy test with probability larger than $1/K$, where $\log K$ is the number of conference key bits obtainable by measuring (the key part of) $\gamma$. Namely we show that $\Tr[\Pi^\gamma\sigma]\leq1/K$ for all biseparable $\sigma$.

As a means of distributing bipartite or multipartite private states among the users, e.g. in a future quantum version of the Internet \cite{DM03,WEH18}, we introduce {\it multiplex quantum channels} that connect a number of parties which have one of three possible roles-- that of only sender or only receiver, or both sender and receiver. We denote senders as Bob $1$ ,..., Bob $k$, and their inputs as $B_1,...,B_k$, receivers as Charlie $1$ ,..., Charlie $m$, and their inputs as $C_1,...,C_m$, as well as parties that are both senders and receivers as Alice $1$ ,..., Alice $n$, with respective inputs $A'_1,...,A'_n$ and outputs $A_1,...,A_n$. See also figure \ref{fig:mpc}. To describe such channel, we use the notation $\mc{N}_{\vv{A'}\vv{B}\to\vv{A}\vv{C}}$, where for sake of brevity we have introduced $\vv{A}\coloneqq A_1,...,A_n$ etc. Further, $:\vv{A}:$ denotes the partition $A_1:...:A_n$ and $:\vv{A}:\vv{B}:$ stands for $A_1:...:A_n:B_1:...:B_k$  etc.

By interleaving the uses of a multiplex quantum channel with local operations and classical communications (LOCC) among the parties, we provide a general framework to describe a number of different quantum protocols. The idea is to construct a multiplex quantum channel in such a way that its use interleaved by LOCC simulates the protocol. For example, in an MDI-QKD setup, where Alice $1$ and Alice $2$ send states to the central measurement unit using respective channels ${\mc N}^{1,2}$, we can define a (bipartite) multiplex quantum channel of the form
\begin{align}
&\mc{N}^{\text{MDI}}_{A_1'A_2'\to A_1A_2}\coloneqq \nonumber\\ 
&\qquad \mc{B}_{X\to A_1A_2}\circ\mc{M}_{A''_1A''_2\to X}\circ{\mc N}^1_{A_1'\to A_1''}\otimes{\mc N}^2_{A_2'\to A_2''}.
\end{align}
Here  $\mc{M}_{A''_1A''_2\to X}$ is the quantum channel performing the central measurement  and $\mc{B}_{X\to A_1A_2}$ a classical broadcast channel sending the result back to Alice $1$ and Alice $2$. Other examples include multipartite  MDI-QKD and secret key agreement protocols over quantum network laced with key repeaters \cite{bauml2015limitations,CM17}.

Generalizing results for point-to-point \cite{PLOB15,WTB16,CM17} and bidirectional \cite{D18thesis,DBW17,BDW18} channels, we derive divergence-based measures for entangling abilities of multiplex quantum channels and show that they provide upper bounds on their secret-key-agreement capacities. The measures we introduce are of the following form: 
\begin{equation}
\tf{E}_{r}(\mc{N})\coloneqq \sup_{\tau\in\FS(:\vv{LA'}:\vv{RB}:)}\tf{E}_{r}(:\!\vv{LA}\!:\!\vv{R}\!:\!\vv{C}\!:)_{\mc{N}(\tau)},
\end{equation} 
where $r=\text{E}$ or $r=\text{GE}$ (E and GE denote entanglement and genuine entanglement, respectively) and $\FS$ denotes the set of fully separable states (see Sections~\ref{sec:ME} and \ref{sec:EM}). Here $\vv{L},\vv{R}$ denote ancillary systems that are kept by the respective parties. For any partition $:\vv{X}:$, we have defined $\tf{E}_{r}$ as the divergence from the convex set $\tf{S}_E$ of fully separable  or the convex set $\tf{S}_{GE}$ of biseparable states, measured by some divergence $\tf{D}$:
\begin{equation}
\tf{E}_{r}(:\vv{X}:)_{\rho}\coloneqq \inf_{\sigma\in\tf{S}_r(:\vv{X}:)}\tf{D}(\rho\Vert\sigma).
\end{equation}

Our main results are the following upper bounds  on secret-key-agreement capacities of a multiplex quantum channel, i.e. on the maximum rates at which multipartite private states can be obtained by using the channel as well as some free operations. In the one-shot case of a multiplex quantum channel with classical preprocessing and postprocessing (cppp), we have the following weak converse result: For any fixed $\varepsilon\in(0,1)$, the achievable region of cppp-assisted secret key agreement over a multiplex channel $\mc{N}$ satisfies 
\begin{equation}
P_{\textnormal{cppp}}^{(1,\varepsilon)}(\mc{N})\leq E^\varepsilon_{h,\GE}(\mc{N}),
\end{equation} 
where $E^\varepsilon_{h,\GE}(\mc{N})$ is the $\varepsilon$-hypothesis testing relative entropy of genuine multipartite entanglement of the multiplex channel $\mc{N}$, which is based on the $\varepsilon$-hypothesis testing divergence \cite{BD10}. In the case of many channel uses, interleaved by LOCC, we can also show the following strong converse bound:
\begin{equation}
P_{\LOCC}(\mc{N})\leq E_{\max,E}(\mc{N}), 
\end{equation}
where $E_{\max,E}(\mc{N})$ is the max-relative entropy of entanglement of the multiplex channel $\mc{N}$, which is based on the max-relative entropy \cite{Dat09}. In the case of finite dimensional Hilbert spaces we can also get a strong converse result in terms of the regularized relative entropy,
\begin{equation}
P_{\LOCC}(\mc{N})\leq E^{\infty}_{R,E}(\mc{N}).
\end{equation}
If $\mc{N}$ is teleportation-simulable \cite{BBPS96,PLOB15}, i.e. it can be simulated by a resource state and an LOCC operation, the bounds on $P_{\LOCC}(\mc{N})$ reduce to the relative entropy of entanglement of the resource state. Our upper bounds on the secret-key-agreement capacities also are upper bounds on the multipartite quantum capacities where goal is to distill GHZ states.

Our technique allows us to compute upper bounds on the rates achievable in MDI-QKD scenarios. For an instance, we consider a dual-rail scheme based on single photons~\cite{DKD18} to determine bounds on the MDI-QKD rates for two users. In this case, the channels between the users and the relay station are describable by erasure channels $\mathcal{E}_i$. We obtain the MDI-QKD capacity to be
\begin{equation}
     \wt{P}_{\LOCC}(\mathcal{N}^{\text{MDI},\mathcal{E}}_{\vv{A}\to \vv{Z}})= q \eta_1\eta_2, 
    \end{equation}
where $\eta_i$'s are the parameters of the erasure channels connecting users to the relay station and $q$ is the probability of success of the Bell-measurement at the relay station (see~\ref{sec:ex-mdi} for precise model of the MDI-QKD setup). Dependence on $\eta_i$ allows us to consider the rate-distance trade-off. We also determine upper bounds on the maximum rates for the MDI-QKD setups where the quantum channels from the users to the relay station are depolarising and dephasing channels.

We also provide lower bounds on the secret-key-agreement rates of multiplex quantum channels that can be achieved by cppp. Our protocols are based on Devetak-Winter (DW) \cite{DW05} and generalize the lower bound for multipartite states presented in Ref.~\cite{AH09} as well as the bound for point-to-point quantum channels presented in Ref.~\cite{pirandola2009direct} to multiplex quantum channels. Our first lower bound is a direct extension of the result for states given in Ref.~\cite{AH09}. The idea is to choose a so-called distributing party that performs the (directed) DW protocol with all remaining parties. The achievable rate is then the worst-case DW rate achievable between the distributing party and any party. Further, we maximize over all choices for the distributing party. Our second protocol is a variation where we have a directed chain of parties in which each party performs the DW protocol with the next party in the chain. The obtainable rate is given by the 'weakest link', i.e., the lowest DW rate, in the chain, and we maximize over all possible permutations of the parties in the chain. 

In the case of a bidirectional network, i.e., a network in which all nodes are connected with their neighbors by a product of point-to-point channels in opposite directions, we provide a tighter bound based on spanning trees. The idea is to find the lowest DW rate in a spanning tree among any pair of the parties and maximise this quantity among all spanning trees. We provide an example where this protocol achieves a higher rate than the previous ones and show that the lower bound can be computed with polynomial complexity.

Finally, we show that the techniques developed in previous sections can also be applied to upper bound the rates at which the conference key can be distilled from multipartite quantum states. In particular, we provide an upper bound on the one-shot distillable conference key in terms of the hypothesis testing relative entropy with respect to biseparable states. Our bound reads
\begin{equation}
      K_{\textnormal{D}}^{(1,\varepsilon)}(\rho)\leq E^\varepsilon_{h,\GE}(\rho).
\end{equation}
Using a particular construction of biseparable states, we provide bounds on this quantity for a number of examples, such as (multiple copies of) GHZ and W states, as well as dephased or depolarised GHZ and W states. We also provide an upper bound on the asymptotic distillable conference key, which is given by the regularized relative entropy with respect to biseparable states.
\begin{equation}
    K_D(\rho)\leq E^\infty_{GE}(\rho),
\end{equation}
which is a generalization of the bipartite bound given in Ref.~\cite{HHHO09}.

\subsection{Relation to prior works}\label{sec:survey}
We briefly sketch some of the major developments that provide upper bounds on the key distillation capacities from states or via LOCC-assisted secret key agreement protocol over a quantum channel. We then compare our bounds on the SKA capacities with those mentioned in prior works. 
 
Conditions and bounds on the distillable key of bipartite states were provided in Refs.~\cite{HHHO05, HHHO09} and \cite{christandl2004squashed}. The former is in terms of the relative entropy of entanglement~\cite{VPRK97,VP98}, the latter in terms of the squashed entanglement~\cite{CW04} (cf.~\cite{Tuc99,Tuc02}). These results were generalized to the conference key in Refs.~\cite{AH09} and \cite{YHH+09}, respectively.

For an LOCC-assisted secret key agreement protocol over a point-to-point channel,  Ref.~\cite{TGW14} provides a weak converse bound in terms of the squashed entanglement, which is generalized to the distribution of bi- and multipartite private states via broadcast channels in Ref.~\cite{seshadreesan2016bounds}. In the case of tele-covariant channels (see Section~\ref{sec:tele-cov}), Ref.~\cite{PLOB15} provides a weak and Ref.~\cite{WTB16} a strong converse bound in terms of the relative entropy of entanglement. This bound has been generalised to the distribution of multiple pairs of bipartite private states states via broadcast channels \cite{LP17,TSW17}, as well as multiple-access and interference channels \cite{LP17}.

For arbitrary point-to-point channels, a strong converse bound in terms of the max-relative entropy of entanglement~\cite{Dat09} is provided in Ref.~\cite{CM17}. Recently, another strong converse bound in terms of the regularized relative entropy was provided in Ref.~\cite{FF19}. For bidirectional channels, strong converse bounds in terms of the max-relative entropy of entanglement, that reduce to the relative entropy of entanglement for tele-covariant channels, have been provided in Refs.~\cite{DBW17,BDW18,D18thesis}.   

In the case where the bipartite key is distributed between two parties using a quantum key repeater, bounds have been provided in Ref.~\cite{CM17} when quantum communication takes place over a point-to-point channel. 
Bounds on rates, at which bipartite and multipartite keys for networks of point-to-point or broadcast channels can be obtained, have been provided in Refs. \cite{pirandola2019capacities,AML16,azuma2017aggregating,RKB+17,bauml2018linear} and \cite{bauml2017fundamental}, respectively. Also, bounds on the rates obtainable in key repeaters that are in terms of entanglement measures of the input states have been obtained in Refs.~\cite{bauml2015limitations,christandl2017private}.

In an LOCC-assisted conference key agreement protocol, the use of a multiplex quantum channel is interleaved with local operations and classical communications (LOCC) among trusted parties. For this scenario, we derive strong converse bounds in terms of the max-relative entropy entanglement for arbitrary multiplex channels. In the case of finite channel dimensions, we also derive bounds in terms of the regularized relative entropy of entanglement. In the case of tele-covariant channels, we obtain bounds in terms of the relative entropy of entanglement. In general, our bounds are not comparable with the squashed entanglement bounds provided in Refs.~\cite{TGW14,seshadreesan2016bounds}. We are able to retain the results of Refs.~\cite{PLOB15,WTB16,CM17,FF19} when multiplex channels are assumed to be point-to-point channels. Our bounds in terms of the max-relative entropy are a direct generalisation of the bounds on bidirectional channels presented in in Refs.~\cite{DBW17,BDW18,D18thesis}, thus we retain those results. By using the recent results Ref.~\cite{FF19}, we further provide bounds in terms of the regularised relative entropy of entanglement, which can provide an improvement.

Concerning quantum key repeaters as well as setups of MDI-QKD, upper bounds on the achievable key rates can be obtained from results bounding key rates achievable in quantum networks, e.g. the one presented in Ref.~\cite{pirandola2019capacities} and subsequently used in Ref.~\cite{van2020extending} or the ones presented in Refs.~\cite{AML16,CM17}. However, we would like to note that by designing the right kind of multiplex channel, we can make more specific assumptions on the operations performed at the relay stations and thus obtain tighter bounds. For example we could design a multiplex channel for a protocol that does not use a quantum memory at the relay station or that performs a particular imperfect measurement at the relay station. The bounds given in Refs.~\cite{pirandola2019capacities,AML16,CM17}, on the other hand, would bound the key rates of a repeater or MDI-QKD setup by finding the weakest link between the nodes, i.e. only take into consideration limitations arising from imperfect point-to-point channels linking Alice and Bob with the central relay station, while assuming unlimited quantum memory at the nodes as well as the possibility to perform perfect measurements, resulting in looser bounds. Hence the bounds given in Refs.~\cite{pirandola2019capacities,AML16,CM17} basically reduce to the minimum of the capacities of the two point-to-point channels, whereas our bounds represent the limitation arising from both imperfect channels and imperfect node operations, which is an important factor when benchmarking experimental implementations.

As for conference key distillation from multipartite states, we provide tighter bounds than those presented in Ref.~\cite{AH09}. As a GHZ state is a special case of a multipartite private state, our bounds can also be applied to the distillation of GHZ states from any pure or mixed multipartite entangled state, both in the asymptotic and finite copies regime. There are a number of results concerned with computing and bounding rates of multipartite entanglement transformation, including \cite{bennett2000exact,SVW,FL-2007,KT-2010,CCL-2011,vrana2015asymptotic,spee2017entangled,vrana2019distillation,streltsov2020rates}. As an example, we consider the non-asymptotic distillation of tripartite conference key from noisy and noiseless W-states and compare our results with Ref.~\cite{SVW}.

\section{Preliminaries}~\label{sec:pre}
In this section, we introduce notations, review basic concepts, and standard definitions to be used frequently in later sections.  
\subsection{Notations and definitions}
We consider quantum systems associated with separable Hilbert spaces. We study both discrete and continuous variable quantum systems, therefore the associated Hilbert spaces can be finite or infinite dimensional. For a composite quantum system $AB$ in a state $\rho_{AB}$, the reduced state $\Tr_{B}[\rho_{AB}]$ of system $A$ is denoted as $\rho_A$. We denote identity operator as $\bbm{1}$. Let $\vv{A'}\coloneqq \{A_a'\}_{a\in\msc{A}}$, $\vv{A}\coloneqq \{A_a\}_{a\in\msc{A}}$, $\vv{B}=\{B_b\}_{b\in\msc{B}}$, $\vv{C}=\{C_c\}_{c\in\msc{C}}$, $\vv{K}=\{K_i\}_{i=1}^M$, denote sets (compositions) of quantum systems, where $\msc{A},\msc{B},\msc{C}$ are finite sets of symbols such that $|\mc{A}|+|\mc{B}|+|\mc{C}|=M$ for some natural number $M \geq 2$. We consider $M$ trusted allies $\{\tf{X}_i\}_{i=1}^M\coloneqq
\{\tf{A}_a\}_{a\in\mathscr{A}}\cup\{\tf{B}_b\}_{b\in\mathscr{B}}\cup\{\tf{C}_c\}_{c\in\mathscr{C
}}$”. Also, $\vv{LA}$ denotes the set $\{L_aA_a\}_{a\in\msc{A}}$, where $L_a$ is a reference system of $A_a$ and held by $\tf{A}_a$, and same follows for $\vv{RB}$, $\vv{PC}$, and $\vv{SK}$. A quantum state $\rho_{\vv{A}}$ denotes a joint state of a system formed by composition of all $A_a$. We use $:\!\vv{A}\!:$ to denote partition with respect to each system in the set $\vv{A}$ as they are held by separate entities, and same follows for $:\!\vv{LA}\!:\!\vv{RB}:$. Each separate elements in a set are held by separate party, in general. For example, let us consider $\vv{A}=\{A_1,A_2,A_3\}$ for $|\msc{A}|=3$, then $\vv{A}$ also depicts composite system $A_1A_2A_3$ and $:\vv{A}:$ denotes the partition $A_1:A_2:A_3$ between each subsystem $A_a$ of $\vv{A}$. In a conference key agreement protocol, each pair $K_i,S_i$ of key and shield systems belongs to respective trusted party $\tf{X}_i$ and fully secure from Eve, while all $A'_a,A_a,B_b,C_c,K_i,S_i$ are physically inaccessible to Eve.

Let $\Phi^{\GHZ}_{\vv{K}}$ denote $M$-partite GHZ state and $\Phi^+_{\vv{L}|\vv{A}}$ denote an Einstein–Podolsky–Rosen (EPR) state~\cite{EPR35}, also called a maximally entangled state, where maximal entanglement is between $\vv{L}$ and $\vv{A}$. It should be noted that $\Phi^+_{\vv{L}|\vv{A}}=\bigotimes_{a\in\msc{A}}\Phi^+_{L_a|A_a}$, where
\begin{equation}
    \Phi^+_{L_a|A_a}=\frac{1}{d}\sum_{i,j=0}^{d-1}\ket{i,i}\bra{j,j}_{L_aA_a}
\end{equation}
for an orthonormal basis $\{\ket{i}\}_{i}$, where $d=\min\{|L_a|,|A_a|\}$. (Without loss of generality, one may assume n EPR state of an even-dimensional qudit system to be a tensor product of EPR states of qubit systems). 

A quantum channel $\mc{M}_{B\to C}$ is a completely positive,  trace-preserving map that acts on trace-class operators defined on the Hilbert space $\mc{H}_B$ and uniquely maps them to trace-class operators defined on the Hilbert space $\mc{H}_C$. For a channel $\mc{M}_{A\to B}$ with $A$ and $B$ as input and output systems, its Choi state $J^\mc{M}_{LB}$ is equal to $\mc{M}(\Phi^+_{LA})$.

A measurement channel $\mc{M}_{A'\to AX}$ is a quantum instrument whose action is expressed as
\begin{equation}
    \mc{M}_{A'\to AX}(\cdot)=\sum_{x}\mathcal{E}^x_{A'\to A}(\cdot)\otimes\op{x}_X,
\end{equation}
where each $\mathcal{E}^x$ is a completely positive, trace non-increasing map such that $\mc{M}$ is a quantum channel and $X$ is a classical register that stores measurement outcomes. A classical register (system) $X$ can be represented with a set of orthogonal quantum states $\{\op{x}_X\}_{x\in\msc{X}}$ defined on the Hilbert space $\mc{H}_X$. 

An LOCC channel $\mc{L}_{\vv{A'}\to \vv{B}}$ can be written as $\sum_{x\in\msc{X}}(\bigotimes_{y\in\msc{Y}}\mc{E}^{y,x}_{A'_y\to B_y})$, where $\vv{A'}=\{A'_{y}\}_y$ and $\vv{B}=\{B_{y}\}_y$ are sets of inputs and outputs, respectively, and $\{\mc{E}^{y,x}\}_{x}$ is a set of completely positive trace non-increasing maps for each $y$ such that $\mc{L}$ is a quantum channel (cf.~\cite{CLM+14}). An LOCC channel does not increase value of entanglement monotones and is deemed as a free operation in resource theory of entanglement~\cite{HHHO09,AH09,CLM+14}.

A quantity is called a generalized divergence \cite{PV10,SW12} if it satisfies the following monotonicity (data-processing) inequality for all density operators $\rho$ and $\sigma$ and quantum channels $\mc{N}$:
\begin{equation}\label{eq:gen-div-mono}
\mathbf{D}(\rho\Vert \sigma)\geq \mathbf{D}(\mathcal{N}(\rho)\Vert \mc{N}(\sigma)).
\end{equation}
Examples include the quantum relative entropy~\cite{Ume62} 
\begin{equation}
D(\rho\Vert\sigma)\coloneqq \Tr[\rho\log_2(\rho-\sigma)],
\end{equation}
for $\supp(\rho)\subseteq \supp(\sigma)$, otherwise it is $\infty$, as well as the sandwiched R\'enyi relative entropy  \cite{MDSFT13,WWY14} which is denoted as $\wt{D}_\alpha(\rho\Vert\sigma)$  and defined for states
$\rho,\sigma$, and  $\forall \alpha\in (0,1)\cup(1,\infty)$ as
\begin{equation}\label{eq:def_sre}
\wt{D}_\alpha(\rho\Vert \sigma):= \frac{1}{\alpha-1}\log_2 \Tr\left[\left(\sigma^{ (1-\alpha) / 2\alpha }\rho\sigma^{ (1-\alpha) / 2\alpha }\right)^\alpha \right] ,
\end{equation}
but it is set to $+\infty$ for $\alpha\in(1,\infty)$ if $\supp(\rho)\nsubseteq \supp(\sigma)$. In the limit $\alpha\to 1$ the sandwiched R\'enyi relative entropy converges to the quantum relative entropy, in the limit  $\alpha\to \infty$, it converges to the max-relative entropy \cite{MDSFT13}, which is defined as \cite{D09,Dat09}
\begin{equation}\label{eq:max-rel}
D_{\max}(\rho\Vert\sigma)\coloneqq \inf\{\lambda\in\mathbb{R}:\ \rho \leq 2^\lambda\sigma\},
\end{equation}
and if $\supp(\rho)\nsubseteq\supp(\sigma)$ then $D_{\max}(\rho\Vert\sigma)=\infty$. Another generalized divergence  is the $\varepsilon$-hypothesis-testing divergence \cite{BD10,WR12},  defined as
\begin{equation}
D^\varepsilon_h\!\(\rho\Vert\sigma\):=-\log_2\inf_{\Lambda: 0\leq\Lambda\leq \bbm{1}}\{\Tr[\Lambda\sigma]:\  \Tr[\Lambda\rho]\geq 1-\varepsilon\},
\end{equation}
for $\varepsilon\in[0,1]$ and density operators $\rho,\sigma$. For a more detailed description and other examples of the generalized divergences like the trace distance $\norm{\rho-\sigma}_{1}$ and negative of fidelity $-F(\rho,\sigma)$ and their properties, see Appendix \ref{App:Div}.

\subsection{Multiplex quantum channels}\label{sec:mc}
We now formally define a general form of network channel which encompasses all other known multiplex quantum channels possible in communication or information processing settings (see Fig.~\ref{fig:mpc}a and Appendix~\ref{app:mqc}). To the best of our knowledge, we have not encountered such a general form of network channel in the literature of quantum communication and computation before. 

\begin{definition}
Consider multipartite quantum channel $\mc{N}_{\vv{A'}\vv{B}\to\vv{A}\vv{C}}$ where each pair $A_a', A_a$ is held by a respective party $\tf{A}_a$ and each $B_b,C_c$ are held by parties $\tf{B}_b,\tf{C}_c$, respectively. While $\tf{A}_a$ is both sender and receiver to the channel, $\tf{B}_b$ is only a sender, and $\tf{C}_c$ is only a receiver to the channel. Such a quantum channel is referred to as the multiplex quantum channel. Any two different systems need not be of the same size in general. 
\end{definition}

The sets $\msc{A}$, $\msc{B}$, or $\msc{C}$ can be empty in such a way that there is at least one input to the channel and one output from the channel. Definition 1 includes all scenarios depicted in Figure 1 (see Appendix B). For example, for a point-to-point channel from Bob to Charlie the set $\msc{A}=\emptyset$ and the sets $\msc{B}$ and $\msc{C}$ are singleton sets.

Also, any physical box with quantum or classical inputs and quantum or classical outputs is a type of a multiplex quantum channel. We may not have an exact description of what is going inside the box except that undergoing process is physical, i.e., described by quantum mechanics. Physical computational devices like a physical black box (oracle) and quantum circuit~\cite{NC00} are also examples of multiplex quantum channels.

\subsection{Conference key and private states}\label{sec:rev-ck-ps}
There are two usual approaches to studying secret key distillation. A direct approach starts by considering purifications of states where the purifying system is accessible to Eve, and all allied parties are allowed to perform local operations and public communication (LOPC). In this approach, we have Eve and $M$ allied parties. Another approach is by considering private states defined below, where all allied parties perform LOCC. We need not consider Eve explicitly in the paradigm of private states, and it is assumed that purifications of states are accessible to Eve. Both approaches are known to be equivalent \cite{HHHO05}. We discuss the equivalence of these two approaches in more detail in Section~\ref{sec:cka}
 
We now review the properties of conference key states discussed in Ref.~\cite{AH09}. Conference key states are a multipartite generalization of secret key shared between two parties.   

\begin{definition}\label{def:key-state}
A conference key state $\gamma^c_{\vv{K}E}$, with $|K_i|=K$ for all $i\in[M]\coloneqq\{1,…,M\}$, is defined as 
\begin{align}
& \mathcal{D}_{K_1}\otimes\mathcal{D}_{K_2}\otimes\cdots\otimes \mathcal{D}_{K_M}\left(\gamma^c_{\vv{K} E}\right) \coloneqq \nonumber \\ 
&\qquad \frac{1}{K}\sum_{k\in\msc{K}}\op{k}_{K_1}\otimes \op{k}_{K_2}\otimes\cdots \otimes \op{k}_{K_M}\otimes\sigma_E, \label{eq:key-state}
\end{align}
where $\sigma_E$ is a state of the system $E$, which is accessible to an eavesdropper Eve, $\mathcal{D}(\cdot)=\sum_{k\in\msc{K}}\op{k}(\cdot)\op{k}$ is a projective measurement channel, $\{\ket{k}_{K_i}\}_{k\in\msc{K}}$ forms an orthonormal basis for each $i\in[M]$. 
\end{definition}
A conference key state $\gamma^c_{\vv{K} E}$ has $\log_2K$ secret bits (key) that are readily accessible. 

A state $\rho_{\vv{K}E}$ is called an $\varepsilon$-approximate conference key state, for $\varepsilon\in[0,1]$, if there exists a conference key state $\gamma^c_{\vv{K} E}$ such that~\cite{AH09}
\begin{equation}
F\(\gamma^c_{\vv{K} E},\rho_{\vv{K}E}\)\geq 1-\varepsilon. 
\end{equation} 

\begin{definition}\label{def:multi-priv-state}
A state $\gamma_{\vv{SM}}$, with $|K_i|=K$ for all $i\in[M]$ is called a ($M$-partite) private state if and only if 
\begin{equation}\label{eq:multi-priv-state}
\gamma_{\vv{SK}}\coloneqq U^{\text{tw}}_{\vv{SK}} (\Phi^{\GHZ}_{\vv{K}}\otimes\omega_{\vv{S}}) (U^{\text{tw}}_{\vv{SK}})^\dag,
\end{equation}
where $U^{\text{tw}}_{\vv{SK}}\coloneqq \sum_{\vv{k}\in\msc{K}^{\times M}}\op{\vv{k}}_{\vv{K}}\otimes U^{\vv{k}}_{\vv{S}}$ is called a twisting unitary operator for some unitary operator $U^{\vv{k}}_{\vv{S}}$ and $\omega$ is some density operator~\cite{AH09}. 
\end{definition}
It should be noted that $\gamma_{\vv{SK}}$ has at least $\log_2K$ secret (key) bits (see \cite{HHHO09} for a discussion of when the private state has exactly $\log_2 K$ bits). Similar to a conference key state, a state $\rho_{\vv{SK}}$ is called an $\varepsilon$-approximate private state for $\varepsilon\in[0,1]$ if there exists a private state $\gamma_{\vv{SK}}$ such that~\cite{AH09} 
\begin{equation}
F\(\gamma_{\vv{SK}},\rho_{\vv{SK}}\)\geq 1-\varepsilon. 
\end{equation}

Any state extension (including purification) $\gamma_{\vv{SK}E}$ of such a private state \eqref{eq:multi-priv-state} necessarily has the following form~\cite{AH09}:
\begin{equation}
\gamma_{\vv{SK}E}\coloneqq U^{\text{tw}}_{\vv{SK}E} \(\Phi_{\vv{K}}\otimes\omega_{\vv{S}E}\) (U^{\text{tw}}_{\vv{SK}E})^\dag,
\end{equation}
where $\omega_{\vv{S}E}$ is a state extension of the density operator $\omega_{\vv{S}}$.

It follows from \cite[Theorem~IV.1]{AH09} that 
$F(\gamma^c_{\vv{K} E},\rho_{\vv{K}E})\geq 1-\varepsilon$
implies $F(\gamma_{\vv{SK}},\rho_{\vv{SK}})\geq 1-\varepsilon$, 
and the converse is also true, i.e., $F(\gamma_{\vv{SK}},\rho_{\vv{SK}})\geq 1-\varepsilon$ implies $F(\gamma^c_{\vv{K}E},\rho_{\vv{K}E})\geq 1-\varepsilon$.

It is known that all perfect private states have nonlocal correlations~\cite{ACPA10}.

\section{Entanglement and privacy test}\label{sec:MEPT}
this section introduces frameworks for the resource theories of multipartite entanglement for the multipartite quantum channels (see Refs.~\cite{D18thesis,DBW17,BDWW19,GS19} for the discussion on bipartite channels).
\subsection{Multipartite entanglement}\label{sec:ME}
Here we provide a short overview of the relevant definitions. For a detailed review of the topic, see Ref.~\cite{friis2019entanglement}. A pure $n$-partite state that can be written as a tensor product $\ket{\psi_1}\otimes\ket{\psi_2}\otimes...\otimes\ket{\psi_m}$ is called \textit{$m$-separable}. If $m<n$, there are partitions of the set of all the parties into two with respect to which the state is entangled. If $n=m$, the pure state is said to be fully separable. If there is no bipartition with respect to which the pure state is a product state, it is called genuinely $n$-partite entangled.

An arbitrary $n$-partite state is $m$-separable if it can be written as following convex composition:
\begin{equation}
\rho_{m-sep}=\sum_{x\in\msc{X}} p_X(x) \op{\psi^x_1}\otimes\op{\psi^x_2} \otimes...\otimes\op{\psi^x_m},
\end{equation}
where $p_X(x)$ is a probability distribution. 
The $m$-separable states form a convex set. Note, however, that the subsystems with respect to which the elements of the decomposition have to be product can differ. 

A  mixed $n$-partite state is considered \textit{genuinely multipartite entangled} (GME) if any decomposition into pure states contains at least one genuinely $n$-partite entangled pure state, i.e. the state is not biseparable. Let a free set $\tf{F}(:\vv{A}:)$ denote the set of all fully separable and biseparable states of system $\vv{A}$ for $\tf{F}=\FS$ and $\tf{F}=\BS$, respectively. Both the sets, $\FS$ and $\BS$, are convex. We note that while $\FS$ is preserved under LOCC operation and tensor product, $\BS$ is preserved under LOCC but not under tensor product, i.e., $\rho^{(x)}_{\vv{A^{(x)}}}\in\BS(:\vv{A^{(x)}}:)$ for $x\in[2]$ but $\rho^{(1)}\otimes\rho^{(2)}$ need not belong to $\BS(:\!\vv{A^{(1)}A^{(2)}}\!:)$. We refer to biseparable quantum states whose biseparability is preserved under tensor products, i.e. $\rho^{(x)}_{\vv{A^{(x)}}}\in\BS(:\vv{A^{(x)}}:)$ and $\rho^{(1)}\otimes...\otimes\rho^{(n)}\in\BS(:\!\vv{A^{(1)}\ldots A^{(n)}}\!:)$ for all $n\in\mathbb{N}$, as \textit{tensor-stable biseparable} states. 

\subsection{Entanglement measures}\label{sec:EM} 
It is pertinent to quantify the resourcefulness of states and channels. The bounds on the capacities that we obtain are in terms of these quantifiers. It is desirable for entanglement quantifiers to be non-negative, attain their minimum for the free states (and separable channels respectively), and be monotone under the action of LOCC.

\begin{definition}\label{def:gen-div-ent}
The generalized divergence of entanglement $\tf{E}_{\text{E}}$ or genuine mulipartite entanglement (GME) $\tf{E}_\text{GE}$ of an arbitrary state $\rho_{\vv{A}}$ is defined as~\cite{CPV19}
\begin{equation}\label{eq:def-gen-ent}
\tf{E}_{r}(:\!\vv{A}\!:)_{\rho}\coloneqq \inf_{\sigma\in\tf{F}(:\!\vv{A}\!:)}\tf{D}(\rho_{\vv{A}}\Vert\sigma_{\vv{A}}),
\end{equation} 
when $\tf{F}=\FS$ or $\tf{F}=\BS$ for $r=\text{E}$ or $r=\text{GE}$, respectively, where $\tf{D}(\rho\Vert\sigma)$ denotes the generalized divergence. 
\end{definition}

The following definition of entanglement measure of a multiplex channel generalizes the notion of entangling power of bipartite quantum channels~\cite{BHLS03} (see also \cite{KW17,DBW17,D18thesis}). 
\begin{definition}
Entangling power of a multiplex channel $\mc{N}_{\vv{A'}\vv{B}\to\vv{A}\vv{C}}$ with respect to entanglement measure $\tf{E}_{r}$ [Eq. ~\eqref{eq:def-gen-ent}] is defined as the maximum possible gain in the entanglement $\tf{E}_r$ when a quantum state when acted upon by the given channel $\mc{N}$, 
\begin{align}
&\tf{E}_{r}^{p}(\mc{N})\coloneqq\nonumber\\
&\qquad \sup_{\rho}\left[\tf{E}_{r}(:\!\vv{LA}\!:\!\vv{R}\!:\!\vv{PC}\!:)_{\mc{N}(\rho)} -\tf{E}_{r}(:\!\vv{LA'}\!:\!\vv{RB}\!:\!\vv{P}\!:)_{\rho}\right],
\end{align}
where optimization is over all possible input states $\rho_{\vv{LA'}\vv{RB}\vv{P}}$. 
\end{definition}

Another way to quantify entanglement measure of a multiplex channel is the following (see Ref.~\cite{D18thesis} for bidirectional channel).
\begin{definition}\label{def:gen-div-ent-c}
The generalized divergence of entanglement $\tf{E}_\text{E}(\mc{N})$ or GME $\tf{E}_{\text{GE}}(\mc{N})$ of a multiplex channel $\mc{N}_{\vv{A'}\vv{B}\to\vv{A}\vv{C}}$
\begin{equation}\label{eq:def-gen-ent-c}
\tf{E}_{r}(\mc{N})\coloneqq \sup_{\rho\in\FS(:\vv{LA'}:\vv{RB}:)}\tf{E}_{r}(:\!\vv{LA}\!:\!\vv{R}\!:\!\vv{C}\!:)_{\mc{N}(\rho)},
\end{equation} 
for $r=\text{E}$ or $r=\text{GE}$, respectively, where $\tf{E}_{r}(:\!\vv{A}\!:)_{\rho}$ is defined in \eqref{eq:def-gen-ent} and GME stands for genuinely multipartite entanglement.
\end{definition}

For $r=\text{E}$, the entanglement measure in \eqref{eq:def-gen-ent} is called $\varepsilon$-hypothesis testing relative entropy of entanglement $E^\varepsilon_{h,\E}$, max-relative entropy of entanglement $E_{\max,\E}$, sandwiched R\`{e}nyi relative entropy of entanglement $\wt{E}_{\alpha, \E}$, or relative entropy of entanglement $E_{\E}$ when the generalized divergence is the $\varepsilon$-hypothesis testing relative entropy, max-relative entropy, sandwiched R\'{enyi} relative entropy, or relative entropy, respectively. For $r=\text{GE}$, the entanglement measure in \eqref{eq:def-gen-ent} is called $\varepsilon$-hypothesis testing relative entropy of GME $E^\varepsilon_{h,\GE}$, max-relative entropy of GME $E_{\max,\GE}$, sandwiched R\`{e}nyi relative entropy of GME $\wt{E}_{\alpha, \GE}$, or relative entropy of GME when the generalized divergence $E_{\GE}$ is the $\varepsilon$-hypothesis testing relative entropy, max-relative entropy, sandwiched R\'{enyi} relative entropy, or relative entropy, respectively. We follow the same procedure for nomenclature of entanglement measures of channels. 

We note that the sets $\FS,\BS$ are convex. Using the data-processed triangle inequality~\cite{CM17} and the argument from the proof of \cite[Proposition~2]{DBW17}, we arrive at the following lemma. 
\begin{lemma}\label{lem:max-amor}
The entangling power of a multiplex channel $\mc{N}_{\vv{A'}\vv{B}\to\vv{A}\vv{C}}$ with respect to the max-relative entropy of entanglement $E_{\max,\E}$ is equal to the max-relative entropy of entanglement of the channel $\mc{N}$,
\begin{equation}
E_{\max,E}^{p}(\mc{N})=E_{\max,E}(\mc{N}).\label{eq:max-amor}
\end{equation}
\end{lemma}

Using a recent result on relative entropies~\cite{fang2019chain}, we can also obtain a result for the relative entropy of entanglement. Let us first define the regularized relative entropy of entanglement of a multiplex channel $\mc{N}_{\vv{A'}\vv{B}\to\vv{A}\vv{C}}$ as
\begin{equation}
E^\infty_{R}(\mc{N}):=\inf_{\Lambda\in\LOCC}D^\infty(\mc{N}_{\vv{A'}\vv{B}\to\vv{A}\vv{C}}||\Lambda_{\vv{A'}\vv{B}\to\vv{A}\vv{C}}),
\end{equation}
where $D^\infty(\mc{N}||\mc{M}):=\lim_{n\to\infty}\frac{1}{n}D(\mc{N}^{\otimes n}||\mc{M}^{\otimes n})$ and 
\begin{equation}
D(\mc{N}||\mc{M}):=\max_{\phi_{\vv{LA'}\vv{RB}\vv{P}}}D(\mc{N}_{\vv{A'}\vv{B}\to\vv{A}\vv{C}}(\phi)||\mc{M}_{\vv{A'}\vv{B}\to\vv{A}\vv{C}}(\phi)),
\end{equation}
where $L\simeq A'$, $R\simeq B$ and $P\simeq C$. We can now show the following relation between the regularized relative entropy of entanglement and the relative entropy of entanglement.
\begin{lemma}\label{lem:ree-amor}
For finite dimensional Hilbert spaces the entangling power of a multiplex channel $\mc{N}_{\vv{A'}\vv{B}\to\vv{A}\vv{C}}$ with respect to the relative entropy of entanglement $E_{E}$ is less than or equal to the regularized relative entropy of entanglement of the channel $\mc{N}$,
\begin{equation}
E_{E}^{p}(\mc{N})\leq E^\infty_{E}(\mc{N}).
\end{equation}
\end{lemma}
\begin{proof}
Let $\rho_{\vv{LA'}\vv{RB}\vv{P}}$ be a state and let $\sigma'\in\FS(:\vv{LA'}:\vv{RB}:\vv{P}:)$. Let $\Lambda_{\vv{A'}\vv{B}\to\vv{A}\vv{C}}$ an LOCC channel. Then the following inequality holds
\begin{multline}
E_{E}(:\vv{LA}:\vv{R}:\vv{PC}:)_{\mc{N}(\rho)}\\ 
 \leq D\(\mc{N}_{\vv{A'}\vv{B}\to\vv{A}\vv{C}}(\rho_{\vv{LA'}\vv{RB}\vv{P}})\Vert \Lambda_{\vv{A'}\vv{B}\to\vv{A}\vv{C}}(\sigma'_{\vv{LA'}\vv{RB}\vv{P}}) \).
 \end{multline}
Applying the the chain rule from Ref.~\cite{fang2019chain}, we find that
\begin{align*}
&D\(\mc{N}_{\vv{A'}\vv{B}\to\vv{A}\vv{C}}(\rho_{\vv{LA'}\vv{RB}\vv{P}})\Vert \Lambda_{\vv{A'}\vv{B}\to\vv{A}\vv{C}}(\sigma'_{\vv{LA'}\vv{RB}\vv{P}}) \)\\
\leq&D\(\rho_{\vv{LA'}\vv{RB}\vv{P}}\Vert \sigma'_{\vv{LA'}\vv{RB}\vv{P}} \)\\ 
& +D^{\infty}\(\mc{N}_{\vv{A'}\vv{B}\to\vv{A}\vv{C}}\Vert \Lambda_{\vv{A'}\vv{B}\to\vv{A}\vv{C}}\).
\end{align*}
Since the above holds for arbitrary fully separable states $\sigma'_{\vv{LA'}\vv{RB}\vv{P}}$ and arbitrary LOCC channels $\Lambda_{\vv{A'}\vv{B}\to\vv{A}\vv{C}}$, we arrive at
\begin{multline}
E_{E}(:\vv{LA}:\vv{R}:\vv{PC}:)_{\mc{N}(\rho)}\\
 \leq  E_{E}(:\vv{LA'}:\vv{RB}:\vv{P}:)_{\rho}+ E^\infty_{E}(\mc{N}),
\end{multline}
finishing the proof.
\end{proof}

\begin{remark}
It suffices to optimize $E^\varepsilon_{h,\E}(\mc{N})$, $E^\varepsilon_{h,\GE}(\mc{N})$, $E_{\max,\E}(\mc{N})$, $E_{E,\E}(\mc{N})$ and $E_{\max,\GE}(\mc{N})$ of a multiplex channel $\mc{N}$ over all pure input states, i.e., $\rho\in\FS(:\vv{LA'}:\vv{RB}:)$ is a pure state in \eqref{eq:def-gen-ent-c} for $E^\varepsilon_{h,\E}(\mc{N}),E^\varepsilon_{h,\GE}(\mc{N}), E_{\max,\E}(\mc{N}), E_{E,\E}(\mc{N}), E_{\max,\GE}(\mc{N})$. This reduction follows from the quasi-convexity of the max-relative entropy~\cite{D09} and $\varepsilon$-hypothesis testing relative entropy~\cite{TT13}, as well as the convexity of the relative entropy of entanglement~\cite{VP98}. Namely, the maximum of a (quasi)-convex function over a convex set will be attained on a boundary point. The boundary points of the set of fully separable density matrices are given by the fully separable pure states. 
\end{remark}

\subsection{Multipartite privacy test}
A $\gamma$-privacy test corresponding to $\gamma_{\vv{SK}}$ is defined as the dichotomic measurement~\cite{WTB16}
$\{\Pi^{\gamma}_{\vv{SK}},\bbm{1}-\Pi^{\gamma}_{\vv{SK}}\}$, where $\Pi^{\gamma}_{\vv{SK}}\coloneqq U^{\text{tw}}_{\vv{SK}} (\Phi_{\vv{K}}\otimes\bbm{1}_{\vv{S}}) (U^{\text{tw}}_{\vv{SK}})^\dag$. 

Using properties of fidelity and form of test measurement, we arrive at following proposition. 
\begin{proposition}\label{thm:priv-test}
If a state $\rho_{\vv{SK}}$ is $\varepsilon$-approximate to $\gamma_{\vv{SK}}$, i.e., $F(\rho_{\vv{SK}},\gamma_{\vv{SK}})\geq 1-\varepsilon$ then $\rho_{\vv{KS}}$ passes $\gamma$-privacy test with success probability $ 1-\varepsilon$, i.e., 
\begin{equation}
\Tr[\Pi^{\gamma}_{\vv{SK}}\rho_{\vv{SK}}] \geq 1-\varepsilon.
\end{equation} 
\end{proposition}
\begin{proof}
\begin{align}
& \Tr[\Pi^{\gamma}_{\vv{SK}}\rho_{\vv{SK}}]\nonumber \\
& \quad = \bra{\Phi^{\GHZ}}_{\vv{K}}\Tr_{\vv{S}}[(U^{\text{tw}}_{\vv{SK}})^\dag \rho_{\vv{SK}}U^{\text{tw}}_{\vv{SK}}]\ket{\Phi^{\GHZ}}_{\vv{K}}\\
&\quad = F\(\Phi^{\GHZ}_{\vv{K}},\Tr_{\vv{S}}[(U^{\text{tw}}_{\vv{SK}})^\dag \rho_{\vv{SK}}U^{\text{tw}}_{K_1\ldots K_MS_1\ldots S_M}]\)\\
&\quad \geq F\(\Phi^{\GHZ}_{\vv{K}}\otimes\omega_{\vv{S}},(U^{\text{tw}}_{\vv{SK}})^\dag \rho_{\vv{SK}}U^{\text{tw}}_{\vv{SK}}\)\\
&\quad = F\(U^{\text{tw}}_{\vv{SK}} \Phi^{\GHZ}_{\vv{K}}\otimes\omega_{\vv{S}}(U^{\text{tw}}_{\vv{SK}})^\dag, \rho_{\vv{SK}}\)\\
& \quad =F(\gamma_{\vv{SK}},\rho_{\vv{SK}})\geq 1-\varepsilon .
\end{align}
\end{proof}

We employ proof arguments similar to bipartite case of Eq. (298) in Ref.~\cite{HHHO09} to arrive at the following theorem, which implies that all private states are necessarily GME states. This is a strict generalization of Eq. (281)~in Ref. \cite{HHHO09}, as a direct generalization would be the same statement for fully separable states instead of biseparable states (cf. Ref. ~\cite{AH09}). See Appendix~\ref{app:priv-t} for the proof.
\begin{theorem}\label{thm:priv-gme}
A biseparable state $\sigma_{\vv{SK}}\in\BS(:\!\vv{SK}\!:)$ can never pass any $\gamma$-privacy test with probability greater $1/K$, i.e.,
\begin{align}
    \Tr[\Pi^{\gamma}_{\vv{SK}}\sigma_{\vv{SK}}]\leq \frac{1}{K}.
\end{align} 
\end{theorem}

\section{Conference key agreement protocol}\label{sec:cka}
In this section, we give a formal description of a secret key agreement protocol for multiple trusted parties, i.e., a conference key agreement protocol. 

We consider an LOCC-assisted secret key agreement protocol among $M$ trusted allies $\{\tf{X}_i\}_{i=1}^M$ over a multiplex quantum channel $\mc{N}_{\vv{A'}\vv{B}\to\vv{A}\vv{C}}$ where each pair $A_a', A_a$ is held by trusted party $\tf{A}_a$ and each $B_b,C_c$ are held by trusted parties $\tf{B}_b,\tf{C}_c$, respectively. Environment part $E$ of an isometric extension $U^\mc{N}_{\vv{A'}\vv{B}\to\vv{A}\vv{C}E}$ of the channel $\mc{N}$ is accessible to Eve along with all classical information communicated among $\tf{X}_i$ while performing LOCC. All other quantum systems that are locally available to $\tf{X}_i$ are said to be secure from Eve, i.e., even if local operations during LOCC are noisy, purifying quantum systems are still within labs of trusted allies, which are off-limits for Eve. This assumption is justifiable because $\tf{X}_i$'s can always abandon performing local operations that would leak information to Eve. 
In an LOCC-assisted protocol, the uses of the multiplex channel $\mc{N}$ are interleaved with LOCC channels. 

In the first round, all $\tf{X}_i$ perform LOCC $\mc{L}^1$ to generate a state $\rho_1\in\FS(:\!\vv{L^{(1)}A^{(1)'}}:\!\vv{R^{(1)}B^{(1)}}\!:\!\vv{P^{(1)}})$. All $\tf{A}_a$ and $\tf{B}_b$ input respective systems to multiplex channel $\mc{N}^1_{\vv{A^{(1)'}}\vv{B^{(1)}}\to \vv{C}^{(1)}}$, and let $\tau_1\coloneqq \mc{N}^1(\rho)$ be the output state after the first use $\mc{N}^{1}$ of multiplex channel. In the second round, an LOCC $\mc{L}^{2}$ is performed on $\tau_1$ and then second use $\mc{N}^{2}$ of multiplex channel is employed on $\rho_2\coloneqq \mc{L}^{2}(\tau_1)$. In the third round, an LOCC $\mc{L}^{3}$ is performed on $\tau_2\coloneqq \mc{N}^2(\rho_2)$ and then the third use $\mc{N}^{3}$ of multiplex channel is employed on $\rho_3\coloneqq \mc{L}^{3}(\tau_2)$. Successively, we continue this procedure for $n$ rounds, where an $\mc{L}$ acts on the output state of previous round, after which multiplex channel is performed on the resultant state. Finally, after $n^{\text{th}}$ round, an LOCC $\mc{L}^{n+1}$ is performed as a decoding channel, which generates the final state $\omega_{\vv{SK}}$.

It can be concluded from the equivalence between private states, and CK states that any protocol of the above form can be purified, i.e., by considering isometric extensions of all channels (LOCC and $\mc{N}$) (the proof arguments are the same as for the purified protocol for LOCC-assisted secret-key-agreement~\cite{DBW17}). At the end of the purified protocol, Eve posses all the environment systems $E^n$ from isometric extension $U^\mc{N}$ of each use of multiplex channel $\mc{N}$ along with coherent copies $Y^{n+1}$ of the classical data exchanged among trusted parties $\tf{X}_i$ during performances of $n+1$ LOCC channels. Whereas, each trusted party $\tf{X}_i$ posses the key system $K_i$ and the shield system $S_i$, which consists of all local reference systems, after the action of decoder. The state at the end of the protocol is a pure state $\omega_{\vv{SK}Y^{n+1}E^n}$ with $F(\gamma_{\vv{SK}},\omega_{\vv{SK}})\geq 1-\varepsilon$. Such a protocol is called an $(n,K,\varepsilon)$ LOCC-assisted secret key agreement protocol. The rate $P$ of a given $(n,K,\varepsilon)$ protocol is equal to the number of conference (secret) bits generated per channel use:
\begin{equation}
P\coloneqq \frac{1}{n}\log_2K. 
\end{equation}

A rate $P$ is achievable if for $\varepsilon\in(0,1),\delta>0$, and sufficiently large $n$, there exists an $(n,2^{n(P-\delta)},\varepsilon)$ LOCC-assisted secret key agreement protocol. The LOCC-assisted secret-key-agreement capacity $\hat{P}_{\LOCC}(\mc{N})$ of a multiplex quantum channel $\mc{N}$ is defined as the supremum of all achievable rates.

A rate $P$ is called a strong converse rate for LOCC-assisted secret key agreement if for all $\varepsilon\in[0,1), \delta>0$, and sufficiently large $n$, there does not exist an $(n, 2^{n(P+\delta)},\varepsilon)$ LOCC-assisted secret key agreement protocol. The strong converse LOCC-assisted secret-key-agreement capacity $\widetilde{P}_{\LOCC}(\mc{N})$ is defined as the infimum of all strong converse rates. 

The following inequality is a direct consequence of the definitions: 
\begin{equation}
\hat{P}_{\LOCC}(\mc{N})\leq \widetilde{P}_{\LOCC}(\mc{N}).
\end{equation}

We can also consider the whole development discussed above for conference key agreement assisted only with classical preprocessing and postprocessing (cppp) communication, i.e., all parties are allowed only two LOCC channels, one for encoding and the other for decoding. A $(n,K,\varepsilon)$ cppp-assisted secret key agreement protocol over $\mc{N}$ is same as a $(1,K,\varepsilon)$ LOCC-assisted secret key agreement protocol over channel $\mc{N}^{\otimes n}$, and for $n=1$ both protocol are the same. The cppp-assisted secret-key-agreement capacity $\hat{P}_{\text{cppp}}$ of the channel $\mc{N}$ is always less than or equal to $\hat{P}_{\LOCC}$,
\begin{equation}
\hat{P}_{\text{cppp}}(\mc{N})\leq \hat{P}_{\LOCC}(\mc{N}). 
\end{equation}
Let $
\hat{P}^{\mc{N}}_{\textnormal{cppp}}(n,\varepsilon)$ be the maximum rate such that $(n,2^{nP},\varepsilon)$ cppp-assisted secret key agreement is achievable for any given $\mc{N}$. 

\begin{remark} 
It should be noted that the maximum rate at which secret key can be distilled using LOCC- or cppp-assisted protocol over a multiplex channel $\mc{N}$ is never less than the maximum rate at which GHZ state can be distilled using LOCC- or cppp-assisted protocol over given channel $\mc{N}$, respectively. This statement holds because GHZ state is a special private state from which secret bits are readily accessible to trusted allies. 
\end{remark}

\begin{remark}
Different physical constraints can be invoked in communication protocols to define constrained protocols and associated capacities. For instance, we can invoke energy constraints on input states and detectors to get energy-constrained protocols and respective capacities (cf.~\cite{GES16,WQ18}).  
\end{remark}

\subsection{Privacy from a single-use of a multiplex channel}
Let $\hat{P}^{\mc{N}}_{\textnormal{cppp}}(n,\varepsilon)$ denote the maximum rate $P$ such that $(n,K,\varepsilon)$ conference key agreement protocol is achievable for any $\mc{N}$ using cppp. The following bound holds for the one-shot secret-key-agreement rate of a multiplex quantum channel $\mc{N}$ (see Appendix~\ref{app:one-shot-cppp} for the proof).  

\begin{theorem}\label{thm:one-shot-cppp}
For any fixed $\varepsilon\in(0,1)$, the achievable region of cppp-assisted secret key agreement over a single use of multiplex channel $\mc{N}_{\vv{A'}\vv{B}\to\vv{A}\vv{C}}$ satisfies 
\begin{equation}\label{eq:one-shot-cppp}
\hat{P}_{\textnormal{cppp}}^{\mc{N}}(1,\varepsilon)\leq E^\varepsilon_{h,\GE}(\mc{N}),
\end{equation} 
where 
\begin{equation}
E^{\varepsilon}_{h,\GE}(\mc{N})  \coloneqq \sup_{\psi\in\FS(:\vv{LA'}:\vv{RB}:)}\inf_{\sigma\in\BS(:\vv{LA}:\vv{R}:\vv{C}:)} D^\varepsilon_h(\mc{N}(\psi)\Vert\sigma)
\end{equation}
is the $\varepsilon$-hypothesis testing relative entropy of genuine entanglement of the multiplex channel $\mc{N}$. It suffices to optimize over pure input states $\psi\in\FS(:\!\vv{LA'}\!:\!\vv{RB}\!:)$.  
\end{theorem}
We can conclude from the above theorem that 
\begin{equation}\label{eq:n-cppp}
\hat{P}^{\mc{N}}_{\textnormal{cppp}}(n,\varepsilon) \leq \frac{1}{n}E^\varepsilon_{h,\GE}(\mc{N}^{\otimes n}),
\end{equation}
which leads to the following corollaries.

\begin{corollary}
A weak converse bound on the cppp-assisted secret-key-agreement capacity of a multiplex channel $\mc{N}$ is given by
\begin{align}
    \hat{P}_{\textnormal{cppp}}(\mc{N})& = \inf_{\varepsilon\in (0,1)}\liminf_{n\to\infty}\hat{P}^{\mc{N}}_{\textnormal{cppp}}(n,\varepsilon)\\ 
    &\leq E^{\infty}_{\GE}(\mc{N}).
\end{align}
\end{corollary}

\begin{corollary}
    Consider a class of multiplex channels $\mc{N}_{\vv{A'}\vv{B}\to\vv{A}\vv{C}}$ such that for all pure input states $\psi\in\FS(\vv{LA'}:\vv{RB}:\vv{P})$, the output states $\mc{N}(\psi)$ are tensor-stable biseparable states with respect to the partition $\vv{LA}:\vv{RB}:\vv{PC}$. The cppp-assisted secret-key-agreement capacities for such class of multiplex channels are zero. 
\end{corollary}

\subsection{Strong converse bounds on LOCC-assisted private capacity of multiplex channel} 
We now derive converse and strong converse bounds on LOCC-assisted secret key agreement protocol over a multiplex channel $\mc{N}$.

Whereas, for LOCC-assisted secret key agreement protocol, by employing Theorem~\ref{thm:priv-gme}, generalizing proof arguments of Ref.~\cite[Theorem~2]{DBW17} (see also \cite{CM17}) to multiplex scenario, we get the following converse bound (proof in Appendix~\ref{app:emax-key-converse}). 
\begin{theorem}\label{thm:emax-key-converse} For a fixed $n,\ K\in\mathbb{N},\ \varepsilon\in(0,1)$, the following bound holds for an $(n,K,\varepsilon)$ protocol for LOCC-assisted secret key agreement over a multiplex $\mc{N}_{\vv{A'}\vv{B}\to\vv{A}\vv{C}}$: 
\begin{equation}\label{eq:emax-key-converse} 
\frac{1}{n}\log_2K\leq E_{\max,E}(\mc{N})+ \frac{1}{n}\log_2\left(\frac{1}{1-\varepsilon}\right), 
\end{equation} 
where the max-relative entropy of entanglement $E_{\max,\E}(\mc{N})$ of the multiplex channel $\mc{N}$ is 
\begin{equation*} 
E_{\max,\E}(\mc{N}) \coloneqq \sup_{\psi\in\FS(:\vv{LA'}:\vv{RB}:)}\inf_{\sigma\in\FS(:\vv{LA}:\vv{R}:\vv{C}:)} D_{\max}(\mc{N}(\psi)\Vert\sigma)
\end{equation*} 
and it suffices to optimize over pure states $\psi$.
\end{theorem}

\begin{remark}
The bound in \eqref{eq:emax-key-converse} can also be rewritten as 
\begin{equation}
1-\varepsilon \leq 2^{-n (P-E_{\max,\E}(\mc{N}))},
\end{equation}
where we have $P=\frac{1}{n}\log_2 K$. Thus, if the secret-key-agreement rate $P$ is strictly greater than the max-relative entropy of entanglement $E_{\max,\E}(\mc{N})$ of the (multiplex) channel $\mc{N}$, then the fidelity of the distillation $(1-\varepsilon)$ decays exponentially fast to zero in the number of channel uses. 
\end{remark} 

An immediate corollary of the above remark is the following strong converse statement. 
\begin{corollary}
The strong converse LOCC-assisted secret-key-agreement capacity of a multiplex channel $\mc{N}$ is bounded from above by its max-relative entropy of entanglement:
\begin{equation}
\wt{P}_{\LOCC}(\mc{N})\leq E_{\max,E}(\mc{N}).
\end{equation}
\end{corollary}


We also have another upper bound on the private capacity of a multiplex channel $\mc{N}_{\vv{A'}\vv{B}\to\vv{A}\vv{C}}$ with finite-dimensional input and output systems in terms of the regularized relative entropy instead of the max-relative entropy (proof in Appendix~\ref{app:ree-key-converse}). 
\begin{theorem}\label{thm:ree-key-converse}  
For finite Hilbert space dimensions the asymptotic LOCC-assisted secret-key-agreement capacity of a  multiplex channel $\mc{N}_{\vv{A'}\vv{B}\to\vv{A}\vv{C}}$ is bounded by its regularized relative entropy of entanglement:
\begin{equation}\label{eq:ree-key-converse} 
\widetilde{\mc{P}}_{\LOCC}(\mc{N})\leq E_{E}^\infty(\mc{N}).
\end{equation} 
\end{theorem}
\subsection{Teleportation-simulable and Tele-covariant multiplex channels}\label{sec:tele-cov}
For a class of multipartite quantum channels obeying certain symmetries, such as teleportation-simulability~\cite{BBPS96}, the LOCC-assistance does not enhance secret-key-agreement capacity, and the original protocol can be reduced to a cppp-assisted secret key agreement protocol. This observation for secret communication between two parties over point-to-point teleportation-simulable channel was first made in Ref.~\cite{PLOB15}.

\begin{definition}
A multipartite quantum channel $\mc{N}_{\vv{A'}\vv{B}\to\vv{A}\vv{C}}$ is teleportation-simulable with associated resource state $\theta_{\vv{LA}\vv{R}\vv{C}}$, where $R_b\simeq B_b$ for all $b\in\msc{B}$ and $L_a\simeq A'_a$ for all $a\in\msc{A}$, if for all input states $\rho_{\vv{A'}\vv{B}}$ the following identity holds
\begin{equation}\label{eq:tele-sim}
\mc{N}_{\vv{A'}\vv{B}\to\vv{A}\vv{C}}(\rho_{\vv{A'}\vv{B}})=\mc{T}_{\vv{A'LA}\vv{BR}\vv{C}\to\vv{A}\vv{C}}(\rho_{\vv{A'}\vv{B}}\otimes\theta_{\vv{LA}\vv{R}\vv{C}})
\end{equation} 
for some LOCC channel $\mc{T}$ with input partition $:\vv{A'LA}:\vv{BR}:\vv{C}:$ and output partition $:\vv{A}:\vv{C}:$.
\end{definition}

\textit{Covariant channels}.--- For each $a\in\msc{A}$ and $b\in\msc{B}$, let $\msc{G}_a$ and $\msc{G}_b$ be finite groups of respective sizes $G_a$ and $G_b$ with respective unitary representations $g_a\rightarrow U_{A'_a}(g_a)$ and $g_b\rightarrow U_{B_b}(g_b)$ for all group elements $g_{a}$ and $g_{b}$. Let $W^{\vv{g}}_{A_a}$ and $W^{\vv{g}}_{C_c}$ be unitary representations for all $a\in\msc{A}$ and $c\in\msc{C}$, where $\vv{g}=\{g_a,g_b\}_{a,b}$. A multiplex quantum channel $\mc{N}_{\vv{A'}\vv{B}\to\vv{A}\vv{C}}$ is \textit{covariant} with respect to these representations if the following relation holds for all input states $\rho_{\vv{A'}\vv{B}}$ and group elements $g_a\in\msc{G}_a$ and $g_b\in\msc{G}_b$ for all $a\in\msc{A}$ and $b\in\msc{B}$: 
\begin{align}
&\mc{N}_{\vv{A'}\vv{B}\to\vv{A}\vv{C}}\(\(\bigotimes_{a\in\msc{A}}\mathcal{U}_{A'_a}^{g_a}\otimes\bigotimes_{b\in\msc{B}}\mathcal{U}_{B_b}^{g_b}\)(\rho_{\vv{A'}\vv{B}})\)\nonumber\\
&=\(\bigotimes_{a\in\msc{A}}\mathcal{W}_{A_a}^{\vv{g}}\otimes\bigotimes_{c\in\msc{C}}\mathcal{W}_{C_c}^{\vv{g}}\)\(\mc{N}_{\vv{A'}\vv{B}\to\vv{A}\vv{C}}(\rho_{\vv{A'}\vv{B}})\),
\end{align}
where we have used the notation $\mc{U}(\cdot)\coloneqq U(\cdot)U^{\dag}$ for unitaries $U$.

\begin{definition}\label{def:tele-cov}
A quantum channel $\mc{N}_{\vv{A'}\vv{B}\to\vv{A}\vv{C}}$ is called tele-covariant if it is covariant with respect to groups $\{\msc{G}_a\}_{a\in\msc{A}}$ and $\{\msc{G}_b\}_{b\in\msc{B}}$ that have representations as unitary one-designs, i.e., for all $a\in\msc{A}$ and $b\in\msc{B}$ as well as states $\rho_{A'_a}$ and $\rho_{B_b}$, it holds $\frac{1}{G_a}\sum_{g_a\in\msc{G}_a}\mathcal{U}_{A'_a}^{g_a}(\rho_{A'_a})=\bbm{1}/|A_a|$ and $\frac{1}{G_b}\sum_{g_b\in\msc{G}_b}\mathcal{U}_{B_b}^{g_b}(\rho_{B_b})=\bbm{1}/|B_b|$, respectively.
\end{definition}

The following observation follows from the definition of tele-covariant channels. 
\begin{remark}\label{rem:comp}
Tele-covariance of a channel is with respect to the groups and their unitary representations on the input and output Hilbert spaces of the channel. If associated unitary representations for the tele-covariant channels $\mc{N}^1$ and $\mc{N}^2$ are respectively same on the output Hilbert spaces of $\mc{N}^1$ that are also the input Hilbert spaces for $\mc{N}^2$, then the composition channel $\mc{N}=\mc{N}^2\circ\mc{N}^1$ is also tele-covariant. 

A quantum channel obtained by the tensor-product (super-operation ``~$\otimes$", which physically means parallel uses) of tele-covariant channels is also a tele-covariant channel.
\end{remark}

The following theorem generalizes the developments in Refs.~\cite{GC99,CDKL01,DBB08,DBW17} (see Appendix~\ref{app:tele-cov} for proof):
\begin{theorem}\label{thm:bicov}
If a multipartite channel $\mc{N}_{\vv{A'}\vv{B}\to\vv{A}\vv{C}}$
is tele-covariant, then it is teleportation-simulable with resource state~\eqref{eq:tele-sim} as its Choi state, i.e., $\theta_{\vv{LA}\vv{R}\vv{C}}=\mc{N}(\Phi^+_{\vv{L}\vv{R}|\vv{A'}\vv{B}})$. 
\end{theorem}

Following the approach in Refs.~\cite{PLOB15,HHHO09}, we obtain

\begin{theorem}\label{theo:telecov-weakconv-GE}
The LOCC-assisted secret-key-agreement capacity of a multiplex quantum channel $\mc{N}_{\vv{A'}\vv{B}\to\vv{A}\vv{C}}$ which is teleportation-simulable with resource state $\theta_{\vv{LA}\vv{R}\vv{C}}$ is upper bounded as 
\begin{equation}
\hat{P}_{\LOCC}(\mc{N})\leq E^{\infty}_{GE}(:\vv{LA}:\vv{R}:\vv{C}:)_{\theta}, 
\end{equation}
where $E^\infty(:\vv{A}:)_{\rho}$ is the regularized relative entropy of entanglement of state $\rho_{\vv{A}}$.
\end{theorem}
For the proof see Appendix \ref{app:tele-cov}. Using the above theorem, we immediately get the following corollary.

\begin{corollary}
For a multiplex quantum channel $\mc{N}_{\vv{A'}\vv{B}\to\vv{A}\vv{C}}$, which is teleportation-simulable with a tensor-stable biseparable resource state it holds $\hat{P}_{\LOCC}(\mc{N})=0$.
\end{corollary}

Let us note that unlike in Refs.~\cite{PLOB15,HHHO09}, which deals with the bipartite relative entropy of entanglement, we do not trivially get a non-regularized bound, which is due to the fact that the definition of biseparability is not tensor-stable. If we consider the relative entropy of entanglement with respect to fully separable states, however, we can employ proof argument of Theorem~4 in Ref. \cite{DBW17} and arrive at the following theorem:

\begin{theorem}\label{theo:E-alpha_bound}
For a fixed $n,\ K\in\mathbb{N},\ \varepsilon\in(0,1)$, the following bound holds for an $(n,M,\varepsilon)$ protocol for LOCC-assisted secret key agreement over a multiplex teleportation-simulable quantum channel $\mc{N}_{\vv{A'}\vv{B}\to\vv{A}\vv{C}}$ with associated resource state $\theta_{\vv{LA}\vv{R}\vv{C}}$,  $\forall \alpha>1$,
\begin{equation}\label{eq:GE-tele-weak-converse}
\frac{1}{n} \log_2 K \leq  \widetilde{E}_{\alpha,E} (:\vv{LA}:\vv{R}:\vv{C}:)_{\theta} 
+\frac{\alpha}{n(\alpha-1)}\log_2 \!\(
\frac{1}{1-\varepsilon}\).
\end{equation} 
\end{theorem}

For a proof see Appendix \ref{app:tele-cov}. Setting $\alpha=1+\frac{1}{\sqrt{n}}$ and letting $n\to\infty$, we obtain

\begin{corollary}\label{thm:tele}
The LOCC-assisted secret-key-agreement capacity of a multiplex channel $\mc{N}_{\vv{A'}\vv{B}\to\vv{A}\vv{C}}$ which is teleportation-simulable with resource state $\theta_{\vv{LA}\vv{R}\vv{C}}$ is upper bounded as 
\begin{equation}\label{eq:tele-bound}
\hat{P}_{\LOCC}(\mc{N})\leq E_E(:\vv{LA}:\vv{R}:\vv{C}:)_{\theta}, 
\end{equation}
where $E(:\vv{A}:)_{\rho}$ is the relative entropy of entanglement of state $\rho_{\vv{A}}$; this bound is also a strong converse bound. 
\end{corollary}


\section{Application to other protocols}\label{sec:Appl}
In this section, we exploit the general nature of an LOCC-assisted secret key agreement protocol over a multiplex quantum channel. We derive upper bounds on the rates for two-party and conference key distribution for a number of seemingly different protocols that are of wide interest. Such seemingly different quantum key distribution and conference key agreement protocols can be shown to be special types of LOCC-assisted secret key agreement protocol over some particular multiplex quantum channels. In particular, we identify protocols like measurement-device-independent quantum key distribution, both in the bipartite \cite{braunstein2012side,LCQ+12} and conference setting \cite{FYCC15,PhysRevA.93.022325,OLLP19}, as well as for quantum key repeaters, i.e., generalized quantum repeaters with the goal of distributing private states \cite{bauml2015limitations,CM17,christandl2017private} to be special types of LOCC-assisted secret key agreement protocol over some particular multiplex quantum channels. We are able to derive upper bounds on the rates achieved in these protocols by exploiting our results in the previous Section. Furthermore, as EPR or GHZ states are special cases of bi- or multipartite private states, respectively, the same holds for LOCC-assisted quantum communication protocols, where the goal is to distill EPR or GHZ states. By providing a unified approach to such a diverse class of private communication setup, we contribute to a better understanding of limitations on respective protocols. These limitations provide benchmarks on experimental realizations of private communication protocols.

\subsection{Measurement-Device-Independent QKD}
Measurement-device-independent (MDI) QKD is a form of QKD, where the honest parties, Alice and Bob, trust their state preparation but do not trust the detectors \cite{braunstein2012side,LCQ+12}. In a typical setup of MDI-QKD, such as the ones described in Refs.~\cite{braunstein2012side,LCQ+12}, Alice and Bob locally prepare states which they send to a relay station, that might be in the hands of Eve, using channels ${\mc N}^1_{A'\to A}$ and ${\mc N}^2_{B'\to B}$. At the relay station, a joint measurement of the systems $AB$ is performed, e.g., in the Bell basis, the results of which are classical values that are then communicated to Alice and Bob. Alice and Bob use the relay many times and perform classical post-processing. 

A way to incorporate such protocols in our scenario is to identify Alice and Bob as two trusted parties and include the measurement performed by the relay, as well as channels ${\mc N}^{1,2}$, into a bipartite quantum-classical (qc) channel 
\begin{equation}\label{eq:MDIchannel}
\mc{N}^{\text{MDI}}_{A'B'\to Z_AZ_B}\coloneqq\mc{B}_{X\to Z_AZ_B}\circ\mc{M}_{AB\to X}\circ{\mc N}^1_{A'\to A}\otimes{\mc N}^2_{B'\to B},
\end{equation}
where $\mc{M}_{AB\to x}$ is the quantum instrument (channel) performing a POVM $\{\Lambda^x\}_x$ and writing the output $x$ into a classical register $X$ and $\mc{B}_{X\to Z_AZ_B}$ a classical broadcast channel sending input $x$ to $Z_A$ and $Z_B$. Registers $Z_A$ and $Z_B$ are received by Alice and Bob respectively.
The channel ${\cal N}_{A'B'\rightarrow Z_AZ_B}^{\text{MDI}}$ is a multiplex channel which is a composition of multiplex channels (see Fig.~\ref{fig:MDI-fig}). 

Application of Theorem \ref{thm:emax-key-converse} for arbitrary systems and Theorem \ref{thm:ree-key-converse} for finite-dimensional systems (as well as the results of Ref.~\cite{DBW17,BDW18,D18thesis}), then provides bounds on the achievable key rate in terms of $E_{\max,\E}(\mc{N}^{\text{MDI}}_{A'B'\to Z_AZ_B})$ and $E^\infty_{\E}(\mc{N}^{\text{MDI}}_{A'B'\to Z_AZ_B})$, respectively, which can be seen as measures of the entangling capabilities of the the measurement $\{\Lambda^x\}_x$. The multiplex quantum channel $\mc{N}^{\text{MDI}}_{A'B'\to Z_AZ_B}$ is tele-covariant if $\mc{N}_{1,2}$ as well as $\mc{M}$ are tele-covariant and the bound reduces to the relative entropy of entanglement of the Choi state of $\mc{N}^{\text{MDI}}_{A'B'\to Z_AZ_B}$.

\subsection{Measurement-Device-Independent Conference Key Agreement}
The concept of MDI-QKD has also been generalized to the multipartite setting \cite{FYCC15,PhysRevA.93.022325,OLLP19}. We assume a setup of MDI conference agreement, where a number of trusted parties $\tf{A}_i$, for $i\in[n]$,  locally prepare a states which they send to a central relay via channels ${\mc N}^1_{A'_1\to A_1},...,{\mc N}^n_{A'_n\to A_n}$. At the relay a joint measurement is performed on $A_1A_2...A_n$, the result of which is broadcasted back to the trusted parties. It is straightforward to generalize ~\eqref{eq:MDIchannel} to the multipartite case and apply Theorems \ref{thm:emax-key-converse} and \ref{thm:ree-key-converse} (or Theorem \ref{thm:tele} for tele-covariant channels)  to obtain bounds on the conference key rates.

\subsection{Quantum Key Repeater}

Let us now consider the quantum key repeater. In its simplest setup, there are three parties, Alice, Bob, and Charlie. Alice and Bob are trusted parties who wish to establish a cryptographic key, whereas Charlie is assumed to be cooperative but is not trusted. One could think of Charlie as a telecom provider. There are two quantum channels, ${\mc N}_1^{A\to C_A}$ from Alice to Charlie and ${\mc N}_2^{B\to C_B}$ from Bob to Charlie. Alice and Bob are not connected by a quantum channel and are assumed not to have any pre-shared entanglement. Instead, Alice and Bob locally prepare quantum states, e.g., two singlets $\Phi^+_{AR_A}$ and $\Phi^+_{BR_B}$, and both send a subsystem to Charlie, using the respective channels. This is then followed by an entanglement swapping operation 
\cite{PhysRevLett.71.4287}, where Charlie performs a joint measurement on the $C_AC_B$ subsystem, communicates the result to Bob, who then performs a unitary on his reference system $R_B$, which should create entanglement that can be used for cryptographic key, between Alice and Bob. The key has to be secure even in the case that Charlie's information falls into the hands of Eve.

If the channels ${\mc N}_1^{A\to C_A}$  and ${\mc N}_2^{B\to C_B}$ are too noisy, it might be necessary to use them multiple times and perform an entanglement purification or error correction protocol before applying the swapping operation. Whereas early quantum repeater protocols \cite{PhysRevLett.81.5932,PhysRevA.59.169} make use of entanglement purification protocols that require two-way classical communication, between Alice and Charlie and between Charlie and Bob, it is also possible to use error correction that only requires one-way classical communication. Such protocols are known as second and third-generation repeater protocols  (see \cite{munro2015inside} and references therein).

By using a large enough number of repeater stations, the key can, in principle distributed across arbitrarily long distances. A way to extend a basic three-party repeater protocol to arbitrarily long repeater chains is known as nested purification \cite{PhysRevA.59.169}. 
More advanced schemes using error correction and one-way communication have also been developed \cite{munro2015inside}.

As in Refs.~\cite{bauml2015limitations,CM17,christandl2017private}, we want to find upper bounds on the rates at which the key can be distributed. Depending on the repeater protocol, there are different ways in which we can describe a quantum key repeater as a multipartite channel and use our results to obtain such bounds. We will now describe how a repeater can be described by a bipartite channel. For an alternative way to describe a repeater we refer to Appendix \ref{App:Repeater}.

In order to describe a repeater as a bipartite channel, we consider two trusted parties, Alice and Bob, and a bipartite quantum-to-classical (qc) channel that takes two quantum (and possibly also classical) inputs from Alice and Bob and returns two classical outputs to Alice and Bob, respectively. Such an operation could include the channels from Alice to Charlie and from Bob to Charlie, the measurement performed by Bob, as well as classical communication of the measurement result from Charlie to Alice and Bob. It could also include an error correction protocol that uses the channels from Alice to Charlie and from Bob to Charlie multiple times and makes use of one-way classical communication from Alice to Charlie and from Bob to Charlie. It is then followed by Charlie's measurement and classical communication to Alice and Bob. Alice and Bob are then allowed to perform LOCC among them but not including Charlie. In the case without error correction, we can define
\begin{equation}
{\mc N}_{AB\to XY}^{\text{repeater}}\coloneqq {\mc M}_{C_AC_B\to  XY}\circ {\mc N}_1^{A\to C_A}\otimes{\mc N}_2^{B\to C_B},
\end{equation}
where ${\mc M}_{C_AC_B\to XY}$ describes the measurement and sending of classical messages $X$ and $Y$ to Alice and Bob, respectively. If we add one-way error correction, we get a bipartite channel of the form
\begin{equation}
{\mc N}_{A^kB^kX'Y'\to XY}^{\text{repeater}}\coloneqq {\mc M}_{\tilde{C}_A\tilde{C}_B\to  XY}\circ {\mc E}_1^{X'A^k\to \tilde{C}_A}\otimes{\mc E}_2^{Y'B^k\to \tilde{C}_B},
\end{equation}
where ${\mc E}_1^{A^k\to \tilde{C}_A}$ includes $k$ instances of the channel $ {\mc N}_1^{A\to C_A}$, the transmission of the classical data $X'$ obtained by Alice's part of the one-way error correction protocol to Charlie, as well as Charlie's part of the error correction protocol (Alice's part of the one-way error correction protocol is included in the LOCC). ${\mc E}_2^{Y'B^k\to \tilde{C}_B}$ is defined in the same way.

By recursively combining the bipartite channels ${\mc N}^{\text{repeater}}$, it is possible to derive a bipartite channel ${\mc N}^{\text{repeater chain}}$ between Alice and Bob that includes a repeater chain with an arbitrary amount of repeater stations.

Using the results of Refs.~\cite{DBW17,BDW18,D18thesis}, or Theorem \ref{thm:ree-key-converse}, we can obtain upper bounds for key repeater protocols that only involve one-way classical communication from Charlie to Alice and Bob, as have been considered in Refs.~\cite{bauml2015limitations,christandl2017private}. The bounds are given by $\min\{ E_{\max,E}(\mc{N}^{\text{repeater (chain)}}),E^{\infty}_{E}(\mc{N}^{\text{repeater (chain)}}\}$. By Remark~\ref{rem:comp}, if $\mc{N}_{1,2}$ as well as $\mc{M}$ are tele-covariant, so will be $\mc{N}^{\text{repeater (chain)}}$. Hence, by Theorem \ref{thm:tele}, the bound reduces to the relative entropy of entanglement of the Choi state of $\mc{N}^{\text{repeater (chain)}}$. Note that, whereas the bounds in Refs.~\cite{bauml2015limitations,christandl2017private} only depend on the initial states shared by Alice and Charlie as well as  Bob and Charlie, the formulation in terms of a bipartite channel can provide bounds that also depend on the measurement performed by Charlie, as well as operations performed during error correction. The new bounds take into account imperfect measurements and error correction, which provide an additional limitation on the obtainable rate in practical implementations. Our bounds can at least shown to be comparable with the results of Refs.~\cite{bauml2015limitations,christandl2017private} under certain situations of practical interest. For an example, our bound is certainly better when ${\mc N}_1^{A\to C_A}$ and ${\mc N}_2^{B\to C_B}$ are identity channels, allowing Alice and Charlie as well as  Bob and Charlie to share maximally entangled states, whereas Charlie's measurement is noisy.

\subsection{Limitations on some practical prototypes}\label{sec:ex-mdi}
In this section, we explore fundamental limitations on some practical prototypes for MDI-QKD protocols between two trusted parties. We begin first by considering photon-based prototypes for which a detailed discussion of quantum system and transmission noise model can be found in Ref.~\cite{DKD18}. In Appendix~\ref{app:mdiqkd}, we consider MDI-QKD prototypes with qubit systems and transmission noise models depicted by dephasing or depolarizing channels.

We begin by considering a dual-rail scheme based on single photons to encode the qubits~\cite{RP10}. The dual-rail encoding of a qubit in two orthogonal optical modes can be represented in the computational basis of the qubit system, where only one of the two modes is occupied by a single photon and another mode is vacuum. When these optical modes are two polarization modes-- horizontal and vertical-- of the light, then we express eigenstates in computational basis as $\ket{H}$ and $\ket{V}$ for horizontal and vertical polarization. It is also possible to consider frequency-offset modes instead of polarization modes for dual-rail encodings. We assume a noise model for the transmission of a photon through the optical fiber to be a pure-loss bosonic channel with transmissivity $\eta$. The inputs to the optical fiber are restricted to a single-photon subspace that is spanned by $\ket{H}$ and $\ket{V}$. The action of this pure-loss channel on qubit encoded with our dual-rail scheme is identical to an erasure channel~\cite{GBP97} $\mathcal{E}$ with erasure parameter $1-\eta$ and erasure state $\ket{e}$, where $\ket{e}$ is the vacuum state, i.e., zero photon in both modes. We note that an erasure channel is tele-covariant. 

\begin{figure}[ht]
        \includegraphics[trim={0.5cm 8cm 1.5cm 0.5cm},clip,width=0.5\textwidth]{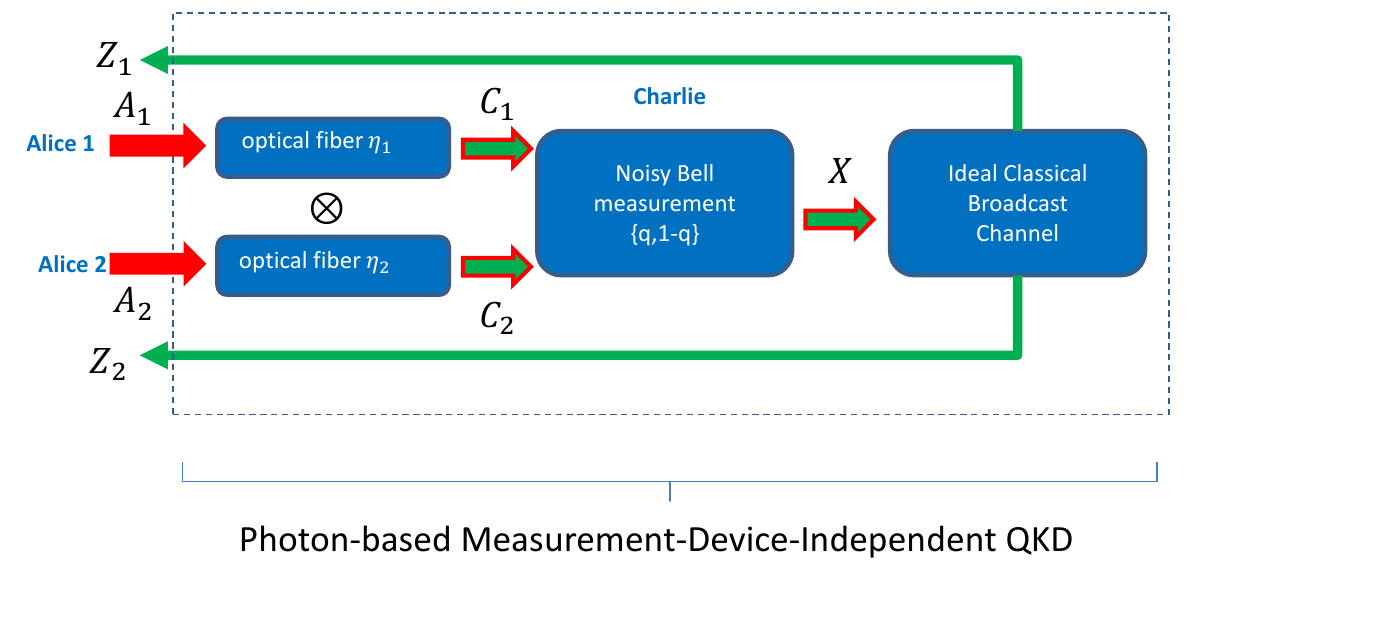}
        \caption{Pictorial illustration of our photon based MDI-QKD between two parties using dual-rail encoding scheme.}\label{fig:mq-2-2}
    \end{figure}
    
Two trusted parties $\tf{A}_i$, $i\in[2]$, use above mentioned polarization-based dual-rail photons to transmit their qubit systems to Charlie at the measurement-relay station, through the optical fibers with respective transmissivities $\eta_i$ (see Fig.~\ref{fig:mq-2-2} for MDI-QKD). We make a simplistic noise model assumption on the measurement channel $\mathcal{M}_{\vv{C_i}\to X}$ by Charlie: it can perform perfect qubit Bell-measurement for bipartite MDI-QKD, respectively, with probability $q$. Whereas with probability $1-q$ for the failed measurement, we assume the relay station signals $\op{\perp}_{X}$ to the users. In addition, we can safely assume classical communication $X\to \vv{Z}$ among all parties to be clean (noiseless) as they do not require any quantum resource. Finally, for simplicity, we assume error-correcting local operations for all parties can be made perfectly. 

To calculate upper bound on the MDI-QKD capacity, it suffices to consider the relative entropy of entanglement of the Choi state of associated multiplex channel $\mathcal{N}^{\text{MDI},\mathcal{E}}_{\vv{A}\to \vv{Z}}$ as it is tele-covariant. Notice that the action of erasure channel $\mathcal{E}_{A_i\to C_i}$ on $D_i\in\{\op{H}_{A_i},\ket{H}\!\bra{V}_{A_i},\ket{V}\!\bra{H}_{A_i},\op{V}_{A_i}\}$ is given as
\begin{equation}
\mathcal{E}_{A_i\to C_i}(D_i)=\eta_iD_i+(1-\eta_i) \Tr[D_i]\op{e}_{C_i}.
\end{equation}
Then, the Choi state $J^\mathcal{E}_{\vv{L}\vv{C}}$ of $\bigotimes_{i=1}^2\mathcal{E}_{A_i\to C_i}$ is
\begin{equation}~\label{eq:e-choi}
   J^\mathcal{E}_{\vv{L}\vv{C}}=\bigotimes_{i=1}^2 \left(\eta_i\Phi^+_{L_iC_i}+ (1-\eta_i)\frac{\mathbbm{1}_{L_i}}{2}\otimes\op{e}_{C_i}\right). 
\end{equation}

For the bipartite MDI-QKD 
\begin{align}
    \mathcal{M}_{C_1C_2\to X}(\cdot) & =   q\sum_{j=1}^4\Tr[\Phi^{(j)}(\cdot)\Phi^{(j)}] \op{j}_X \nonumber \\ 
    &\qquad +(1-q)\Tr[\cdot]\otimes\op{\perp}_X, \label{eq:bell-m}
\end{align}
where $\{\Phi^{(j)}_{C_1C_2}\}_{j=1}^4$ is the Bell-measurement, a projective measurement. $\{\Phi^{(j)}_{C_1C_2}\}_{j=1}^4$ represents the set of maximally entangled states $\{\Phi^+,\Phi^-,\Psi^+,\Psi^-\}$ for two-qubit systems and $\ket{\perp}\perp\ket{j}$. We note that the Bell-measurement is tele-covariant. Upon action of the measurement channel $\mc{M}_{C_1C_2\to X}$ on the state $J^{\mathcal{E}}_{L_1L_2C_1C_2}$~(Eq. \eqref{eq:e-choi}), the output state is essentially of the form (see~Ref.~\cite{DKD18}) 
\begin{equation}
q\eta_1\eta_2\frac{1}{4} \sum_{j=1}^4 \Phi^{(j)}_{L_1L_2}\otimes\op{j}_X+(1-q\eta_1\eta_2)\frac{\mathbbm{1}_{L_1L_2}}{4}\otimes\op{\perp}_X. 
\end{equation}
This implies, the relative entropy of entanglement of the Choi state of $\mathcal{N}^{\text{MDI},\mathcal{E}}_{A_1A_2\to Z_1Z_2}$ is $q\eta_1\eta_2$. Employing Theorem~\ref{theo:E-alpha_bound}, the bipartite MDI-QKD capacity for the given MDI-QKD prototype with erasure channels is  
\begin{equation}\label{eq:dbwhh}
    \wt{P}_{\LOCC}(\mathcal{N}^{\text{MDI},\{\mathcal{E}_i\}_{i=1}^2})=q\eta_1\eta_2,
\end{equation}
as $q\eta_1\eta_2$ \textit{bits} is an achievable rate for the given setup  (see Refs.~\cite{GES16,PLOB15,WTB16} for the private capacities of $\mathcal{E}_{A_i\to C_i}$). Notice that $q\eta_1\eta_2$ is a strong converse bound.

\begin{figure}[ht]
\centering
 \includegraphics[trim=4.5cm 7.5cm 4.5cm 6.5cm,width=6cm]{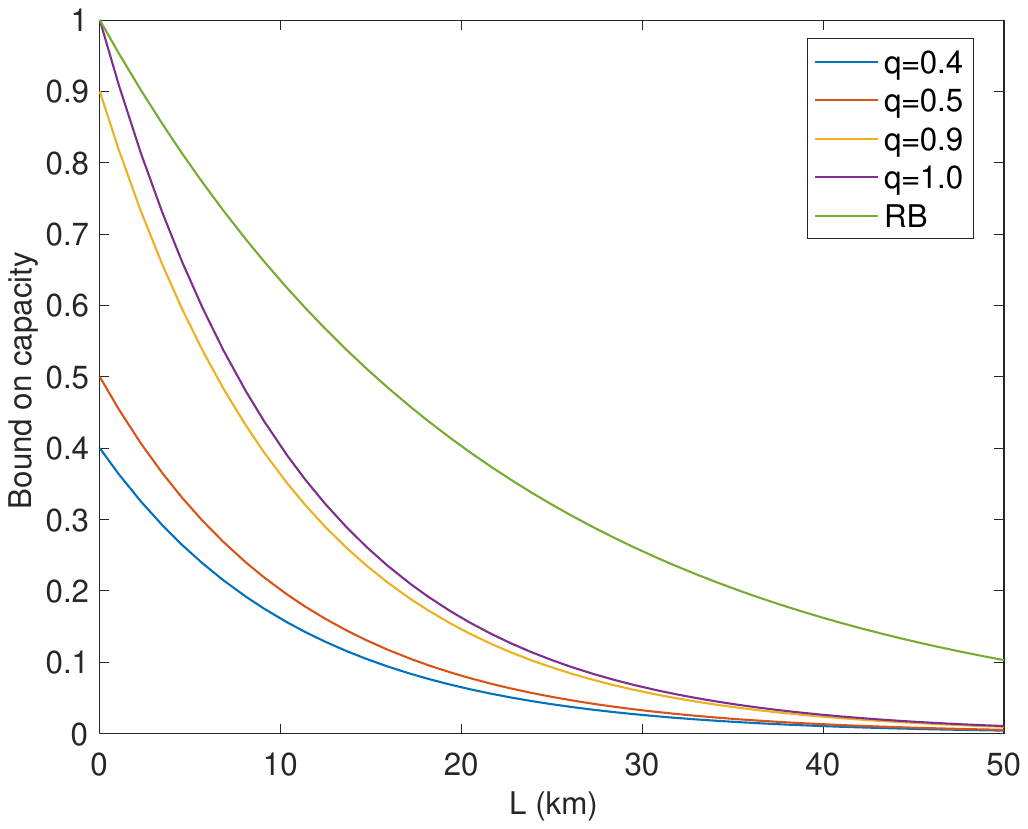}
        \caption{Rate-distance tradeoff comparison between our bound~\eqref{eq:dbwhh} (blue, red, yellow, purple lines) and RB bound (green line) for MDI-QKD protocol for our photon-based prototype.}\label{fig:MDI_2D} 
\end{figure}

For bipartite MDI-QKD (see Fig.~\ref{fig:mq-2-2}), using the results of Ref.~\cite{GES16,PLOB15} we get upper bound (RB) on the bipartite MDI-QKD capacity to be $\min\{\eta_1,\eta_2\}$ (e.g., see~\cite{CM17,pirandola2019capacities}). This bound is always looser than our strong converse upper bound $q\eta_1\eta_2$ bits for all practical purposes. In Fig.~\ref{fig:MDI_2D}, we plot rate-distance trade-off (secret key capacity versus distance L in km) for our bound in Eq.~\eqref{eq:dbwhh} when $n=2$, $\eta_1=\eta_2=\exp(-\alpha L)$, and $\alpha=\frac{1}{22 \textnormal{km}}$ and compare it with the upper bound (RB) $\eta_{1}$ (since $\eta_1=\eta_2$).

 We would like to note that, whereas there now exist variants of MDI-QKD schemes or setups that can achieve the repeaterless bound, e.g.~\cite{MZZ18,LYDS18,LL18}, the dual-rail protocols we consider here, while being sub-optimal, may be easier to implement practically. In particular, implementation of a twin-field protocol requires long distance phase-stabilization, which can be challenging~\cite{LYW+19}. We showcase here the ability to get non-trivial upper bounds for a specific, sub-optimal implementation of QKD schemes. These non-trivial upper bounds are derived from a universal framework, which illustrates the usefulness of the framework we have proposed.

\section{Lower bounds on privacy}\label{sec:DevetakWinter}

In this section, we will derive lower bounds on the secret-key-agreement rate of a multiplex channel achievable by means of cppp, in the sense of Ref.~\cite{DW05}. This is a generalization of the lower bound presented in Ref.~\cite{AH09} from multipartite states to multiplex channels, as well as a generalization of the lower bounds on one-to-one channels presented in Ref.~\cite{pirandola2009direct} to the multiplex case.

The Devetak-Winter (DW) protocol \cite{DW05}, which is considered with bipartite states, only uses one-way communication from Alice to Bob. In Ref.~\cite{pirandola2009direct}, which is concerned with one-to-one channels, a {\it direct} and {\it reverse} scenarios are considered. The former corresponding to the case where the quantum channel and the classical communication are oriented in the same direction. The latter corresponding to the case where the two are oriented in opposite directions. In Ref.~\cite{AH09} the DW protocol is generalized to multipartite states by selecting one {\it distributing party}, which performs the DW protocol with all remaining parties simultaneously.

We will now generalize this result to the setting of multiplex channels. We begin with a fully separable pure state $\phi^n\in\FS\left(:\vv{A'^{ n}L}:\vv{B^{ n}R}:\vv{P}:\right)$. Here the notation $\vv{X^{ n}}$ means we consider $n$ copies of all subsystems $X_1,...,X_M$. Application of $n$ copies of the isometric extension of multiplex channel $\mc{N}_{\vv{A'B}\to\vv{AC}}$ results in a pure state $\psi^n_{:\vv{A^{ n}L}:\vv{R}:\vv{C^{ n}P}:E^{n}}$. Let us now choose one party, ${\bf X}_i$, $i\in\{1,...,M\}$ as the distributing party. Party ${\bf X}_i$ performs a POVM $\mc{Q}=\left\{Q_x\right\}$ with corresponding random variable $X=\{x,p(x)\}$ on her subsystem, resulting in a classical-quantum-...-quantum (c{\bf q}) state
\begin{equation}
\omega_{c{\bf q}}=\sum_xp(x)\op{x}_{X}\otimes\omega^x,
\end{equation}
where $\omega^x$ is the post measurement state of the remaining parties and Eve. Party ${\bf X}_i$ then processes $X$ using classical channels $X\to Y$ and $Y\to Z$, where $Y=\{y,q(y)\}$ and  $Z=\{z,r(z)\}$ are classical random variables. $Y$ is kept by party ${\bf X}_i$ (to be used for the key) and $Z$ is broadcasted to all other trusted parties (and Eve). Upon receiving $Z$, the other parties then perform their respective POVMs with the goal of estimating the key variable $Y$. Thus, as shown in Ref.~\cite{DW05}, every trusted party ${\bf X}_j$, where $i\neq j\in\{1,...,M\}$, obtains a common key with ${\bf X}$ at a rate $r_n^{i\to j}$ of 
\begin{equation}\label{eq59}
r_n^{i\to j}=\frac{1}{n}\(I(Y:{\bf X}_j|Z)_{\tilde{\omega}_{c{\bf q}}}-I(Y:E^n|Z)_{\tilde{\omega}_{c{\bf q}}}\),
\end{equation}
where in a slight abuse of notation we use ${\bf X}_j$ as a placeholder for $A^n_jL_j$, $R_j$ or $C^n_jP_j$, depending if ${\bf X}_j$ is in $\{{\bf A}_a\}_a$, $\{{\bf B}_b\}_b$ or $\{{\bf C}_c\}_c$, respectively. The second and third case corresponds of the reverse and direct scenarios in Ref.~\cite{pirandola2009direct}, respectively. Whereas,
\begin{equation}
\tilde{\omega}_{c{\bf q}}=\sum_{xyz}r(z|y)q(y|x)p(x)\op{xyz}\otimes\omega^x.
\end{equation}
Eq.~\eqref{eq59} has to be maximized over all free input states $\phi^n\in\FS\(:\vv{A'^{ n}L}:\vv{B^{ n}R}:\vv{P}:\)$, POVMs $\mc{Q}$ as well as classical channels $X\to Y$ and $Y\to Z$. As discussed in Ref.~\cite{AH09}, a conference key among all trusted parties can be obtained at the worst case rate between any pair $({\bf X}_i,{\bf X}_j)$. We also have the freedom to choose the distributing party. Putting it all together, we can achieve the following rate of conference key:
\begin{equation}
\hat{P}_{\textnormal{cppp}}^{\mc{N}}\geq\max_i\min_j\lim_{n\to\infty}\max_{\substack{\phi^n,\mc{Q}\text{ POVM}\\X\to Y,Y\to Z} } r_n^{i\to j},
\label{eq:lb-star}
\end{equation}
with $\phi^n\in\FS\(:\vv{A'^{ n}L}:\vv{B^{ n}R}:\vv{P}:\)$. Note that in the case of a single-sender-single-receiver channel $\mc{N}:B\to C$ this reduces to the maximum of the direct and reverse key rates presented in Ref.~\cite{pirandola2009direct}.

Next, we propose an alternative generalization of the DW protocol to the case of multipartite states and multiplex channels. The rough idea is that instead of performing the DW protocol simultaneously with all other parties after her measurement, the distributing party performs a one-way protocol with a second party, who then performs a one-way protocol with a third party, and the iteration continues. In particular, the random variables obtained in all previous measurements can be passed on in every classical communication step, so that a party can adapt her measurement depending on all previous measurements instead of the first measurement as in the protocol described in Ref.~\cite{AH09}. 

We will now describe the protocol in detail: As before, we begin with a fully separable pure state $\phi^n\in\FS\(:\vv{A'^{ n}L}:\vv{B^{ n}R}:\vv{P}:\)$ and apply $n$ copies of the isometric extension of multiplex channel $\mc{N}_{\vv{A'B}\to\vv{AC}}$, resulting in a pure state $\psi^n_{:\vv{A^{ n}L}:\vv{R}:\vv{C^{ n}P}:E^{n}}$. 

Assume now we are given some permutation $\sigma:\{1,...,M\}\to\{\sigma(1),...,\sigma(M)\}$, which determines the order in which the parties participate in the protocol. Party ${\bf X}_{\sigma(1)}$ begins by performing a POVM $\mc{Q}^{(1)}$ on her share of $\psi^n$, i.e. on subsystem $A^n_{\sigma(1)}L_{\sigma(1)}$, $R_{\sigma(1)}$ or $C^n_{\sigma(1)}P_{\sigma(1)}$, depending on which kind of party ${\bf X}_{\sigma(1)}$ is. This results in a random variable $X^{(1)}=\{p_1(x_1),x_1\}$. The corresponding classical-quantum-...-quantum (c{\bf q}) state is
\begin{equation}
\omega^{(1)}_{c{\bf q}}=\sum_{x_1}p_1(x_1)\op{x_1}_{X^{(1)}}\otimes\omega^{x_1}.
\end{equation}
Party ${\bf X}_{\sigma(1)}$ then performs classical channels $X^{(1)}\to Y^{(1)}\to Z^{(1)}$, keeping random variable $Y^{(1)}$ and sending $Z^{(1)}$ to party ${\bf X}_{\sigma(2)}$. The corresponding c{\bf q} state is then given by 
\begin{equation}\label{eq:cq1}
\tilde{\omega}^{(1)}_{c{\bf q}}=\sum_{x_1y_1z_1}r_1(z_1|y_1)q_1(y_1|x_1)p_1(x_1)\op{x_1y_1z_1}\otimes\omega^{x_1},
\end{equation}
where $\omega^{x_1}$ is the state of the remaining parties and Eve. Next, party ${\bf X}_{\sigma(2)}$ performs a POVM $\mc{Q}_{Z^{(1)}}^{(2)}$ on her share of $\omega^{x_{1}}$, which provides random variable ($X^{(2)}$. Party ${\bf X}_{\sigma(2)}$ then performs classical channels $Z^{(1)}X^{(2)}\to Y^{(2)}\to Z^{(2)}$, keeps $Y^{(2)}$ for herself and sends $Z^{(2)}$ to the next party ${\bf X}_{\sigma(3)}$, who applies the same procedure. The protocol is repeated until party ${\bf X}_{\sigma(M)}$ receives $Z^{(M-1)}$, followed by her POVM and post-processing. The c{\bf q} after $k\in\{1,...,M\}$ measurements and post-processing steps is given by
\begin{align}\label{eq:cqk}
\tilde{\omega}^{(k)}_{c{\bf q}}=&\sum_{\substack{x_1...x_k\\y_1...y_k\\z_1...z_k}}\tilde{p}_{x_1y_1z_1...x_ky_kz_k}\\
&\times\op{x_1y_1z_1...x_ky_kz_k}\otimes\omega^{x_1...x_k},\nonumber
\end{align}
where we have defined recursively
\begin{align}
\tilde{p}_{x_1y_1z_1...x_ky_kz_k}=&r_k(z_k|y_k)q_k(y_k|x_kz_{k-1})p_k(x_k)\\
&\times\tilde{p}_{x_1y_1z_1...x_{k-1}y_{k-1}z_{k-1}}.\nonumber
\end{align}

Parties ${\bf X}_{\sigma(k)}$ and ${\bf X}_{\sigma(k+1)}$ can establish a key rate of \cite{DW05}

\begin{align}
r^{\sigma(k)\to\sigma(k+1)}=\frac{1}{n}&\(I(Y^{(k)}:{\bf X}_{\sigma(k+1)}|Z^{(k)})_{\tilde{\omega}^{(k)}_{c{\bf q}}}\right.\\\
&\left.-I(Y^{(k)}:E^n|Z^{(k)})_{\tilde{\omega}^{(k)}_{c{\bf q}}}\)\nonumber
\end{align}

We can again maximize over all free input states, POVMs as well as classical channels and consider the worst-case rate between any pair $({\bf X}_i,{\bf X}_j)$. Further we have the freedom to choose the order of the parties. Putting it all together, we can achieve the following rate of conference key:
\begin{equation}
\hat{P}_{\textnormal{cppp}}^{\mc{N}}\geq\max_{\sigma\in\text{perm}}\min_k\lim_{n\to\infty}
\max_{\substack{\phi^n,\mc{Q}^{(1)},...,\mc{Q}^{(k)}\text{ POVM}\\X^{(1)}\to Y^{(1)}\to Z^{(1)},\\X^{(2)}Z^{(1)}\to Y^{(2)}\to Z^{(2)},\\...,\\X^{(k)}Z^{(k-1)}\to Y^{(k)}\to Z^{(k)}}} r^{\sigma(k)\to\sigma(k+1)},
\label{eq:lb-path}
\end{equation}
with $\phi^n\in\FS\(:\vv{A'^{ n}L}:\vv{B^{ n}R}:\vv{P}:\)$. 

\subsection{Lower bound for Bidirectional Network via spanning tree}
In this section, we observe that one can tighten the lower bounds presented in the previous Section for a particular multiplex channel called {\it bidirectional network} (BN).
In the BN, each of the nodes is connected with its neighbors by product bidirectional channels, which are specific bidirectional channels that is a tensor product of two point-to-point channels directed in opposite ways from each other.

\begin{figure}[t]
    \center{\includegraphics[trim=3cm 7cm 13cm -1cm,width=5.5cm]{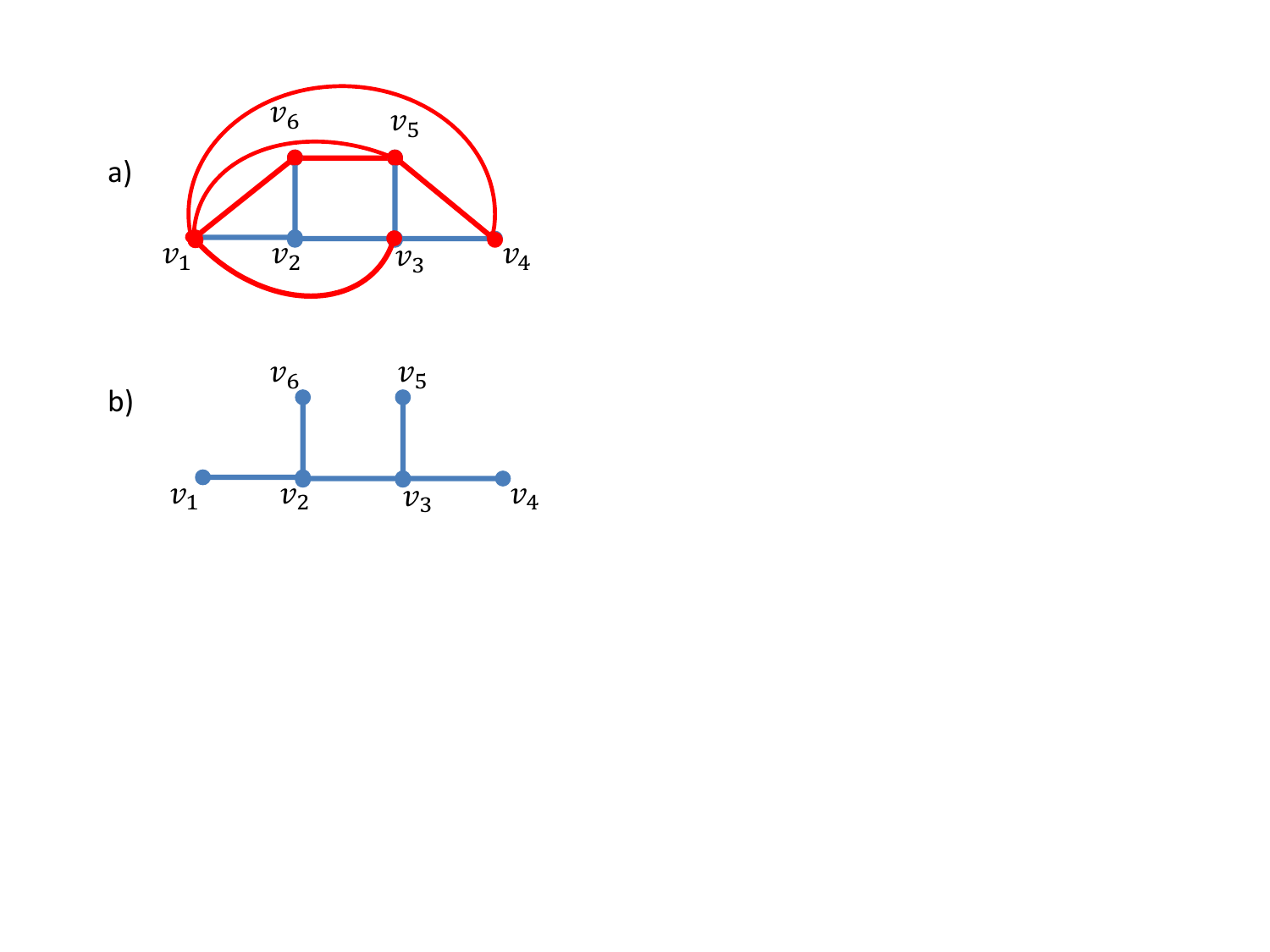}}
    \caption{An exemplary graph on a). Red edges correspond to private capacity $1$ and blue to private capacity $2$. The first strategy of obtaining conference key uses a vertex connected to all
    others, and reaches sub-optimal rate $\min\{w(e_{ij}): (v_{1},v_6), (v_1,v_5),(v_1,v_4),(v_1,v_3),(v_1,v_2)\}=1$.
    The same happens for any path, which inevitably has to pass through some red edge. The solution is a tree,
    which is a spanning tree of this graph, and contains no red edge b). Traversing edges of this tree is equivalent to the breadth-first search.
 }
    \label{fig:multigraph}
\end{figure}

We first observe that BN is a particular case of a multiplex channel (call it $\cal N$). Indeed, in this case, all the parties are of type $\cal A$; i.e., they can read and write. The rule is that each party represented in the network as a vertex $v$ has $deg(v)$ of neighbors (see Ref.~\cite{Wil96book} for introduction to graph theory). Each party is assumed to write to her neighbors and also receive from these neighbors some quantum data.
We present now a tighter bound on the private capacity of $\cal N$ based on the above exemplary graph.

To be more specific, the BN can be represented by a weighted, directed multi-graph $G=(E,V)$ in which each edge $e_{ij}=(v_i,v_j) \in E$ represents a product bidirectional channel $\Lambda_{ij}=\Lambda_{i\rightarrow j} \otimes \Lambda_{j\rightarrow i}$ with weight $W:E\mapsto R_+$ such that $W(e_{ij})= W(e_{ji})={\cal P}(\Lambda_{i\rightarrow j})={\cal P}(\Lambda_{j\rightarrow i})$ (this edge can be represented by two directed edges: one from $v_i$ to $v_j$ and the other vice versa, hence the structure is directed multi-graph). Each product bidirectional channel has in both directions the same private capacity (that however may differ for different channels). By convention, we consider edges with index $i>j$ only. The number of nodes in the network is denoted as $|V|:=n$ and the number of edges as $|E|:=m$.

As a motivation for the next consideration, there comes the fact that for such multiplex channels, the bounds given in inequalities~\eqref{eq:lb-star} and \eqref{eq:lb-path} above are not tight. We exemplify this on the graph presented in Fig.~\ref{fig:multigraph}(a). Namely, we assume that each red edge of the graph $G$ depicted there represents a (bidirectional) channel with private capacity $1$, while each blue - with this capacity equal to $2$. We do not depict all other edges (connections) as they have zero private capacity by assumption. We are ready to make two observations (i) approach of inequality~\eqref{eq:lb-star} would
    yield overall secret key agreement at rate $1$, as the only node connected all others in $G$ ($v_1$) contains
    (in fact, more than one) red edge. (ii) We observe by direct inspection that every path connecting all vertices also contains at least one red edge. On the other hand, there is a set of vertices [depicted with edges on Fig.~\ref{fig:multigraph}(b)] that forms the so-called {\it spanning tree} $T:=(V_T,E_T)\subset G$ of the graph $G$. Spanning tree is an acyclic connected subgraph of $G$, and the word "spanning" refers to the fact that all the vertices of the graph $G$ belong to $V_T$.
    It is easy to see that starting from any vertex of this tree, by the {\it breadth-first search } algorithm, one can visit all its edges, and one can obtain conference key at rate $2$ (see Ref.~\cite{LRCSbook} for introduction to algorithms). 

As a generalization of this idea, one easily comes up with the following lower bound, which is the main result of this section:
\begin{widetext}
 \begin{align}
    \hat{P}_{\textnormal{cppp}}^{\mc{N}}\geq\max_{T\subseteq G}\min_{t\in V_T,t'\in N[t]} \lim_{n\to\infty}
    \max_{\substack{\phi^n,\mc{Q}^{(1)},...,\mc{Q}^{(|V_T|)}\text{ POVM} \\X^{(v_1)}\to Y^{(v_1)}\to Z^{(v_1)},\\
        X^{(N[v_1])}Z^{(v_1)}\to Y^{(N[v_1])} \to (Z^{(N[v_1])})|_{\mbox{deg}\geq 2},\\
        X^{(N_2[v_1])}Z^{(N[v_1])} \to Y^{(N_2[v_1])}\to (Z^{(N_2[v_1])})|_{\mbox{deg}\geq 2},\\
            ...,\\
            X^{N_l[v_1]}Z^{N_l[v_1]} \to Y^{N_l[v_1]} 
        }}  r^{\sigma(t) \to \sigma(t')},
    \label{eq:lb-general}
    \end{align}
\end{widetext}
    where $1\leq l\leq n$ is an index that counts how many times the breadth-first search needs to be invoked in order to traverse all the edges of the spanning tree $T$.
    For the ease of notation $T$ is meant to be a rooted, without loss of generality, at vertex $v_1$.
    By $N[v]$ we mean the proper neighborhood of the node $v$ (i.e. the set of all vertices that are connected by a single edge with $v$). 
    In rooted tree every vertex is reachable from the root vertex
    by a path. By $N_i[v_1]$ we will
    mean the set of vertices reachable from vertex $v_1$ by a path of length $i$. Owing to this notation $N[v_1]\equiv N_1[v_1]$, while
    all vertices achievable from $v_1$ by traversing two edges belong to
    $N_2[v_1]$ and so on.

    The first inner maximization needs to be understood inductively. The first step is obvious:
    we begin with an arbitrary vertex  $v_1\in V_T$. The party $X_{v_1}$ who is at node $v_1$ performs a POVM $Q_{1}$
    which produces a random variable $X^{(v_1)}$. She processes this variable further to obtain $Y^{(v_1)}$
    and sends a communication in the form of a variable $Z^{(v_1)}$. The latter variable is broadcasted to all the next neighbors of $v_1$ i.e. $N[v_1] \setminus \{v_1\}$. Further, if at step $m-1$ the form of operations 
    and communication between the nodes has concise notation $X^{S_m}Z^{S_{m-1}}\rightarrow Y^{S_m} \rightarrow Z^{S_m}|_{\mbox{deg}\geq 2}$,
    then the next level of nesting i.e. 
    \begin{equation}
    X^{N[S_m]}Z^{S_{m}}\rightarrow Y^{N[S_m]} \rightarrow \left(Z^{N[S_m]}\right)|_{\mbox{deg}\geq 2}
    \end{equation}
    has to be understood as a short notation of the following postprocessing at a number of nodes
    from the set $N[S_m] = \{s_1,...,s_{r}\}$ with $r=|N[S_m]|$: 
    \begin{eqnarray}
    \forall_{s_i \in N[S_m]\,:\, \mbox{deg}(s_i)\geq 2}: \,X^{(s_i)} Z^{N[s_i]\cap p(s_i)} \rightarrow  Y^{(s_i)}\rightarrow Z^{(s_i)} \nonumber \\
    \forall_{s_i \in N[S_m]\,:\, \mbox{deg}(s_i)=1}: \,X^{(s_i)} Z^{N[s_i]\cap p(s_i)} \rightarrow  Y^{(s_i)}, \nonumber
    \end{eqnarray}
    where $p(s_i)$ denotes the parent vertex of the vertex $s_i$, that is the unique vertex belonging to the neighbourhood which is the closest to the root $v_1$ in terms of traversed edges.

The above description means, that if some vertex of the tree is of degree equal to one, it has no further
children in the tree to pass useful information contained in $Z$-type variable, while all vertices with
larger degree than $1$ need to broadcast appropriate data to their further neighbors in the tree.

We exemplify the lower bound given in inequality~\eqref{eq:lb-general} with the Broadcast Network depicted on Fig.~\ref{fig:multigraph}. Let us first
focus on involved sets of vertices in the process of the breadth-first search over the tree $T$. The set of vertices of the spanning tree $T$ reads $\{v_1,...,v_6\}$. As the root vertex we choose $v_1$. Next $N_1[v_1] =\{v_2\}$, $N_2[v_1] =\{v_3,v_6\}$ and $N_3[v_1] =\{v_4,v_5\}$.
The presented lower bound reads in this case:
 \begin{align}
&\hat{P}_{\textnormal{cppp}}^{\mc{N}}\geq \nonumber
\\
&\max_{T\subseteq G}\min_{t\in V_T,t'\in N[t]} \lim_{n\to\infty}
\max_{\substack{\phi^n,\mc{Q}^{(1)},...,\mc{Q}^{(6)}\text{ POVM}\\X^{(v_1)}\to Y^{(v_1)}\to Z^{(v_1)},\\
        X^{(v_2)}Z^{(v_1)}\to Y^{(v_2)}\to Z^{(v_2)},\\
        X^{(v_3)}Z^{(v_2)}\to Y^{(v_3)}\to Z^{(v_3)},\\
            X^{(v_6)}Z^{(v_2)}\to Y^{(v_6)},\\
                X^{(v_5)}Z^{(v_3)}\to Y^{(v_5)},\\
                    X^{(v_4)}Z^{(v_3)}\to Y^{(v_4)}}}  r^{\sigma(t) \to \sigma(t')}
\label{eq:lb-example}
\end{align}

In Appendix~\ref{app:complexity}, we briefly comment on the complexity of finding a sub-graph, which allows us to realize the
Conference Key Agreement with the capacity indicated by the inequality~\eqref{eq:lb-general}.

\section{Key distillation from states} \label{sec:dist-state}

In this section, we concentrate on the subject of the distillation of secret keys from quantum states. An $(n,K,\varepsilon)$ LOCC conference key distillation begins with $M$ parties $\tf{A}_i$ for $i\in[M]$ sharing $n$ copies of $M$-partite quantum state $\rho_{\vv{A}}$, to which they apply an LOCC channel $\mathcal{L}_{\vv{A^{\otimes n}}\to \vv{SK}}$. The resulting output state satisfies the following condition:
\begin{equation}
    F(\mathcal{L}_{\vv{A^{\otimes n}}\to \vv{SK}}(\rho^{\otimes n}_{\vv{A}}), \gamma_{\vv{KS}})\geq 1-\varepsilon.
\end{equation}
The one-shot secret key distillation rate from a single copy of a multipartite quantum state $K_{\textnormal{D}}^{(1,\varepsilon)}$ is upper bounded as follows (cf.~Section~\ref{sec:cka}).
\begin{theorem}\label{cor:states}
    For any fixed $\varepsilon\in (0,1)$, the achievable region of secret key agreement from a single copy of an arbitrary multipartite quantum state $\rho_{\vv{A}}$ satisfies 
    \begin{equation}\label{eqn:advanatge1}
        K_{\textnormal{D}}^{(1,\varepsilon)}(\rho)\leq E^\varepsilon_{h,\GE}(:\vv{A}:)_{\rho},
    \end{equation} 
    where 
    \begin{equation}\label{eqn:advanatge2}
        E^{\varepsilon}_{h,\GE}(:\vv{A}:)_{\rho}\coloneqq \inf_{\sigma\in\BS(:\vv{A}:)} D^\varepsilon_h(\rho\Vert\sigma).
    \end{equation}
    is the $\varepsilon$-hypothesis testing relative entropy of genuine entanglement of multipartite state $\rho_{\vv{A}}$.   
\end{theorem}
\begin{proof}
    The proof argument is the same as that of Theorem \ref{thm:one-shot-cppp}, so we omit proof here. 
\end{proof}

In the asymptotic limit the rate $K_{\textnormal{D}}^{(n,\varepsilon)}$ satisfies 
\begin{align}
    \inf_{\varepsilon>0} \limsup_{n \to \infty}\frac 1n K_{\textnormal{D}}^{(n,\varepsilon)}(\rho^{\otimes n})  = K_{\textnormal{D}}(\rho),
\end{align}
which follows directly from the definition of the secret key rate $K_{\textnormal{D}}$ \cite{AH09}.


Using the same argument as in the proof of Theorem \ref{theo:telecov-weakconv-GE} in Section \ref{sec:tele-cov}, we can also get the following asymptotic bound, which generalizes Theorem 9 in \cite{HHHO09}:
\begin{proposition}\label{bisep:theorem9}
For an $m$-partite state $\rho_{\vv{A}}$ it holds that
\begin{equation}
K_D(\rho_{\vv{A}})\leq E^\infty_{GE}(\rho_{\vv{A}}).
\end{equation}
\end{proposition}

In general, to share the conference key, it is necessary for the honest parties to distill genuine multipartite entanglement. 

\begin{corollary}
For a tensor-stable biseparable state $\rho_{\vv{A}}$ it holds $K_D(\rho_{\vv{A}})=0$.
\end{corollary}

The above Corollary of Theorem~\ref{cor:states}, is precisely due to the infimum over biseparable states. However already in tripartite setting there are two non equivalent families of three-partite genuinely entangled states, that is $\Phi_M^\mathrm{GHZ}$ type and $\Phi_M^\mathrm{W}$ type states \cite{nonEquivWandGHZ,bennett2000exact,HHHH09,amico2008entanglement,HSD15,spee2017entangled}. Both families of states contain states which are maximally entangled; however, they can not be transformed with LOCC one into another at unit rate \cite{FL-2007,FL-2008,Hoi-Kwong2010,vrana2015asymptotic,vrana2019distillation,streltsov2020rates}. As the perfect $\Phi_M^\mathrm{GHZ}$ state plays a role of the honest (or perfect) implementation of conference quantum key agreement protocols, the distillation of $\Phi_3^\mathrm{GHZ}$ states from $\Phi_3^\mathrm{W}$ states has been intensively studied \cite{SVW,FL-2007,KT-2010,CCL-2011,vrana2015asymptotic,vrana2019distillation}. In particular, recalling Example 11 of Ref. \cite{SVW}, it is known that one can not transform a single $\Phi_3^\mathrm{W}$ state into $\Phi_3^\mathrm{GHZ}$ state even in a probabilistic manner. However, according to \cite[Theorem~2]{SVW}, the calculated asymptotic rate for conversion form $\Phi_3^\mathrm{W}$ to $\Phi_3^\mathrm{GHZ}$ due to certain protocol is approximately $0.643$ (per copy), what constitutes a lower bound for the general case. Another, complementary lower bound has been provided in Ref.~\cite{vrana2019distillation}.

Surprisingly, in the one shot regime, distillation of $\Phi_3^\mathrm{GHZ}$ states from $\Phi_3^\mathrm{W}$ states, and
therefore of secret key is still possible. To accomplish this task, it
is sufficient to consider the initial state as being made up of two copies of the $\Phi_3^\mathrm{W}$ state. Then using results in Ref.~\cite{FL-2007}, it follows that we can obtain two $\Phi_2^+$ states in two distinct bipartite systems with a probability that is arbitrarily close to $\frac 23$; having this in mind, one can obtain $\Phi_3^\mathrm{GHZ}$ by employing ancilla and the entanglement swapping protocol~\cite{PhysRevLett.71.4287}. In this way, we calculated a lower bound on the distillation of $\Phi_3^\mathrm{GHZ}$ states from tow copies of $\Phi_3^\mathrm{W}$ state in one-shot regime (one $\Phi_3^\mathrm{GHZ}$ state with probability $\frac 23$ from two $\Phi_3^\mathrm{W}$ states). This lower bound can be compared with the upper bound in Theorem~\ref{cor:states} given above. 

Nevertheless distillation of $\Phi_M^\mathrm{GHZ}$ states is only an example of key distillation technique \cite{Cabello2000,Scarani2001,Chen2005,AH09,Augusiak2009W,Grasselli2019,vrana2015asymptotic,vrana2019distillation}. A more general conference key agreement scenario of our interest incorporates distillation of twisted $\Phi_M^\mathrm{GHZ}$ states (see Definition~\ref{def:multi-priv-state}) \cite{Pankowski2008,HHHH09,AH09,Pankowski2010,bauml2017fundamental}. In that case, an approach for upper bounding conference key rates that is different than the estimation of $\Phi_M^\mathrm{W}$ to $\Phi_M^\mathrm{GHZ}$ conversion rates is required. This corresponds to a possible gap between rates of $\Phi_M^\mathrm{GHZ}$ (that can be distilled) and secret key distillation. Since the $\Phi_M^\mathrm{GHZ}$ state is an instance of a private state, an upper bound on the conference key rate is also an upper bound on the distillation rate from any state. For plotting our numerical results, we do concentrate on secret key distillation from $n$ copies of $\Phi_M^\mathrm{W}$ state in order to compare with other limitations discussed in this section.   

The upper bound in Theorem 8 has optimization over all possible biseparable states. Computation of the exact value of the bound given in Eq. (\ref{eqn:advanatge1}) need not be feasible in general. As we take the infimum in Eq. (\ref{eqn:advanatge2}), we can obtain non-trivial upper bounds on the upper bound given in Eq. (\ref{eqn:advanatge1}) by considering optimization over suitable subsets of biseparable states. We make an educated guess for the form of biseparable state to yield a non-trivial upper bound. We remark here that the set of biseparable states is not closed under tensor product so that we have to find different states for any tensor power $n$ of $\Phi_M^\mathrm{W}$ or $\Phi_M^\mathrm{GHZ}$ states. We devise two families of biseparable states $\pi_\textnormal{W}^{n,M}$, $\pi_\textnormal{GHZ}^{n,M}$ adjusted to both number of copies $n$ and number of parties $M$.
\begin{align}
    &\pi_\textnormal{GHZ}^{n,M}:= \frac{1}{M}\sum_{i=1}^M \left ( \mathcal{S}_{1,i} \left( \frac{I}{2}\otimes \Phi_{M-1}^\mathrm{GHZ}\right)\right)^{\otimes n}, \label{eqn:bisepGHZ}\\
    &\pi_\textnormal{W}^{n,M}:= \frac{1}{M}\sum_{i=1}^M \left ( \mathcal{S}_{1,i} \left( \op{0}\otimes \Phi_{M-1}^\mathrm{W}\right)\right)^{\otimes n},\label{eqn:bisepW}
\end{align}
where the operator $\mathcal{S}_{1,i}$ swaps the qubit of the first party with qubit of the $i$th party. The choice of $\pi_\textnormal{GHZ}^{n,M}$ and $\pi_\textnormal{W}^{n,M}$ states is motivated by keeping correlation between $M-1$ parties most similar to these in $\Phi_M^\mathrm{GHZ}$ or $\Phi_M^\mathrm{W}$ states, while keeping one party explicitly separated. Additionally $\pi_\textnormal{GHZ}^{n,M}$ and $\pi_\textnormal{W}^{n,M}$ states by definition are symmetric with respect to permutation of parties, due to permutations with $\mathcal{S}_{1,i}$. 

We would like to point out here that $\pi_\textnormal{W}^{1,3}$ presented here is closer to $\Phi_{3}^\mathrm{W}$ state in the Hilbert-Schmidt norm than the state (let us call it $\Upsilon$) in Ref.~\cite{VDM-2002}, even though the state constructed there was supposed to biseparable state closest to $\Phi_{3}^\mathrm{W}$ in the Hilbert-Schmidt norm. This result is due to different definitions of biseparability; the state in Ref.~\cite{VDM-2002} is a tensor product with respect to one of the cuts, whereas we make use of convexity of the set of biseparable states. Indeed our states are biseparable by construction (see Sec. \ref{sec:ME}).

The upper bound on asymptotic secret key rate can be compared with lower bound on asymptotic $\Phi_{3}^\mathrm{GHZ}$ states from $\Phi_{3}^\mathrm{W}$ sates distillation \cite{SVW}. This can be done in the following way. First, we notice that if two parties unite, then $M-1$-partite key is no less than initial $M$-partite key because the set of operations of $M$-partite LOCC protocol is a strict subset the set of operations for the case in which two parties $i$ and $j$, are in the same laboratory. We have the following Proposition:
\begin{proposition}\label{prop:statesmerging} For any $M$-partite state $\rho_{[M]}$, the asymptotic secret-key-agreement rate satisfies the following inequality:
\begin{align}
    \max_{k}K_D(\rho_{[M+1]_{k}})\leq K_D(\rho_{[M]}) \le \min_{i,j} K_D(\rho_{[M-1]_{ij}}), 
\end{align}
where $[M]=[1,..,M]$ and $[M-1]_{ij}=[1,..,i-1,(i,j),i+1,...,j-1,j+1,...,M]$ indicate a state $\rho_{[M-1]}$ in which subsystems $i$ and $j$ are merged. Analogously $[M+1]_{k}=
[1,...,k-1,k_1,k_2,k+1,...,M+1]$
indicates the state in which
subsystem $k$ was split into systems $k_1$ and $k_2$.
\end{proposition}
\begin{proof}
It is enough to notice that class of LOCC protocols involved in definition of $K_D(\rho_{[M]})$ is strictly contained in the class of the protocols involved in definition of $K_D(\rho_{[M-1]_{ij}})$. Indeed, the merged parties can still simulate any operation from the former class; however, together, they can perform many more operations including global quantum operations on all merged subsystems together. Since $K_D$ is defined as the supremum of the key rate over such protocols, the upper bound follows. For the
lower bound it is enough to 
notice that by splitting subsystem(s) of $\rho$ we restrict
the class of operations that
can be used to distill key.
\end{proof}
We immediately observe that Proposition \ref{prop:statesmerging} provides a whole family of nonequivalent upper bounds. To see this one can consider of a state that is not invariant under permutations. What is more, one can continue merging as long as there is still two or more subsystems left.
\begin{corollary}\label{cor:dist-state}
For any $M$-partite state  $\rho_{[M]}$ defined on the Hilbert space ${\mathcal H}$, the asymptotic secret-key-agreement rate satisfies the following inequality:
\begin{align}
    \max_L K_D(\rho_{[L]})\leq K_D(\rho_{[M]}) \leq \min_{N} K_D(\rho_{[N]}),
\end{align}
where the state $\rho_{[L]}$ is 
obtained from the state $\rho_{[M]}$ by splitting its subsystems so that $L\geq \log \dim({\mathcal H})$.
Analogously the state $\rho_{[N]}$ is obtained via any merging of subsystems of $\rho_{[M]}$, such that $\rho_{[N]}$ has at least two subsystems.
\end{corollary}

Hence in particular case of $\Phi_{3}^\mathrm{W}$ state we can also skip minimization with respect to $i$, $j$ since the state is symmetric. Using properties of entanglement measures \cite{Horodecki_2000,Ishizaka_2005,plenio2005introduction}, we have
\begin{align}
&K_D\left( \Phi_{3}^\mathrm{W}\right) \le 
K_D\left( \Phi_{2+1}^\mathrm{W}\right) \le E_r^{\infty}\left(\Phi_{2+1}^\mathrm{W}\right)\\
&=h_2\left(\frac{1}{3}\right)\approx 0.9183~ \mathrm{bit},
\end{align}
where $h_2(x)$ is the binary entropy function.

The asymptotic key rate and bounds on it are usually noninteger real numbers. In the one-shot regime, expressing these quantities in a similar manner, instead of integers obtained with floor or ceiling functions, is no less meaningful because the amount of secret key and the value of bounds are functions of privacy test parameter $\varepsilon$, which can vary, yielding, in general, different values of these quantities. Therefore dependence of the scenario on the privacy parameter $\varepsilon$ is interesting on its own. See Appendix~\ref{app:StatesPlots} and Ref. \footnote{See Supplemental Material at \url{https://journals.aps.org/prx/abstract/10.1103/PhysRevX.11.041016#supplemental} for codes to
get plots.}.

\begin{remark}
	It is natural that the analogies of Proposition~\ref{prop:statesmerging} and Corollary~\ref{cor:dist-state} hold for the multiplex quantum channel $\mc{N}$. The upper bound on $M$-partite multiplex quantum channel takes the form of $M-1$-partite multiplex channel, where the new party's type is determined according to the following rule: If the two parties are of the same type (say $B$), then the new type is the same as the same ($B$ in that case). If the types are different, then the new type becomes always $A$ because, e.g., when $B$ and $C$ are merged, they have the ability to both read and write.
\end{remark}

\section{Discussion}
\label{sec:dis}
We have provided universal limits on the rates at which one can distribute conference key over a quantum network described by a multiplex quantum channel. We have shown that multipartite private states are necessarily genuine multipartite entangled. As a consequence, it is not possible to distill multipartite private states from tensor-stable biseparable states. We have obtained an upper bound on the single-shot, classical preprocessing and postprocessing assisted secret-key-agreement capacity. The bound is in terms of the hypothesis testing divergence with respect to biseparable states of the output state of the multiplex channel, maximized over all fully separable input states. We have further provided strong-converse bounds on the LOCC-assisted private capacity of multiplex channels that are in terms of the max-relative entropy of entanglement as well as the regularized relative entropy of entanglement. In the case of tele-covariant multiplex channels, we have also obtained bounds in terms of the relative entropy of entanglement of the resource state. We have shown the versatility of our bounds by applying it to several communication scenarios, including measurement-device independent QKD and conference key agreement as well as quantum key repeaters. In addition to our upper bounds, we have also provided lower bounds on asymptotic conference key rates, that are asymptotically achievable in Devetak-Winter-like protocols. We have also derived an upper bound on the secret key that can be distilled from finite copies of multipartite states via LOCC and show some numerical examples. The task of distillation of $\Phi_3^\mathrm{GHZ}$ from $\Phi_3^\mathrm{W}$ was extensively studied in the literature \cite{FL-2007,FL-2008,Hoi-Kwong2010}. Here we initiate the study on the distillation of the key rather than $\Phi_3^\mathrm{GHZ}$ distillation from the $\Phi_3^\mathrm{W}$ state. This is the rate of the distillation of "twisted" $\Phi_3^\mathrm{GHZ}$ being private states - a class to which $\Phi_3^\mathrm{GHZ}$ belongs. It would be interesting to find if the distillation of the key from $\Phi_3^\mathrm{W}$ is just equivalent to the distillation of $\Phi_3^\mathrm{GHZ}$ (see recent result on this topic \cite{Grasselli2019}).

Distillation of secret key allows trusted parties to access private random bits. Our lower bound on an asymptotic LOCC-assisted secret-key-agreement capacity over a multiplex channel also provides an asymptotic achievable rate of private random bits for trusted parties over a multiplex channel with classical preprocessing and postprocessing. 

Our work also provides frameworks for the resource theories of multipartite entanglement for quantum multipartite channels (analogous to bipartite channels as discussed in Refs.~\cite{D18thesis,DBW17,BDWW19,GS19}). In this context, it is natural to extend the results of Ref.~\cite{PHM18} where the so-called layered QKD is considered, to the noisy case of multipartite private states. It would be interesting to systematically consider other frameworks in the resource theory of multipartite entanglement. An important future direction for application purposes is to identify new information processing tasks and determine bounds on the rate regions of classical and quantum communication protocols over a multiplex channel (e.g., see Refs.~\cite{Sha61,GKbook,BHTW10,WDW16,D18thesis,LALS19,TR19,DW19}).

\medskip

\begin{acknowledgments}

S. D. is grateful to Jonathan P. Dowling (3 April 1955 – 5 June 2020) for insightful discussions.

The authors thank Koji Azuma, Nicolas Cerf, Marcus Huber, Liang Jiang, Sumeet Khatri, Glaucia Murta, Mark M. Wilde, and Paweł Żyliński for valuable discussions. S. D. acknowledges individual fellowships atUniversit\'{e} libre de Bruxelles; this project received funding from the European Union’s Horizon 2020 research and innovation program under the Marie Skłodowska-Curie Grant Agreement No. 801505. S. B. acknowledges funding from the European Union’s Horizon 2020 research and innovation program, Grant Agreement No. 820466 (project CiViQ), the postdoctoral fellowships program Beatriu de Pinós, funded by the Secretary of Universities and Research (Government of Catalonia), and by the Horizon 2020 program of research and innovation of the European Union under the Marie Sklodowska-Curie Grant Agreement No. 801370 (2019 BP 00097), as well as from the Government of Spain (FIS2020-TRANQI and Severo Ochoa CEX2019-000910-S), Fundació Cellex, Fundació Mir-Puig, Generalitat de Catalunya (CERCA, AGAUR SGR 1381, and QuantumCAT). K. H. and M.W. acknowledge support from the grant Sonata Bis 5 (Grant No. 2015/18/E/ST2/00327) from the National Science Center. We acknowledge partial support by the Foundation for Polish Science (IRAP project, ICTQT, contract no. MAB/2018/5, co-financed by the EU within Smart Growth Operational Programme). The ’International Centre for Theory of Quantum Technologies’ project (contract no. MAB/2018/5) is carried out within the International Research Agendas Programme of the Foundation for Polish Science co-financed by the European Union from the funds of the Smart Growth Operational Programme, axis IV: Increasing the research potential (Measure 4.3).



\end{acknowledgments}

\begin{appendix}

\section{Generalized divergences and their properties}\label{App:Div}
Any generalized divergence $\tf{D}(\cdot\Vert\cdot)$ satisfies the following two properties for an isometry $U$ and a state~$\tau$ \cite{WWY14}:
\begin{align}
\mathbf{D}(\rho\Vert \sigma) & = \mathbf{D}(U\rho U^\dag\Vert U \sigma U^\dag),\label{eq:gen-div-unitary}\\
\mathbf{D}(\rho\Vert \sigma) & = \mathbf{D}(\rho \otimes \tau \Vert \sigma \otimes \tau).\label{eq:gen-div-prod}
\end{align}

The sandwiched R\'enyi relative entropy obeys the following ``monotonicity in $\alpha$'' inequality \cite{MDSFT13}:
\begin{equation}\label{eq:mono_sre}
\wt{D}_\alpha(\rho\Vert\sigma)\leq \wt{D}_\beta(\rho\Vert\sigma) \text{ if }  \alpha\leq \beta, \text{ for } \alpha,\beta\in(0,1)\cup(1,\infty).
\end{equation}
The following inequality states that the sandwiched R\'enyi relative entropy $\wt{D}_\alpha(\rho\Vert\sigma)$ between states $\rho,\sigma$ is a particular generalized divergence for certain values of $\alpha$ \cite{FL13,Bei13}. For a quantum channel $\mc{N}$,
\begin{equation}
\wt{D}_\alpha(\rho\Vert\sigma)\geq \wt{D}_\alpha(\mc{N}(\rho)\Vert\mc{N}(\sigma)), \ \forall \alpha\in \[1/2,1\)\cup (1,\infty).
\end{equation}

In the limit $\alpha\to 1$, the sandwiched R\'enyi relative entropy $\wt{D}_\alpha(\rho\Vert\sigma)$ between quantum states $\rho,\sigma$ converges to the quantum relative entropy \cite{MDSFT13,WWY14}:
\begin{equation}\label{eq:mono_renyi}
\lim_{\alpha\to 1}\wt{D}_\alpha(\rho\Vert\sigma)=D(\rho\Vert\sigma),
\end{equation}
and quantum relative entropy~\cite{Ume62} between states is
\begin{equation}
D(\rho\Vert\sigma)\coloneqq \Tr[\rho\log_2(\rho-\sigma)]
\end{equation}
for $\supp(\rho)\subseteq \supp(\sigma)$ and otherwise it is $\infty$. 

In the limit $\alpha\to 1/2$, the sandwiched R\'enyi relative entropy $\wt{D}_\alpha(\rho\Vert\sigma)$ converges to $-\log_2 F(\rho,\sigma)$, where $F(\rho,\sigma)$ is the fidelity between $\rho, \sigma$ and defined as
\begin{equation}
F(\rho,\sigma)\coloneqq \left[\Tr\left[\sqrt{\sqrt{\sigma}\rho\sqrt{\sigma}}\right]\right]^2.
\end{equation}

The following inequality relates $D^\varepsilon_h(\rho\Vert \sigma)$ to $\wt{D}_{\alpha}(\rho\Vert\sigma)$ for density operators $\rho,\sigma$, $\alpha\in(1,\infty)$ and $\varepsilon\in(0,1)$ \cite{HP91,N01,ON00},\cite[Lemma~5]{CMW14}:
\begin{equation}\label{eq:hyp-sand-rel}
D^\varepsilon_h(\rho\Vert \sigma)\leq  \wt{D}_{\alpha}(\rho\Vert\sigma)+\frac{\alpha}{\alpha-1}\log\left(\frac{1}{1-\varepsilon}\right).
\end{equation}
The following inequality also holds~\cite{WR12}:
\begin{equation}
D^\varepsilon_h(\rho\Vert\sigma)\leq \frac{1}{1-\varepsilon}\left(D(\rho\Vert\sigma)+h_2(\varepsilon)\right),
\end{equation}
where $h_2(\varepsilon)\coloneqq -\varepsilon\log_2\varepsilon-(1-\varepsilon)\log_2 (1-\varepsilon)$ is the binary entropy function. 

In a specific case $\varepsilon$-hypothesis testing relative entropy can be calculated exactly.
\begin{lemma}\label{lem:exactD} If $\rho$ is a pure state and it is one of the eigenvectors of $\sigma$ , i.e., there exists decomposition $\sigma = p_0 \rho + \sum_{i=1} p_i \gamma_i^\perp$, with $\sum_{i=0}p_i=1$, $0\le p_i \le 1$, $p_0 \neq 0$ and states $\gamma_i^\perp$ orthogonal to $\rho$ then for any $\epsilon \in [0,1]$:
\begin{align}
	D_h^\varepsilon \left(\rho \Vert \sigma  \right) = -\log_2 \Tr \left[\Omega \sigma\right],
\end{align}
with $\Omega = (1-\varepsilon)\rho$.
\end{lemma}

\section{Multiplex quantum channels}\label{app:mqc}
All network channels possible in a communication setting are special cases of multiplex quantum channels $\mc{N}_{\vv{A'}\vv{B}\to\vv{A}\vv{C}}$ (see Fig.~\ref{fig:mpc}), e.g.,
\begin{enumerate}
\item Point-to-point quantum channel: This is a quantum channel of the form $\mathcal{N}_{B_b\to C_c}$ with a single sender and
a single receiver. When a multiplex quantum channel has the form $\mathcal{N}_{B_b\to
C_c}$ then $\mathscr{A}=\emptyset$ and $|\mathscr{B}|=1=|\mathscr{C}|$. This is arguably
the simplest form of a communication (network) channel as it involves only two parties with
one party sending input to the channel and the other receiving the output from the channel.
\item Bidirectional quantum channel: This is a multiplex quantum channel of the form $\mc{N}_{A_1'A_2'\to A_1A_2}$ with two parties who are
both senders and receivers, i.e. $|\mathscr{A}|=2$ and $\mathscr{B}=\emptyset=\mathscr{C}$ (cf.~\cite{D18thesis,BHLS03}).


\item Quantum interference channel: This is a bipartite quantum channel of the form $\mc{N}_{B_1B_2\to C_1C_2}$ with two senders and two receivers (cf.~\cite{FHSSW11}). We may also call $\mc{N}_{\vv{B}\to \vv{C}}$ with an equal number of senders and receivers as quantum interference channel. \item Broadcast quantum channel: This is a multipartite quantum channel of the form $\mc{N}_{B_b\to \vv{C}}$ with a single sender and multiple receivers (cf.~\cite{D05,YHD2006}). We may also call $\mc{N}_{\vv{B}\to\vv{C}}$ as a broadcast channel if the number of senders is less than the number of receivers. 
\item Multiple access quantum channel: This is a multipartite quantum channel of the form $\mc{N}_{\vv{B}\to C_c}$ with multiple senders and a single receiver (cf.~\cite{YHD05MQAC}). We may also call $\mc{N}_{\vv{B}\to\vv{C}}$ as a multiple access channel if the number of senders is more than the number of receivers. 
\item Physical box: Any physical box with quantum or classical inputs and quantum or classical outputs.
\item Network quantum channels of types $\mc{N}_{\vv{A'}\to \vv{A}\vv{C}}$ and $\mc{N}_{\vv{A'}\vv{B}\to \vv{A}}$. 
\end{enumerate}

If inputs and outputs to a multiplex channel are classical systems and underlying processes are governed by classical physics, then the channel is called \textit{classical multiplex channel} (see \cite{GKbook} for examples of such network channels). If input and output to the channel are quantum and classical systems, respectively, then the channel is called quantum to classical channel. If input and output to the channel are classical and quantum systems, respectively, then the channel is called classical to a quantum channel.

\section{Privacy test}\label{app:priv-t}

Recall the definition the twisting operation
\begin{equation}
U^\text{tw}_{\vv{KS}}=\sum_{i_1,...,i_M=0}^{K-1}\ket{i_1...i_M}\bra{i_1...i_M}_{\vv{K}}\otimes U_{\vv{S}}^{(i_1...i_M)}
\end{equation}
and a privacy test as
\begin{align}
\Pi^{\gamma,K}_{\vv{KS}}&=U^\text{tw}_{\vv{KS}}\(\Phi^{\GHZ}_{\vv{K}}\otimes\bbm{1}_{\vv{S}}\)U^{\text{tw}\dag}_{\vv{KS}}\\
&=\frac{1}{K}\sum_{i,k=0}^{K-1}\(\ket{i}\bra{k}\)^{\otimes M}_{\vv{K}}\otimes U_{\vv{S}}^{(i^M)}U_{\vv{S}}^{(k^M)\dag},
\end{align}
where we have defined the notation $i^M:=\underbrace{i...i}_{M\text{ times}}$. We will now provide the proof of Theorem~\ref{thm:priv-gme}:
\begin{proof}[Proof of Theorem~\ref{thm:priv-gme}]
We begin by showing the bound for pure biseparable states $\ket{\varphi}_{\vv{KS}}$. For such a state there exists a bipartition of the parties, defined by nonempty index sets $I\subset\{1,...,M\}$ and $J=\{1,...,M\}\setminus I$, such that the state is product with respect to that bipartition.  Namely $\ket{\varphi}_{\vv{KS}}=\ket{\widetilde\varphi}_{S_{I}K_{I}}\otimes\ket{\overline\varphi}_{S_{J}K_{J}}$, where we have defined $\mc{H}_{S_{I}K_{I}}=\bigotimes_{i\in I}\mc{H}_{S_{i}K_{i}}$ and $\mc{H}_{S_{J}K_{J}}=\bigotimes_{j\in J}\mc{H}_{S_{j}K_{j}}$.
Let us also define $m:=|I|$ and $n:=|J|$ and note that $M=m+n$. We can expand
\begin{align}
&\ket{\widetilde\varphi}_{S_{I}K_{I}}=\sum_{i_1,...,i_{m}=0}^{K-1}\widetilde{\alpha}_{i_1...i_{m}}\ket{i_1...i_{m}}_{K_{I}}\otimes\ket{\widetilde{\phi}_{i_1...i_{m}}}_{S_{I}}\\
&\ket{\overline\varphi}_{S_{J}K_{J}}=\sum_{j_1,...,j_{n}=0}^{K-1}\overline{\alpha}_{j_1...j_{n}}\ket{j_1...j_{n}}_{K_J}\otimes\ket{\overline{\phi}_{j_1...j_{n}}}_{S_J}.
\end{align}
Here $\widetilde{\alpha}_{i_1...i_{m}}\in\bbm{C}$ such that $\sum_{i_1,...,i_{m}=0}^{K-1}|\widetilde{\alpha}_{i_1...i_{m}}|^2=1$ and $\overline{\alpha}_{j_1...j_{n}}\in\bbm{C}$ such that $\sum_{j_1,...,j_{n}=0}^{K-1}|\overline{\alpha}_{j_1...j_{n}}|^2=1$. Further it holds

\begin{widetext}
\begin{align}
\Tr\left[\Pi^{\gamma,K}_{\vv{KS}}\varphi_{\vv{KS}}\right]
&=\Tr\left[\(\frac{1}{K}\sum_{i,k=0}^{K-1}\(\ket{i}\bra{k}\)^{\otimes M}_{\vv{K}}\otimes U_{\vv{S}}^{(i^M)}U_{\vv{S}}^{(k^M)\dag}\)\widetilde\varphi_{K_IS_I}\otimes\overline\varphi_{K_JS_J}\right]\\
&=\frac{1}{K}\sum_{i,k=0}^{K-1}\widetilde{\alpha}_{i^{m}}\overline{\alpha}_{i^{n}}(\widetilde{\alpha}_{k^{m}})^*(\overline{\alpha}_{k^{n}})^*\Tr\left[U^{(i^M)\dag}\ket{\widetilde{\phi}_{i^m}}\bra{\widetilde{\phi}_{k^m}}_{S_I}\otimes\ket{\overline{\phi}_{i^n}}\bra{\overline{\phi}_{k^n}}_{S_J}U^{(k^M)}\right]\\
&=\frac{1}{K}\sum_{i,k=0}^{K-1}\widetilde{\alpha}_{i^m}\overline{\alpha}_{i^n}(\widetilde{\alpha}_{k^m})^*(\overline{\alpha}_{k^n})^*\bra{\zeta_k}\ket{\zeta_i},\label{eq:9}
\end{align}
\end{widetext}
where we have defined the state
\begin{equation}
\ket{\zeta_i}_{\vv{S}} \coloneqq U^{(i^M)\dag}\ket{\widetilde{\phi}_{i^m}}_{S_I}\otimes\ket{\overline{\phi}_{i^n}}_{S_J}.
\end{equation}
We note that eq. (\ref{eq:9}) is a probability, in particular it is real and non-negative. Hence it holds
\begin{align}
&\frac{1}{K}\sum_{i,k=0}^{K-1}\widetilde{\alpha}_{i^m}\overline{\alpha}_{i^n}(\widetilde{\alpha}_{k^m})^*(\overline{\alpha}_{k^n})^*\bra{\zeta_k}\ket{\zeta_i}\\
&=\left|\frac{1}{K}\sum_{i,k=0}^{K-1}\widetilde{\alpha}_{i^m}\overline{\alpha}_{i^n}(\widetilde{\alpha}_{k^m})^*(\overline{\alpha}_{k^n})^*\bra{\zeta_k}\ket{\zeta_i}\right|\\
&\leq\frac{1}{K}\sum_{i,k=0}^{K-1}\left|\widetilde{\alpha}_{i^m}\right|\left|\overline{\alpha}_{i^n}\right|\left|\widetilde{\alpha}_{k^m}\right|\left|\overline{\alpha}_{k^n}\right|\left|\bra{\zeta_k}\ket{\zeta_i}\right|
\end{align}
where in the first inequality we have used the subadditivity and multiplicity of the absolute value of complex numbers. We note that for all $i,k$ in the sum $\left|\bra{\zeta_k}\ket{\zeta_i}\right|\leq 1$. Let us define $p_i=\left|\widetilde{\alpha}_{i^m}\right|^2$ and note that $p_i\geq0$ and $\sum_{i=0}^{K-1}p_i\leq1$. Let us also define $q_i=\left|\overline{\alpha}_{i^n}\right|^2$ and note that $q_i\geq0$ and $\sum_{i=0}^{K-1}q_i\leq1$. Hence there exist respective probability distributions $\{\hat{p}_i\}$ and $\{\hat{q}_i\}$ over $\{0,...,K-1\}$ such that $p_i\leq\hat{p}_i$ and $q_i\leq\hat{q}_i$ for all $i=0,...,K-1$. We then obtain
\begin{align}
&\frac{1}{K}\sum_{i,k=0}^{K-1}\left|\widetilde{\alpha}_{i^m}\right|\left|\overline{\alpha}_{i^n}\right|\left|\widetilde{\alpha}_{k^m}\right|\left|\overline{\alpha}_{k^n}\right|\left|\bra{\zeta_k}\ket{\zeta_i}\right|\\
&\leq\frac{1}{K}\sum_{i,k=0}^{K-1}\sqrt{p_iq_ip_kq_k} =\frac{1}{K}\left[\sum_{i=0}^{K-1}\sqrt{p_iq_i}\right]^2\\
&\leq\frac{1}{K}\left[\sum_{i=0}^{K-1}\sqrt{\hat{p}_i\hat{q}_i}\right]^2\leq\frac{1}{K},
\end{align}
where we have used that the classical fidelity between two probability distributions is upper bounded by $1$. This establishes the theorem for pure biseparable states with respect to arbitrary bipartitions. Noting that every mixed biseparable state $\sigma_{\vv{KS}}\in\BS(:\vv{KS}:)$ can be expressed as a convex sum of pure biseparable states finishes the proof.
\end{proof}


\section{Upper bounds on the CKA rates of multiplex channels}
\subsection{Proof of Theorem~\ref{thm:one-shot-cppp}}\label{app:one-shot-cppp}
\begin{proof}
Let us consider any cppp-assisted protocol that achieves a rate $\hat{P}_{\textnormal{cppp}}^{\mc{N}}\equiv \hat{P}$. Let $\rho^{(1)}\in\FS(:\!\vv{LA'}\!:\!\vv{RB}\!:\!\vv{P}\!:)$ be a fully separable state generated by the first use of LOCC among all spatially separated allies. Let 
\begin{equation}
\tau^{(1)}_{\vv{LA}\vv{R}\vv{PC}}\coloneqq \mc{N}(\rho^{(1)}_{\vv{LA}\vv{RB}\vv{P}}).
\end{equation}  
We note that $\tau^{(1)}$ is a separable state with respect to bipartition $\vv{LA}\vv{R}\vv{C}\!:\!\vv{P}$. The action of the decoder channel $\mc{D}\coloneqq \mc{L}^{(2)}_{\vv{LA}\vv{R}\vv{PC}\to \vv{SK}}$ on $\tau^{(1)}$ yields the state
\begin{equation}
\omega_{\vv{SK}}\coloneqq \mc{L}^{(2)}(\tau^{(1)}_{\vv{LA}\vv{R}\vv{PC}}).
\end{equation}
By assumption we have that
\begin{equation}
F(\gamma_{\vv{SK}}, \omega_{\vv{SK}})\geq 1-\varepsilon,
\end{equation}
for some ($M$-partite) private state $\gamma$. This implies that there exists a projector $\Pi^\gamma_{\vv{SK}}$ corresponding to a $\gamma$-privacy test such that (see Proposition~\ref{thm:priv-test})
\begin{equation}
\Tr[\Pi^\gamma_{\vv{SK}}\omega_{\vv{SK}}]\geq 1-\varepsilon. 
\end{equation}
From Theorem~\ref{thm:priv-gme},
\begin{equation}
\Tr[\Pi^\gamma_{\vv{SK}}\sigma'_{\vv{SK}}]\leq \frac{1}{K}=2^{-\hat{P}},
\end{equation}
for any $\sigma'\in\BS(:\!\vv{SK}\!:)$. 

Let us suppose a state $\sigma_{\vv{LA}\vv{R}\vv{PC}}\in\BS(:\vv{LA}:\vv{R}:\vv{PC}:)$ of the form $\sigma_{\vv{LA}\vv{R}\vv{PC}}=\sigma_{\vv{LA}\vv{R}\vv{C}}\otimes\sigma_{\vv{P}}$, where $\sigma{\vv{LA}\vv{R}\vv{C}}$ is arbitrary. It holds that $\sigma_{\vv{SK}}\coloneqq \mc{L}^{(2)}(\sigma_{\vv{LA}\vv{R}\vv{PC}})\in\BS(:\vv{SK}:)$. Thus, the privacy test is feasible for $D^\varepsilon_h(\omega\Vert \sigma)$ and we find that
\begin{align}
\hat{P} &\leq D^\varepsilon_h(\omega_{\vv{SK}}\Vert \sigma_{\vv{SK}})\\
&\leq  D^\varepsilon_h(\tau^{(1)}_{\vv{LA}\vv{R}\vv{PC}})\Vert \sigma_{\vv{LA}\vv{R}\vv{PC}} )\label{eq:40}\\
&\leq \sup_{\psi\in\FS(\vv{LA'}:\vv{RB}:\vv{P})}D^\varepsilon_h(\mc{N}(\psi_{\vv{LA'}:\vv{RB}:\vv{P}})\Vert \sigma_{\vv{LA}\vv{R}\vv{PC}})\\
&= \sup_{\psi\in\FS(\vv{LA'}:\vv{RB})}D^\varepsilon_h(\mc{N}(\psi_{\vv{LA'}:\vv{RB}})\Vert \sigma_{\vv{LA}\vv{R}\vv{C}}).\label{eq:one-shot-proof-part}
\end{align}

The second inequality follows from data processing inequality. The third inequality follows from the quasi-convexity of $D^\varepsilon_h$. The equality follows from Eq.~\eqref{eq:gen-div-prod} and suitable choice of $\sigma_{\vv{P}}$ that always exists because for any pure state $\psi\in\FS(\vv{LA'}:\vv{RB}:\vv{P})$, the output state $\mc{N}(\psi)$ is separable with respect to bipartition $\vv{LA}\vv{R}\vv{C}:\vv{P}$. 

Since inequality \eqref{eq:one-shot-proof-part} also holds for an arbitrary $\sigma\in\BS(:\!\vv{LA}\!:\!\vv{R}\!:\vv{C}\!:)$, we can conclude that
\begin{equation}
\hat{P}\leq E^\varepsilon_{h,\GE}(\mc{N}).
\end{equation} 
\end{proof}

\subsection{Proof of Theorem~\ref{thm:emax-key-converse}}\label{app:emax-key-converse}
\begin{proof}
The following inequality holds for an $(n,K,\varepsilon)$ LOCC-assisted secret key agreement protocol over a multiplex channel $\mc{N}$:
\begin{equation}\label{eq:fid-conv-proof}
F(\omega_{\vv{SK}},\gamma_{\vv{SK}})\geq 1-\varepsilon. 
\end{equation}
For any $\sigma_{\vv{SK}}\in\FS(:\vv{SK}:)$, we have following bound due to inequality \eqref{eq:fid-conv-proof} and Theorem~\ref{thm:priv-gme}:
\begin{equation}
\log_2 K\leq D^\varepsilon_h(\omega_{\vv{SK}}\Vert\sigma_{\vv{SK}}). 
\end{equation}
Employing inequality \eqref{eq:hyp-sand-rel} in the limit $\alpha\to+\infty$, we obtain
\begin{align}
\log_2 K &\leq D^\varepsilon_h(\omega_{\vv{SK}}\Vert\sigma_{\vv{SK}}) \\
&\leq D_{\max}(\omega_{\vv{SK}}\Vert\sigma_{\vv{SK}})+\log_2\left(\frac{1}{1-\varepsilon}\right).
\end{align}
Above inequality holds for arbitrary $\sigma\in\FS(:\vv{SK}:)$, therefore
\begin{equation}\label{eq:con-max-proof}
\log_2 K \leq E_{\max,\E}(:\vv{SK}:)_{\omega}+\log_2\left(\frac{1}{1-\varepsilon}\right),
\end{equation}
where $E_{\max,\E}(:\vv{SK}:)_{\omega}$ is the max-relative entropy of entanglement of the state $\omega_{\vv{SK}}$.

The max-relative entropy of entanglement $E_{\max,\E}$ of a state is monotonically non-increasing under the action of LOCC channels and it is zero for states that are fully separable. Using these facts, we get that
\begin{align}
& E_{\max,\E}(:\vv{SK}:)_{\omega}\nonumber \\
&\leq E_{\max,\E}(:\vv{L^{(n)}A^{(n)}}:\vv{R^{(n)}}:\vv{P^{(n)}C^{(n)}}:)_{\tau_n} \label{eq:52}\\
&=  E_{\max,\E}(:\vv{L^{(n)}A^{(n)}}:\vv{R^{(n)}}:\vv{P^{(n)}C^{(n)}}:)_{\tau_n} \nonumber\\
&\qquad - E_{\max,\E}(:\vv{L^{(1)}A^{(1)'}}:\vv{R^{(1)}B^{(1)}}:\vv{P^{(1)}}:)_{\rho_1}\\ 
&=E_{\max,\E}(:\vv{L^{(n)}A^{(n)}}:\vv{R^{(n)}}:\vv{P^{(n)}C^{(n)}}:)_{\tau_n} \nonumber\\ 
&\qquad +\left[\sum_{i=2}^n E_{\max,\E}(:\vv{L^{(i)}A^{(i)'}}:\vv{R^{(i)}B^{(i)}}:\vv{P^{(i)}}:)_{\rho_{i}} \right.\nonumber \\
&\qquad  -\left. \sum_{i=2}^n E_{\max,\E}(:\vv{L^{(i)}A^{(i)'}}:\vv{R^{(i)}B^{(i)}}:\vv{P^{(i)}}:)_{\rho_{i}} \right]\nonumber \\
&\qquad - E_{\max,\E}(:\vv{L^{(1)}A^{(1)'}}:\vv{R^{(1)}B^{(1)}}:\vv{P^{(1)}}:)_{\rho_{1}}\\
& \leq \sum_{i=1}^n \left[ E_{\max,\E}(:\vv{L^{(i)}A^{(i)}}:\vv{R^{(i)}}:\vv{P^{(i)}C^{(i)}}:)_{\tau_{i}} \right. \nonumber\\
&\qquad -\left.  E_{\max,\E}(:\vv{L^{(i)}A^{(i)'}}:\vv{R^{(i)}B^{(i)}}:\vv{P^{(i)}}:)_{\rho_{i}} \right]\\
&\leq n E_{\max,\E}(\mc{N}).\label{eq:amor-proof}
\end{align}
The first equality follows because $E_{\max,\E}(:\vv{L^{(1)}A^{(1)'}}:\vv{R^{(1)}B^{(1)}}:\vv{P^{(1)}}:)_{\rho_{1}}=0$. The second inequality follows because $E_{\max,\GE}$ is monotone under LOCC channels and $\rho_{i}=\mc{L}^{i}(\tau_{i-1})$ for all $i\in\{2,3,\ldots,n\}$. The final inequality follows from Lemma~\ref{lem:max-amor}.

From inequalities \eqref{eq:con-max-proof} and \eqref{eq:amor-proof}, we conclude that
\begin{equation}
\log_2K\leq nE_{\max,\E}(\mc{N})+\log_2\left(\frac{1}{1-\varepsilon}\right).
\end{equation}
\end{proof}

\subsection{Proof of Theorem~\ref{thm:ree-key-converse}}
\begin{proof}\label{app:ree-key-converse}
For an $(n,K,\varepsilon)$ LOCC-assisted secret key agreement protocol over a multiplex channel $\mc{N}$, such that $F(\omega_{\vv{SK}},\gamma_{\vv{SK}})\geq 1-\varepsilon$, due to inequality \eqref{eq:fid-conv-proof} and Theorem~\ref{thm:priv-gme} it holds for any $\sigma_{\vv{SK}}\in\FS(:\vv{SK}:)$:
\begin{equation}\label{eq:68}
\log_2 K\leq D^\varepsilon_h(\omega_{\vv{SK}}\Vert\sigma_{\vv{SK}}). 
\end{equation}
Using the fact that~\cite{WR12}
\begin{equation}
D^\varepsilon_h(\omega_{\vv{SK}}\Vert\sigma_{\vv{SK}}) \leq \frac{1}{1-\varepsilon}\(D(\omega_{\vv{SK}}\Vert\sigma_{\vv{SK}})+h(\varepsilon)\),
\end{equation}
where $h$ is the binary entropy function, and that the bound~\eqref{eq:68} holds for arbitrary $\sigma\in\FS(:\vv{SK}:)$, we obtain
\begin{equation}\label{eq:con-ree-proof-1}
\log_2 K \leq \frac{1}{1-\varepsilon}\(E_{\E}(:\vv{SK}:)_{\omega}+h(\varepsilon)\)\end{equation}
As the relative entropy of entanglement of a state is monotonically non-increasing under the action of LOCC channels and vanishes for states that are fully separable we can repeat the argument in inequalities~\eqref{eq:52}--\eqref{eq:amor-proof} and obtain
\begin{equation}\label{eq:amor-proof_ree}
E_{\E}(:\vv{SK}:)_{\omega}\leq n E^p_{\E}(\mc{N})\leq n E^\infty_{\E}(\mc{N}),
\end{equation}
where the second inequality follows from Lemma~\ref{lem:ree-amor}, Taking the limits $\varepsilon\to0$ and $n\to\infty$, we obtain \begin{equation}
\hat{\mc{P}}_{\LOCC}(\mc{N})\leq E^\infty_{\E}(\mc{N}),
\end{equation}
showing the converse. As for the strong converse, we follow the argument used in Ref.~\cite{WTB16}: From inequalities~\eqref{eq:68} and  \eqref{eq:hyp-sand-rel} we obtain
\begin{equation}\label{eq:con-ree-proof}
\log_2 K \leq \widetilde{E}_{\alpha,\E}(:\vv{SK}:)_{\omega}+\frac{\alpha}{\alpha-1}\log_2\left(\frac{1}{1-\varepsilon}\right),
\end{equation}
where $\alpha\in(1,\infty)$ and $\widetilde{E}_{\alpha,\E}(:\vv{SK}:)_{\omega}$ is the sandwiched R\'{e}nyi relative entropy of entanglement of the state $\omega_{\vv{SK}}$. Rewriting inequality~\eqref{eq:con-ree-proof} we obtain
\begin{equation}
\varepsilon\geq 1- 2^{-n\(\frac{\alpha-1}{\alpha}\)\(\frac{\log_2 K}{n}-\frac{1}{n}\widetilde{E}_{\alpha,\E}(:\vv{SK}:)_{\omega}\)}.
\end{equation}
Assuming that the rate $\frac{\log_2 K}{n}$ exceeds $E_{E}^\infty(\mc{N})$, by inequality~\eqref{eq:amor-proof_ree} it will be larger than $\frac{1}{n}E_{\E}(:\vv{SK}:)_{\omega}$. Hence, there exists an $\alpha>1$, such that $\frac{\log_2 K}{n}-\frac{1}{n}\widetilde{E}_{\alpha,\E}(:\vv{SK}:)_{\omega}>0$ and the error increases to $1$ exponentially.
\end{proof}
\subsection{Proof of Theorem~\ref{thm:bicov}}\label{app:tele-cov}
Let $\mc{N}_{\vv{A'}\vv{B}\to\vv{A}\vv{C}}$ be a multipartite 
quantum channel that is tele-covariant with respect to groups $\{\msc{G}_a\}_{a\in\msc{A}}$ and $\{\msc{G}_b\}_{b\in\msc{B}}$ as defined in Section~\ref{sec:tele-cov}.
By definition, for all $a\in\msc{A}$ and $b\in\msc{B}$, we have
\begin{align}
&\frac{1}{G_a}\sum_{g_a}\mathcal{U}_{A''_a}^{g_a}(\Phi^+_{A''_aL_a})=\frac{\bbm{1}_{A''_a}}{|A''_a|}\otimes\frac{\bbm{1}_{L_a}}{|L_a|},\label{eq:max-ent-cov-action_A}\\
&\frac{1}{G_b}\sum_{g_b}\mathcal{U}_{B'_b}^{g_b}(\Phi^+_{B'_bR_b})=\frac{\bbm{1}_{B'_b}}{|B'_b|}\otimes\frac{\bbm{1}_{R_b}}{|R_b|},\label{eq:max-ent-cov-action_B}
\end{align}
respectively, where $A''_a\simeq L_a$, $B'_b\simeq R_b$ and $\Phi^+$ denotes an EPR state. Note that in order for each $\{U_{A''_a}^{g_a}\}$ and $\{U_{B'_b}^{g_b}\}$ to be one-designs, it is necessary that $\left\vert A''_a\right\vert ^{2}\leq G_a $ and $\left\vert B'_b\right\vert ^{2}\leq G_b $ \cite{AMTW00}.

For every $a\in\msc{A}$ and every $b\in\msc{B}$, we can now define $\{E_{A''_aL_a}^{g_a}\}_{g_a}$ and $\{E_{B_b'R_b}^{g_b}\}_{g_b}$,
with respective elements defined as
\begin{align}
&E_{A''_aL_a}^{g_a}:=\frac{\left\vert A'_a\right\vert
^{2}}{G_a }U_{A''_a}^{g_a}\Phi^+_{A''_aL_a}\left(  U_{A''_a}^{g_a}\right)  ^{\dag},\\
&E_{B'_bR_b}^{g_b}:=\frac{\left\vert B_b\right\vert
^{2}}{G_b }U_{B'_b}^{g_b}\Phi^+_{B'_bR_b}\left(  U_{B'_b}^{g_b}\right)  ^{\dag},
\end{align}
where $A'_a\simeq A''_a$ and $B_b\simeq B'_b$. It follows from the fact that $\left\vert A'_a\right\vert ^{2}
\leq G_a$ and $\left\vert B_b\right\vert ^{2}
\leq G_b$ as well as  \eqref{eq:max-ent-cov-action_A} and \eqref{eq:max-ent-cov-action_B} that
$\{E_{A''_aL_a}^{g_a}\}_{g_a}$ and $\{E_{B_b'R_b}^{g_b}\}_{g_b}$ are valid POVMs for all $a\in\msc{A}$ and $b\in\msc{B}$.

The simulation of the channel $\mc{N}_{\vv{A'}\vv{B}\to\vv{A}\vv{C}}$ via teleportation begins with a state $\rho_{\vv{A''}\vv{B'}}$ and a shared resource $\theta_{\vv{LA}\vv{R}\vv{C}}=\mc{N}_{\vv{A'}\vv{B}\to\vv{A}\vv{C}}\(\Phi^+_{\vv{L}\vv{R}|\vv{A'}\vv{B}}\)$. The desired outcome is for the receivers to receive
the state $\mathcal{N}(\rho_{\vv{A''}\vv{B'}})$ and for the protocol to work independently of the input state
$\rho_{\vv{A''}\vv{B'}}$. The first step is for senders $\tf{A}_a$ and $\tf{B}_b$ to locally
perform the measurement $\left\{\bigotimes_{a\in\msc{A}}E_{A''_aL_a}^{g_a}\otimes\bigotimes_{b\in\msc{B}}E_{B_b'R_b}^{g_b}\right\}_{\vv{g}}$ and then send the outcomes $\vv{g}$ to the
receivers. Based on the outcomes $\vv{g}$, the receivers $\tf{A}_a$ and $\tf{C}_c$ then perform
$W_{A_a}^{\vv{g}}$ and $W_{C_c}^{\vv{g}}$, respectively. The following analysis demonstrates that this
protocol works, by simplifying the form of the post-measurement state:

\begin{widetext}
\begin{align}
&  \left(\prod_{a\in\msc{A}}G_a\prod_{b\in\msc{B}}G_b\right) \Tr_{\vv{A''L}\vv{B'R}}\left[\left(\bigotimes_{a\in\msc{A}}E_{A''_aL_a}^{g_a}\otimes\bigotimes_{b\in\msc{B}}E_{B_b'R_b}^{g_b}\right)(\rho_{\vv{A"}\vv{B'}}\otimes\theta_{\vv{LA}\vv{R}\vv{C}})\right]\nonumber\\
&  =\left(\prod_{a\in\msc{A}}\left\vert A'_a\right\vert ^{2} \prod_{b\in\msc{B}}\left\vert B_b\right\vert ^{2}\right)\Tr_{\vv{A''L}\vv{B'R}}\left[\left(\bigotimes_{a\in\msc{A}}{U}_{A''_a}^{g_a}\Phi^+_{A''_aL_a}{{U}_{A''_a}^{g_a}}^\dag\otimes\bigotimes_{b\in\msc{B}}{U}_{B_b'}^{g_b}\Phi^+_{B_b'R_b}{{U}_{B_b'}^{g_b}}^\dag\right)(\rho_{\vv{A"}\vv{B'}}\otimes\theta_{\vv{LA}\vv{R}\vv{C}})\right]\\
&  =\left(\prod_{a\in\msc{A}}\left\vert A'_a\right\vert ^{2} \prod_{b\in\msc{B}}\left\vert B_b\right\vert ^{2}\right)\bra{\Phi^+}_{\vv{A''}\vv{B'}|\vv{L}\vv{R}}\left( \bigotimes_{a\in\msc{A}} U_{A''_a}^{g_a}\otimes \bigotimes_{b\in\msc{B}} U_{B_b'}^{g_b}\right) ^{\dag}\rho_{\vv{A"}\vv{B'}}\otimes\theta_{\vv{LA}\vv{R}\vv{C}}\left( \bigotimes_{a\in\msc{A}} U_{A''_a}^{g_a}\otimes \bigotimes_{b\in\msc{B}} U_{B_b'}^{g_b}\right)\ket{\Phi^+}_{\vv{A''}\vv{B'}|\vv{L}\vv{R}}\\
&  =\left(\prod_{a\in\msc{A}}\left\vert A'_a\right\vert ^{2} \prod_{b\in\msc{B}}\left\vert B_b\right\vert ^{2}\right)\bra{\Phi^+}_{\vv{A''}\vv{B'}|\vv{L}\vv{R}}\left( \bigotimes_{a\in\msc{A}} U_{A''_a}^{g_a}\otimes \bigotimes_{b\in\msc{B}} U_{B_b'}^{g_b}\right) ^{\dag}\rho_{\vv{A"}\vv{B'}}\left( \bigotimes_{a\in\msc{A}} U_{A''_a}^{g_a}\otimes \bigotimes_{b\in\msc{B}} U_{B_b'}^{g_b}\right)\otimes\theta_{\vv{LA}\vv{R}\vv{C}}\ket{\Phi^+}_{\vv{A''}\vv{B'}|\vv{L}\vv{R}}\\
&  =\left(\prod_{a\in\msc{A}}\left\vert A'_a\right\vert ^{2} \prod_{b\in\msc{B}}\left\vert B_b\right\vert ^{2}\right)\bra{\Phi^+}_{\vv{A''}\vv{B'}|\vv{L}\vv{R}}\left[\left( \bigotimes_{a\in\msc{A}} U_{L_a}^{g_a}\otimes \bigotimes_{b\in\msc{B}} U_{R_b}^{g_b}\right) ^{\dag}\rho_{\vv{L}\vv{R}}\left( \bigotimes_{a\in\msc{A}} U_{L_a}^{g_a}\otimes \bigotimes_{b\in\msc{B}} U_{R_b}^{g_b}\right)\right]^\ast\theta_{\vv{LA}\vv{R}\vv{C}}\ket{\Phi^+}_{\vv{A''}\vv{B'}|\vv{L}\vv{R}}
.\label{eq:cov-tp-simul-block-1}%
\end{align}
\end{widetext}
The first three equalities follow by substitution and some rewriting. The
fourth equality follows from the fact that%

\begin{equation}
\langle\Phi|_{A^{\prime}A}M_{A^{\prime}}=\langle\Phi|_{A^{\prime}A}M_{A}%
^{\ast}\label{eq:ricochet-prop}%
\end{equation}
for any operator $M$ and where $\ast$ denotes the complex conjugate, taken
with respect to the basis in which $\ket{\Phi}_{A^{\prime}A}$ is defined. 
Continuing, we have that

\begin{widetext}
\begin{align}
\text{Eq.~}\eqref{eq:cov-tp-simul-block-1} & =\left(\prod_{a\in\msc{A}}\left\vert A'_a\right\vert  \prod_{b\in\msc{B}}\left\vert B_b\right\vert \right)\Tr_{\vv{L}\vv{R}}\left[\left[\left( \bigotimes_{a\in\msc{A}} U_{L_a}^{g_a}\otimes \bigotimes_{b\in\msc{B}} U_{R_b}^{g_b}\right) ^{\dag}\rho_{\vv{L}\vv{R}}\left( \bigotimes_{a\in\msc{A}} U_{L_a}^{g_a}\otimes \bigotimes_{b\in\msc{B}} U_{R_b}^{g_b}\right)\right]^\ast\mc{N}_{\vv{A'}\vv{B}\to\vv{A}\vv{C}}\(\Phi^+_{\vv{L}\vv{R}|\vv{A'}\vv{B}}\)\right]\\
& =\left(\prod_{a\in\msc{A}}\left\vert A'_a\right\vert  \prod_{b\in\msc{B}}\left\vert B_b\right\vert\right)\Tr_{\vv{L}\vv{R}}\left[\mc{N}_{\vv{A'}\vv{B}\to\vv{A}\vv{C}}\left(\left( \bigotimes_{a\in\msc{A}} U_{{A'}_a}^{g_a}\otimes \bigotimes_{b\in\msc{B}} U_{{B}_b}^{g_b}\right) ^{\dag}\rho_{\vv{{A'}}\vv{{B}}}\left( \bigotimes_{a\in\msc{A}} U_{{A'}_a}^{g_a}\otimes \bigotimes_{b\in\msc{B}} U_{{B}_b}^{g_b}\right)\Phi^+_{\vv{L}\vv{R}|\vv{A'}\vv{B}}\right)\right]\\
&= \mc{N}_{\vv{A'}\vv{B}\to\vv{A}\vv{C}}\left(\left( \bigotimes_{a\in\msc{A}} U_{{A'}_a}^{g_a}\otimes \bigotimes_{b\in\msc{B}} U_{{B}_b}^{g_b}\right) ^{\dag}\rho_{\vv{{A'}}\vv{{B}}}\left( \bigotimes_{a\in\msc{A}} U_{{A'}_a}^{g_a}\otimes \bigotimes_{b\in\msc{B}} U_{{B}_b}^{g_b}\right)\right)\\
&=\(\bigotimes_{a\in\msc{A}}{W}_{A_a}^{\vv{g}}\otimes\bigotimes_{c\in\msc{C}}{W}_{C_c}^{\vv{g}}\)^\dag\mc{N}_{\vv{A'}\vv{B}\to\vv{A}\vv{C}}(\rho_{\vv{A'}\vv{B}})\(\bigotimes_{a\in\msc{A}}{W}_{A_a}^{\vv{g}}\otimes\bigotimes_{c\in\msc{C}}{W}_{C_c}^{\vv{g}}\).
\end{align}
\end{widetext}
The first equality follow because $\left\vert A\right\vert \langle
\Phi|_{A^{\prime}A}\left(  \bbm{1}_{A^{\prime}}\otimes M_{AB}\right)  |\Phi
\rangle_{A^{\prime}A}=\operatorname{Tr}_{A}\{M_{AB}\}$ for any operator
$M_{AB}$. The second equality follows by applying the conjugate transpose of
\eqref{eq:ricochet-prop}. The final equality follows from the covariance
property of the channel.

Thus, if the receivers finally perform the unitaries $\bigotimes_{a\in\msc{A}}{W}_{A_a}^{\vv{g}}\otimes\bigotimes_{c\in\msc{C}}{W}_{C_c}^{\vv{g}}$ upon receiving $\vv{g}$ via a classical channel from the
senders, then the output of the protocol is $\mc{N}_{\vv{A'}\vv{B}\to\vv{A}\vv{C}}(\rho_{\vv{A'}\vv{B}})$, so that this
protocol simulates the action of the multipartite channel $\mathcal{N}$ on the state $\rho$.
\qed

\subsection{Proof of Theorem \ref{theo:telecov-weakconv-GE}}
Before proving Theorem \ref{theo:telecov-weakconv-GE}, we need the following lemma, which generalizes Lemma 7 in \cite{HHHO09}:
\begin{lemma}\label{bisep:lemma7}
Let $\mc{T}=\{{U^{tw}}^\dagger\rho_{\vv{SK}}U^{tw}:\rho_{\vv{SK}}\in \BS(:\vv{SK}:)\}$ be the set of twisted biseparable states. Then for any $\sigma_{\vv{SK}}\in\mc{T}$ it holds
\begin{equation}
D(\Phi_{\vv{K}}||\sigma_{\vv{K}})\geq\log{K}.
\end{equation}
\end{lemma}

\begin{proof}
Let $\sigma_{\vv{SK}}\in\mc{T}$, i.e. $\sigma_{\vv{SK}}={U^{tw}}^\dagger\rho_{\vv{SK}}U^{tw}$ for some twisting unitary $U^{tw}$ and bisparable $\rho_{\vv{SK}}$. $U^{tw}$ defines a privacy test $\Pi^\gamma_{\vv{SK}}=U^{tw}(\Phi_{\vv{K}}\otimes\bbm{1}_{\vv{S}}){U^{tw}}^\dagger$. By Theorem \ref{theo:telecov-weakconv-GE}  it then holds
\begin{equation}
\Tr[\Phi_{\vv{K}}\sigma_{\vv{K}}]=\Tr[\Pi^\gamma_{\vv{SK}}\rho_{\vv{SK}}]\leq\frac{1}{K}.
\end{equation}
By the concavity of the logarithm it then holds
\begin{align}
D(\Phi_{\vv{K}}||\sigma_{\vv{K}})&=-S(\Phi_{\vv{K}})-\tr[\Phi_{\vv{K}} \log{\sigma_{\vv{K}}}]\\
&\geq-\log\Tr[\Phi_{\vv{K}}\sigma_{\vv{K}}]\\
&\geq \log K,
\end{align}
finishing the proof.
\end{proof}

Now we can follow \cite{PLOB15} to prove Theorem \ref{theo:telecov-weakconv-GE}:

\begin{proof}[Proof of Theorem \ref{theo:telecov-weakconv-GE}]
Let $\epsilon>0$ and $n\in\mathbb{N}$. We begin by noting that in the case of teleportation-simulable multiplex channels LOCC-assistance does not enhance secret-key-agreement capacity, and the original protocol can be reduced to a cppp-assisted secret key agreement protocol \cite{PLOB15}. Namely, in every round $1\leq i \leq n$ it holds 
\begin{align}
\rho_i&=\mc{L}^i(\tau_i)=\mc{L}^i(\mc{N}_{\vv{{A'}^{(i-1)}}\vv{B^{(i-1)}}\to\vv{A^{(i-1)}}\vv{C^{(i-1)}}}(\rho_{i-1}))\\
&=\mc{L}^i(\mc{T}_{\vv{{A'}^{(i-1)}LA}\vv{B^{(i-1)}R}\vv{C}\to\vv{A^{(i-1)}}\vv{C^{(i-1)}}}(\theta_{\vv{LA}\vv{R}\vv{C}}\otimes\rho_{i-1})),
\end{align}
where $\mc{L}^i$ and $\mc{T}$ are LOCC. As the initial state $\rho_0$ is assumed to be fully separable, we can get that the final state $\omega_{\vv{SK}}=\rho_n$ of an adaptive LOCC CKA-protocol involving $n$ uses of teleportation-simulable multiplex channel $\mc{N}_{\vv{A'}\vv{B}\to\vv{A}\vv{C}}$, can be expressed as
\begin{equation}
\omega_{\vv{SK}}=\mc{L}_{\vv{L^nA^n}\vv{R^n}\vv{C^n}\to\vv{SK}}\(\theta^{\otimes n}_{\vv{LA}\vv{R}\vv{C}}\),
\end{equation}
where $\mc{L}$ is an LOCC operation with respect to the partition $:\vv{L^nA^n}:\vv{R^n}:\vv{C^n}:$. 
By assumption it holds $\|\omega_{\vv{SK}}-\gamma_{\vv{KS}}\|_1\leq\epsilon$ for some $m$-partite private state $\gamma_{\vv{KS}}=U^{tw}(\Phi_{\vv{K}}\otimes\tau_{\vv{S}}){U^{tw}}^\dagger$, where $m$ is the number of parties. Let $\tilde{\sigma}_{\vv{L^nA^n}\vv{R^n}\vv{C^n}}\in\BS(:\vv{L^nA^n}:\vv{R^n}:\vv{C^n}:)$. Following the proof of Theorem 9 in \cite{HHHO09}, we obtain 
\begin{align}
&D(\theta^{\otimes n}_{\vv{LA}\vv{R}\vv{C}}||\tilde{\sigma}) \nonumber\\
&\quad \geq D(\omega_{\vv{SK}}||\mc{L}(\tilde{\sigma}_{\vv{L^nA^n}\vv{R^n}\vv{C^n}}))\\
&~~~=D({U^{tw}}^\dagger\omega_{\vv{SK}}{U^{tw}}||{U^{tw}}^\dagger\mc{L}(\tilde{\sigma}_{\vv{L^nA^n}\vv{R^n}\vv{C^n}}){U^{tw}})\\
&\quad \geq\inf_{\sigma_{\vv{SK}}\in\mc{T}}D(\Tr_{\vv{S}}[{U^{tw}}^\dagger\omega_{\vv{SK}}{U^{tw}}]||\sigma_{\vv{K}})\\
&\quad \geq\inf_{\sigma_{\vv{SK}}\in\mc{T}}D(\Phi_{\vv{K}}||\sigma_{\vv{K}})-4m\epsilon\log{K}-h(\epsilon)\\
&\quad \geq(1-4m\epsilon)\log{K}-h(\epsilon),
\end{align}
where in the last two inequalities we have used the asymptotic continuity of the relative entropy and Lemma \ref{bisep:lemma7}, respectively. Letting $n\to\infty$ and $\epsilon\to0$ finishes the proof.

\end{proof}

\subsection{Proof of Theorem \ref{theo:E-alpha_bound}}
As in the proof of Theorem \ref{theo:telecov-weakconv-GE} we have 
\begin{equation}
\omega_{\vv{SK}}=\mc{L}_{\vv{L^nA^n}\vv{R^n}\vv{C^n}\to\vv{SK}}\(\theta^{\otimes n}_{\vv{LA}\vv{R}\vv{C}}\),
\end{equation}
where $\mc{L}$ is an LOCC operation with respect to the partition $:\vv{L^nA^n}:\vv{R^n}:\vv{C^n}:$. Now following the proof of Theorem \ref{thm:one-shot-cppp}, we have that
\begin{equation}
F(\gamma_{\vv{SK}}, \omega_{\vv{SK}})\geq 1-\varepsilon,
\end{equation}
for some private state $\gamma$, hence there exists a projector $\Pi^\gamma_{\vv{SK}}$ corresponding to a $\gamma$-privacy test such that (see Proposition~\ref{thm:priv-test})
\begin{equation}
\Tr[\Pi^\gamma_{\vv{SK}}\omega_{\vv{SK}}]\geq 1-\varepsilon. 
\end{equation}
On the other hand, from Theorem~\ref{thm:priv-gme}, we have
\begin{equation}
\Tr[\Pi^\gamma_{\vv{SK}}\sigma_{\vv{SK}}]\leq \frac{1}{K},
\end{equation}
for any $\sigma\in\FS(:\!\vv{SK}\!:)$. Let us suppose a state $\sigma'_{\vv{LA}\vv{R}\vv{C}}\in\FS(:\vv{LA}:\vv{R}:\vv{C}:)$ and let us define $\sigma_{\vv{SK}}=\mc{L}_{\vv{L^nA^n}\vv{R^n}\vv{C^n}\to\vv{SK}}\({\sigma'}^{\otimes n}_{\vv{LA}\vv{R}\vv{C}}\)$, which is in $\FS(:\!\vv{SK}\!:)$. Hence for all $\alpha>1$ it holds,

\begin{align}
& \log_2 K\nonumber\\ 
&\leq D^\varepsilon_h\(\omega_{\vv{SK}}\Vert \sigma_{\vv{SK}}\)\\
&\leq D^\varepsilon_h\(\theta^{\otimes n}_{\vv{LA}\vv{R}\vv{C}}\Vert{\sigma'}^{\otimes n}_{\vv{LA}\vv{R}\vv{C}}\)\\
&\leq \widetilde{D}_\alpha\(\theta^{\otimes n}_{\vv{LA}\vv{R}\vv{C}}\Vert{\sigma'}^{\otimes n}_{\vv{LA}\vv{R}\vv{C}}\)+\frac{\alpha}{\alpha-1}\log_2\(\frac{1}{1-\epsilon}\)\\
&= n\widetilde{D}_\alpha\(\theta_{\vv{LA}\vv{R}\vv{C}}\Vert{\sigma'}_{\vv{LA}\vv{R}\vv{C}}\)+\frac{\alpha}{\alpha-1}\log_2\(\frac{1}{1-\epsilon}\)
\end{align}
The first inequality holds for any $\sigma\in\FS(\vv{SK})$. The second inequality follows from data processing inequality. The third inequality follows from eq. (\ref{eq:hyp-sand-rel}). The equality is due to the additivity of $\widetilde{D}_\alpha$ \cite{MDSFT13}. As the above holds for any $\sigma'_{\vv{LA}\vv{R}\vv{C}}\in\FS(:\vv{LA}:\vv{R}:\vv{C}:)$, we obtain Theorem~\ref{theo:E-alpha_bound}.

\qed

\section{Repeater as multipartite channel}\label{App:Repeater}
In order to provide bounds for more repeater protocols that involve two-way communication between Alice and Charlie or between Bob and Charlie before Charlie's measurement, we will have to slightly generalize our results in section \ref{sec:cka}. Namely, in addition to trusted parties $\{\tf{X}_i\}_{i=1}^M= \{\tf{A}_a\}_a\cup\{\tf{B}_b\}_b\cup\{\tf{C}_c\}_c$, we can add a number of cooperative but untrusted parties $\{\widetilde{\tf{X}}_i\}_{i=1}^{\widetilde{M}}\coloneqq \{\widetilde{\tf{A}}_{\tilde a}\}_{\tilde a\in \widetilde{\msc{A}}}\cup\{\widetilde{\tf{B}}_{\tilde b\in \widetilde{\msc{B}}}\}_{\tilde b}\cup\{\widetilde{\tf{C}}_{\tilde c\in \widetilde{\msc{C}}}\}_{\tilde c}$. 
Let us denote the quantum systems hold by respective untrusted parties as $\widetilde{A'}_{\tilde a}, \widetilde{L}_{\tilde a},\widetilde{A}_{\tilde a}, \widetilde{B}_{\tilde b}, \widetilde{R}_{\tilde b},\widetilde{C}_{\tilde c},\widetilde{P}_{\tilde c}$ and redefine 
\begin{align*}
&\vv{A'}\coloneqq \{A'_a\}_{a\in\msc{A}}\cup\{\widetilde{A'}_{\tilde{a}}\}_{\tilde{a}\in\widetilde{\msc{A}}},
\vv{A}\coloneqq \{A_a\}_{a\in\msc{A}}\cup\{\widetilde{A}_{\tilde{a}}\}_{\tilde{a}\in\widetilde{\msc{A}}},\\
&\vv{L}\coloneqq \{L_a\}_{a\in\msc{A}}\cup\{\widetilde{L}_{\tilde{a}}\}_{\tilde{a}\in\widetilde{\msc{A}}},\\
&\vv{B}\coloneqq \{B_b\}_{b\in\msc{B}}\cup\{\widetilde{B}_{\tilde{b}}\}_{\tilde{b}\in\widetilde{\msc{B}}},
\vv{R}\coloneqq \{R_b\}_{b\in\msc{B}}\cup\{\widetilde{R}_{\tilde{b}}\}_{\tilde{b}\in\widetilde{\msc{B}}},\\
&\vv{C}\coloneqq \{C_c\}_{c\in\msc{C}}\cup\{\widetilde{C}_{\tilde{c}}\}_{\tilde{c}\in\widetilde{\msc{C}}},
\vv{P}\coloneqq \{P_c\}_{c\in\msc{B}}\cup\{\widetilde{P}_{\tilde{c}}\}_{\tilde{c}\in\widetilde{\msc{C}}},
\end{align*} 
while keeping the old definitions for $\vv{K}$ and $\vv{S}$. We then assume we have a multiplex channel $\mc{N}_{\vv{A'}\vv{B}\to\vv{A}\vv{C}}$ and LOCC operations $\mc{L}^i$
, for $i=1,..,n$, among trusted and untrusted parties. However, we assume that as part of the last round of LOCC, $\mc{L}^{n+1}$ all subsystems belonging to untrusted parties are traced out, resulting in a state $\omega_{\vv{SK}}$ among the trusted parties only. It is now easy to show that the proofs of Theorems \ref{thm:emax-key-converse} and \ref{thm:ree-key-converse}  also go through in this slightly generalized scenario. Namely, tracing out parties in a fully separable state results in a fully separable state on the remaining parties, and by the monotonicity of the generalized divergences, inequalities~\eqref{eq:52} and \eqref{eq:amor-proof_ree} also hold if we trace out the untrusted parties in order to obtain $\omega$. Note that the same does not hold true in the case of Theorem \ref{thm:one-shot-cppp}, where we are dealing with the distance to the set of biseparable states, which is not preserved under traceout. 

Returning to the quantum key repeater, we can now identify Alice and Bob as two trusted parties and Charlie as an untrusted party and define a multiplex channel as the tensor product of the two channels from Alice to Charlie and Bob to Charlie, namely: $\mc{N}^{\text{repeater}}_{AB\to C}\coloneqq {\mc N}^1_{A\to C_A}\otimes{\mc N}^2_{B\to C_B}$, with $C\coloneqq C_AC_B$. We include the local state preparation by Alice and Bob, the LOCC performed by Alice, Charlie, and Bob during key distillation protocols, as well as Bob's entanglement-swapping measurement and subsequent classical communication into the LOCC operations that interleave the uses of $\mc{N}^{\text{repeater}}_{AB\to C}$. Crucially, the final LOCC operation has to include the trace-out of Charlie's system, as he is an untrusted party. 
 Application of the generalized versions of Theorem \ref{thm:emax-key-converse} or Theorem \ref{thm:ree-key-converse} then provides us with an upper bound on the achievable key rate in terms of $\min\{ E_{\max,E}(\mc{N}^{\text{repeater}}_{AB\to C}),E^{\infty}_{E}(\mc{N}^{\text{repeater}}_{AB\to C})\}$. As has been shown in Ref.~\cite{FF19}, there are examples of channels acting on finite-dimensional systems where the regularized relative entropy of entanglement is strictly less than max-relative entropy of entanglement, in which case Theorem~\ref{thm:ree-key-converse} provides tighter bounds than the ones provided in Ref.~\cite{CM17}. For tele-covariant channels, we can invoke Remark~\ref{rem:comp} and Theorem~\ref{thm:tele} to obtain bounds in terms of the relative entropy of entanglement.

Let us now consider repeater chains with more than a single repeater station. We assume a protocol where each channel has to be used the same number of times to get the desired fidelity. We consider Alice and Bob as trusted parties and the repeater stations $C_1,...,C_l$ as cooperative but untrusted parties. Defining a multiplex channel ${\mc N}^{\text{repeater chain}}_{AC'_1,...,C'_l\to C_1,...,C_lB}\coloneqq {\mc N}^1_{A\to C_1}\otimes{\mc N}^2_{C'_1\to C_2}\otimes...\otimes {\mc N}^{l}_{C'_{l-1}\to C'_l}\otimes{\mc N}^{l+1}_{C'_{l}\to B}$ and including entanglement purification and swapping operations of all nesting levels into the LOCC operations, we then apply Theorem \ref{thm:emax-key-converse} or Theorem \ref{thm:ree-key-converse} to bound the achievable key rate between Alice and Bob by $\min\{E_{\max,E}({\mc N}^{\text{repeater chain}}),E^\infty_{E}({\mc N}^{\text{repeater chain}})\}$. If involved channels are tele-covariant then we obtain bounds in terms of the relative entropy of entanglement. 

\section{Limitations on some MDI-QKD prototypes}\label{app:mdiqkd}
Following discussion in Section~\ref{sec:ex-mdi}, let us now consider MDI-QKD settings with noise model for transmission of qubit systems from both $\textbf{A}_{a_1}$ and $\textbf{A}_{a_2}$ to Charlie through qubit channels given by either depolarizing channel $\mc{D}^{l}_{A_i\to C_i}$ or dephasing channel $\mc{D}^s_{A_i\to C_i}$:
\begin{align}
   \mc{D}^l_{A_i\to C_i}(\rho_{A_i}) &= \lambda_{l}\rho_{C_i}+\frac{1-\lambda_l}{2}\bbm{1}_{C_i},\label{eq:ds}\\
    \mc{D}^s_{A_i\to C_i}(\rho_{A_i}) &= \lambda_{s}\rho_{C_i}+(1-\lambda_s)\widehat{Z}\rho_{C_i}\widehat{Z}^\dag,\label{eq:dl}
\end{align}
where 
\begin{align}
-\frac{1}{3}\leq \lambda_l\leq 1,\qquad\qquad 0\leq\lambda_s\leq 1,
\end{align} 
$\widehat{Z}$ is a Pauli-$Z$ operator, and $\rho$ is an arbitrary input state. Same as for MDI-QKD setup with erasure channels discussed earlier, we assume that Charlie can perform perfect Bell measurement $\mc{M}_{\vv{C}\to X}$ with probability $q$ and failure probability to be $1-q$. We notice that the multiplex channels $\mc{N}^{\text{MDI},\mc{D}^{l}}_{\vv{A}\to\vv{Z}}, \mc{N}^{\text{MDI},\mc{D}^{s}}_{\vv{A}\to\vv{Z}}$ for these MDI-QKD prototypes are also tele-covariant. This implies, the MDI-QKD capacities for respective MDI-QKD settings, i.e., with depolarizing channels and dephasing channels, to be upper bounded as (see following subsections for proofs and plots for some values of $q$):
\begin{enumerate}
    \item MDI-QKD with depolarizing channels $\mc{D}^{l}$~\eqref{eq:dl}, where $-\frac{1}{3}\leq \lambda_l\leq 1$, 
    \begin{align}\label{eq:mdi-dl}
        \widetilde{P}_{\LOCC}(\mc{N}^{\text{MDI},\mc{D}^{l}})\leq q\left(1-h_2\left(\frac{3}{4}\lambda_l^2+\frac{1}{4}\right)\right)
    \end{align}
    for $\frac{1}{\sqrt{3}}<\lambda_l\leq 1$, and $0$ otherwise.
    \item MDI-QKD with dephasing channels $\mc{D}^s$~\eqref{eq:ds}, where $0\leq\lambda_s\leq 1$,
      \begin{align}\label{eq:MDI-ds}
   & \widetilde{P}_{\LOCC}(\mc{N}^{\text{MDI},\mc{D}^{s}})\leq \nonumber\\
& \left\{
\begin{array}{ll}
q(1- h_2({1\over 2}p_-(\lambda_s)) & \textnormal{for}\quad \lambda_s > \frac{3}{4},\\
0 & \textnormal{for}\quad \frac{1}{4}\leq \lambda_s\leq \frac{3}{4},\\
q(1- h_2({1\over 2}p_-(1-\lambda_s)) & \textnormal{for}\quad \lambda_s< \frac{1}{4},\\
\end{array}
\right.
\end{align}
where $p_-(x)\coloneqq 4x^2-3x +1$.
    \end{enumerate}
    
\subsection{MDI-QKD via depolarizing channels}
In this section we will show a bound on MDI-QKD (or equivalently on a particular type of a quantum repeater). In the latter setup there are three stations: $A,B$ and an intermediate one $C\equiv C_AC_B$. We will consider the links $AC_A$ and $C_BB$ be {\it depolarising } channels $\mc{D}^s$ both with the same parameter $\lambda_l$, see~\eqref{eq:dl}. We consider also that the Bell measurement followed by communication of the results to both the parties happens only with probability $q$. With probability $(1-q)$ the state of $C$ is just traced out. We will call the multiplex channel for given MDI-QKD setup composed of depolarizing channels $\mc{D}^l$ with Bell measurement which happens with probability $q$ in total a $q$-depolarizing-MDIQKD channel. The upper bound which we derive below will demonstrate quantitatively that the operation of distillation of entanglement along the links does not commute with the operation of entanglement swapping. Indeed, even for $q=1$, if one does the Bell measurement first, the output key is zero for $\lambda_l\leq {1\over \sqrt{3}}$.

We are interested in the Choi-Jamiolkowski state of the $q$-depolarising-MDIQKD channel, which we obtain from the Choi states (up to local unitary as the input state is $\Psi^-$) of the two depolarising channels. The latter two states read $\lambda_l \Psi^- + (1-\lambda_l){\bbm{1} \over 4}$. The Choi state $\rho_{AB}^{out}$ reads
\begin{align}
& \rho_{AB}^{out} \coloneqq \nonumber\\
& {\lambda_l^2 q\over 4}\left[\Psi^-_{AB}\otimes \op{00}_{I_AI_B} +\Psi^+_{AB}\otimes\op{11}_{I_AI_B}+\right.\nonumber\\
&\left. \Phi^-_{AB}\otimes\op{22}_{I_AI_B}+ \Phi^+_{AB}\otimes\op{33}_{I_AI_B}\right]\otimes \op{00}_{I'_AI'_B}  \nonumber\\
&+ (1-\lambda_l^2) q {\bbm{1}_{AB}\over 4}\otimes {1\over 4}\sum_{i=0}^{3}\op{ii}_{I_AI_B} \otimes \op{00}_{I'_AI'_B} \nonumber \\ 
& + (1-q) {\bbm{1}_{AB}\over 4}\otimes \op{\perp}_{I_AI_B}\otimes \op{11}_{I'_AI'_B}.
\end{align} 

Let us examine this case. First, with probability $(1-q)$, the parties are left with the initial state on $AB$ which is $\bbm{1} \over 4$ and the ``flag'' $\op{11}_{I'_AI'_B}$ reporting error in the Bell measurement. With probability $q$ they obtain a flag $\op{00}_{I'_AI'_B}$, which informs that the Bell measurement was successful. They also receive the classical result
tof the Bell measurement was the outcome: $\{\op{ii}_{I_AI_B}\}_{i=0}^3$. Only with probability $\lambda_l^2$ this measurement results in output of appropriate Bell state on $AB$.
With probability $(1-\lambda_l^2)=(1-\lambda_l)\lambda_l+ \lambda_l(1-\lambda_l)+(1-\lambda_l)^2$ there happens one of three possibilities with respective probabilities: (i) teleportation of ${\bbm{1}_{C_B}\over 2}$ from $C_A$ to $A$ with probability $\lambda_l(1-\lambda_l)$, (ii) teleportation of ${\bbm{1}_{C_A}\over 2}$ from $C_B$ to $B$ with probability $(1-\lambda_l)\lambda_l$, and a Bell measurement on systems $C_AC_B$ of the state ${\bbm{1}_{AC_A}\over 4}\otimes {\bbm{1}_{C_BB}\over 4}$ followed by communication of the outcomes (with probability $(1-\lambda_l)^2$). As one can check by inspection, all the three operations result in the state ${\bbm{1}\over 4}$ on system $AB$.

\begin{figure}[ht]
        \centering
            \includegraphics[trim=4.5cm 7.5cm 4.5cm 6.5cm,width=6cm]{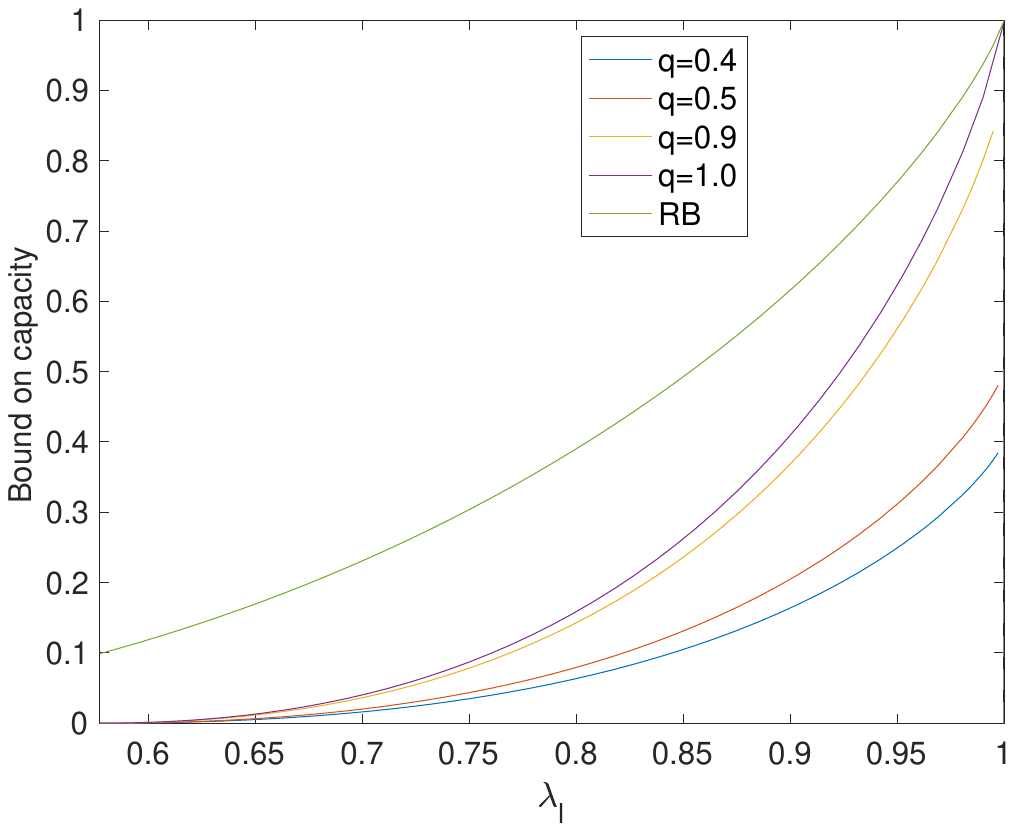}
              \caption{Upper bounds~\eqref{eq:mdi-dl} on the secret key capacities for the MDI-QKD protocol with depolarizing channels for different values of parameters $q$ and $\lambda_l$, in comparison to the RB bound \cite{PLOB15}.} \label{fig:MDI_depola}
    \end{figure}
The relative entropy of the $\rho_{AB}^{out}$ reads:
\begin{align}
E_R(\rho^{out}_{AB}) &\leq q E_R(\rho_{AB|00}^{out}) + (1-q)E_R(\rho_{AB|11}^{out})\\ 
& = q E_R(\rho^{out}_{AB|00}),
\label{eq:qbound}
\end{align}
where $\rho_{AB|11}^{out} = {\bbm{1}_{AB}\over 4}\otimes \op{\perp}_{I_AI_B}\otimes \op{11}_{I'_AI'_B}$ and $\rho_{AB|00}^{out}$ is such that
$(1-q)\rho_{AB|11}+q\rho_{AB|00} = \rho_{AB}$. We have used there convexity of the relative entropy and the fact that it is zero for 
maximally mixed state. We then observe that
\begin{align}
    & E_R(\rho_{AB|00}^{out}) \nonumber\\
   &\quad = E_R\Bigg(\bigg(\sum_{i=0}^3 \lambda_l^2 \op{\psi_i}_{AB} + (1-\lambda_l^2){\bbm{1}_{AB} \over 4}\bigg)\otimes \nonumber\\ &\qquad\qquad \qquad\qquad  \op{ii}_{I_AI_B} \otimes \op{00}_{I'_AI'_B}\Bigg),
\end{align}
where $\op{\psi_i}$ are the Bell states. We then use the fact that for each $i$ the state $\lambda_l^2\op{\psi_i}_{AB} + (1-\lambda_l^2){\bbm{1} \over 4}$ 
is a Bell diagonal state. A Bell diagonal state of the form $\sum_{j} p_j \op{\psi_j}$
has $E_R$ equal to $1- h(p_{max})$ where $p_{max} = \max_j p_j$ is maximal of the weights of the Bell state $\op{\psi_j}$ in the mixture, or $0$ if $p_{max}\leq {1\over 2}$. 
In our case $p_{max} =\lambda_l^2 + (1-\lambda_l^2)/4 $. Thus, via convexity and Eq.~\eqref{eq:qbound} we obtain that
\begin{equation}
E_R(\rho_{AB}^{out}) \leq q\left(1- h_2(\lambda_l^2 + {(1-\lambda_l^2)\over 4})\right)
\end{equation}
for $\lambda_l^2 + (1-\lambda_l^2)/4 > 1/2$, and $0$ otherwise. The condition $\lambda_l^2 + (1-\lambda_l^2)/4 > 1/2$ on $\lambda_l$ is equivalent to $\lambda_l > {1\over \sqrt{3}}$. This implies that for $q=1$ the bound is zero for
$\lambda_l\in ({1\over 3},{1\over \sqrt{3}}]$, for which depolarizing channel is non-zero, and hence its private capacity is non-zero as well. We interpret this as noncommutativity of the independent and identically distributed (i.i.d.) Bell measurement and entanglement distillation. Indeed, for this range of $\lambda_l$ given access to an isotropic state $\rho(\lambda_l)$ one can distil $E_D(\rho(\lambda_l))=(1- h_2(\lambda_l))$ of entanglement, and hence the quantum capacity ${\cal Q}(\mc{D}^{l}) = 1- h_2(\lambda_l)$ (or zero for $\lambda_l\leq 1/3$). On the other hand, this amount of key becomes inaccessible when the Bell measurement is done first.

\subsection{MDI-QKD via dephasing channels}

In this section we consider two dephasing channels~\eqref{eq:ds} between Alice and Charlie and Bob and Charlie. We will again observe that the operation of distillation and \textit{iid} entanglement swapping via Bell measurement do not commute. Altering them leads to different amount of key in the output.
We will use the fact that MDI-QKD via dephasing channel is teleportation-covariant.

 \begin{figure}[ht]
        \centering
            \includegraphics[trim=4.5cm 7.5cm 4.5cm 6.5cm,width=6cm]{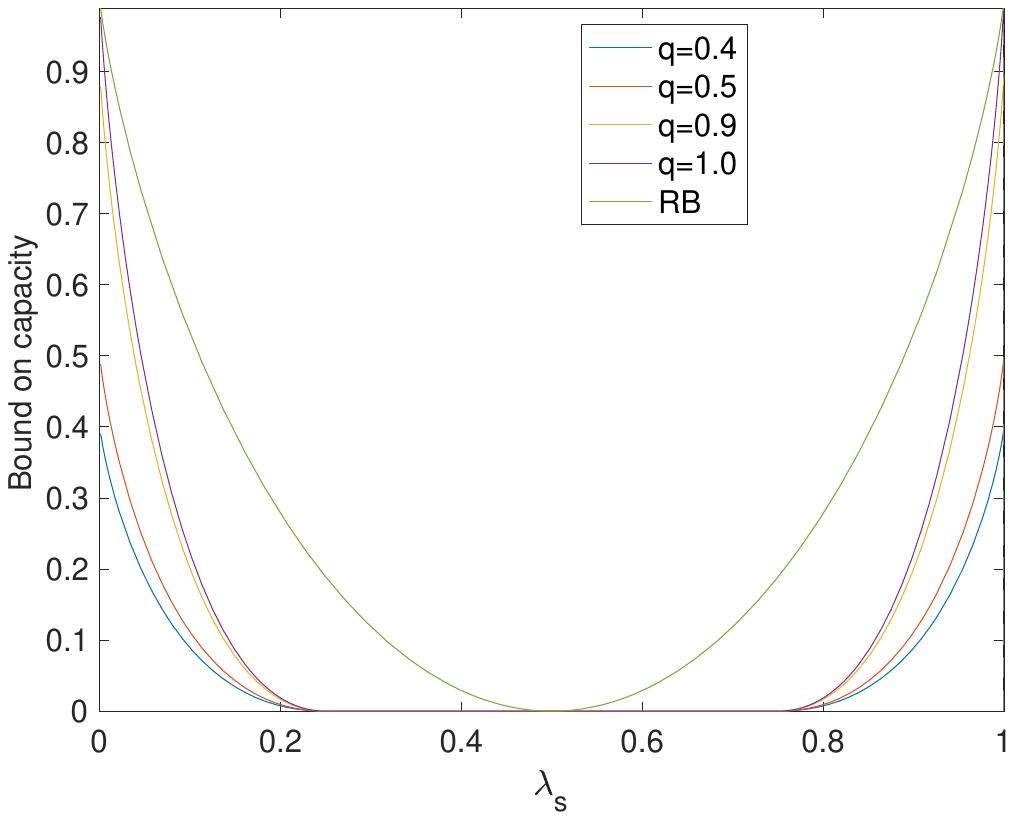}
              \caption{Upper bounds~\eqref{eq:MDI-ds} on the secret key capacities for the MDI-QKD protocol with dephasing channels for different values of parameters $q$ and $\lambda_s$, in comparison to the RB bound \cite{PLOB15}.}
              \label{fig:MDI_depha}
    \end{figure}

Note that the Choi-Jamiolkowski state (up to local unitary operation as the input state is $\Psi^-$) of the dephasing channel equals $\lambda_s\Psi^- +(1-\lambda_s)\Psi^+=(2\lambda_s-1)\Psi^- + (2-2\lambda_s)\rho_{cl}$ with $\rho_{cl} = {1\over 2}\left(\op{01}+\op{10}\right)$. Hence, the Choi-Jamiolkowski state of the dephasing-MDIQKD channel reads:

\begin{align}
&\rho_{AB}^{out} \coloneqq \nonumber\\ 
& (2\lambda_s-1)^2 q\Psi^-_{AB}\otimes \sum_{i=0}^{3}\op{ii}_{I_AI_B}\otimes \op{00}_{I'_AI'_B} +\nonumber\\
& (2-2\lambda_s)(2\lambda_s-1)q \rho_{cl}^{AB}\otimes {1 \over 4}\sum_{i=0}^{3}\op{ii}_{I_AI_B} \otimes \op{00}_{I'_AI'_B}\nonumber\\ &  + 
(2-2\lambda_s)q {\bbm{1}_{AB}\over 4}\otimes {1\over 4}\sum_{i=0}^{3}\op{ii}_{I_AI_B} \otimes \op{00}_{I'_AI'_B} + \nonumber\\
&(1-q) {\bbm{1}_{AB}\over 4}\otimes \op{\perp}_{I_AI_B}\otimes \op{11}_{I'_AI'_B},
\end{align} 

given that Alice has performed the control-Pauli operations on her systems $AI_A$. We can safely assume that this decoding has been done, because local unitary operation does not change the relative entropy of entanglement. The first case is a straightforward result of correct entanglement swapping. Regarding the next term, with probability $(2-2\lambda_s)(2\lambda_s-1)$ a subsystem $C_A$ of the state $\rho_{cl}$ gets correctly teleported to $A$, and hence finally $\rho_{cl}^{AB}$ is shared by Alice and Bob. However with probability $(2-2\lambda_s)=(2-2\lambda_s)^2 +(2-2\lambda_s)(2\lambda_s-1)$, resulting state is maximally mixed. This is because with probability $(2-2\lambda_s)^2$ the state on system $C$ is traced out, hence a product of subsystems of $\rho_{cl}^{AB}$  is an output. On the other hand, with probability $(2-2\lambda_s)(2\lambda_s-1)$, subsystem $C_B$ of the state  $\rho_{cl}$ is teleported to Bob, however Bob does not do the decoding. It is strightforward then to check that ${1\over 4}\sum_{i=0}^1 \sigma_i^{B}\otimes\bbm{1}_A\rho_{cl}^{AB}\widehat{\sigma}_i^B\otimes \bbm{1}_A$ with $\widehat{\sigma}_i$ being Pauli operators, is the maximally mixed state of two qubits.

The relative entropy of the $\rho_{AB}^{out}$ reads:
\begin{align}
E_R(\rho^{out}_{AB}) & \leq q E_R(\rho_{AB|00}^{out}) + (1-q)E_R(\rho_{AB|11}^{out})\\ 
& =  q E_R(\rho^{out}_{AB|00}),
\label{eq:qbound2}
\end{align}
where $\rho_{AB|11}^{out} = {\bbm{1}_{AB}\over 4}\otimes \op{\perp}_{I_AI_B}\otimes \op{11}_{I'_AI'_B}$ and $\rho_{AB|00}^{out}$ is such that
$(1-q)\rho_{AB|11}+q\rho_{AB|00} = \rho_{AB}$. We have used again the convexity of the relative entropy and the fact that it is zero for 
maximally mixed state. We then observe that 
\begin{align}
& E_R(\rho_{AB|00}^{out})  \nonumber\\
 &=  E_R\left({(2\lambda_s-1)^2}\op{\Psi^-}_{AB}  + \right. \nonumber\\ 
&\quad \left. (2-2\lambda_s)(2\lambda_s-1) \rho_{cl}^{AB} + 
(2-2\lambda_s) {\bbm{1}_{AB}\over 4}\right)
\end{align}

where we have neglected systems $I_AI_B$ and $I_A'I_B'$ due to subadditivity of $E_R$ and the fact that it is zero for both the state $\sum_{i=0}^{3}\op{ii}_{I_AI_B}$ and $\op{00}_{I_A'I_B'}$.
Resulting state is Bell diagonal (note that $\rho_{cl}^{AB} = {1\over 2}(\op{\Psi^-}+\op{\Psi^+})$), it is thus sufficient to find the largest weight of a Bell state to compute its relative entropy. Bell diagonal states are separable if the largest weight is less than or equal to half, i.e., when none of the Bell states ($\Phi^+,\Phi^-,\Psi^+,\Psi^-$) has weight greater than $1/2$.

For the case $\lambda_s\geq \frac{1}{2}$, the state $\op{\Psi^-}$ is 
in the mixed state $\rho^{out}_{AB|00}$  with probability $(2\lambda_s-1)^2 + (2-2\lambda_s)(2\lambda_s-1) +(2-2\lambda_s)/4={1\over 2}(4\lambda_s^2-3\lambda_s +1)$.

Thus, keeping the structure of the Choi state of the dephasing channel in mind, we arrive at the following bound:

\begin{equation}
    E_R(\rho^{out}_{AB})\leq
\left\{
\begin{array}{ll}
q(1- h_2({1\over 2}p_-(\lambda_s)) & \textnormal{for}\quad \lambda_s > \frac{3}{4},\\
0 & \textnormal{for}\quad \frac{1}{4}\leq \lambda_s\leq \frac{3}{4},\\
q(1- h_2({1\over 2}p_-(1-\lambda_s)) & \textnormal{for}\quad \lambda_s< \frac{1}{4},\\
\end{array}
\right.
\end{equation}
where $p_-(x)\coloneqq 4x^2-3x +1$.

\section{On the complexity of finding lower bound of SKA rate for Bidirectional Network}\label{app:complexity}

Here, we briefly comment on the complexity of finding a sub-graph, which allows us to realize the Conference Key Agreement with the capacity indicated by the inequality~\eqref{eq:lb-general}. 
As we show, the complexity is a polynomial
of low degree $O(n^2)$. In what 
follows a {\it minimum spanning tree} is
 a tree with a minimal sum of the weights
 of its edges. A {\it minimum bottleneck
 spanning tree} is the one in which the
 edge with the highest weight has the lowest possible value for the considered graph.

The algorithm of finding maximal of
the minimal edges over all spanning 
trees of the graph is as follows.
\begin{enumerate}[label={(\arabic*)}]
\item Find maximal weight
of the edges of $G$ (denoted as $M$)
\item Find
{\it minimum spanning tree} $T_{MST}$ in the graph $G'=(V_G,E_G)$, which is same as $G$, but with weights of edges changed from
$w(e)$ to $M-w(e)$, where $M\equiv \max_{e'\in E_G} w(e')$.
\item Find minimal
weight of the edges in $T$, denoted $w_{min}$. Return $M-w_{min}$.
\end{enumerate}

The correctness of this algorithm follows
from the fact that every minimal spanning tree is a {\it minimal bottleneck spanning tree}. Finding the highest weight of edges of this tree that is as low as possible is the opposite task to ours. Indeed, we aim
at finding trees with the lowest weight
over its edges to be as high as possible.
This is why we search for minimal spanning tree in the graph with converted edges 
to $M\equiv \max_{e'\in E_G} w(e') - w(e)$.
Next, we use the fact that $\min_{T\subseteq G} \max_{e \in E_T} [M- w(e)] = 
M- \max_{T\subseteq G} \min_{e \in E_T} w(e)$, so the $M- w_{min}$ is the solution.
The overall time complexity of this algorithm is $O(m + n\log n)$. Indeed, the first step takes $O(m)$ time. The next two
take $O(m +n\log n)$, where the finding
of minimum spanning tree is via the Prim's algorithm based on the data structure called Fibonacci heap \cite{LRCSbook}. The final step takes $O(n\log n)$, which is the time of sorting the weights of edges (e.g. by QuickSort algorithm). Taking into account
that $m$ scales pessimistically as $n^2$,
we obtain $O(n^2)$ worst-case complexity.

To summarize, the value of the lower bound can be found efficiently on a classical computer, given all the
capacities describing the Bidirectional Network are known and represented in the form of a graph.

\section{Key distillation from states-- plots}\label{app:StatesPlots}

\begin{figure}[h]
    \center{\includegraphics[trim=4.5cm 7cm 4.5cm 6.5cm,width=5.5cm]{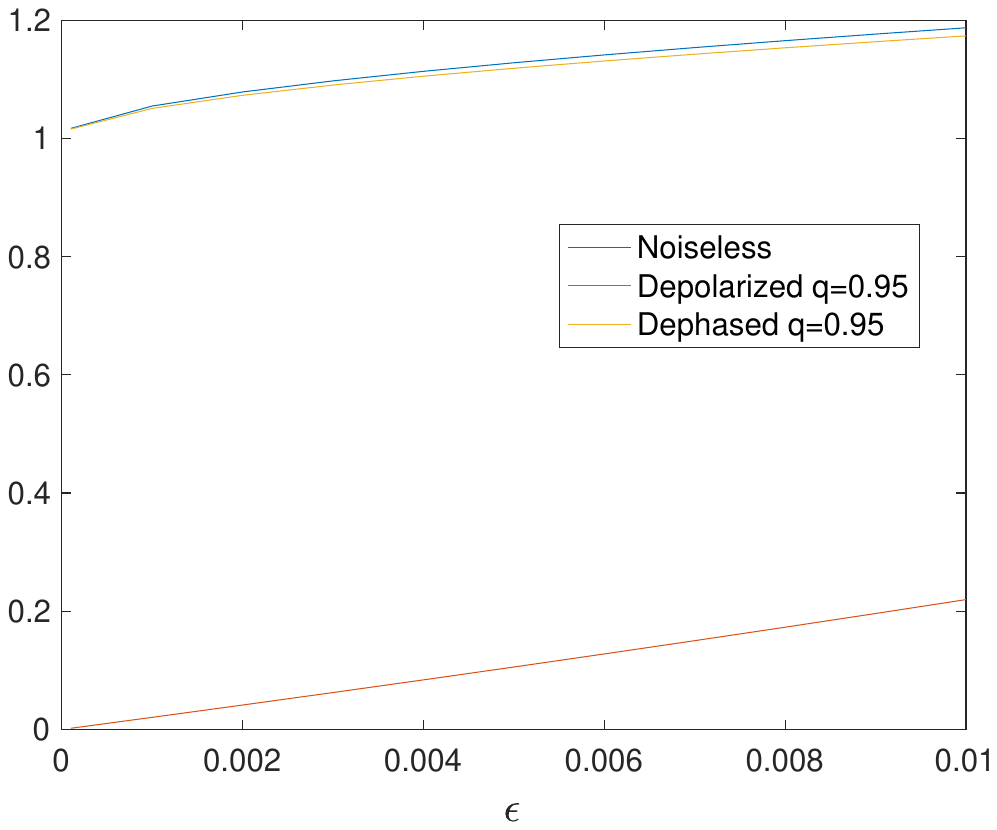}}
    \caption{Plot of $\varepsilon$-hypothesis testing upper bound on conference key rate for single copy of $\Phi_{3}^\mathrm{GHZ}$ state, for noiseless, dephased and depolarized case.}
    \label{fig:GHZ_3p_1c}
\end{figure}
To calculate our upper bounds, we utilize the technique of semidefinite programming (SDP) with MatLab (version) library ``SDPT3 4.0" \cite{SDPT3}. We calculate upper bounds for several cases incorporating both $\Phi_{M}^\mathrm{GHZ}$ states and $\Phi_{M}^\mathrm{W}$ states. Firstly we vary the number of copies of the state that entered the protocol; secondly, we make calculations for multipartite states with the number of parties exceeding three. Finally, we extend our consideration to states subjected to dephasing or depolarizing noises characterized in Eq.~\eqref{eqn:noise} (each qubit is subjected to noise separately). We investigate the effect of noise in the case of a different number of copies and different number of parties. 
\begin{align}\label{eqn:noise}
    &\rho_\mathrm{noisy}=\mc{D}^{\otimes M}(\rho),
\end{align}
    for $D$ given by
\begin{align}
    &\mc{D}^q_{deph.}(\omega)=q\omega+ (1-q)  \sigma_z  \omega  \sigma_z,\label{n:deph} \\
    &\mc{D}^q_{depol.}(\omega)=q \omega + (1-q)\frac{\mathbbm{1}}{2},\label{n:depol} 
\end{align}
where $\sigma_z$ is the Pauli Z matrix, and $q$ is the noise parameter.

We present the plots for the upper bound on the key rate distilled from both $\Phi_{M}^\mathrm{GHZ}$, $\Phi_{M}^\mathrm{W}$ states and tensor powers of them. The plots are a function of $\varepsilon$ parameter controlling the fidelity of the target state $\rho_{\vec{A}}$ with respect to a private state.


\begin{figure}[h]
    \center{\includegraphics[trim=4.5cm 7.5cm 4.5cm 6.5cm,width=5.5cm]{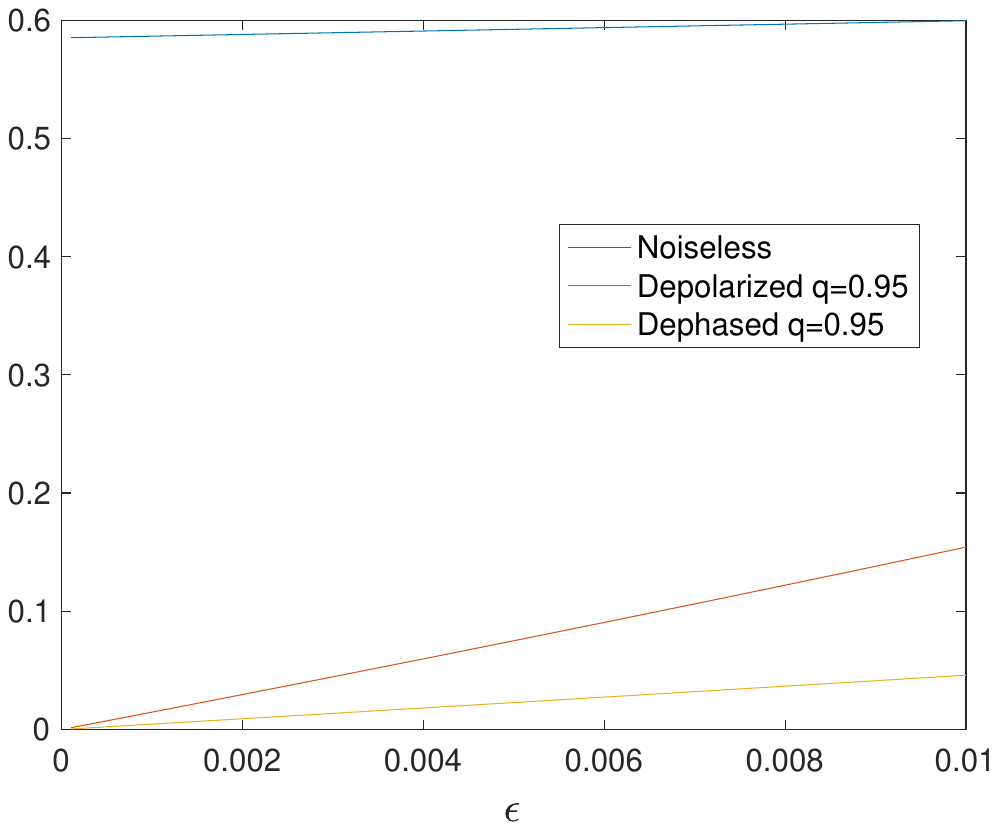}}
    \caption{Plot of $\varepsilon$-hypothesis testing upper bound on conference key rate for single copy of $\Phi_{3}^\mathrm{W}$ state, for noiseless, dephased and depolarized case.}
    \label{fig:W_3p_1c}
\end{figure}

\begin{figure}[h]  \center{\includegraphics[trim=4.5cm 7.5cm 4.5cm 6.5cm,width=5.5cm]{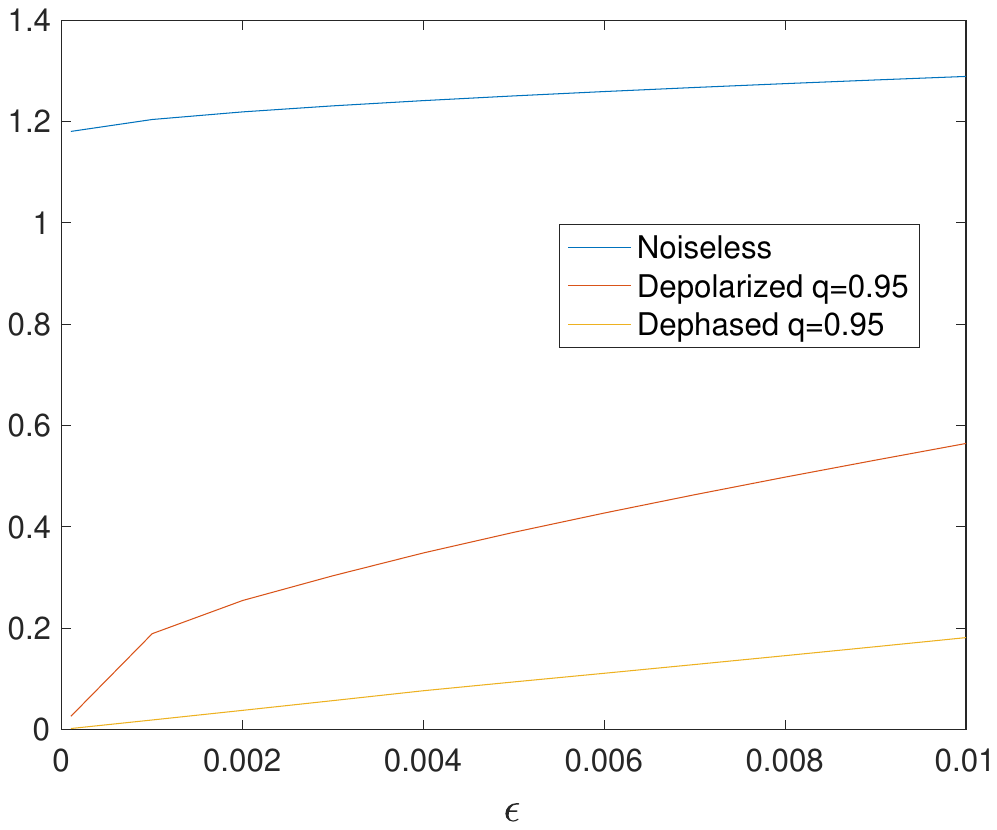}} \caption{Plot of $\varepsilon$-hypothesis testing upper bound on conference key rate for two copies of $\Phi_{3}^\mathrm{W}$ state, for noiseless, dephased and depolarized case.}     \label{fig:W_3p_2c} \end{figure}

We compare performance of our upper bound and choice of biseparable states for tripartite single copy state in the plots in Figs.~\ref{fig:GHZ_3p_1c} and \ref{fig:W_3p_1c}. In the control plot in Fig. \ref{fig:GHZ_3p_1c} for the noiseless $\Phi_M^\mathrm{GHZ}$ state the upper bound as expected, exhibits the value to be just above 1 for the chosen range of $\varepsilon$. This indicates that the $\varepsilon$-hypothesis testing upper bound is not too loose. For single copy tripartite $\Phi_M^\mathrm{W}$ state value of the upper bound in Fig.~\ref{fig:W_3p_1c} for $\varepsilon \approx 0$ is below 0.6, what is below the value of the rate of the optimal LOCC asymptotic protocol being approximately $0.643$ per copy \cite{SVW}. In the case of two copies of bipartite $\Phi_M^\mathrm{W}$ in Fig.~\ref{fig:W_3p_2c} state we obtain an upper bound that for $\varepsilon \approx 0$ has the value around 1.18 what is significantly above $\frac{2}{3}$ achieved by the protocol described earlier in this section and 1.286 what is an asymptotic limit for state being two copies of $\Phi_M^\mathrm{W}$ state \cite[Theorem~2]{SVW}. Both these results stand with an agreement with the fact that single copy and two copies one-shot protocols constitute very limited class of protocols compared to those available for calculating the asymptotic limit. For two copies of $\Phi_M^\mathrm{W}$ state, the large gap between our upper bound for the conference key rate and rate of $\Phi_M^\mathrm{GHZ}$ states distillation protocol makes us think that indeed the former is larger than the latter. However, formal proof is still missing. Moreover, we notice that optimal protocol $\Phi_M^\mathrm{W}$ to $\Phi_M^\mathrm{GHZ}$ conversion have to incorporate at least three copies of $\Phi_M^\mathrm{W}$ state. This is because our $\varepsilon$-hypothesis testing upper bound is smaller than the asymptotic limit for $\Phi_M^\mathrm{GHZ}$ distillation.

\end{appendix}

\bibliography{qr}{}

\begin{thebibliography}{154}%
\makeatletter
\providecommand \@ifxundefined [1]{%
 \@ifx{#1\undefined}
}%
\providecommand \@ifnum [1]{%
 \ifnum #1\expandafter \@firstoftwo
 \else \expandafter \@secondoftwo
 \fi
}%
\providecommand \@ifx [1]{%
 \ifx #1\expandafter \@firstoftwo
 \else \expandafter \@secondoftwo
 \fi
}%
\providecommand \natexlab [1]{#1}%
\providecommand \enquote  [1]{``#1''}%
\providecommand \bibnamefont  [1]{#1}%
\providecommand \bibfnamefont [1]{#1}%
\providecommand \citenamefont [1]{#1}%
\providecommand \href@noop [0]{\@secondoftwo}%
\providecommand \href [0]{\begingroup \@sanitize@url \@href}%
\providecommand \@href[1]{\@@startlink{#1}\@@href}%
\providecommand \@@href[1]{\endgroup#1\@@endlink}%
\providecommand \@sanitize@url [0]{\catcode `\\12\catcode `\$12\catcode
  `\&12\catcode `\#12\catcode `\^12\catcode `\_12\catcode `\%12\relax}%
\providecommand \@@startlink[1]{}%
\providecommand \@@endlink[0]{}%
\providecommand \url  [0]{\begingroup\@sanitize@url \@url }%
\providecommand \@url [1]{\endgroup\@href {#1}{\urlprefix }}%
\providecommand \urlprefix  [0]{URL }%
\providecommand \Eprint [0]{\href }%
\providecommand \doibase [0]{http://dx.doi.org/}%
\providecommand \selectlanguage [0]{\@gobble}%
\providecommand \bibinfo  [0]{\@secondoftwo}%
\providecommand \bibfield  [0]{\@secondoftwo}%
\providecommand \translation [1]{[#1]}%
\providecommand \BibitemOpen [0]{}%
\providecommand \bibitemStop [0]{}%
\providecommand \bibitemNoStop [0]{.\EOS\space}%
\providecommand \EOS [0]{\spacefactor3000\relax}%
\providecommand \BibitemShut  [1]{\csname bibitem#1\endcsname}%
\let\auto@bib@innerbib\@empty
\bibitem [{\citenamefont {Bennett}\ and\ \citenamefont
  {Brassard}(1984)}]{BB84}%
  \BibitemOpen
  \bibfield  {author} {\bibinfo {author} {\bibfnamefont {Charles~H.}\
  \bibnamefont {Bennett}}\ and\ \bibinfo {author} {\bibfnamefont {Gilles}\
  \bibnamefont {Brassard}},\ }\bibfield  {title} {\enquote {\bibinfo {title}
  {Quantum cryptography: Public key distribution and coin tossing},}\ }in\
  \href@noop {} {\emph {\bibinfo {booktitle} {International Conference on
  Computer System and Signal Processing, IEEE, 1984}}}\ (\bibinfo {year}
  {1984})\ pp.\ \bibinfo {pages} {175--179}\BibitemShut {NoStop}%
\bibitem [{\citenamefont {Ekert}(1991)}]{Eke91}%
  \BibitemOpen
  \bibfield  {author} {\bibinfo {author} {\bibfnamefont {Artur~K.}\
  \bibnamefont {Ekert}},\ }\bibfield  {title} {\enquote {\bibinfo {title}
  {Quantum cryptography based on {Bell's} theorem},}\ }\href {\doibase
  10.1103/PhysRevLett.67.661} {\bibfield  {journal} {\bibinfo  {journal}
  {Physical Review Letters}\ }\textbf {\bibinfo {volume} {67}},\ \bibinfo
  {pages} {661--663} (\bibinfo {year} {1991})}\BibitemShut {NoStop}%
\bibitem [{\citenamefont {Dowling}\ and\ \citenamefont {Milburn}(2003)}]{DM03}%
  \BibitemOpen
  \bibfield  {author} {\bibinfo {author} {\bibfnamefont {Jonathan~P.}\
  \bibnamefont {Dowling}}\ and\ \bibinfo {author} {\bibfnamefont {Gerard~J.}\
  \bibnamefont {Milburn}},\ }\bibfield  {title} {\enquote {\bibinfo {title}
  {Quantum technology: the second quantum revolution},}\ }\href {\doibase
  10.1098/rsta.2003.1227} {\bibfield  {journal} {\bibinfo  {journal}
  {Philosophical Transactions of the Royal Society of London. Series A:
  Mathematical, Physical and Engineering Sciences}\ }\textbf {\bibinfo {volume}
  {361}},\ \bibinfo {pages} {1655--1674} (\bibinfo {year} {2003})},\ \bibinfo
  {note} {arXiv:quant-ph/0206091}\BibitemShut {NoStop}%
\bibitem [{\citenamefont {Renner}(2005)}]{RennerThesis}%
  \BibitemOpen
  \bibfield  {author} {\bibinfo {author} {\bibfnamefont {Renato}\ \bibnamefont
  {Renner}},\ }\emph {\bibinfo {title} {Security of Quantum Key
  Distribution}},\ \href {https://arxiv.org/abs/quant-ph/0512258} {Ph.D.
  thesis},\ \bibinfo  {school} {ETH Z\"urich} (\bibinfo {year} {2005}),\
  \bibinfo {note} {arXiv:quant-ph/0512258}\BibitemShut {NoStop}%
\bibitem [{\citenamefont {Chen}\ \emph {et~al.}(2009)\citenamefont {Chen},
  \citenamefont {Liang}, \citenamefont {Liu}, \citenamefont {Cai},
  \citenamefont {Ju}, \citenamefont {Liu}, \citenamefont {Wang}, \citenamefont
  {Yin}, \citenamefont {Chen}, \citenamefont {Chen}, \citenamefont {Peng},\
  and\ \citenamefont {Pan}}]{CLL+09}%
  \BibitemOpen
  \bibfield  {author} {\bibinfo {author} {\bibfnamefont {Teng-Yun}\
  \bibnamefont {Chen}}, \bibinfo {author} {\bibfnamefont {Hao}\ \bibnamefont
  {Liang}}, \bibinfo {author} {\bibfnamefont {Yang}\ \bibnamefont {Liu}},
  \bibinfo {author} {\bibfnamefont {Wen-Qi}\ \bibnamefont {Cai}}, \bibinfo
  {author} {\bibfnamefont {Lei}\ \bibnamefont {Ju}}, \bibinfo {author}
  {\bibfnamefont {Wei-Yue}\ \bibnamefont {Liu}}, \bibinfo {author}
  {\bibfnamefont {Jian}\ \bibnamefont {Wang}}, \bibinfo {author} {\bibfnamefont
  {Hao}\ \bibnamefont {Yin}}, \bibinfo {author} {\bibfnamefont {Kai}\
  \bibnamefont {Chen}}, \bibinfo {author} {\bibfnamefont {Zeng-Bing}\
  \bibnamefont {Chen}}, \bibinfo {author} {\bibfnamefont {Cheng-Zhi}\
  \bibnamefont {Peng}}, \ and\ \bibinfo {author} {\bibfnamefont {Jian-Wei}\
  \bibnamefont {Pan}},\ }\bibfield  {title} {\enquote {\bibinfo {title} {Field
  test of a practical secure communication network with decoy-state quantum
  cryptography},}\ }\href {\doibase 10.1364/oe.17.006540} {\bibfield  {journal}
  {\bibinfo  {journal} {Optics Express}\ }\textbf {\bibinfo {volume} {17}},\
  \bibinfo {pages} {6540} (\bibinfo {year} {2009})},\ \bibinfo {note}
  {arXiv:0810.1264}\BibitemShut {NoStop}%
\bibitem [{\citenamefont {Vallone}\ \emph {et~al.}(2015)\citenamefont
  {Vallone}, \citenamefont {Bacco}, \citenamefont {Dequal}, \citenamefont
  {Gaiarin}, \citenamefont {Luceri}, \citenamefont {Bianco},\ and\
  \citenamefont {Villoresi}}]{VBD+15}%
  \BibitemOpen
  \bibfield  {author} {\bibinfo {author} {\bibfnamefont {Giuseppe}\
  \bibnamefont {Vallone}}, \bibinfo {author} {\bibfnamefont {Davide}\
  \bibnamefont {Bacco}}, \bibinfo {author} {\bibfnamefont {Daniele}\
  \bibnamefont {Dequal}}, \bibinfo {author} {\bibfnamefont {Simone}\
  \bibnamefont {Gaiarin}}, \bibinfo {author} {\bibfnamefont {Vincenza}\
  \bibnamefont {Luceri}}, \bibinfo {author} {\bibfnamefont {Giuseppe}\
  \bibnamefont {Bianco}}, \ and\ \bibinfo {author} {\bibfnamefont {Paolo}\
  \bibnamefont {Villoresi}},\ }\bibfield  {title} {\enquote {\bibinfo {title}
  {Experimental satellite quantum communications},}\ }\href {\doibase
  10.1103/PhysRevLett.115.040502} {\bibfield  {journal} {\bibinfo  {journal}
  {Physical Review Letters}\ }\textbf {\bibinfo {volume} {115}},\ \bibinfo
  {pages} {040502} (\bibinfo {year} {2015})},\ \bibinfo {note}
  {arXiv:1406.4051}\BibitemShut {NoStop}%
\bibitem [{\citenamefont {Zhang}\ \emph {et~al.}(2018)\citenamefont {Zhang},
  \citenamefont {Xu}, \citenamefont {Chen}, \citenamefont {Peng},\ and\
  \citenamefont {Pan}}]{ZXC+18}%
  \BibitemOpen
  \bibfield  {author} {\bibinfo {author} {\bibfnamefont {Qiang}\ \bibnamefont
  {Zhang}}, \bibinfo {author} {\bibfnamefont {Feihu}\ \bibnamefont {Xu}},
  \bibinfo {author} {\bibfnamefont {Yu-Ao}\ \bibnamefont {Chen}}, \bibinfo
  {author} {\bibfnamefont {Cheng-Zhi}\ \bibnamefont {Peng}}, \ and\ \bibinfo
  {author} {\bibfnamefont {Jian-Wei}\ \bibnamefont {Pan}},\ }\bibfield  {title}
  {\enquote {\bibinfo {title} {Large scale quantum key distribution: challenges
  and solutions},}\ }\href {\doibase 10.1364/oe.26.024260} {\bibfield
  {journal} {\bibinfo  {journal} {Optics Express}\ }\textbf {\bibinfo {volume}
  {26}},\ \bibinfo {pages} {24260} (\bibinfo {year} {2018})},\ \bibinfo {note}
  {arXiv:1809.02291}\BibitemShut {NoStop}%
\bibitem [{\citenamefont {Hensen}\ \emph {et~al.}(2015)\citenamefont {Hensen},
  \citenamefont {Bernien}, \citenamefont {Dr{\'{e}}au}, \citenamefont
  {Reiserer}, \citenamefont {Kalb}, \citenamefont {Blok}, \citenamefont
  {Ruitenberg}, \citenamefont {Vermeulen}, \citenamefont {Schouten},
  \citenamefont {Abell{\'{a}}n}, \citenamefont {Amaya}, \citenamefont
  {Pruneri}, \citenamefont {Mitchell}, \citenamefont {Markham}, \citenamefont
  {Twitchen}, \citenamefont {Elkouss}, \citenamefont {Wehner}, \citenamefont
  {Taminiau},\ and\ \citenamefont {Hanson}}]{HBD+15}%
  \BibitemOpen
  \bibfield  {author} {\bibinfo {author} {\bibfnamefont {B.}~\bibnamefont
  {Hensen}}, \bibinfo {author} {\bibfnamefont {H.}~\bibnamefont {Bernien}},
  \bibinfo {author} {\bibfnamefont {A.~E.}\ \bibnamefont {Dr{\'{e}}au}},
  \bibinfo {author} {\bibfnamefont {A.}~\bibnamefont {Reiserer}}, \bibinfo
  {author} {\bibfnamefont {N.}~\bibnamefont {Kalb}}, \bibinfo {author}
  {\bibfnamefont {M.~S.}\ \bibnamefont {Blok}}, \bibinfo {author}
  {\bibfnamefont {J.}~\bibnamefont {Ruitenberg}}, \bibinfo {author}
  {\bibfnamefont {R.~F.~L.}\ \bibnamefont {Vermeulen}}, \bibinfo {author}
  {\bibfnamefont {R.~N.}\ \bibnamefont {Schouten}}, \bibinfo {author}
  {\bibfnamefont {C.}~\bibnamefont {Abell{\'{a}}n}}, \bibinfo {author}
  {\bibfnamefont {W.}~\bibnamefont {Amaya}}, \bibinfo {author} {\bibfnamefont
  {V.}~\bibnamefont {Pruneri}}, \bibinfo {author} {\bibfnamefont {M.~W.}\
  \bibnamefont {Mitchell}}, \bibinfo {author} {\bibfnamefont {M.}~\bibnamefont
  {Markham}}, \bibinfo {author} {\bibfnamefont {D.~J.}\ \bibnamefont
  {Twitchen}}, \bibinfo {author} {\bibfnamefont {D.}~\bibnamefont {Elkouss}},
  \bibinfo {author} {\bibfnamefont {S.}~\bibnamefont {Wehner}}, \bibinfo
  {author} {\bibfnamefont {T.~H.}\ \bibnamefont {Taminiau}}, \ and\ \bibinfo
  {author} {\bibfnamefont {R.}~\bibnamefont {Hanson}},\ }\bibfield  {title}
  {\enquote {\bibinfo {title} {Loophole-free bell inequality violation using
  electron spins separated by 1.3 kilometres},}\ }\href {\doibase
  10.1038/nature15759} {\bibfield  {journal} {\bibinfo  {journal} {Nature}\
  }\textbf {\bibinfo {volume} {526}},\ \bibinfo {pages} {682--686} (\bibinfo
  {year} {2015})},\ \bibinfo {note} {arXiv:1508.05949}\BibitemShut {NoStop}%
\bibitem [{\citenamefont {Pauka}\ \emph {et~al.}(2019)\citenamefont {Pauka},
  \citenamefont {Das}, \citenamefont {Kalra}, \citenamefont {Moini},
  \citenamefont {Yang}, \citenamefont {Trainer}, \citenamefont {Bousquet},
  \citenamefont {Cantaloube}, \citenamefont {Dick}, \citenamefont {Gardner},
  \citenamefont {Manfra},\ and\ \citenamefont {Reilly}}]{Mst19}%
  \BibitemOpen
  \bibfield  {author} {\bibinfo {author} {\bibfnamefont {S.~J.}\ \bibnamefont
  {Pauka}}, \bibinfo {author} {\bibfnamefont {K.}~\bibnamefont {Das}}, \bibinfo
  {author} {\bibfnamefont {R.}~\bibnamefont {Kalra}}, \bibinfo {author}
  {\bibfnamefont {A.}~\bibnamefont {Moini}}, \bibinfo {author} {\bibfnamefont
  {Y.}~\bibnamefont {Yang}}, \bibinfo {author} {\bibfnamefont {M.}~\bibnamefont
  {Trainer}}, \bibinfo {author} {\bibfnamefont {A.}~\bibnamefont {Bousquet}},
  \bibinfo {author} {\bibfnamefont {C.}~\bibnamefont {Cantaloube}}, \bibinfo
  {author} {\bibfnamefont {N.}~\bibnamefont {Dick}}, \bibinfo {author}
  {\bibfnamefont {G.~C.}\ \bibnamefont {Gardner}}, \bibinfo {author}
  {\bibfnamefont {M.~J.}\ \bibnamefont {Manfra}}, \ and\ \bibinfo {author}
  {\bibfnamefont {D.~J.}\ \bibnamefont {Reilly}},\ }\bibfield  {title}
  {\enquote {\bibinfo {title} {A cryogenic interface for controlling many
  qubits},}\ }\href@noop {} {\  (\bibinfo {year} {2019})},\ \bibinfo {note}
  {arXiv:1912.01299}\BibitemShut {NoStop}%
\bibitem [{\citenamefont {Bradley}\ \emph {et~al.}(2019)\citenamefont
  {Bradley}, \citenamefont {Randall}, \citenamefont {Abobeih}, \citenamefont
  {Berrevoets}, \citenamefont {Degen}, \citenamefont {Bakker}, \citenamefont
  {Markham}, \citenamefont {Twitchen},\ and\ \citenamefont
  {Taminiau}}]{BRA+19}%
  \BibitemOpen
  \bibfield  {author} {\bibinfo {author} {\bibfnamefont {C.~E.}\ \bibnamefont
  {Bradley}}, \bibinfo {author} {\bibfnamefont {J.}~\bibnamefont {Randall}},
  \bibinfo {author} {\bibfnamefont {M.~H.}\ \bibnamefont {Abobeih}}, \bibinfo
  {author} {\bibfnamefont {R.~C.}\ \bibnamefont {Berrevoets}}, \bibinfo
  {author} {\bibfnamefont {M.~J.}\ \bibnamefont {Degen}}, \bibinfo {author}
  {\bibfnamefont {M.~A.}\ \bibnamefont {Bakker}}, \bibinfo {author}
  {\bibfnamefont {M.}~\bibnamefont {Markham}}, \bibinfo {author} {\bibfnamefont
  {D.~J.}\ \bibnamefont {Twitchen}}, \ and\ \bibinfo {author} {\bibfnamefont
  {T.~H.}\ \bibnamefont {Taminiau}},\ }\bibfield  {title} {\enquote {\bibinfo
  {title} {A ten-qubit solid-state spin register with quantum memory up to one
  minute},}\ }\href {\doibase 10.1103/PhysRevX.9.031045} {\bibfield  {journal}
  {\bibinfo  {journal} {Physical Review X}\ }\textbf {\bibinfo {volume} {9}},\
  \bibinfo {pages} {031045} (\bibinfo {year} {2019})}\BibitemShut {NoStop}%
\bibitem [{\citenamefont {Chen}\ \emph {et~al.}(2021)\citenamefont {Chen},
  \citenamefont {Zhang}, \citenamefont {Chen}, \citenamefont {Cai},
  \citenamefont {Liao}, \citenamefont {Zhang}, \citenamefont {Chen},
  \citenamefont {Yin}, \citenamefont {Ren}, \citenamefont {Chen}, \citenamefont
  {Han}, \citenamefont {Yu}, \citenamefont {Liang}, \citenamefont {Zhou},
  \citenamefont {Yuan}, \citenamefont {Zhao}, \citenamefont {Wang},
  \citenamefont {Jiang}, \citenamefont {Zhang}, \citenamefont {Liu},
  \citenamefont {Li}, \citenamefont {Shen}, \citenamefont {Cao}, \citenamefont
  {Lu}, \citenamefont {Shu}, \citenamefont {Wang}, \citenamefont {Li},
  \citenamefont {Liu}, \citenamefont {Xu}, \citenamefont {Wang}, \citenamefont
  {Peng},\ and\ \citenamefont {Pan}}]{CZC+21}%
  \BibitemOpen
  \bibfield  {author} {\bibinfo {author} {\bibfnamefont {Yu-Ao}\ \bibnamefont
  {Chen}}, \bibinfo {author} {\bibfnamefont {Qiang}\ \bibnamefont {Zhang}},
  \bibinfo {author} {\bibfnamefont {Teng-Yun}\ \bibnamefont {Chen}}, \bibinfo
  {author} {\bibfnamefont {Wen-Qi}\ \bibnamefont {Cai}}, \bibinfo {author}
  {\bibfnamefont {Sheng-Kai}\ \bibnamefont {Liao}}, \bibinfo {author}
  {\bibfnamefont {Jun}\ \bibnamefont {Zhang}}, \bibinfo {author} {\bibfnamefont
  {Kai}\ \bibnamefont {Chen}}, \bibinfo {author} {\bibfnamefont {Juan}\
  \bibnamefont {Yin}}, \bibinfo {author} {\bibfnamefont {Ji-Gang}\ \bibnamefont
  {Ren}}, \bibinfo {author} {\bibfnamefont {Zhu}\ \bibnamefont {Chen}},
  \bibinfo {author} {\bibfnamefont {Sheng-Long}\ \bibnamefont {Han}}, \bibinfo
  {author} {\bibfnamefont {Qing}\ \bibnamefont {Yu}}, \bibinfo {author}
  {\bibfnamefont {Ken}\ \bibnamefont {Liang}}, \bibinfo {author} {\bibfnamefont
  {Fei}\ \bibnamefont {Zhou}}, \bibinfo {author} {\bibfnamefont {Xiao}\
  \bibnamefont {Yuan}}, \bibinfo {author} {\bibfnamefont {Mei-Sheng}\
  \bibnamefont {Zhao}}, \bibinfo {author} {\bibfnamefont {Tian-Yin}\
  \bibnamefont {Wang}}, \bibinfo {author} {\bibfnamefont {Xiao}\ \bibnamefont
  {Jiang}}, \bibinfo {author} {\bibfnamefont {Liang}\ \bibnamefont {Zhang}},
  \bibinfo {author} {\bibfnamefont {Wei-Yue}\ \bibnamefont {Liu}}, \bibinfo
  {author} {\bibfnamefont {Yang}\ \bibnamefont {Li}}, \bibinfo {author}
  {\bibfnamefont {Qi}~\bibnamefont {Shen}}, \bibinfo {author} {\bibfnamefont
  {Yuan}\ \bibnamefont {Cao}}, \bibinfo {author} {\bibfnamefont {Chao-Yang}\
  \bibnamefont {Lu}}, \bibinfo {author} {\bibfnamefont {Rong}\ \bibnamefont
  {Shu}}, \bibinfo {author} {\bibfnamefont {Jian-Yu}\ \bibnamefont {Wang}},
  \bibinfo {author} {\bibfnamefont {Li}~\bibnamefont {Li}}, \bibinfo {author}
  {\bibfnamefont {Nai-Le}\ \bibnamefont {Liu}}, \bibinfo {author}
  {\bibfnamefont {Feihu}\ \bibnamefont {Xu}}, \bibinfo {author} {\bibfnamefont
  {Xiang-Bin}\ \bibnamefont {Wang}}, \bibinfo {author} {\bibfnamefont
  {Cheng-Zhi}\ \bibnamefont {Peng}}, \ and\ \bibinfo {author} {\bibfnamefont
  {Jian-Wei}\ \bibnamefont {Pan}},\ }\bibfield  {title} {\enquote {\bibinfo
  {title} {An integrated space-to-ground quantum communication network over
  4,600 kilometres},}\ }\href {\doibase 10.1038/s41586-020-03093-8} {\bibfield
  {journal} {\bibinfo  {journal} {Nature}\ }\textbf {\bibinfo {volume} {589}},\
  \bibinfo {pages} {214--219} (\bibinfo {year} {2021})}\BibitemShut {NoStop}%
\bibitem [{\citenamefont {Shor}(1994)}]{Sho94}%
  \BibitemOpen
  \bibfield  {author} {\bibinfo {author} {\bibfnamefont {Peter~W.}\
  \bibnamefont {Shor}},\ }\bibfield  {title} {\enquote {\bibinfo {title}
  {Algorithms for quantum computation: Discrete logarithms and factoring},}\
  }in\ \href@noop {} {\emph {\bibinfo {booktitle} {Proceedings of the 35th
  Annual Symposium on Foundations of Computer Science}}}\ (\bibinfo
  {organization} {IEEE Computer Society Press},\ \bibinfo {address} {Los
  Alamitos, California},\ \bibinfo {year} {1994})\ pp.\ \bibinfo {pages}
  {124–--134}\BibitemShut {NoStop}%
\bibitem [{\citenamefont {{Kai Chen}}\ and\ \citenamefont {{Hoi-Kwong
  Lo}}(2005)}]{Chen2005}%
  \BibitemOpen
  \bibfield  {author} {\bibinfo {author} {\bibnamefont {{Kai Chen}}}\ and\
  \bibinfo {author} {\bibnamefont {{Hoi-Kwong Lo}}},\ }\bibfield  {title}
  {\enquote {\bibinfo {title} {Conference key agreement and quantum sharing of
  classical secrets with noisy {GHZ} states},}\ }in\ \href {\doibase
  10.1109/ISIT.2005.1523616} {\emph {\bibinfo {booktitle} {Proceedings.
  International Symposium on Information Theory, 2005. ISIT 2005.}}}\ (\bibinfo
  {year} {2005})\ pp.\ \bibinfo {pages} {1607--1611}\BibitemShut {NoStop}%
\bibitem [{\citenamefont {Augusiak}\ and\ \citenamefont
  {Horodecki}(2009{\natexlab{a}})}]{AH09}%
  \BibitemOpen
  \bibfield  {author} {\bibinfo {author} {\bibfnamefont {Remigiusz}\
  \bibnamefont {Augusiak}}\ and\ \bibinfo {author} {\bibfnamefont {Pawe\l{}}\
  \bibnamefont {Horodecki}},\ }\bibfield  {title} {\enquote {\bibinfo {title}
  {Multipartite secret key distillation and bound entanglement},}\ }\href
  {\doibase 10.1103/PhysRevA.80.042307} {\bibfield  {journal} {\bibinfo
  {journal} {Physical Review A}\ }\textbf {\bibinfo {volume} {80}},\ \bibinfo
  {pages} {042307} (\bibinfo {year} {2009}{\natexlab{a}})},\ \bibinfo {note}
  {arXiv:0811.3603}\BibitemShut {NoStop}%
\bibitem [{\citenamefont {Greenberger}\ \emph {et~al.}(1989)\citenamefont
  {Greenberger}, \citenamefont {Horne},\ and\ \citenamefont
  {Zeilinger}}]{GHZ89}%
  \BibitemOpen
  \bibfield  {author} {\bibinfo {author} {\bibfnamefont {Daniel~M}\
  \bibnamefont {Greenberger}}, \bibinfo {author} {\bibfnamefont {Michael~A}\
  \bibnamefont {Horne}}, \ and\ \bibinfo {author} {\bibfnamefont {Anton}\
  \bibnamefont {Zeilinger}},\ }\bibfield  {title} {\enquote {\bibinfo {title}
  {Going beyond {B}ell’s theorem},}\ }in\ \href@noop {} {\emph {\bibinfo
  {booktitle} {Bell’s theorem, quantum theory and conceptions of the
  universe}}}\ (\bibinfo  {publisher} {Springer},\ \bibinfo {year} {1989})\
  pp.\ \bibinfo {pages} {69--72},\ \bibinfo {note}
  {arXiv:0712.0921}\BibitemShut {NoStop}%
\bibitem [{\citenamefont {Kimble}(2008)}]{Kim08}%
  \BibitemOpen
  \bibfield  {author} {\bibinfo {author} {\bibfnamefont {H.~J.}\ \bibnamefont
  {Kimble}},\ }\bibfield  {title} {\enquote {\bibinfo {title} {The quantum
  internet},}\ }\href {\doibase 10.1038/nature07127} {\bibfield  {journal}
  {\bibinfo  {journal} {Nature}\ }\textbf {\bibinfo {volume} {453}},\ \bibinfo
  {pages} {1023--1030} (\bibinfo {year} {2008})}\BibitemShut {NoStop}%
\bibitem [{\citenamefont {Wehner}\ \emph {et~al.}(2018)\citenamefont {Wehner},
  \citenamefont {Elkouss},\ and\ \citenamefont {Hanson}}]{WEH18}%
  \BibitemOpen
  \bibfield  {author} {\bibinfo {author} {\bibfnamefont {Stephanie}\
  \bibnamefont {Wehner}}, \bibinfo {author} {\bibfnamefont {David}\
  \bibnamefont {Elkouss}}, \ and\ \bibinfo {author} {\bibfnamefont {Ronald}\
  \bibnamefont {Hanson}},\ }\bibfield  {title} {\enquote {\bibinfo {title}
  {Quantum internet: A vision for the road ahead},}\ }\href {\doibase
  10.1126/science.aam9288} {\bibfield  {journal} {\bibinfo  {journal}
  {Science}\ }\textbf {\bibinfo {volume} {362}},\ \bibinfo {pages} {eaam9288}
  (\bibinfo {year} {2018})}\BibitemShut {NoStop}%
\bibitem [{\citenamefont {Ma}\ \emph {et~al.}(2012)\citenamefont {Ma},
  \citenamefont {Herbst}, \citenamefont {Scheidl}, \citenamefont {Wang},
  \citenamefont {Kropatschek}, \citenamefont {Naylor}, \citenamefont
  {Wittmann}, \citenamefont {Mech}, \citenamefont {Kofler}, \citenamefont
  {Anisimova}, \citenamefont {Makarov}, \citenamefont {Jennewein},
  \citenamefont {Ursin},\ and\ \citenamefont {Zeilinger}}]{MHS+12}%
  \BibitemOpen
  \bibfield  {author} {\bibinfo {author} {\bibfnamefont {Xiao-Song}\
  \bibnamefont {Ma}}, \bibinfo {author} {\bibfnamefont {Thomas}\ \bibnamefont
  {Herbst}}, \bibinfo {author} {\bibfnamefont {Thomas}\ \bibnamefont
  {Scheidl}}, \bibinfo {author} {\bibfnamefont {Daqing}\ \bibnamefont {Wang}},
  \bibinfo {author} {\bibfnamefont {Sebastian}\ \bibnamefont {Kropatschek}},
  \bibinfo {author} {\bibfnamefont {William}\ \bibnamefont {Naylor}}, \bibinfo
  {author} {\bibfnamefont {Bernhard}\ \bibnamefont {Wittmann}}, \bibinfo
  {author} {\bibfnamefont {Alexandra}\ \bibnamefont {Mech}}, \bibinfo {author}
  {\bibfnamefont {Johannes}\ \bibnamefont {Kofler}}, \bibinfo {author}
  {\bibfnamefont {Elena}\ \bibnamefont {Anisimova}}, \bibinfo {author}
  {\bibfnamefont {Vadim}\ \bibnamefont {Makarov}}, \bibinfo {author}
  {\bibfnamefont {Thomas}\ \bibnamefont {Jennewein}}, \bibinfo {author}
  {\bibfnamefont {Rupert}\ \bibnamefont {Ursin}}, \ and\ \bibinfo {author}
  {\bibfnamefont {Anton}\ \bibnamefont {Zeilinger}},\ }\bibfield  {title}
  {\enquote {\bibinfo {title} {Quantum teleportation over 143 kilometres using
  active feed-forward},}\ }\href {\doibase 10.1038/nature11472} {\bibfield
  {journal} {\bibinfo  {journal} {Nature}\ }\textbf {\bibinfo {volume} {489}},\
  \bibinfo {pages} {269--273} (\bibinfo {year} {2012})}\BibitemShut {NoStop}%
\bibitem [{\citenamefont {Liao}\ \emph {et~al.}(2017)\citenamefont {Liao},
  \citenamefont {Cai}, \citenamefont {Liu}, \citenamefont {Zhang},
  \citenamefont {Li}, \citenamefont {Ren}, \citenamefont {Yin}, \citenamefont
  {Shen}, \citenamefont {Cao}, \citenamefont {Li}, \citenamefont {Li},
  \citenamefont {Chen}, \citenamefont {Sun}, \citenamefont {Jia}, \citenamefont
  {Wu}, \citenamefont {Jiang}, \citenamefont {Wang}, \citenamefont {Huang},
  \citenamefont {Wang}, \citenamefont {Zhou}, \citenamefont {Deng},
  \citenamefont {Xi}, \citenamefont {Ma}, \citenamefont {Hu}, \citenamefont
  {Zhang}, \citenamefont {Chen}, \citenamefont {Liu}, \citenamefont {Wang},
  \citenamefont {Zhu}, \citenamefont {Lu}, \citenamefont {Shu}, \citenamefont
  {Peng}, \citenamefont {Wang},\ and\ \citenamefont {Pan}}]{LCL+17}%
  \BibitemOpen
  \bibfield  {author} {\bibinfo {author} {\bibfnamefont {Sheng-Kai}\
  \bibnamefont {Liao}}, \bibinfo {author} {\bibfnamefont {Wen-Qi}\ \bibnamefont
  {Cai}}, \bibinfo {author} {\bibfnamefont {Wei-Yue}\ \bibnamefont {Liu}},
  \bibinfo {author} {\bibfnamefont {Liang}\ \bibnamefont {Zhang}}, \bibinfo
  {author} {\bibfnamefont {Yang}\ \bibnamefont {Li}}, \bibinfo {author}
  {\bibfnamefont {Ji-Gang}\ \bibnamefont {Ren}}, \bibinfo {author}
  {\bibfnamefont {Juan}\ \bibnamefont {Yin}}, \bibinfo {author} {\bibfnamefont
  {Qi}~\bibnamefont {Shen}}, \bibinfo {author} {\bibfnamefont {Yuan}\
  \bibnamefont {Cao}}, \bibinfo {author} {\bibfnamefont {Zheng-Ping}\
  \bibnamefont {Li}}, \bibinfo {author} {\bibfnamefont {Feng-Zhi}\ \bibnamefont
  {Li}}, \bibinfo {author} {\bibfnamefont {Xia-Wei}\ \bibnamefont {Chen}},
  \bibinfo {author} {\bibfnamefont {Li-Hua}\ \bibnamefont {Sun}}, \bibinfo
  {author} {\bibfnamefont {Jian-Jun}\ \bibnamefont {Jia}}, \bibinfo {author}
  {\bibfnamefont {Jin-Cai}\ \bibnamefont {Wu}}, \bibinfo {author}
  {\bibfnamefont {Xiao-Jun}\ \bibnamefont {Jiang}}, \bibinfo {author}
  {\bibfnamefont {Jian-Feng}\ \bibnamefont {Wang}}, \bibinfo {author}
  {\bibfnamefont {Yong-Mei}\ \bibnamefont {Huang}}, \bibinfo {author}
  {\bibfnamefont {Qiang}\ \bibnamefont {Wang}}, \bibinfo {author}
  {\bibfnamefont {Yi-Lin}\ \bibnamefont {Zhou}}, \bibinfo {author}
  {\bibfnamefont {Lei}\ \bibnamefont {Deng}}, \bibinfo {author} {\bibfnamefont
  {Tao}\ \bibnamefont {Xi}}, \bibinfo {author} {\bibfnamefont {Lu}~\bibnamefont
  {Ma}}, \bibinfo {author} {\bibfnamefont {Tai}\ \bibnamefont {Hu}}, \bibinfo
  {author} {\bibfnamefont {Qiang}\ \bibnamefont {Zhang}}, \bibinfo {author}
  {\bibfnamefont {Yu-Ao}\ \bibnamefont {Chen}}, \bibinfo {author}
  {\bibfnamefont {Nai-Le}\ \bibnamefont {Liu}}, \bibinfo {author}
  {\bibfnamefont {Xiang-Bin}\ \bibnamefont {Wang}}, \bibinfo {author}
  {\bibfnamefont {Zhen-Cai}\ \bibnamefont {Zhu}}, \bibinfo {author}
  {\bibfnamefont {Chao-Yang}\ \bibnamefont {Lu}}, \bibinfo {author}
  {\bibfnamefont {Rong}\ \bibnamefont {Shu}}, \bibinfo {author} {\bibfnamefont
  {Cheng-Zhi}\ \bibnamefont {Peng}}, \bibinfo {author} {\bibfnamefont
  {Jian-Yu}\ \bibnamefont {Wang}}, \ and\ \bibinfo {author} {\bibfnamefont
  {Jian-Wei}\ \bibnamefont {Pan}},\ }\bibfield  {title} {\enquote {\bibinfo
  {title} {Satellite-to-ground quantum key distribution},}\ }\href {\doibase
  10.1038/nature23655} {\bibfield  {journal} {\bibinfo  {journal} {Nature}\
  }\textbf {\bibinfo {volume} {549}},\ \bibinfo {pages} {43--47} (\bibinfo
  {year} {2017})},\ \bibinfo {note} {arXiv:1707.00542}\BibitemShut {NoStop}%
\bibitem [{\citenamefont {Azuma}\ \emph {et~al.}(2015)\citenamefont {Azuma},
  \citenamefont {Tamaki},\ and\ \citenamefont {Lo}}]{ATL15}%
  \BibitemOpen
  \bibfield  {author} {\bibinfo {author} {\bibfnamefont {Koji}\ \bibnamefont
  {Azuma}}, \bibinfo {author} {\bibfnamefont {Kiyoshi}\ \bibnamefont {Tamaki}},
  \ and\ \bibinfo {author} {\bibfnamefont {Hoi-Kwong}\ \bibnamefont {Lo}},\
  }\bibfield  {title} {\enquote {\bibinfo {title} {All-photonic quantum
  repeaters},}\ }\href {\doibase 10.1038/ncomms7787} {\bibfield  {journal}
  {\bibinfo  {journal} {Nature Communications}\ }\textbf {\bibinfo {volume}
  {6}} (\bibinfo {year} {2015}),\ 10.1038/ncomms7787},\ \bibinfo {note}
  {arXiv:1309.7207}\BibitemShut {NoStop}%
\bibitem [{\citenamefont {Briegel}\ \emph {et~al.}(1998)\citenamefont
  {Briegel}, \citenamefont {D\"ur}, \citenamefont {Cirac},\ and\ \citenamefont
  {Zoller}}]{PhysRevLett.81.5932}%
  \BibitemOpen
  \bibfield  {author} {\bibinfo {author} {\bibfnamefont {H.-J.}\ \bibnamefont
  {Briegel}}, \bibinfo {author} {\bibfnamefont {W.}~\bibnamefont {D\"ur}},
  \bibinfo {author} {\bibfnamefont {J.~I.}\ \bibnamefont {Cirac}}, \ and\
  \bibinfo {author} {\bibfnamefont {P.}~\bibnamefont {Zoller}},\ }\bibfield
  {title} {\enquote {\bibinfo {title} {Quantum repeaters: The role of imperfect
  local operations in quantum communication},}\ }\href {\doibase
  10.1103/PhysRevLett.81.5932} {\bibfield  {journal} {\bibinfo  {journal}
  {Physical Review Letters}\ }\textbf {\bibinfo {volume} {81}},\ \bibinfo
  {pages} {5932--5935} (\bibinfo {year} {1998})}\BibitemShut {NoStop}%
\bibitem [{\citenamefont {D\"ur}\ \emph {et~al.}(1999)\citenamefont {D\"ur},
  \citenamefont {Briegel}, \citenamefont {Cirac},\ and\ \citenamefont
  {Zoller}}]{PhysRevA.59.169}%
  \BibitemOpen
  \bibfield  {author} {\bibinfo {author} {\bibfnamefont {W.}~\bibnamefont
  {D\"ur}}, \bibinfo {author} {\bibfnamefont {H.-J.}\ \bibnamefont {Briegel}},
  \bibinfo {author} {\bibfnamefont {J.~I.}\ \bibnamefont {Cirac}}, \ and\
  \bibinfo {author} {\bibfnamefont {P.}~\bibnamefont {Zoller}},\ }\bibfield
  {title} {\enquote {\bibinfo {title} {Quantum repeaters based on entanglement
  purification},}\ }\href {\doibase 10.1103/PhysRevA.59.169} {\bibfield
  {journal} {\bibinfo  {journal} {Physical Review A}\ }\textbf {\bibinfo
  {volume} {59}},\ \bibinfo {pages} {169--181} (\bibinfo {year}
  {1999})}\BibitemShut {NoStop}%
\bibitem [{\citenamefont {Munro}\ \emph {et~al.}(2015)\citenamefont {Munro},
  \citenamefont {Azuma}, \citenamefont {Tamaki},\ and\ \citenamefont
  {Nemoto}}]{munro2015inside}%
  \BibitemOpen
  \bibfield  {author} {\bibinfo {author} {\bibfnamefont {William~J}\
  \bibnamefont {Munro}}, \bibinfo {author} {\bibfnamefont {Koji}\ \bibnamefont
  {Azuma}}, \bibinfo {author} {\bibfnamefont {Kiyoshi}\ \bibnamefont {Tamaki}},
  \ and\ \bibinfo {author} {\bibfnamefont {Kae}\ \bibnamefont {Nemoto}},\
  }\bibfield  {title} {\enquote {\bibinfo {title} {Inside quantum repeaters},}\
  }\href@noop {} {\bibfield  {journal} {\bibinfo  {journal} {IEEE Journal of
  Selected Topics in Quantum Electronics}\ }\textbf {\bibinfo {volume} {21}},\
  \bibinfo {pages} {78--90} (\bibinfo {year} {2015})}\BibitemShut {NoStop}%
\bibitem [{\citenamefont {Makarov}\ \emph {et~al.}(2006)\citenamefont
  {Makarov}, \citenamefont {Anisimov},\ and\ \citenamefont {Skaar}}]{MAS06}%
  \BibitemOpen
  \bibfield  {author} {\bibinfo {author} {\bibfnamefont {Vadim}\ \bibnamefont
  {Makarov}}, \bibinfo {author} {\bibfnamefont {Andrey}\ \bibnamefont
  {Anisimov}}, \ and\ \bibinfo {author} {\bibfnamefont {Johannes}\ \bibnamefont
  {Skaar}},\ }\bibfield  {title} {\enquote {\bibinfo {title} {Effects of
  detector efficiency mismatch on security of quantum cryptosystems},}\ }\href
  {\doibase 10.1103/PhysRevA.74.022313} {\bibfield  {journal} {\bibinfo
  {journal} {Physical Review A}\ }\textbf {\bibinfo {volume} {74}},\ \bibinfo
  {pages} {022313} (\bibinfo {year} {2006})}\BibitemShut {NoStop}%
\bibitem [{\citenamefont {Qi}\ \emph {et~al.}(2007)\citenamefont {Qi},
  \citenamefont {Fung}, \citenamefont {Lo},\ and\ \citenamefont {Ma}}]{QFLM05}%
  \BibitemOpen
  \bibfield  {author} {\bibinfo {author} {\bibfnamefont {Bing}\ \bibnamefont
  {Qi}}, \bibinfo {author} {\bibfnamefont {Chi-Hang~Fred}\ \bibnamefont
  {Fung}}, \bibinfo {author} {\bibfnamefont {Hoi-Kwong}\ \bibnamefont {Lo}}, \
  and\ \bibinfo {author} {\bibfnamefont {Xiongfeng}\ \bibnamefont {Ma}},\
  }\bibfield  {title} {\enquote {\bibinfo {title} {Time-shift attack in
  practical quantum cryptosystems},}\ }\href@noop {} {\bibfield  {journal}
  {\bibinfo  {journal} {Quantum Information and Computation}\ ,\ \bibinfo
  {pages} {073--082}} (\bibinfo {year} {2007})},\ \bibinfo {note} {arXiv:
  quant-ph/0512080}\BibitemShut {NoStop}%
\bibitem [{\citenamefont {Diamanti}\ \emph {et~al.}(2016)\citenamefont
  {Diamanti}, \citenamefont {Lo}, \citenamefont {Qi},\ and\ \citenamefont
  {Yuan}}]{DLQY16}%
  \BibitemOpen
  \bibfield  {author} {\bibinfo {author} {\bibfnamefont {Eleni}\ \bibnamefont
  {Diamanti}}, \bibinfo {author} {\bibfnamefont {Hoi-Kwong}\ \bibnamefont
  {Lo}}, \bibinfo {author} {\bibfnamefont {Bing}\ \bibnamefont {Qi}}, \ and\
  \bibinfo {author} {\bibfnamefont {Zhiliang}\ \bibnamefont {Yuan}},\
  }\bibfield  {title} {\enquote {\bibinfo {title} {Practical challenges in
  quantum key distribution},}\ }\href@noop {} {\bibfield  {journal} {\bibinfo
  {journal} {npj Quantum Information}\ }\textbf {\bibinfo {volume} {2}},\
  \bibinfo {pages} {16025} (\bibinfo {year} {2016})}\BibitemShut {NoStop}%
\bibitem [{\citenamefont {Lo}\ \emph {et~al.}(2012)\citenamefont {Lo},
  \citenamefont {Curty},\ and\ \citenamefont {Qi}}]{LCQ+12}%
  \BibitemOpen
  \bibfield  {author} {\bibinfo {author} {\bibfnamefont {Hoi-Kwong}\
  \bibnamefont {Lo}}, \bibinfo {author} {\bibfnamefont {Marcos}\ \bibnamefont
  {Curty}}, \ and\ \bibinfo {author} {\bibfnamefont {Bing}\ \bibnamefont
  {Qi}},\ }\bibfield  {title} {\enquote {\bibinfo {title}
  {Measurement-device-independent quantum key distribution},}\ }\href {\doibase
  10.1103/PhysRevLett.108.130503} {\bibfield  {journal} {\bibinfo  {journal}
  {Physical Review Letters}\ }\textbf {\bibinfo {volume} {108}},\ \bibinfo
  {pages} {130503} (\bibinfo {year} {2012})}\BibitemShut {NoStop}%
\bibitem [{\citenamefont {Braunstein}\ and\ \citenamefont
  {Pirandola}(2012)}]{braunstein2012side}%
  \BibitemOpen
  \bibfield  {author} {\bibinfo {author} {\bibfnamefont {Samuel~L}\
  \bibnamefont {Braunstein}}\ and\ \bibinfo {author} {\bibfnamefont {Stefano}\
  \bibnamefont {Pirandola}},\ }\bibfield  {title} {\enquote {\bibinfo {title}
  {Side-channel-free quantum key distribution},}\ }\href@noop {} {\bibfield
  {journal} {\bibinfo  {journal} {Physical Review Letters}\ }\textbf {\bibinfo
  {volume} {108}},\ \bibinfo {pages} {130502} (\bibinfo {year}
  {2012})}\BibitemShut {NoStop}%
\bibitem [{\citenamefont {Pirandola}\ \emph {et~al.}(2015)\citenamefont
  {Pirandola}, \citenamefont {Ottaviani}, \citenamefont {Spedalieri},
  \citenamefont {Weedbrook}, \citenamefont {Braunstein}, \citenamefont {Lloyd},
  \citenamefont {Gehring}, \citenamefont {Jacobsen},\ and\ \citenamefont
  {Andersen}}]{pirandola2015high}%
  \BibitemOpen
  \bibfield  {author} {\bibinfo {author} {\bibfnamefont {Stefano}\ \bibnamefont
  {Pirandola}}, \bibinfo {author} {\bibfnamefont {Carlo}\ \bibnamefont
  {Ottaviani}}, \bibinfo {author} {\bibfnamefont {Gaetana}\ \bibnamefont
  {Spedalieri}}, \bibinfo {author} {\bibfnamefont {Christian}\ \bibnamefont
  {Weedbrook}}, \bibinfo {author} {\bibfnamefont {Samuel~L}\ \bibnamefont
  {Braunstein}}, \bibinfo {author} {\bibfnamefont {Seth}\ \bibnamefont
  {Lloyd}}, \bibinfo {author} {\bibfnamefont {Tobias}\ \bibnamefont {Gehring}},
  \bibinfo {author} {\bibfnamefont {Christian~S}\ \bibnamefont {Jacobsen}}, \
  and\ \bibinfo {author} {\bibfnamefont {Ulrik~L}\ \bibnamefont {Andersen}},\
  }\bibfield  {title} {\enquote {\bibinfo {title} {High-rate
  measurement-device-independent quantum cryptography},}\ }\href@noop {}
  {\bibfield  {journal} {\bibinfo  {journal} {Nature Photonics}\ }\textbf
  {\bibinfo {volume} {9}},\ \bibinfo {pages} {397--402} (\bibinfo {year}
  {2015})}\BibitemShut {NoStop}%
\bibitem [{\citenamefont {Fu}\ \emph {et~al.}(2015)\citenamefont {Fu},
  \citenamefont {Yin}, \citenamefont {Chen},\ and\ \citenamefont
  {Chen}}]{FYCC15}%
  \BibitemOpen
  \bibfield  {author} {\bibinfo {author} {\bibfnamefont {Yao}\ \bibnamefont
  {Fu}}, \bibinfo {author} {\bibfnamefont {Hua-Lei}\ \bibnamefont {Yin}},
  \bibinfo {author} {\bibfnamefont {Teng-Yun}\ \bibnamefont {Chen}}, \ and\
  \bibinfo {author} {\bibfnamefont {Zeng-Bing}\ \bibnamefont {Chen}},\
  }\bibfield  {title} {\enquote {\bibinfo {title} {Long-distance
  measurement-device-independent multiparty quantum communication},}\ }\href
  {\doibase 10.1103/PhysRevLett.114.090501} {\bibfield  {journal} {\bibinfo
  {journal} {Physical Review Letters}\ }\textbf {\bibinfo {volume} {114}},\
  \bibinfo {pages} {090501} (\bibinfo {year} {2015})}\BibitemShut {NoStop}%
\bibitem [{\citenamefont {Lucamarini}\ \emph {et~al.}(2018)\citenamefont
  {Lucamarini}, \citenamefont {Yuan}, \citenamefont {Dynes},\ and\
  \citenamefont {Shields}}]{LYDS18}%
  \BibitemOpen
  \bibfield  {author} {\bibinfo {author} {\bibfnamefont {Marco}\ \bibnamefont
  {Lucamarini}}, \bibinfo {author} {\bibfnamefont {Zhiliang~L}\ \bibnamefont
  {Yuan}}, \bibinfo {author} {\bibfnamefont {James~F}\ \bibnamefont {Dynes}}, \
  and\ \bibinfo {author} {\bibfnamefont {Andrew~J}\ \bibnamefont {Shields}},\
  }\bibfield  {title} {\enquote {\bibinfo {title} {Overcoming the
  rate--distance limit of quantum key distribution without quantum
  repeaters},}\ }\href@noop {} {\bibfield  {journal} {\bibinfo  {journal}
  {Nature}\ }\textbf {\bibinfo {volume} {557}},\ \bibinfo {pages} {400}
  (\bibinfo {year} {2018})}\BibitemShut {NoStop}%
\bibitem [{\citenamefont {Ma}\ \emph {et~al.}(2018)\citenamefont {Ma},
  \citenamefont {Zeng},\ and\ \citenamefont {Zhou}}]{MZZ18}%
  \BibitemOpen
  \bibfield  {author} {\bibinfo {author} {\bibfnamefont {Xiongfeng}\
  \bibnamefont {Ma}}, \bibinfo {author} {\bibfnamefont {Pei}\ \bibnamefont
  {Zeng}}, \ and\ \bibinfo {author} {\bibfnamefont {Hongyi}\ \bibnamefont
  {Zhou}},\ }\bibfield  {title} {\enquote {\bibinfo {title} {Phase-matching
  quantum key distribution},}\ }\href {\doibase 10.1103/PhysRevX.8.031043}
  {\bibfield  {journal} {\bibinfo  {journal} {Physical Review X}\ }\textbf
  {\bibinfo {volume} {8}},\ \bibinfo {pages} {031043} (\bibinfo {year}
  {2018})}\BibitemShut {NoStop}%
\bibitem [{\citenamefont {Tamaki}\ \emph {et~al.}(2018)\citenamefont {Tamaki},
  \citenamefont {Lo}, \citenamefont {Wang},\ and\ \citenamefont
  {Lucamarini}}]{tamaki2018information}%
  \BibitemOpen
  \bibfield  {author} {\bibinfo {author} {\bibfnamefont {Kiyoshi}\ \bibnamefont
  {Tamaki}}, \bibinfo {author} {\bibfnamefont {Hoi-Kwong}\ \bibnamefont {Lo}},
  \bibinfo {author} {\bibfnamefont {Wenyuan}\ \bibnamefont {Wang}}, \ and\
  \bibinfo {author} {\bibfnamefont {Marco}\ \bibnamefont {Lucamarini}},\
  }\bibfield  {title} {\enquote {\bibinfo {title} {Information theoretic
  security of quantum key distribution overcoming the repeaterless secret key
  capacity bound},}\ }\href@noop {} {\bibfield  {journal} {\bibinfo  {journal}
  {arXiv preprint arXiv:1805.05511}\ } (\bibinfo {year} {2018})}\BibitemShut
  {NoStop}%
\bibitem [{\citenamefont {Lin}\ and\ \citenamefont
  {L\"utkenhaus}(2018)}]{LL18}%
  \BibitemOpen
  \bibfield  {author} {\bibinfo {author} {\bibfnamefont {Jie}\ \bibnamefont
  {Lin}}\ and\ \bibinfo {author} {\bibfnamefont {Norbert}\ \bibnamefont
  {L\"utkenhaus}},\ }\bibfield  {title} {\enquote {\bibinfo {title} {Simple
  security analysis of phase-matching measurement-device-independent quantum
  key distribution},}\ }\href {\doibase 10.1103/PhysRevA.98.042332} {\bibfield
  {journal} {\bibinfo  {journal} {Physical Review A}\ }\textbf {\bibinfo
  {volume} {98}},\ \bibinfo {pages} {042332} (\bibinfo {year}
  {2018})}\BibitemShut {NoStop}%
\bibitem [{\citenamefont {Cui}\ \emph {et~al.}(2019)\citenamefont {Cui},
  \citenamefont {Yin}, \citenamefont {Wang}, \citenamefont {Chen},
  \citenamefont {Wang}, \citenamefont {Guo},\ and\ \citenamefont
  {Han}}]{cui2019twin}%
  \BibitemOpen
  \bibfield  {author} {\bibinfo {author} {\bibfnamefont {Chaohan}\ \bibnamefont
  {Cui}}, \bibinfo {author} {\bibfnamefont {Zhen-Qiang}\ \bibnamefont {Yin}},
  \bibinfo {author} {\bibfnamefont {Rong}\ \bibnamefont {Wang}}, \bibinfo
  {author} {\bibfnamefont {Wei}\ \bibnamefont {Chen}}, \bibinfo {author}
  {\bibfnamefont {Shuang}\ \bibnamefont {Wang}}, \bibinfo {author}
  {\bibfnamefont {Guang-Can}\ \bibnamefont {Guo}}, \ and\ \bibinfo {author}
  {\bibfnamefont {Zheng-Fu}\ \bibnamefont {Han}},\ }\bibfield  {title}
  {\enquote {\bibinfo {title} {Twin-field quantum key distribution without
  phase postselection},}\ }\href@noop {} {\bibfield  {journal} {\bibinfo
  {journal} {Physical Review Applied}\ }\textbf {\bibinfo {volume} {11}},\
  \bibinfo {pages} {034053} (\bibinfo {year} {2019})}\BibitemShut {NoStop}%
\bibitem [{\citenamefont {Curty}\ \emph {et~al.}(2019)\citenamefont {Curty},
  \citenamefont {Azuma},\ and\ \citenamefont {Lo}}]{curty2019simple}%
  \BibitemOpen
  \bibfield  {author} {\bibinfo {author} {\bibfnamefont {Marcos}\ \bibnamefont
  {Curty}}, \bibinfo {author} {\bibfnamefont {Koji}\ \bibnamefont {Azuma}}, \
  and\ \bibinfo {author} {\bibfnamefont {Hoi-Kwong}\ \bibnamefont {Lo}},\
  }\bibfield  {title} {\enquote {\bibinfo {title} {Simple security proof of
  twin-field type quantum key distribution protocol},}\ }\href@noop {}
  {\bibfield  {journal} {\bibinfo  {journal} {npj Quantum Information}\
  }\textbf {\bibinfo {volume} {5}},\ \bibinfo {pages} {1--6} (\bibinfo {year}
  {2019})}\BibitemShut {NoStop}%
\bibitem [{\citenamefont {Liu}\ \emph {et~al.}(2019)\citenamefont {Liu},
  \citenamefont {Wang}, \citenamefont {Wei}, \citenamefont {Fang},
  \citenamefont {Li}, \citenamefont {Liu}, \citenamefont {Liang}, \citenamefont
  {Zhang}, \citenamefont {Zhang}, \citenamefont {Li}, \citenamefont {You},
  \citenamefont {Wang}, \citenamefont {Lo}, \citenamefont {Chen}, \citenamefont
  {Xu},\ and\ \citenamefont {Pan}}]{LWW+19}%
  \BibitemOpen
  \bibfield  {author} {\bibinfo {author} {\bibfnamefont {Hui}\ \bibnamefont
  {Liu}}, \bibinfo {author} {\bibfnamefont {Wenyuan}\ \bibnamefont {Wang}},
  \bibinfo {author} {\bibfnamefont {Kejin}\ \bibnamefont {Wei}}, \bibinfo
  {author} {\bibfnamefont {Xiao-Tian}\ \bibnamefont {Fang}}, \bibinfo {author}
  {\bibfnamefont {Li}~\bibnamefont {Li}}, \bibinfo {author} {\bibfnamefont
  {Nai-Le}\ \bibnamefont {Liu}}, \bibinfo {author} {\bibfnamefont {Hao}\
  \bibnamefont {Liang}}, \bibinfo {author} {\bibfnamefont {Si-Jie}\
  \bibnamefont {Zhang}}, \bibinfo {author} {\bibfnamefont {Weijun}\
  \bibnamefont {Zhang}}, \bibinfo {author} {\bibfnamefont {Hao}\ \bibnamefont
  {Li}}, \bibinfo {author} {\bibfnamefont {Lixing}\ \bibnamefont {You}},
  \bibinfo {author} {\bibfnamefont {Zhen}\ \bibnamefont {Wang}}, \bibinfo
  {author} {\bibfnamefont {Hoi-Kwong}\ \bibnamefont {Lo}}, \bibinfo {author}
  {\bibfnamefont {Teng-Yun}\ \bibnamefont {Chen}}, \bibinfo {author}
  {\bibfnamefont {Feihu}\ \bibnamefont {Xu}}, \ and\ \bibinfo {author}
  {\bibfnamefont {Jian-Wei}\ \bibnamefont {Pan}},\ }\bibfield  {title}
  {\enquote {\bibinfo {title} {Experimental demonstration of high-rate
  measurement-device-independent quantum key distribution over asymmetric
  channels},}\ }\href {\doibase 10.1103/PhysRevLett.122.160501} {\bibfield
  {journal} {\bibinfo  {journal} {Physical Review Letters}\ }\textbf {\bibinfo
  {volume} {122}},\ \bibinfo {pages} {160501} (\bibinfo {year}
  {2019})}\BibitemShut {NoStop}%
\bibitem [{\citenamefont {Minder}\ \emph {et~al.}(2019)\citenamefont {Minder},
  \citenamefont {Pittaluga}, \citenamefont {Roberts}, \citenamefont
  {Lucamarini}, \citenamefont {Dynes}, \citenamefont {Yuan},\ and\
  \citenamefont {Shields}}]{MPR+19}%
  \BibitemOpen
  \bibfield  {author} {\bibinfo {author} {\bibfnamefont {M}~\bibnamefont
  {Minder}}, \bibinfo {author} {\bibfnamefont {M}~\bibnamefont {Pittaluga}},
  \bibinfo {author} {\bibfnamefont {GL}~\bibnamefont {Roberts}}, \bibinfo
  {author} {\bibfnamefont {M}~\bibnamefont {Lucamarini}}, \bibinfo {author}
  {\bibfnamefont {JF}~\bibnamefont {Dynes}}, \bibinfo {author} {\bibfnamefont
  {ZL}~\bibnamefont {Yuan}}, \ and\ \bibinfo {author} {\bibfnamefont
  {AJ}~\bibnamefont {Shields}},\ }\bibfield  {title} {\enquote {\bibinfo
  {title} {Experimental quantum key distribution beyond the repeaterless secret
  key capacity},}\ }\href@noop {} {\bibfield  {journal} {\bibinfo  {journal}
  {Nature Photonics}\ }\textbf {\bibinfo {volume} {13}},\ \bibinfo {pages}
  {334} (\bibinfo {year} {2019})}\BibitemShut {NoStop}%
\bibitem [{\citenamefont {Primaatmaja}\ \emph {et~al.}(2019)\citenamefont
  {Primaatmaja}, \citenamefont {Lavie}, \citenamefont {Goh}, \citenamefont
  {Wang},\ and\ \citenamefont {Lim}}]{PLG+19}%
  \BibitemOpen
  \bibfield  {author} {\bibinfo {author} {\bibfnamefont {Ignatius~William}\
  \bibnamefont {Primaatmaja}}, \bibinfo {author} {\bibfnamefont {Emilien}\
  \bibnamefont {Lavie}}, \bibinfo {author} {\bibfnamefont {Koon~Tong}\
  \bibnamefont {Goh}}, \bibinfo {author} {\bibfnamefont {Chao}\ \bibnamefont
  {Wang}}, \ and\ \bibinfo {author} {\bibfnamefont {Charles Ci~Wen}\
  \bibnamefont {Lim}},\ }\bibfield  {title} {\enquote {\bibinfo {title}
  {Versatile security analysis of measurement-device-independent quantum key
  distribution},}\ }\href {\doibase 10.1103/PhysRevA.99.062332} {\bibfield
  {journal} {\bibinfo  {journal} {Physical Review A}\ }\textbf {\bibinfo
  {volume} {99}},\ \bibinfo {pages} {062332} (\bibinfo {year}
  {2019})}\BibitemShut {NoStop}%
\bibitem [{\citenamefont {Horodecki}\ \emph {et~al.}(2005)\citenamefont
  {Horodecki}, \citenamefont {Horodecki}, \citenamefont {Horodecki},\ and\
  \citenamefont {Oppenheim}}]{HHHO05}%
  \BibitemOpen
  \bibfield  {author} {\bibinfo {author} {\bibfnamefont {Karol}\ \bibnamefont
  {Horodecki}}, \bibinfo {author} {\bibfnamefont {Micha\l{}}\ \bibnamefont
  {Horodecki}}, \bibinfo {author} {\bibfnamefont {Pawe\l{}}\ \bibnamefont
  {Horodecki}}, \ and\ \bibinfo {author} {\bibfnamefont {Jonathan}\
  \bibnamefont {Oppenheim}},\ }\bibfield  {title} {\enquote {\bibinfo {title}
  {Secure key from bound entanglement},}\ }\href {\doibase
  10.1103/PhysRevLett.94.160502} {\bibfield  {journal} {\bibinfo  {journal}
  {Physical Review Letters}\ }\textbf {\bibinfo {volume} {94}},\ \bibinfo
  {pages} {160502} (\bibinfo {year} {2005})},\ \bibinfo {note}
  {arXiv:quant-ph/0309110}\BibitemShut {NoStop}%
\bibitem [{\citenamefont {Christandl}\ and\ \citenamefont
  {Winter}(2004{\natexlab{a}})}]{christandl2004squashed}%
  \BibitemOpen
  \bibfield  {author} {\bibinfo {author} {\bibfnamefont {Matthias}\
  \bibnamefont {Christandl}}\ and\ \bibinfo {author} {\bibfnamefont {Andreas}\
  \bibnamefont {Winter}},\ }\bibfield  {title} {\enquote {\bibinfo {title}
  {Squashed entanglement: an additive entanglement measure},}\ }\href@noop {}
  {\bibfield  {journal} {\bibinfo  {journal} {Journal of mathematical physics}\
  }\textbf {\bibinfo {volume} {45}},\ \bibinfo {pages} {829--840} (\bibinfo
  {year} {2004}{\natexlab{a}})}\BibitemShut {NoStop}%
\bibitem [{\citenamefont {Bennett}\ \emph
  {et~al.}(1996{\natexlab{a}})\citenamefont {Bennett}, \citenamefont
  {DiVincenzo}, \citenamefont {Smolin},\ and\ \citenamefont
  {Wootters}}]{BDSW96}%
  \BibitemOpen
  \bibfield  {author} {\bibinfo {author} {\bibfnamefont {Charles~H.}\
  \bibnamefont {Bennett}}, \bibinfo {author} {\bibfnamefont {David~P.}\
  \bibnamefont {DiVincenzo}}, \bibinfo {author} {\bibfnamefont {John~A.}\
  \bibnamefont {Smolin}}, \ and\ \bibinfo {author} {\bibfnamefont {William~K.}\
  \bibnamefont {Wootters}},\ }\bibfield  {title} {\enquote {\bibinfo {title}
  {Mixed-state entanglement and quantum error correction},}\ }\href {\doibase
  10.1103/PhysRevA.54.3824} {\bibfield  {journal} {\bibinfo  {journal}
  {Physical Review A}\ }\textbf {\bibinfo {volume} {54}},\ \bibinfo {pages}
  {3824--3851} (\bibinfo {year} {1996}{\natexlab{a}})},\ \bibinfo {note}
  {arXiv:quant-ph/9604024}\BibitemShut {NoStop}%
\bibitem [{\citenamefont {Vedral}\ \emph {et~al.}(1997)\citenamefont {Vedral},
  \citenamefont {Plenio}, \citenamefont {Rippin},\ and\ \citenamefont
  {Knight}}]{VPRK97}%
  \BibitemOpen
  \bibfield  {author} {\bibinfo {author} {\bibfnamefont {V.}~\bibnamefont
  {Vedral}}, \bibinfo {author} {\bibfnamefont {M.~B.}\ \bibnamefont {Plenio}},
  \bibinfo {author} {\bibfnamefont {M.~A.}\ \bibnamefont {Rippin}}, \ and\
  \bibinfo {author} {\bibfnamefont {P.~L.}\ \bibnamefont {Knight}},\ }\bibfield
   {title} {\enquote {\bibinfo {title} {Quantifying entanglement},}\ }\href
  {\doibase 10.1103/PhysRevLett.78.2275} {\bibfield  {journal} {\bibinfo
  {journal} {Physical Review Letters}\ }\textbf {\bibinfo {volume} {78}},\
  \bibinfo {pages} {2275--2279} (\bibinfo {year} {1997})},\ \bibinfo {note}
  {arXiv:quant-ph/9702027}\BibitemShut {NoStop}%
\bibitem [{\citenamefont {Vedral}\ and\ \citenamefont {Plenio}(1998)}]{VP98}%
  \BibitemOpen
  \bibfield  {author} {\bibinfo {author} {\bibfnamefont {Vlatko}\ \bibnamefont
  {Vedral}}\ and\ \bibinfo {author} {\bibfnamefont {Martin~B.}\ \bibnamefont
  {Plenio}},\ }\bibfield  {title} {\enquote {\bibinfo {title} {Entanglement
  measures and purification procedures},}\ }\href@noop {} {\bibfield  {journal}
  {\bibinfo  {journal} {Physical Review A}\ }\textbf {\bibinfo {volume} {57}},\
  \bibinfo {pages} {1619--1633} (\bibinfo {year} {1998})},\ \bibinfo {note}
  {arXiv:quant-ph/9707035}\BibitemShut {NoStop}%
\bibitem [{\citenamefont {Horodecki}\ \emph {et~al.}(1999)\citenamefont
  {Horodecki}, \citenamefont {Horodecki},\ and\ \citenamefont
  {Horodecki}}]{HHH99}%
  \BibitemOpen
  \bibfield  {author} {\bibinfo {author} {\bibfnamefont {Micha\l{}}\
  \bibnamefont {Horodecki}}, \bibinfo {author} {\bibfnamefont {Pawe\l{}}\
  \bibnamefont {Horodecki}}, \ and\ \bibinfo {author} {\bibfnamefont {Ryszard}\
  \bibnamefont {Horodecki}},\ }\bibfield  {title} {\enquote {\bibinfo {title}
  {General teleportation channel, singlet fraction, and quasidistillation},}\
  }\href {\doibase 10.1103/PhysRevA.60.1888} {\bibfield  {journal} {\bibinfo
  {journal} {Physical Review A}\ }\textbf {\bibinfo {volume} {60}},\ \bibinfo
  {pages} {1888--1898} (\bibinfo {year} {1999})},\ \bibinfo {note}
  {arXiv:quant-ph/9807091}\BibitemShut {NoStop}%
\bibitem [{\citenamefont {Datta}(2009{\natexlab{a}})}]{Dat09}%
  \BibitemOpen
  \bibfield  {author} {\bibinfo {author} {\bibfnamefont {Nilanjana}\
  \bibnamefont {Datta}},\ }\bibfield  {title} {\enquote {\bibinfo {title}
  {Max-relative entropy of entanglement, alias log robustness},}\ }\href@noop
  {} {\bibfield  {journal} {\bibinfo  {journal} {International Journal of
  Quantum Information}\ }\textbf {\bibinfo {volume} {7}},\ \bibinfo {pages}
  {475--491} (\bibinfo {year} {2009}{\natexlab{a}})},\ \bibinfo {note}
  {arXiv:0807.2536}\BibitemShut {NoStop}%
\bibitem [{\citenamefont {Takeoka}\ \emph {et~al.}(2014)\citenamefont
  {Takeoka}, \citenamefont {Guha},\ and\ \citenamefont {Wilde}}]{TGW14}%
  \BibitemOpen
  \bibfield  {author} {\bibinfo {author} {\bibfnamefont {Masahiro}\
  \bibnamefont {Takeoka}}, \bibinfo {author} {\bibfnamefont {Saikat}\
  \bibnamefont {Guha}}, \ and\ \bibinfo {author} {\bibfnamefont {Mark~M.}\
  \bibnamefont {Wilde}},\ }\bibfield  {title} {\enquote {\bibinfo {title}
  {Fundamental rate-loss tradeoff for optical quantum key distribution},}\
  }\href {\doibase 10.1038/ncomms6235} {\bibfield  {journal} {\bibinfo
  {journal} {Nature Communications}\ }\textbf {\bibinfo {volume} {5}} (\bibinfo
  {year} {2014}),\ 10.1038/ncomms6235}\BibitemShut {NoStop}%
\bibitem [{\citenamefont {{Pirandola}}\ \emph {et~al.}(2017)\citenamefont
  {{Pirandola}}, \citenamefont {{Laurenza}}, \citenamefont {{Ottaviani}},\ and\
  \citenamefont {{Banchi}}}]{PLOB15}%
  \BibitemOpen
  \bibfield  {author} {\bibinfo {author} {\bibfnamefont {Stefano}\ \bibnamefont
  {{Pirandola}}}, \bibinfo {author} {\bibfnamefont {Riccardo}\ \bibnamefont
  {{Laurenza}}}, \bibinfo {author} {\bibfnamefont {Carlo}\ \bibnamefont
  {{Ottaviani}}}, \ and\ \bibinfo {author} {\bibfnamefont {Leonardo}\
  \bibnamefont {{Banchi}}},\ }\bibfield  {title} {\enquote {\bibinfo {title}
  {Fundamental limits of repeaterless quantum communications},}\ }\href@noop {}
  {\bibfield  {journal} {\bibinfo  {journal} {Nature Communications}\ }\textbf
  {\bibinfo {volume} {8}},\ \bibinfo {pages} {15043} (\bibinfo {year}
  {2017})},\ \bibinfo {note} {see also, arXiv:1512.04945}\BibitemShut {NoStop}%
\bibitem [{\citenamefont {Wilde}\ \emph {et~al.}(2017)\citenamefont {Wilde},
  \citenamefont {Tomamichel},\ and\ \citenamefont {Berta}}]{WTB16}%
  \BibitemOpen
  \bibfield  {author} {\bibinfo {author} {\bibfnamefont {Mark~M.}\ \bibnamefont
  {Wilde}}, \bibinfo {author} {\bibfnamefont {Marco}\ \bibnamefont
  {Tomamichel}}, \ and\ \bibinfo {author} {\bibfnamefont {Mario}\ \bibnamefont
  {Berta}},\ }\bibfield  {title} {\enquote {\bibinfo {title} {Converse bounds
  for private communication over quantum channels},}\ }\href@noop {} {\bibfield
   {journal} {\bibinfo  {journal} {IEEE Transactions on Information Theory}\
  }\textbf {\bibinfo {volume} {63}},\ \bibinfo {pages} {1792--1817} (\bibinfo
  {year} {2017})},\ \bibinfo {note} {arXiv:1602.08898}\BibitemShut {NoStop}%
\bibitem [{\citenamefont {Christandl}\ and\ \citenamefont
  {M\"{u}ller-Hermes}(2017)}]{CM17}%
  \BibitemOpen
  \bibfield  {author} {\bibinfo {author} {\bibfnamefont {Matthias}\
  \bibnamefont {Christandl}}\ and\ \bibinfo {author} {\bibfnamefont
  {Alexander}\ \bibnamefont {M\"{u}ller-Hermes}},\ }\bibfield  {title}
  {\enquote {\bibinfo {title} {Relative entropy bounds on quantum, private and
  repeater capacities},}\ }\href {\doibase 10.1007/s00220-017-2885-y}
  {\bibfield  {journal} {\bibinfo  {journal} {Communications in Mathematical
  Physics}\ }\textbf {\bibinfo {volume} {353}},\ \bibinfo {pages} {821--852}
  (\bibinfo {year} {2017})},\ \bibinfo {note} {arXiv:1604.03448}\BibitemShut
  {NoStop}%
\bibitem [{\citenamefont {Das}\ \emph {et~al.}(2020)\citenamefont {Das},
  \citenamefont {B\"auml},\ and\ \citenamefont {Wilde}}]{DBW17}%
  \BibitemOpen
  \bibfield  {author} {\bibinfo {author} {\bibfnamefont {Siddhartha}\
  \bibnamefont {Das}}, \bibinfo {author} {\bibfnamefont {Stefan}\ \bibnamefont
  {B\"auml}}, \ and\ \bibinfo {author} {\bibfnamefont {Mark~M.}\ \bibnamefont
  {Wilde}},\ }\bibfield  {title} {\enquote {\bibinfo {title} {Entanglement and
  secret-key-agreement capacities of bipartite quantum interactions and
  read-only memory devices},}\ }\href {\doibase 10.1103/PhysRevA.101.012344}
  {\bibfield  {journal} {\bibinfo  {journal} {Physical Review A}\ }\textbf
  {\bibinfo {volume} {101}},\ \bibinfo {pages} {012344} (\bibinfo {year}
  {2020})},\ \bibinfo {note} {arXiv:1712.00827}\BibitemShut {NoStop}%
\bibitem [{\citenamefont {B\"auml}\ \emph {et~al.}(2018)\citenamefont
  {B\"auml}, \citenamefont {Das},\ and\ \citenamefont {Wilde}}]{BDW18}%
  \BibitemOpen
  \bibfield  {author} {\bibinfo {author} {\bibfnamefont {Stefan}\ \bibnamefont
  {B\"auml}}, \bibinfo {author} {\bibfnamefont {Siddhartha}\ \bibnamefont
  {Das}}, \ and\ \bibinfo {author} {\bibfnamefont {Mark~M.}\ \bibnamefont
  {Wilde}},\ }\bibfield  {title} {\enquote {\bibinfo {title} {Fundamental
  limits on the capacities of bipartite quantum interactions},}\ }\href
  {\doibase 10.1103/PhysRevLett.121.250504} {\bibfield  {journal} {\bibinfo
  {journal} {Physical Review Letters}\ }\textbf {\bibinfo {volume} {121}},\
  \bibinfo {pages} {250504} (\bibinfo {year} {2018})},\ \bibinfo {note}
  {arXiv:1812.08223}\BibitemShut {NoStop}%
\bibitem [{\citenamefont {Das}(2018)}]{D18thesis}%
  \BibitemOpen
  \bibfield  {author} {\bibinfo {author} {\bibfnamefont {Siddhartha}\
  \bibnamefont {Das}},\ }\emph {\bibinfo {title} {Bipartite Quantum
  Interactions: Entangling and Information Processing Abilities}},\ \href
  {https://arxiv.org/abs/1901.05895} {Ph.D. thesis},\ \bibinfo  {school}
  {Louisiana State University} (\bibinfo {year} {2018}),\ \bibinfo {note}
  {arXiv:1901.05895}\BibitemShut {NoStop}%
\bibitem [{\citenamefont {Laurenza}\ and\ \citenamefont
  {Pirandola}(2017)}]{LP17}%
  \BibitemOpen
  \bibfield  {author} {\bibinfo {author} {\bibfnamefont {Riccardo}\
  \bibnamefont {Laurenza}}\ and\ \bibinfo {author} {\bibfnamefont {Stefano}\
  \bibnamefont {Pirandola}},\ }\bibfield  {title} {\enquote {\bibinfo {title}
  {General bounds for sender-receiver capacities in multipoint quantum
  communications},}\ }\href {\doibase 10.1103/PhysRevA.96.032318} {\bibfield
  {journal} {\bibinfo  {journal} {Physical Review A}\ }\textbf {\bibinfo
  {volume} {96}},\ \bibinfo {pages} {032318} (\bibinfo {year} {2017})},\
  \bibinfo {note} {arXiv:1603.07262}\BibitemShut {NoStop}%
\bibitem [{\citenamefont {Seshadreesan}\ \emph {et~al.}(2016)\citenamefont
  {Seshadreesan}, \citenamefont {Takeoka},\ and\ \citenamefont
  {Wilde}}]{seshadreesan2016bounds}%
  \BibitemOpen
  \bibfield  {author} {\bibinfo {author} {\bibfnamefont {Kaushik~P}\
  \bibnamefont {Seshadreesan}}, \bibinfo {author} {\bibfnamefont {Masahiro}\
  \bibnamefont {Takeoka}}, \ and\ \bibinfo {author} {\bibfnamefont {Mark~M}\
  \bibnamefont {Wilde}},\ }\bibfield  {title} {\enquote {\bibinfo {title}
  {Bounds on entanglement distillation and secret key agreement for quantum
  broadcast channels},}\ }\href@noop {} {\bibfield  {journal} {\bibinfo
  {journal} {IEEE Transactions on Information Theory}\ }\textbf {\bibinfo
  {volume} {62}},\ \bibinfo {pages} {2849--2866} (\bibinfo {year}
  {2016})}\BibitemShut {NoStop}%
\bibitem [{\citenamefont {Takeoka}\ \emph {et~al.}(2017)\citenamefont
  {Takeoka}, \citenamefont {Seshadreesan},\ and\ \citenamefont
  {Wilde}}]{TSW17}%
  \BibitemOpen
  \bibfield  {author} {\bibinfo {author} {\bibfnamefont {Masahiro}\
  \bibnamefont {Takeoka}}, \bibinfo {author} {\bibfnamefont {Kaushik~P.}\
  \bibnamefont {Seshadreesan}}, \ and\ \bibinfo {author} {\bibfnamefont
  {Mark~M.}\ \bibnamefont {Wilde}},\ }\bibfield  {title} {\enquote {\bibinfo
  {title} {Unconstrained capacities of quantum key distribution and
  entanglement distillation for pure-loss bosonic broadcast channels},}\ }\href
  {\doibase 10.1103/PhysRevLett.119.150501} {\bibfield  {journal} {\bibinfo
  {journal} {Physical Review Letters}\ }\textbf {\bibinfo {volume} {119}},\
  \bibinfo {pages} {150501} (\bibinfo {year} {2017})}\BibitemShut {NoStop}%
\bibitem [{\citenamefont {B{\"a}uml}\ \emph {et~al.}(2015)\citenamefont
  {B{\"a}uml}, \citenamefont {Christandl}, \citenamefont {Horodecki},\ and\
  \citenamefont {Winter}}]{bauml2015limitations}%
  \BibitemOpen
  \bibfield  {author} {\bibinfo {author} {\bibfnamefont {Stefan}\ \bibnamefont
  {B{\"a}uml}}, \bibinfo {author} {\bibfnamefont {Matthias}\ \bibnamefont
  {Christandl}}, \bibinfo {author} {\bibfnamefont {Karol}\ \bibnamefont
  {Horodecki}}, \ and\ \bibinfo {author} {\bibfnamefont {Andreas}\ \bibnamefont
  {Winter}},\ }\bibfield  {title} {\enquote {\bibinfo {title} {Limitations on
  quantum key repeaters},}\ }\href@noop {} {\bibfield  {journal} {\bibinfo
  {journal} {Nature communications}\ }\textbf {\bibinfo {volume} {6}},\
  \bibinfo {pages} {6908} (\bibinfo {year} {2015})}\BibitemShut {NoStop}%
\bibitem [{\citenamefont {Azuma}\ \emph {et~al.}(2016)\citenamefont {Azuma},
  \citenamefont {Mizutani},\ and\ \citenamefont {Lo}}]{AML16}%
  \BibitemOpen
  \bibfield  {author} {\bibinfo {author} {\bibfnamefont {Koji}\ \bibnamefont
  {Azuma}}, \bibinfo {author} {\bibfnamefont {Akihiro}\ \bibnamefont
  {Mizutani}}, \ and\ \bibinfo {author} {\bibfnamefont {Hoi-Kwong}\
  \bibnamefont {Lo}},\ }\bibfield  {title} {\enquote {\bibinfo {title}
  {Fundamental rate-loss tradeoff for the quantum internet},}\ }\href@noop {}
  {\bibfield  {journal} {\bibinfo  {journal} {Nature Communications}\ }\textbf
  {\bibinfo {volume} {7}},\ \bibinfo {pages} {13523} (\bibinfo {year}
  {2016})}\BibitemShut {NoStop}%
\bibitem [{\citenamefont {Rigovacca}\ \emph {et~al.}(2018)\citenamefont
  {Rigovacca}, \citenamefont {Kato}, \citenamefont {B\"{a}uml}, \citenamefont
  {Kim}, \citenamefont {Munro},\ and\ \citenamefont {Azuma}}]{RKB+17}%
  \BibitemOpen
  \bibfield  {author} {\bibinfo {author} {\bibfnamefont {Luca}\ \bibnamefont
  {Rigovacca}}, \bibinfo {author} {\bibfnamefont {Go}~\bibnamefont {Kato}},
  \bibinfo {author} {\bibfnamefont {Stefan}\ \bibnamefont {B\"{a}uml}},
  \bibinfo {author} {\bibfnamefont {Myunghik}\ \bibnamefont {Kim}}, \bibinfo
  {author} {\bibfnamefont {William~J.}\ \bibnamefont {Munro}}, \ and\ \bibinfo
  {author} {\bibfnamefont {Koji}\ \bibnamefont {Azuma}},\ }\bibfield  {title}
  {\enquote {\bibinfo {title} {Versatile relative entropy bounds for quantum
  networks},}\ }\href@noop {} {\bibfield  {journal} {\bibinfo  {journal} {New
  Journal of Physics}\ }\textbf {\bibinfo {volume} {20}},\ \bibinfo {pages}
  {013033} (\bibinfo {year} {2018})},\ \bibinfo {note}
  {arXiv:1707.05543}\BibitemShut {NoStop}%
\bibitem [{\citenamefont {Pirandola}(2019)}]{pirandola2019capacities}%
  \BibitemOpen
  \bibfield  {author} {\bibinfo {author} {\bibfnamefont {Stefano}\ \bibnamefont
  {Pirandola}},\ }\bibfield  {title} {\enquote {\bibinfo {title} {End-to-end
  capacities of a quantum communication network},}\ }\href@noop {} {\bibfield
  {journal} {\bibinfo  {journal} {Communications Physics}\ }\textbf {\bibinfo
  {volume} {2}},\ \bibinfo {pages} {51} (\bibinfo {year} {2019})},\ \bibinfo
  {note} {see also arXiv:1601.00966}\BibitemShut {NoStop}%
\bibitem [{\citenamefont {B{\"a}uml}\ and\ \citenamefont
  {Azuma}(2017)}]{bauml2017fundamental}%
  \BibitemOpen
  \bibfield  {author} {\bibinfo {author} {\bibfnamefont {Stefan}\ \bibnamefont
  {B{\"a}uml}}\ and\ \bibinfo {author} {\bibfnamefont {Koji}\ \bibnamefont
  {Azuma}},\ }\bibfield  {title} {\enquote {\bibinfo {title} {Fundamental
  limitation on quantum broadcast networks},}\ }\href@noop {} {\bibfield
  {journal} {\bibinfo  {journal} {Quantum Science and Technology}\ }\textbf
  {\bibinfo {volume} {2}},\ \bibinfo {pages} {024004} (\bibinfo {year}
  {2017})}\BibitemShut {NoStop}%
\bibitem [{\citenamefont {Horodecki}\ \emph
  {et~al.}(2009{\natexlab{a}})\citenamefont {Horodecki}, \citenamefont
  {Horodecki}, \citenamefont {Horodecki},\ and\ \citenamefont
  {Oppenheim}}]{HHHO09}%
  \BibitemOpen
  \bibfield  {author} {\bibinfo {author} {\bibfnamefont {Karol}\ \bibnamefont
  {Horodecki}}, \bibinfo {author} {\bibfnamefont {Michal}\ \bibnamefont
  {Horodecki}}, \bibinfo {author} {\bibfnamefont {Pawel}\ \bibnamefont
  {Horodecki}}, \ and\ \bibinfo {author} {\bibfnamefont {Jonathan}\
  \bibnamefont {Oppenheim}},\ }\bibfield  {title} {\enquote {\bibinfo {title}
  {General paradigm for distilling classical key from quantum states},}\
  }\href@noop {} {\bibfield  {journal} {\bibinfo  {journal} {IEEE Transactions
  on Information Theory}\ }\textbf {\bibinfo {volume} {55}},\ \bibinfo {pages}
  {1898--1929} (\bibinfo {year} {2009}{\natexlab{a}})},\ \bibinfo {note}
  {arXiv:quant-ph/0506189}\BibitemShut {NoStop}%
\bibitem [{\citenamefont {Wilde}\ \emph {et~al.}(2014)\citenamefont {Wilde},
  \citenamefont {Winter},\ and\ \citenamefont {Yang}}]{WWY14}%
  \BibitemOpen
  \bibfield  {author} {\bibinfo {author} {\bibfnamefont {Mark~M.}\ \bibnamefont
  {Wilde}}, \bibinfo {author} {\bibfnamefont {Andreas}\ \bibnamefont {Winter}},
  \ and\ \bibinfo {author} {\bibfnamefont {Dong}\ \bibnamefont {Yang}},\
  }\bibfield  {title} {\enquote {\bibinfo {title} {Strong converse for the
  classical capacity of entanglement-breaking and {Hadamard} channels via a
  sandwiched {R\'enyi} relative entropy},}\ }\href@noop {} {\bibfield
  {journal} {\bibinfo  {journal} {Communications in Mathematical Physics}\
  }\textbf {\bibinfo {volume} {331}},\ \bibinfo {pages} {593--622} (\bibinfo
  {year} {2014})},\ \bibinfo {note} {arXiv:1306.1586}\BibitemShut {NoStop}%
\bibitem [{\citenamefont {{M\"uller}-Lennert}\ \emph
  {et~al.}(2013)\citenamefont {{M\"uller}-Lennert}, \citenamefont {Dupuis},
  \citenamefont {Szehr}, \citenamefont {Fehr},\ and\ \citenamefont
  {Tomamichel}}]{MDSFT13}%
  \BibitemOpen
  \bibfield  {author} {\bibinfo {author} {\bibfnamefont {Martin}\ \bibnamefont
  {{M\"uller}-Lennert}}, \bibinfo {author} {\bibfnamefont {Fr\'ed\'eric}\
  \bibnamefont {Dupuis}}, \bibinfo {author} {\bibfnamefont {Oleg}\ \bibnamefont
  {Szehr}}, \bibinfo {author} {\bibfnamefont {Serge}\ \bibnamefont {Fehr}}, \
  and\ \bibinfo {author} {\bibfnamefont {Marco}\ \bibnamefont {Tomamichel}},\
  }\bibfield  {title} {\enquote {\bibinfo {title} {On quantum {R\'enyi}
  entropies: a new definition and some properties},}\ }\href@noop {} {\bibfield
   {journal} {\bibinfo  {journal} {Journal of Mathematical Physics}\ }\textbf
  {\bibinfo {volume} {54}},\ \bibinfo {pages} {122203} (\bibinfo {year}
  {2013})},\ \bibinfo {note} {arXiv:1306.3142}\BibitemShut {NoStop}%
\bibitem [{\citenamefont {Buscemi}\ and\ \citenamefont {Datta}(2010)}]{BD10}%
  \BibitemOpen
  \bibfield  {author} {\bibinfo {author} {\bibfnamefont {Francesco}\
  \bibnamefont {Buscemi}}\ and\ \bibinfo {author} {\bibfnamefont {Nilanjana}\
  \bibnamefont {Datta}},\ }\bibfield  {title} {\enquote {\bibinfo {title} {The
  quantum capacity of channels with arbitrarily correlated noise},}\ }\href
  {\doibase 10.1109/TIT.2009.2039166} {\bibfield  {journal} {\bibinfo
  {journal} {IEEE Transactions on Information Theory}\ }\textbf {\bibinfo
  {volume} {56}},\ \bibinfo {pages} {1447--1460} (\bibinfo {year} {2010})},\
  \bibinfo {note} {arXiv:0902.0158}\BibitemShut {NoStop}%
\bibitem [{\citenamefont {Bennett}\ \emph
  {et~al.}(1996{\natexlab{b}})\citenamefont {Bennett}, \citenamefont
  {Bernstein}, \citenamefont {Popescu},\ and\ \citenamefont
  {Schumacher}}]{BBPS96}%
  \BibitemOpen
  \bibfield  {author} {\bibinfo {author} {\bibfnamefont {Charles~H.}\
  \bibnamefont {Bennett}}, \bibinfo {author} {\bibfnamefont {Herbert~J.}\
  \bibnamefont {Bernstein}}, \bibinfo {author} {\bibfnamefont {Sandu}\
  \bibnamefont {Popescu}}, \ and\ \bibinfo {author} {\bibfnamefont {Benjamin}\
  \bibnamefont {Schumacher}},\ }\bibfield  {title} {\enquote {\bibinfo {title}
  {Concentrating partial entanglement by local operations},}\ }\href {\doibase
  10.1103/PhysRevA.53.2046} {\bibfield  {journal} {\bibinfo  {journal}
  {Physical Review A}\ }\textbf {\bibinfo {volume} {53}},\ \bibinfo {pages}
  {2046--2052} (\bibinfo {year} {1996}{\natexlab{b}})},\ \bibinfo {note}
  {arXiv:quant-ph/9511030}\BibitemShut {NoStop}%
\bibitem [{\citenamefont {Das}\ \emph {et~al.}(2018)\citenamefont {Das},
  \citenamefont {Khatri},\ and\ \citenamefont {Dowling}}]{DKD18}%
  \BibitemOpen
  \bibfield  {author} {\bibinfo {author} {\bibfnamefont {Siddhartha}\
  \bibnamefont {Das}}, \bibinfo {author} {\bibfnamefont {Sumeet}\ \bibnamefont
  {Khatri}}, \ and\ \bibinfo {author} {\bibfnamefont {Jonathan~P.}\
  \bibnamefont {Dowling}},\ }\bibfield  {title} {\enquote {\bibinfo {title}
  {Robust quantum network architectures and topologies for entanglement
  distribution},}\ }\href {\doibase 10.1103/physreva.97.012335} {\bibfield
  {journal} {\bibinfo  {journal} {Physical Review A}\ }\textbf {\bibinfo
  {volume} {97}},\ \bibinfo {pages} {012335} (\bibinfo {year} {2018})},\
  \bibinfo {note} {arXiv:1709.07404}\BibitemShut {NoStop}%
\bibitem [{\citenamefont {Devetak}\ and\ \citenamefont {Winter}(2005)}]{DW05}%
  \BibitemOpen
  \bibfield  {author} {\bibinfo {author} {\bibfnamefont {Igor}\ \bibnamefont
  {Devetak}}\ and\ \bibinfo {author} {\bibfnamefont {Andreas}\ \bibnamefont
  {Winter}},\ }\bibfield  {title} {\enquote {\bibinfo {title} {Distillation of
  secret key and entanglement from quantum states},}\ }\href {\doibase
  10.1098/rspa.2004.1372} {\bibfield  {journal} {\bibinfo  {journal}
  {Proceedings of the Royal Society of London A: Mathematical, Physical and
  Engineering Sciences}\ }\textbf {\bibinfo {volume} {461}},\ \bibinfo {pages}
  {207--235} (\bibinfo {year} {2005})},\ \bibinfo {note}
  {arXiv:quant-ph/0306078}\BibitemShut {NoStop}%
\bibitem [{\citenamefont {Pirandola}\ \emph {et~al.}(2009)\citenamefont
  {Pirandola}, \citenamefont {Garc{\'\i}a-Patr{\'o}n}, \citenamefont
  {Braunstein},\ and\ \citenamefont {Lloyd}}]{pirandola2009direct}%
  \BibitemOpen
  \bibfield  {author} {\bibinfo {author} {\bibfnamefont {Stefano}\ \bibnamefont
  {Pirandola}}, \bibinfo {author} {\bibfnamefont {Raul}\ \bibnamefont
  {Garc{\'\i}a-Patr{\'o}n}}, \bibinfo {author} {\bibfnamefont {Samuel~L}\
  \bibnamefont {Braunstein}}, \ and\ \bibinfo {author} {\bibfnamefont {Seth}\
  \bibnamefont {Lloyd}},\ }\bibfield  {title} {\enquote {\bibinfo {title}
  {Direct and reverse secret-key capacities of a quantum channel},}\
  }\href@noop {} {\bibfield  {journal} {\bibinfo  {journal} {Physical Review
  Letters}\ }\textbf {\bibinfo {volume} {102}},\ \bibinfo {pages} {050503}
  (\bibinfo {year} {2009})}\BibitemShut {NoStop}%
\bibitem [{\citenamefont {Christandl}\ and\ \citenamefont
  {Winter}(2004{\natexlab{b}})}]{CW04}%
  \BibitemOpen
  \bibfield  {author} {\bibinfo {author} {\bibfnamefont {Matthias}\
  \bibnamefont {Christandl}}\ and\ \bibinfo {author} {\bibfnamefont {Andreas}\
  \bibnamefont {Winter}},\ }\bibfield  {title} {\enquote {\bibinfo {title}
  {``{S}quashed entanglement'': An additive entanglement measure},}\
  }\href@noop {} {\bibfield  {journal} {\bibinfo  {journal} {Journal of
  Mathematical Physics}\ }\textbf {\bibinfo {volume} {45}},\ \bibinfo {pages}
  {829--840} (\bibinfo {year} {2004}{\natexlab{b}})},\ \bibinfo {note}
  {arXiv:quant-ph/0308088}\BibitemShut {NoStop}%
\bibitem [{\citenamefont {Tucci}(1999)}]{Tuc99}%
  \BibitemOpen
  \bibfield  {author} {\bibinfo {author} {\bibfnamefont {Robert~R.}\
  \bibnamefont {Tucci}},\ }\bibfield  {title} {\enquote {\bibinfo {title}
  {Quantum entanglement and conditional information transmission},}\
  }\href@noop {} {\  (\bibinfo {year} {1999})},\ \bibinfo {note} {arXiv:
  quant-ph/9909041}\BibitemShut {NoStop}%
\bibitem [{\citenamefont {Tucci}(2002)}]{Tuc02}%
  \BibitemOpen
  \bibfield  {author} {\bibinfo {author} {\bibfnamefont {Robert~R.}\
  \bibnamefont {Tucci}},\ }\bibfield  {title} {\enquote {\bibinfo {title}
  {Entanglement of distillation and conditional mutual information},}\
  }\href@noop {} {\  (\bibinfo {year} {2002})},\ \bibinfo {note}
  {arXiv:quant-ph/0202144}\BibitemShut {NoStop}%
\bibitem [{\citenamefont {Yang}\ \emph {et~al.}(2009)\citenamefont {Yang},
  \citenamefont {Horodecki}, \citenamefont {Horodecki}, \citenamefont
  {Horodecki}, \citenamefont {Oppenheim},\ and\ \citenamefont {Song}}]{YHH+09}%
  \BibitemOpen
  \bibfield  {author} {\bibinfo {author} {\bibfnamefont {Dong}\ \bibnamefont
  {Yang}}, \bibinfo {author} {\bibfnamefont {Karol}\ \bibnamefont {Horodecki}},
  \bibinfo {author} {\bibfnamefont {Michal}\ \bibnamefont {Horodecki}},
  \bibinfo {author} {\bibfnamefont {Pawel}\ \bibnamefont {Horodecki}}, \bibinfo
  {author} {\bibfnamefont {Jonathan}\ \bibnamefont {Oppenheim}}, \ and\
  \bibinfo {author} {\bibfnamefont {Wei}\ \bibnamefont {Song}},\ }\bibfield
  {title} {\enquote {\bibinfo {title} {Squashed entanglement for multipartite
  states and entanglement measures based on the mixed convex roof},}\ }\href
  {\doibase 10.1109/tit.2009.2021373} {\bibfield  {journal} {\bibinfo
  {journal} {{IEEE} Transactions on Information Theory}\ }\textbf {\bibinfo
  {volume} {55}},\ \bibinfo {pages} {3375--3387} (\bibinfo {year} {2009})},\
  \bibinfo {note} {arXiv:0704.2236}\BibitemShut {NoStop}%
\bibitem [{\citenamefont {Fang}\ and\ \citenamefont {Fawzi}(2019)}]{FF19}%
  \BibitemOpen
  \bibfield  {author} {\bibinfo {author} {\bibfnamefont {Kun}\ \bibnamefont
  {Fang}}\ and\ \bibinfo {author} {\bibfnamefont {Hamza}\ \bibnamefont
  {Fawzi}},\ }\bibfield  {title} {\enquote {\bibinfo {title} {Geometric
  {R\'enyi} divergence and its applications in quantum channel capacities},}\
  }\href@noop {} {\  (\bibinfo {year} {2019})},\ \bibinfo {note}
  {arXiv:1909.05758}\BibitemShut {NoStop}%
\bibitem [{\citenamefont {Azuma}\ and\ \citenamefont
  {Kato}(2017)}]{azuma2017aggregating}%
  \BibitemOpen
  \bibfield  {author} {\bibinfo {author} {\bibfnamefont {Koji}\ \bibnamefont
  {Azuma}}\ and\ \bibinfo {author} {\bibfnamefont {Go}~\bibnamefont {Kato}},\
  }\bibfield  {title} {\enquote {\bibinfo {title} {Aggregating quantum
  repeaters for the quantum internet},}\ }\href@noop {} {\bibfield  {journal}
  {\bibinfo  {journal} {Physical Review A}\ }\textbf {\bibinfo {volume} {96}},\
  \bibinfo {pages} {032332} (\bibinfo {year} {2017})}\BibitemShut {NoStop}%
\bibitem [{\citenamefont {B{\"a}uml}\ \emph {et~al.}(2020)\citenamefont
  {B{\"a}uml}, \citenamefont {Azuma}, \citenamefont {Kato},\ and\ \citenamefont
  {Elkouss}}]{bauml2018linear}%
  \BibitemOpen
  \bibfield  {author} {\bibinfo {author} {\bibfnamefont {Stefan}\ \bibnamefont
  {B{\"a}uml}}, \bibinfo {author} {\bibfnamefont {Koji}\ \bibnamefont {Azuma}},
  \bibinfo {author} {\bibfnamefont {Go}~\bibnamefont {Kato}}, \ and\ \bibinfo
  {author} {\bibfnamefont {David}\ \bibnamefont {Elkouss}},\ }\bibfield
  {title} {\enquote {\bibinfo {title} {Linear programs for entanglement and key
  distribution in the quantum internet},}\ }\href@noop {} {\bibfield  {journal}
  {\bibinfo  {journal} {Communications Physics}\ }\textbf {\bibinfo {volume}
  {3}},\ \bibinfo {pages} {1--12} (\bibinfo {year} {2020})}\BibitemShut
  {NoStop}%
\bibitem [{\citenamefont {Christandl}\ and\ \citenamefont
  {Ferrara}(2017)}]{christandl2017private}%
  \BibitemOpen
  \bibfield  {author} {\bibinfo {author} {\bibfnamefont {Matthias}\
  \bibnamefont {Christandl}}\ and\ \bibinfo {author} {\bibfnamefont {Roberto}\
  \bibnamefont {Ferrara}},\ }\bibfield  {title} {\enquote {\bibinfo {title}
  {Private states, quantum data hiding, and the swapping of perfect secrecy},}\
  }\href@noop {} {\bibfield  {journal} {\bibinfo  {journal} {Physical Review
  Letters}\ }\textbf {\bibinfo {volume} {119}},\ \bibinfo {pages} {220506}
  (\bibinfo {year} {2017})}\BibitemShut {NoStop}%
\bibitem [{\citenamefont {van Loock}\ \emph {et~al.}(2020)\citenamefont {van
  Loock}, \citenamefont {Alt}, \citenamefont {Becher}, \citenamefont {Benson},
  \citenamefont {Boche}, \citenamefont {Deppe}, \citenamefont {Eschner},
  \citenamefont {H{\"o}fling}, \citenamefont {Meschede}, \citenamefont
  {Michler} \emph {et~al.}}]{van2020extending}%
  \BibitemOpen
  \bibfield  {author} {\bibinfo {author} {\bibfnamefont {Peter}\ \bibnamefont
  {van Loock}}, \bibinfo {author} {\bibfnamefont {Wolfgang}\ \bibnamefont
  {Alt}}, \bibinfo {author} {\bibfnamefont {Christoph}\ \bibnamefont {Becher}},
  \bibinfo {author} {\bibfnamefont {Oliver}\ \bibnamefont {Benson}}, \bibinfo
  {author} {\bibfnamefont {Holger}\ \bibnamefont {Boche}}, \bibinfo {author}
  {\bibfnamefont {Christian}\ \bibnamefont {Deppe}}, \bibinfo {author}
  {\bibfnamefont {J{\"u}rgen}\ \bibnamefont {Eschner}}, \bibinfo {author}
  {\bibfnamefont {Sven}\ \bibnamefont {H{\"o}fling}}, \bibinfo {author}
  {\bibfnamefont {Dieter}\ \bibnamefont {Meschede}}, \bibinfo {author}
  {\bibfnamefont {Peter}\ \bibnamefont {Michler}},  \emph {et~al.},\ }\bibfield
   {title} {\enquote {\bibinfo {title} {Extending quantum links: Modules for
  fiber-and memory-based quantum repeaters},}\ }\href@noop {} {\bibfield
  {journal} {\bibinfo  {journal} {Advanced Quantum Technologies}\ }\textbf
  {\bibinfo {volume} {3}},\ \bibinfo {pages} {1900141} (\bibinfo {year}
  {2020})}\BibitemShut {NoStop}%
\bibitem [{\citenamefont {Bennett}\ \emph {et~al.}(2000)\citenamefont
  {Bennett}, \citenamefont {Popescu}, \citenamefont {Rohrlich}, \citenamefont
  {Smolin},\ and\ \citenamefont {Thapliyal}}]{bennett2000exact}%
  \BibitemOpen
  \bibfield  {author} {\bibinfo {author} {\bibfnamefont {Charles~H}\
  \bibnamefont {Bennett}}, \bibinfo {author} {\bibfnamefont {Sandu}\
  \bibnamefont {Popescu}}, \bibinfo {author} {\bibfnamefont {Daniel}\
  \bibnamefont {Rohrlich}}, \bibinfo {author} {\bibfnamefont {John~A}\
  \bibnamefont {Smolin}}, \ and\ \bibinfo {author} {\bibfnamefont {Ashish~V}\
  \bibnamefont {Thapliyal}},\ }\bibfield  {title} {\enquote {\bibinfo {title}
  {Exact and asymptotic measures of multipartite pure-state entanglement},}\
  }\href@noop {} {\bibfield  {journal} {\bibinfo  {journal} {Physical Review
  A}\ }\textbf {\bibinfo {volume} {63}},\ \bibinfo {pages} {012307} (\bibinfo
  {year} {2000})}\BibitemShut {NoStop}%
\bibitem [{\citenamefont {{Smolin}}\ \emph {et~al.}(2005)\citenamefont
  {{Smolin}}, \citenamefont {{Verstraete}},\ and\ \citenamefont
  {{Winter}}}]{SVW}%
  \BibitemOpen
  \bibfield  {author} {\bibinfo {author} {\bibfnamefont {John~A.}\ \bibnamefont
  {{Smolin}}}, \bibinfo {author} {\bibfnamefont {Frank}\ \bibnamefont
  {{Verstraete}}}, \ and\ \bibinfo {author} {\bibfnamefont {Andreas}\
  \bibnamefont {{Winter}}},\ }\bibfield  {title} {\enquote {\bibinfo {title}
  {{Entanglement of assistance and multipartite state distillation}},}\ }\href
  {\doibase 10.1103/PhysRevA.72.052317} {\bibfield  {journal} {\bibinfo
  {journal} {\pra}\ }\textbf {\bibinfo {volume} {72}},\ \bibinfo {eid} {052317}
  (\bibinfo {year} {2005})},\ \Eprint {http://arxiv.org/abs/quant-ph/0505038}
  {arXiv:quant-ph/0505038 [quant-ph]} \BibitemShut {NoStop}%
\bibitem [{\citenamefont {Fortescue}\ and\ \citenamefont {Lo}(2007)}]{FL-2007}%
  \BibitemOpen
  \bibfield  {author} {\bibinfo {author} {\bibfnamefont {Ben}\ \bibnamefont
  {Fortescue}}\ and\ \bibinfo {author} {\bibfnamefont {Hoi-Kwong}\ \bibnamefont
  {Lo}},\ }\bibfield  {title} {\enquote {\bibinfo {title} {Random bipartite
  entanglement from {W} and {W}-like states},}\ }\href {\doibase
  10.1103/PhysRevLett.98.260501} {\bibfield  {journal} {\bibinfo  {journal}
  {Physical Review Letters}\ }\textbf {\bibinfo {volume} {98}},\ \bibinfo
  {pages} {260501} (\bibinfo {year} {2007})}\BibitemShut {NoStop}%
\bibitem [{\citenamefont {Kıntaş}\ and\ \citenamefont
  {Turgut}(2010)}]{KT-2010}%
  \BibitemOpen
  \bibfield  {author} {\bibinfo {author} {\bibfnamefont {S.}~\bibnamefont
  {Kıntaş}}\ and\ \bibinfo {author} {\bibfnamefont {S.}~\bibnamefont
  {Turgut}},\ }\bibfield  {title} {\enquote {\bibinfo {title} {Transformations
  of {W}-type entangled states},}\ }\href {\doibase 10.1063/1.3481573}
  {\bibfield  {journal} {\bibinfo  {journal} {Journal of Mathematical Physics}\
  }\textbf {\bibinfo {volume} {51}},\ \bibinfo {pages} {092202} (\bibinfo
  {year} {2010})},\ \Eprint
  {http://arxiv.org/abs/https://doi.org/10.1063/1.3481573}
  {https://doi.org/10.1063/1.3481573} \BibitemShut {NoStop}%
\bibitem [{\citenamefont {{Cui}}\ \emph {et~al.}(2011)\citenamefont {{Cui}},
  \citenamefont {{Chitambar}},\ and\ \citenamefont {{Lo}}}]{CCL-2011}%
  \BibitemOpen
  \bibfield  {author} {\bibinfo {author} {\bibfnamefont {W.}~\bibnamefont
  {{Cui}}}, \bibinfo {author} {\bibfnamefont {E.}~\bibnamefont {{Chitambar}}},
  \ and\ \bibinfo {author} {\bibfnamefont {H.~K.}\ \bibnamefont {{Lo}}},\
  }\bibfield  {title} {\enquote {\bibinfo {title} {{Randomly distilling
  {W}-class states into general configurations of two-party entanglement}},}\
  }\href {\doibase 10.1103/PhysRevA.84.052301} {\bibfield  {journal} {\bibinfo
  {journal} {Physical Review A}\ }\textbf {\bibinfo {volume} {84}},\ \bibinfo
  {eid} {052301} (\bibinfo {year} {2011})},\ \bibinfo {note}
  {arXiv:1106.1209}\BibitemShut {NoStop}%
\bibitem [{\citenamefont {Vrana}\ and\ \citenamefont
  {Christandl}(2015)}]{vrana2015asymptotic}%
  \BibitemOpen
  \bibfield  {author} {\bibinfo {author} {\bibfnamefont {P{\'e}ter}\
  \bibnamefont {Vrana}}\ and\ \bibinfo {author} {\bibfnamefont {Matthias}\
  \bibnamefont {Christandl}},\ }\bibfield  {title} {\enquote {\bibinfo {title}
  {Asymptotic entanglement transformation between {W} and {GHZ} states},}\
  }\href@noop {} {\bibfield  {journal} {\bibinfo  {journal} {Journal of
  Mathematical Physics}\ }\textbf {\bibinfo {volume} {56}},\ \bibinfo {pages}
  {022204} (\bibinfo {year} {2015})}\BibitemShut {NoStop}%
\bibitem [{\citenamefont {Spee}\ \emph {et~al.}(2017)\citenamefont {Spee},
  \citenamefont {de~Vicente}, \citenamefont {Sauerwein},\ and\ \citenamefont
  {Kraus}}]{spee2017entangled}%
  \BibitemOpen
  \bibfield  {author} {\bibinfo {author} {\bibfnamefont {C}~\bibnamefont
  {Spee}}, \bibinfo {author} {\bibfnamefont {JI}~\bibnamefont {de~Vicente}},
  \bibinfo {author} {\bibfnamefont {D}~\bibnamefont {Sauerwein}}, \ and\
  \bibinfo {author} {\bibfnamefont {B}~\bibnamefont {Kraus}},\ }\bibfield
  {title} {\enquote {\bibinfo {title} {Entangled pure state transformations via
  local operations assisted by finitely many rounds of classical
  communication},}\ }\href@noop {} {\bibfield  {journal} {\bibinfo  {journal}
  {Physical Review Letters}\ }\textbf {\bibinfo {volume} {118}},\ \bibinfo
  {pages} {040503} (\bibinfo {year} {2017})}\BibitemShut {NoStop}%
\bibitem [{\citenamefont {Vrana}\ and\ \citenamefont
  {Christandl}(2019)}]{vrana2019distillation}%
  \BibitemOpen
  \bibfield  {author} {\bibinfo {author} {\bibfnamefont {P{\'e}ter}\
  \bibnamefont {Vrana}}\ and\ \bibinfo {author} {\bibfnamefont {Matthias}\
  \bibnamefont {Christandl}},\ }\bibfield  {title} {\enquote {\bibinfo {title}
  {Distillation of greenberger--horne--zeilinger states by combinatorial
  methods},}\ }\href@noop {} {\bibfield  {journal} {\bibinfo  {journal} {IEEE
  Transactions on Information Theory}\ }\textbf {\bibinfo {volume} {65}},\
  \bibinfo {pages} {5945--5958} (\bibinfo {year} {2019})}\BibitemShut {NoStop}%
\bibitem [{\citenamefont {Streltsov}\ \emph {et~al.}(2020)\citenamefont
  {Streltsov}, \citenamefont {Meignant},\ and\ \citenamefont
  {Eisert}}]{streltsov2020rates}%
  \BibitemOpen
  \bibfield  {author} {\bibinfo {author} {\bibfnamefont {Alexander}\
  \bibnamefont {Streltsov}}, \bibinfo {author} {\bibfnamefont {Clement}\
  \bibnamefont {Meignant}}, \ and\ \bibinfo {author} {\bibfnamefont {Jens}\
  \bibnamefont {Eisert}},\ }\bibfield  {title} {\enquote {\bibinfo {title}
  {Rates of multipartite entanglement transformations},}\ }\href@noop {}
  {\bibfield  {journal} {\bibinfo  {journal} {Physical Review Letters}\
  }\textbf {\bibinfo {volume} {125}},\ \bibinfo {pages} {080502} (\bibinfo
  {year} {2020})}\BibitemShut {NoStop}%
\bibitem [{\citenamefont {Einstein}\ \emph {et~al.}(1935)\citenamefont
  {Einstein}, \citenamefont {Podolsky},\ and\ \citenamefont {Rosen}}]{EPR35}%
  \BibitemOpen
  \bibfield  {author} {\bibinfo {author} {\bibfnamefont {Albert}\ \bibnamefont
  {Einstein}}, \bibinfo {author} {\bibfnamefont {Boris}\ \bibnamefont
  {Podolsky}}, \ and\ \bibinfo {author} {\bibfnamefont {Nathan}\ \bibnamefont
  {Rosen}},\ }\bibfield  {title} {\enquote {\bibinfo {title} {Can
  quantum-mechanical description of physical reality be considered complete?}}\
  }\href {\doibase 10.1103/PhysRev.47.777} {\bibfield  {journal} {\bibinfo
  {journal} {Physical Review}\ }\textbf {\bibinfo {volume} {47}},\ \bibinfo
  {pages} {777--780} (\bibinfo {year} {1935})}\BibitemShut {NoStop}%
\bibitem [{\citenamefont {Chitambar}\ \emph {et~al.}(2014)\citenamefont
  {Chitambar}, \citenamefont {Leung}, \citenamefont {Man{\v{c}}inska},
  \citenamefont {Ozols},\ and\ \citenamefont {Winter}}]{CLM+14}%
  \BibitemOpen
  \bibfield  {author} {\bibinfo {author} {\bibfnamefont {Eric}\ \bibnamefont
  {Chitambar}}, \bibinfo {author} {\bibfnamefont {Debbie}\ \bibnamefont
  {Leung}}, \bibinfo {author} {\bibfnamefont {Laura}\ \bibnamefont
  {Man{\v{c}}inska}}, \bibinfo {author} {\bibfnamefont {Maris}\ \bibnamefont
  {Ozols}}, \ and\ \bibinfo {author} {\bibfnamefont {Andreas}\ \bibnamefont
  {Winter}},\ }\bibfield  {title} {\enquote {\bibinfo {title} {Everything you
  always wanted to know about {LOCC} (but were afraid to ask)},}\ }\href
  {\doibase 10.1007/s00220-014-1953-9} {\bibfield  {journal} {\bibinfo
  {journal} {Communications in Mathematical Physics}\ }\textbf {\bibinfo
  {volume} {328}},\ \bibinfo {pages} {303--326} (\bibinfo {year} {2014})},\
  \bibinfo {note} {arXiv:1210.4583}\BibitemShut {NoStop}%
\bibitem [{\citenamefont {Polyanskiy}\ and\ \citenamefont
  {Verd\'u}(2010)}]{PV10}%
  \BibitemOpen
  \bibfield  {author} {\bibinfo {author} {\bibfnamefont {Yury}\ \bibnamefont
  {Polyanskiy}}\ and\ \bibinfo {author} {\bibfnamefont {Sergio}\ \bibnamefont
  {Verd\'u}},\ }\bibfield  {title} {\enquote {\bibinfo {title} {Arimoto channel
  coding converse and {R\'enyi} divergence},}\ }in\ \href@noop {} {\emph
  {\bibinfo {booktitle} {Proceedings of the 48th Annual Allerton Conference on
  Communication, Control, and Computation}}}\ (\bibinfo {year} {2010})\ pp.\
  \bibinfo {pages} {1327--1333}\BibitemShut {NoStop}%
\bibitem [{\citenamefont {Sharma}\ and\ \citenamefont {Warsi}(2012)}]{SW12}%
  \BibitemOpen
  \bibfield  {author} {\bibinfo {author} {\bibfnamefont {Naresh}\ \bibnamefont
  {Sharma}}\ and\ \bibinfo {author} {\bibfnamefont {Naqueeb~Ahmad}\
  \bibnamefont {Warsi}},\ }\bibfield  {title} {\enquote {\bibinfo {title} {On
  the strong converses for the quantum channel capacity theorems},}\ }\href
  {\doibase 10.1103/PhysRevLett.110.080501} {\bibfield  {journal} {\bibinfo
  {journal} {Physical Review Letters}\ } (\bibinfo {year} {2012}),\
  10.1103/PhysRevLett.110.080501},\ \bibinfo {note}
  {arXiv:1205.1712}\BibitemShut {NoStop}%
\bibitem [{\citenamefont {Umegaki}(1962)}]{Ume62}%
  \BibitemOpen
  \bibfield  {author} {\bibinfo {author} {\bibfnamefont {Hisaharu}\
  \bibnamefont {Umegaki}},\ }\bibfield  {title} {\enquote {\bibinfo {title}
  {Conditional expectations in an operator algebra, {IV} (entropy and
  information)},}\ }\href {\doibase 10.2996/kmj/1138844604} {\bibfield
  {journal} {\bibinfo  {journal} {Kodai Mathematical Seminar Reports}\ }\textbf
  {\bibinfo {volume} {14}},\ \bibinfo {pages} {59--85} (\bibinfo {year}
  {1962})}\BibitemShut {NoStop}%
\bibitem [{\citenamefont {Datta}(2009{\natexlab{b}})}]{D09}%
  \BibitemOpen
  \bibfield  {author} {\bibinfo {author} {\bibfnamefont {Nilanjana}\
  \bibnamefont {Datta}},\ }\bibfield  {title} {\enquote {\bibinfo {title} {Min-
  and max-relative entropies and a new entanglement monotone},}\ }\href@noop {}
  {\bibfield  {journal} {\bibinfo  {journal} {IEEE Transactions on Information
  Theory}\ }\textbf {\bibinfo {volume} {55}},\ \bibinfo {pages} {2816--2826}
  (\bibinfo {year} {2009}{\natexlab{b}})},\ \bibinfo {note}
  {arXiv:0803.2770}\BibitemShut {NoStop}%
\bibitem [{\citenamefont {Wang}\ and\ \citenamefont {Renner}(2012)}]{WR12}%
  \BibitemOpen
  \bibfield  {author} {\bibinfo {author} {\bibfnamefont {Ligong}\ \bibnamefont
  {Wang}}\ and\ \bibinfo {author} {\bibfnamefont {Renato}\ \bibnamefont
  {Renner}},\ }\bibfield  {title} {\enquote {\bibinfo {title} {One-shot
  classical-quantum capacity and hypothesis testing},}\ }\href {\doibase
  10.1103/PhysRevLett.108.200501} {\bibfield  {journal} {\bibinfo  {journal}
  {Physical Review Letters}\ }\textbf {\bibinfo {volume} {108}},\ \bibinfo
  {pages} {200501} (\bibinfo {year} {2012})},\ \bibinfo {note}
  {arXiv:1007.5456}\BibitemShut {NoStop}%
\bibitem [{\citenamefont {Nielsen}\ and\ \citenamefont {Chuang}(2000)}]{NC00}%
  \BibitemOpen
  \bibfield  {author} {\bibinfo {author} {\bibfnamefont {Michael~A.}\
  \bibnamefont {Nielsen}}\ and\ \bibinfo {author} {\bibfnamefont {Isaac~L.}\
  \bibnamefont {Chuang}},\ }\href@noop {} {\emph {\bibinfo {title} {Quantum
  computation and quantum information}}}\ (\bibinfo  {publisher} {Cambridge
  University Press},\ \bibinfo {year} {2000})\BibitemShut {NoStop}%
\bibitem [{\citenamefont {Augusiak}\ \emph {et~al.}(2010)\citenamefont
  {Augusiak}, \citenamefont {Cavalcanti}, \citenamefont {Prettico},\ and\
  \citenamefont {Ac\'{\i}n}}]{ACPA10}%
  \BibitemOpen
  \bibfield  {author} {\bibinfo {author} {\bibfnamefont {Remigiusz}\
  \bibnamefont {Augusiak}}, \bibinfo {author} {\bibfnamefont {Daniel}\
  \bibnamefont {Cavalcanti}}, \bibinfo {author} {\bibfnamefont {Giuseppe}\
  \bibnamefont {Prettico}}, \ and\ \bibinfo {author} {\bibfnamefont {Antonio}\
  \bibnamefont {Ac\'{\i}n}},\ }\bibfield  {title} {\enquote {\bibinfo {title}
  {Perfect quantum privacy implies nonlocality},}\ }\href {\doibase
  10.1103/PhysRevLett.104.230401} {\bibfield  {journal} {\bibinfo  {journal}
  {Physical Review Letters}\ }\textbf {\bibinfo {volume} {104}},\ \bibinfo
  {pages} {230401} (\bibinfo {year} {2010})}\BibitemShut {NoStop}%
\bibitem [{\citenamefont {B{\"a}uml}\ \emph {et~al.}(2019)\citenamefont
  {B{\"a}uml}, \citenamefont {Das}, \citenamefont {Wang},\ and\ \citenamefont
  {Wilde}}]{BDWW19}%
  \BibitemOpen
  \bibfield  {author} {\bibinfo {author} {\bibfnamefont {Stefan}\ \bibnamefont
  {B{\"a}uml}}, \bibinfo {author} {\bibfnamefont {Siddhartha}\ \bibnamefont
  {Das}}, \bibinfo {author} {\bibfnamefont {Xin}\ \bibnamefont {Wang}}, \ and\
  \bibinfo {author} {\bibfnamefont {Mark~M}\ \bibnamefont {Wilde}},\ }\bibfield
   {title} {\enquote {\bibinfo {title} {Resource theory of entanglement for
  bipartite quantum channels},}\ }\href@noop {} {\  (\bibinfo {year} {2019})},\
  \bibinfo {note} {arXiv:1907.04181}\BibitemShut {NoStop}%
\bibitem [{\citenamefont {Gour}\ and\ \citenamefont {Scandolo}(2019)}]{GS19}%
  \BibitemOpen
  \bibfield  {author} {\bibinfo {author} {\bibfnamefont {Gilad}\ \bibnamefont
  {Gour}}\ and\ \bibinfo {author} {\bibfnamefont {Carlo~Maria}\ \bibnamefont
  {Scandolo}},\ }\bibfield  {title} {\enquote {\bibinfo {title} {The
  entanglement of a bipartite channel},}\ }\href@noop {} {\  (\bibinfo {year}
  {2019})},\ \bibinfo {note} {arXiv:1907.02552}\BibitemShut {NoStop}%
\bibitem [{\citenamefont {Friis}\ \emph {et~al.}(2019)\citenamefont {Friis},
  \citenamefont {Vitagliano}, \citenamefont {Malik},\ and\ \citenamefont
  {Huber}}]{friis2019entanglement}%
  \BibitemOpen
  \bibfield  {author} {\bibinfo {author} {\bibfnamefont {Nicolai}\ \bibnamefont
  {Friis}}, \bibinfo {author} {\bibfnamefont {Giuseppe}\ \bibnamefont
  {Vitagliano}}, \bibinfo {author} {\bibfnamefont {Mehul}\ \bibnamefont
  {Malik}}, \ and\ \bibinfo {author} {\bibfnamefont {Marcus}\ \bibnamefont
  {Huber}},\ }\bibfield  {title} {\enquote {\bibinfo {title} {Entanglement
  certification from theory to experiment},}\ }\href@noop {} {\bibfield
  {journal} {\bibinfo  {journal} {Nature Reviews Physics}\ }\textbf {\bibinfo
  {volume} {1}},\ \bibinfo {pages} {72--87} (\bibinfo {year}
  {2019})}\BibitemShut {NoStop}%
\bibitem [{\citenamefont {Contreras-Tejada}\ \emph {et~al.}(2019)\citenamefont
  {Contreras-Tejada}, \citenamefont {Palazuelos},\ and\ \citenamefont
  {de~Vicente}}]{CPV19}%
  \BibitemOpen
  \bibfield  {author} {\bibinfo {author} {\bibfnamefont {Patricia}\
  \bibnamefont {Contreras-Tejada}}, \bibinfo {author} {\bibfnamefont {Carlos}\
  \bibnamefont {Palazuelos}}, \ and\ \bibinfo {author} {\bibfnamefont
  {Julio~I.}\ \bibnamefont {de~Vicente}},\ }\bibfield  {title} {\enquote
  {\bibinfo {title} {Resource theory of entanglement with a unique multipartite
  maximally entangled state},}\ }\href {\doibase
  10.1103/PhysRevLett.122.120503} {\bibfield  {journal} {\bibinfo  {journal}
  {Physical Review Letters}\ }\textbf {\bibinfo {volume} {122}},\ \bibinfo
  {pages} {120503} (\bibinfo {year} {2019})}\BibitemShut {NoStop}%
\bibitem [{\citenamefont {Bennett}\ \emph {et~al.}(2003)\citenamefont
  {Bennett}, \citenamefont {Harrow}, \citenamefont {Leung},\ and\ \citenamefont
  {Smolin}}]{BHLS03}%
  \BibitemOpen
  \bibfield  {author} {\bibinfo {author} {\bibfnamefont {Charles~H.}\
  \bibnamefont {Bennett}}, \bibinfo {author} {\bibfnamefont {Aram~W.}\
  \bibnamefont {Harrow}}, \bibinfo {author} {\bibfnamefont {Debbie~W.}\
  \bibnamefont {Leung}}, \ and\ \bibinfo {author} {\bibfnamefont {John~A.}\
  \bibnamefont {Smolin}},\ }\bibfield  {title} {\enquote {\bibinfo {title} {On
  the capacities of bipartite {Hamiltonians} and unitary gates},}\ }\href@noop
  {} {\bibfield  {journal} {\bibinfo  {journal} {IEEE Transactions on
  Information Theory}\ }\textbf {\bibinfo {volume} {49}},\ \bibinfo {pages}
  {1895--1911} (\bibinfo {year} {2003})},\ \bibinfo {note}
  {arXiv:quant-ph/0205057}\BibitemShut {NoStop}%
\bibitem [{\citenamefont {Kaur}\ and\ \citenamefont {Wilde}(2017)}]{KW17}%
  \BibitemOpen
  \bibfield  {author} {\bibinfo {author} {\bibfnamefont {Eneet}\ \bibnamefont
  {Kaur}}\ and\ \bibinfo {author} {\bibfnamefont {Mark~M.}\ \bibnamefont
  {Wilde}},\ }\bibfield  {title} {\enquote {\bibinfo {title} {Amortized
  entanglement of a quantum channel and approximately teleportation-simulable
  channels},}\ }\href {http://iopscience.iop.org/10.1088/1751-8121/aa9da7}
  {\bibfield  {journal} {\bibinfo  {journal} {Journal of Physics A:
  Mathematical and Theoretical}\ } (\bibinfo {year} {2017})},\ \bibinfo {note}
  {arXiv:1707.07721}\BibitemShut {NoStop}%
\bibitem [{\citenamefont {Fang}\ \emph {et~al.}(2019)\citenamefont {Fang},
  \citenamefont {Fawzi}, \citenamefont {Renner},\ and\ \citenamefont
  {Sutter}}]{fang2019chain}%
  \BibitemOpen
  \bibfield  {author} {\bibinfo {author} {\bibfnamefont {Kun}\ \bibnamefont
  {Fang}}, \bibinfo {author} {\bibfnamefont {Omar}\ \bibnamefont {Fawzi}},
  \bibinfo {author} {\bibfnamefont {Renato}\ \bibnamefont {Renner}}, \ and\
  \bibinfo {author} {\bibfnamefont {David}\ \bibnamefont {Sutter}},\ }\bibfield
   {title} {\enquote {\bibinfo {title} {A chain rule for the quantum relative
  entropy},}\ }\href@noop {} {\  (\bibinfo {year} {2019})},\ \bibinfo {note}
  {arXiv:1909.05826}\BibitemShut {NoStop}%
\bibitem [{\citenamefont {Tomamichel}\ and\ \citenamefont {Tan}(2015)}]{TT13}%
  \BibitemOpen
  \bibfield  {author} {\bibinfo {author} {\bibfnamefont {Marco}\ \bibnamefont
  {Tomamichel}}\ and\ \bibinfo {author} {\bibfnamefont {Vincent Y.~F.}\
  \bibnamefont {Tan}},\ }\bibfield  {title} {\enquote {\bibinfo {title}
  {Second-order asymptotics for the classical capacity of image-additive
  quantum channels},}\ }\href@noop {} {\bibfield  {journal} {\bibinfo
  {journal} {Communication in Mathematical Physics}\ }\textbf {\bibinfo
  {volume} {338}},\ \bibinfo {pages} {103--137} (\bibinfo {year} {2015})},\
  \bibinfo {note} {arXiv:1308.6503}\BibitemShut {NoStop}%
\bibitem [{\citenamefont {Goodenough}\ \emph {et~al.}(2016)\citenamefont
  {Goodenough}, \citenamefont {Elkouss},\ and\ \citenamefont {Wehner}}]{GES16}%
  \BibitemOpen
  \bibfield  {author} {\bibinfo {author} {\bibfnamefont {K.}~\bibnamefont
  {Goodenough}}, \bibinfo {author} {\bibfnamefont {D.}~\bibnamefont {Elkouss}},
  \ and\ \bibinfo {author} {\bibfnamefont {S.}~\bibnamefont {Wehner}},\
  }\bibfield  {title} {\enquote {\bibinfo {title} {Assessing the performance of
  quantum repeaters for all phase-insensitive gaussian bosonic channels},}\
  }\href {\doibase 10.1088/1367-2630/18/6/063005} {\bibfield  {journal}
  {\bibinfo  {journal} {New Journal of Physics}\ }\textbf {\bibinfo {volume}
  {18}},\ \bibinfo {pages} {063005} (\bibinfo {year} {2016})}\BibitemShut
  {NoStop}%
\bibitem [{\citenamefont {Wilde}\ and\ \citenamefont {Qi}(2018)}]{WQ18}%
  \BibitemOpen
  \bibfield  {author} {\bibinfo {author} {\bibfnamefont {Mark~M.}\ \bibnamefont
  {Wilde}}\ and\ \bibinfo {author} {\bibfnamefont {Haoyu}\ \bibnamefont {Qi}},\
  }\bibfield  {title} {\enquote {\bibinfo {title} {Energy-constrained private
  and quantum capacities of quantum channels},}\ }\href {\doibase
  10.1109/tit.2018.2854766} {\bibfield  {journal} {\bibinfo  {journal} {{IEEE}
  Transactions on Information Theory}\ }\textbf {\bibinfo {volume} {64}},\
  \bibinfo {pages} {7802--7827} (\bibinfo {year} {2018})}\BibitemShut {NoStop}%
\bibitem [{\citenamefont {Gottesman}\ and\ \citenamefont
  {Chuang}(1999)}]{GC99}%
  \BibitemOpen
  \bibfield  {author} {\bibinfo {author} {\bibfnamefont {Daniel}\ \bibnamefont
  {Gottesman}}\ and\ \bibinfo {author} {\bibfnamefont {Isaac~L.}\ \bibnamefont
  {Chuang}},\ }\bibfield  {title} {\enquote {\bibinfo {title} {Demonstrating
  the viability of universal quantum computation using teleportation and
  single-qubit operations},}\ }\href {\doibase 10.1038/46503} {\bibfield
  {journal} {\bibinfo  {journal} {Nature}\ }\textbf {\bibinfo {volume} {402}},\
  \bibinfo {pages} {390--393} (\bibinfo {year} {1999})},\ \bibinfo {note}
  {arXiv:quant-ph/9908010}\BibitemShut {NoStop}%
\bibitem [{\citenamefont {Cirac}\ \emph {et~al.}(2001)\citenamefont {Cirac},
  \citenamefont {D\"{u}r}, \citenamefont {Kraus},\ and\ \citenamefont
  {Lewenstein}}]{CDKL01}%
  \BibitemOpen
  \bibfield  {author} {\bibinfo {author} {\bibfnamefont {Juan~Ignacio}\
  \bibnamefont {Cirac}}, \bibinfo {author} {\bibfnamefont {Wolfgang}\
  \bibnamefont {D\"{u}r}}, \bibinfo {author} {\bibfnamefont {Barabara}\
  \bibnamefont {Kraus}}, \ and\ \bibinfo {author} {\bibfnamefont {Maciej}\
  \bibnamefont {Lewenstein}},\ }\bibfield  {title} {\enquote {\bibinfo {title}
  {Entangling operations and their implementation using a small amount of
  entanglement},}\ }\href {\doibase 10.1103/PhysRevLett.86.544} {\bibfield
  {journal} {\bibinfo  {journal} {Physical Review Letters}\ }\textbf {\bibinfo
  {volume} {86}},\ \bibinfo {pages} {544--547} (\bibinfo {year} {2001})},\
  \bibinfo {note} {arXiv:quant-ph/0007057}\BibitemShut {NoStop}%
\bibitem [{\citenamefont {D\"ur}\ \emph {et~al.}(2008)\citenamefont {D\"ur},
  \citenamefont {Bremner},\ and\ \citenamefont {Briegel}}]{DBB08}%
  \BibitemOpen
  \bibfield  {author} {\bibinfo {author} {\bibfnamefont {Wolfgang}\
  \bibnamefont {D\"ur}}, \bibinfo {author} {\bibfnamefont {Michael~J.}\
  \bibnamefont {Bremner}}, \ and\ \bibinfo {author} {\bibfnamefont {Hans~J.}\
  \bibnamefont {Briegel}},\ }\bibfield  {title} {\enquote {\bibinfo {title}
  {Quantum simulation of interacting high-dimensional systems: The influence of
  noise},}\ }\href {\doibase 10.1103/PhysRevA.78.052325} {\bibfield  {journal}
  {\bibinfo  {journal} {Physical Review A}\ }\textbf {\bibinfo {volume} {78}},\
  \bibinfo {pages} {052325} (\bibinfo {year} {2008})},\ \bibinfo {note}
  {arXiv:0706.0154}\BibitemShut {NoStop}%
\bibitem [{\citenamefont {Wu}\ \emph {et~al.}(2016)\citenamefont {Wu},
  \citenamefont {Zhou}, \citenamefont {Gong}, \citenamefont {Guo},
  \citenamefont {Zhang},\ and\ \citenamefont {He}}]{PhysRevA.93.022325}%
  \BibitemOpen
  \bibfield  {author} {\bibinfo {author} {\bibfnamefont {Yadong}\ \bibnamefont
  {Wu}}, \bibinfo {author} {\bibfnamefont {Jian}\ \bibnamefont {Zhou}},
  \bibinfo {author} {\bibfnamefont {Xinbao}\ \bibnamefont {Gong}}, \bibinfo
  {author} {\bibfnamefont {Ying}\ \bibnamefont {Guo}}, \bibinfo {author}
  {\bibfnamefont {Zhi-Ming}\ \bibnamefont {Zhang}}, \ and\ \bibinfo {author}
  {\bibfnamefont {Guangqiang}\ \bibnamefont {He}},\ }\bibfield  {title}
  {\enquote {\bibinfo {title} {Continuous-variable
  measurement-device-independent multipartite quantum communication},}\ }\href
  {\doibase 10.1103/PhysRevA.93.022325} {\bibfield  {journal} {\bibinfo
  {journal} {Physical Review A}\ }\textbf {\bibinfo {volume} {93}},\ \bibinfo
  {pages} {022325} (\bibinfo {year} {2016})}\BibitemShut {NoStop}%
\bibitem [{\citenamefont {Ottaviani}\ \emph {et~al.}(2019)\citenamefont
  {Ottaviani}, \citenamefont {Lupo}, \citenamefont {Laurenza},\ and\
  \citenamefont {Pirandola}}]{OLLP19}%
  \BibitemOpen
  \bibfield  {author} {\bibinfo {author} {\bibfnamefont {Carlo}\ \bibnamefont
  {Ottaviani}}, \bibinfo {author} {\bibfnamefont {Cosmo}\ \bibnamefont {Lupo}},
  \bibinfo {author} {\bibfnamefont {Riccardo}\ \bibnamefont {Laurenza}}, \ and\
  \bibinfo {author} {\bibfnamefont {Stefano}\ \bibnamefont {Pirandola}},\
  }\bibfield  {title} {\enquote {\bibinfo {title} {Modular network for
  high-rate quantum conferencing},}\ }\href {\doibase
  10.1038/s42005-019-0209-6} {\bibfield  {journal} {\bibinfo  {journal}
  {Communications Physics}\ }\textbf {\bibinfo {volume} {2}} (\bibinfo {year}
  {2019}),\ 10.1038/s42005-019-0209-6}\BibitemShut {NoStop}%
\bibitem [{\citenamefont {\ifmmode~\dot{Z}\else \.{Z}\fi{}ukowski}\ \emph
  {et~al.}(1993)\citenamefont {\ifmmode~\dot{Z}\else \.{Z}\fi{}ukowski},
  \citenamefont {Zeilinger}, \citenamefont {Horne},\ and\ \citenamefont
  {Ekert}}]{PhysRevLett.71.4287}%
  \BibitemOpen
  \bibfield  {author} {\bibinfo {author} {\bibfnamefont {M.}~\bibnamefont
  {\ifmmode~\dot{Z}\else \.{Z}\fi{}ukowski}}, \bibinfo {author} {\bibfnamefont
  {A.}~\bibnamefont {Zeilinger}}, \bibinfo {author} {\bibfnamefont {M.~A.}\
  \bibnamefont {Horne}}, \ and\ \bibinfo {author} {\bibfnamefont {A.~K.}\
  \bibnamefont {Ekert}},\ }\bibfield  {title} {\enquote {\bibinfo {title}
  {``event-ready-detectors'' bell experiment via entanglement swapping},}\
  }\href {\doibase 10.1103/PhysRevLett.71.4287} {\bibfield  {journal} {\bibinfo
   {journal} {Physical Review Letters}\ }\textbf {\bibinfo {volume} {71}},\
  \bibinfo {pages} {4287--4290} (\bibinfo {year} {1993})}\BibitemShut {NoStop}%
\bibitem [{\citenamefont {Ralph}\ and\ \citenamefont {Pryde}(2010)}]{RP10}%
  \BibitemOpen
  \bibfield  {author} {\bibinfo {author} {\bibfnamefont {Tim~C.}\ \bibnamefont
  {Ralph}}\ and\ \bibinfo {author} {\bibfnamefont {Geoff~J.}\ \bibnamefont
  {Pryde}},\ }\bibfield  {title} {\enquote {\bibinfo {title} {Optical quantum
  computation},}\ }in\ \href {\doibase 10.1016/s0079-6638(10)05409-0} {\emph
  {\bibinfo {booktitle} {Progress in Optics}}}\ (\bibinfo  {publisher}
  {Elsevier},\ \bibinfo {year} {2010})\ pp.\ \bibinfo {pages}
  {209--269}\BibitemShut {NoStop}%
\bibitem [{\citenamefont {Grassl}\ \emph {et~al.}(1997)\citenamefont {Grassl},
  \citenamefont {Beth},\ and\ \citenamefont {Pellizzari}}]{GBP97}%
  \BibitemOpen
  \bibfield  {author} {\bibinfo {author} {\bibfnamefont {Markus}\ \bibnamefont
  {Grassl}}, \bibinfo {author} {\bibfnamefont {Thomas}\ \bibnamefont {Beth}}, \
  and\ \bibinfo {author} {\bibfnamefont {Thomas}\ \bibnamefont {Pellizzari}},\
  }\bibfield  {title} {\enquote {\bibinfo {title} {Codes for the quantum
  erasure channel},}\ }\href {\doibase 10.1103/PhysRevA.56.33} {\bibfield
  {journal} {\bibinfo  {journal} {Physical Review A}\ }\textbf {\bibinfo
  {volume} {56}},\ \bibinfo {pages} {33--38} (\bibinfo {year} {1997})},\
  \bibinfo {note} {arXiv:quant-ph/9610042}\BibitemShut {NoStop}%
\bibitem [{\citenamefont {Lu}\ \emph {et~al.}(2019)\citenamefont {Lu},
  \citenamefont {Yin}, \citenamefont {Wang}, \citenamefont {Fan-Yuan},
  \citenamefont {Wang}, \citenamefont {He}, \citenamefont {Chen}, \citenamefont
  {Huang}, \citenamefont {Xu}, \citenamefont {Guo},\ and\ \citenamefont
  {Han}}]{LYW+19}%
  \BibitemOpen
  \bibfield  {author} {\bibinfo {author} {\bibfnamefont {Feng-Yu}\ \bibnamefont
  {Lu}}, \bibinfo {author} {\bibfnamefont {Zhen-Qiang}\ \bibnamefont {Yin}},
  \bibinfo {author} {\bibfnamefont {Rong}\ \bibnamefont {Wang}}, \bibinfo
  {author} {\bibfnamefont {Guan-Jie}\ \bibnamefont {Fan-Yuan}}, \bibinfo
  {author} {\bibfnamefont {Shuang}\ \bibnamefont {Wang}}, \bibinfo {author}
  {\bibfnamefont {De-Yong}\ \bibnamefont {He}}, \bibinfo {author}
  {\bibfnamefont {Wei}\ \bibnamefont {Chen}}, \bibinfo {author} {\bibfnamefont
  {Wei}\ \bibnamefont {Huang}}, \bibinfo {author} {\bibfnamefont {Bing-Jie}\
  \bibnamefont {Xu}}, \bibinfo {author} {\bibfnamefont {Guang-Can}\
  \bibnamefont {Guo}}, \ and\ \bibinfo {author} {\bibfnamefont {Zheng-Fu}\
  \bibnamefont {Han}},\ }\bibfield  {title} {\enquote {\bibinfo {title}
  {Practical issues of twin-field quantum key distribution},}\ }\href {\doibase
  10.1088/1367-2630/ab5a97} {\bibfield  {journal} {\bibinfo  {journal} {New
  Journal of Physics}\ }\textbf {\bibinfo {volume} {21}},\ \bibinfo {pages}
  {123030} (\bibinfo {year} {2019})}\BibitemShut {NoStop}%
\bibitem [{\citenamefont {Wilson}(1996)}]{Wil96book}%
  \BibitemOpen
  \bibfield  {author} {\bibinfo {author} {\bibfnamefont {R.J.}\ \bibnamefont
  {Wilson}},\ }\href {https://books.google.be/books?id=tSolAQAAIAAJ} {\emph
  {\bibinfo {title} {Introduction to Graph Theory}}}\ (\bibinfo  {publisher}
  {Longman},\ \bibinfo {year} {1996})\BibitemShut {NoStop}%
\bibitem [{\citenamefont {Leiserson}\ \emph {et~al.}(2001)\citenamefont
  {Leiserson}, \citenamefont {Rivest}, \citenamefont {Cormen},\ and\
  \citenamefont {Stein}}]{LRCSbook}%
  \BibitemOpen
  \bibfield  {author} {\bibinfo {author} {\bibfnamefont {Charles~Eric}\
  \bibnamefont {Leiserson}}, \bibinfo {author} {\bibfnamefont {Ronald~L}\
  \bibnamefont {Rivest}}, \bibinfo {author} {\bibfnamefont {Thomas~H}\
  \bibnamefont {Cormen}}, \ and\ \bibinfo {author} {\bibfnamefont {Clifford}\
  \bibnamefont {Stein}},\ }\href@noop {} {\emph {\bibinfo {title} {Introduction
  to algorithms}}},\ Vol.~\bibinfo {volume} {6}\ (\bibinfo  {publisher} {MIT
  press Cambridge, MA},\ \bibinfo {year} {2001})\BibitemShut {NoStop}%
\bibitem [{\citenamefont {{D{\"u}r}}\ \emph {et~al.}(2000)\citenamefont
  {{D{\"u}r}}, \citenamefont {{Vidal}},\ and\ \citenamefont
  {{Cirac}}}]{nonEquivWandGHZ}%
  \BibitemOpen
  \bibfield  {author} {\bibinfo {author} {\bibfnamefont {W.}~\bibnamefont
  {{D{\"u}r}}}, \bibinfo {author} {\bibfnamefont {G.}~\bibnamefont {{Vidal}}},
  \ and\ \bibinfo {author} {\bibfnamefont {J.~I.}\ \bibnamefont {{Cirac}}},\
  }\bibfield  {title} {\enquote {\bibinfo {title} {{Three qubits can be
  entangled in two inequivalent ways}},}\ }\href {\doibase
  10.1103/PhysRevA.62.062314} {\bibfield  {journal} {\bibinfo  {journal}
  {\pra}\ }\textbf {\bibinfo {volume} {62}},\ \bibinfo {eid} {062314} (\bibinfo
  {year} {2000})},\ \Eprint {http://arxiv.org/abs/quant-ph/0005115}
  {arXiv:quant-ph/0005115 [quant-ph]} \BibitemShut {NoStop}%
\bibitem [{\citenamefont {Horodecki}\ \emph
  {et~al.}(2009{\natexlab{b}})\citenamefont {Horodecki}, \citenamefont
  {Horodecki}, \citenamefont {Horodecki},\ and\ \citenamefont
  {Horodecki}}]{HHHH09}%
  \BibitemOpen
  \bibfield  {author} {\bibinfo {author} {\bibfnamefont {Ryszard}\ \bibnamefont
  {Horodecki}}, \bibinfo {author} {\bibfnamefont {Pawe\l{}}\ \bibnamefont
  {Horodecki}}, \bibinfo {author} {\bibfnamefont {Micha\l{}}\ \bibnamefont
  {Horodecki}}, \ and\ \bibinfo {author} {\bibfnamefont {Karol}\ \bibnamefont
  {Horodecki}},\ }\bibfield  {title} {\enquote {\bibinfo {title} {Quantum
  entanglement},}\ }\href {\doibase 10.1103/RevModPhys.81.865} {\bibfield
  {journal} {\bibinfo  {journal} {Review of Modern Physics}\ }\textbf {\bibinfo
  {volume} {81}},\ \bibinfo {pages} {865--942} (\bibinfo {year}
  {2009}{\natexlab{b}})},\ \bibinfo {note} {arXiv:quant-ph/0702225}\BibitemShut
  {NoStop}%
\bibitem [{\citenamefont {Amico}\ \emph {et~al.}(2008)\citenamefont {Amico},
  \citenamefont {Fazio}, \citenamefont {Osterloh},\ and\ \citenamefont
  {Vedral}}]{amico2008entanglement}%
  \BibitemOpen
  \bibfield  {author} {\bibinfo {author} {\bibfnamefont {Luigi}\ \bibnamefont
  {Amico}}, \bibinfo {author} {\bibfnamefont {Rosario}\ \bibnamefont {Fazio}},
  \bibinfo {author} {\bibfnamefont {Andreas}\ \bibnamefont {Osterloh}}, \ and\
  \bibinfo {author} {\bibfnamefont {Vlatko}\ \bibnamefont {Vedral}},\
  }\bibfield  {title} {\enquote {\bibinfo {title} {Entanglement in many-body
  systems},}\ }\href {\doibase 10.1103/RevModPhys.80.517} {\bibfield  {journal}
  {\bibinfo  {journal} {Review of Modern Physics}\ }\textbf {\bibinfo {volume}
  {80}},\ \bibinfo {pages} {517--576} (\bibinfo {year} {2008})}\BibitemShut
  {NoStop}%
\bibitem [{\citenamefont {Home}\ \emph {et~al.}(2015)\citenamefont {Home},
  \citenamefont {Saha},\ and\ \citenamefont {Das}}]{HSD15}%
  \BibitemOpen
  \bibfield  {author} {\bibinfo {author} {\bibfnamefont {Dipankar}\
  \bibnamefont {Home}}, \bibinfo {author} {\bibfnamefont {Debashis}\
  \bibnamefont {Saha}}, \ and\ \bibinfo {author} {\bibfnamefont {Siddhartha}\
  \bibnamefont {Das}},\ }\bibfield  {title} {\enquote {\bibinfo {title}
  {Multipartite {B}ell-type inequality by generalizing {W}igner's argument},}\
  }\href {\doibase 10.1103/PhysRevA.91.012102} {\bibfield  {journal} {\bibinfo
  {journal} {Physical Review A}\ }\textbf {\bibinfo {volume} {91}},\ \bibinfo
  {pages} {012102} (\bibinfo {year} {2015})},\ \bibinfo {note}
  {arXiv:1410.7936}\BibitemShut {NoStop}%
\bibitem [{\citenamefont {Fortescue}\ and\ \citenamefont {Lo}(2008)}]{FL-2008}%
  \BibitemOpen
  \bibfield  {author} {\bibinfo {author} {\bibfnamefont {Ben}\ \bibnamefont
  {Fortescue}}\ and\ \bibinfo {author} {\bibfnamefont {Hoi-Kwong}\ \bibnamefont
  {Lo}},\ }\bibfield  {title} {\enquote {\bibinfo {title} {Random-party
  entanglement distillation in multiparty states},}\ }\href {\doibase
  10.1103/PhysRevA.78.012348} {\bibfield  {journal} {\bibinfo  {journal}
  {Physical Review A}\ }\textbf {\bibinfo {volume} {78}},\ \bibinfo {pages}
  {012348} (\bibinfo {year} {2008})}\BibitemShut {NoStop}%
\bibitem [{\citenamefont {Cui}\ \emph {et~al.}(2010)\citenamefont {Cui},
  \citenamefont {Helwig},\ and\ \citenamefont {Lo}}]{Hoi-Kwong2010}%
  \BibitemOpen
  \bibfield  {author} {\bibinfo {author} {\bibfnamefont {Wei}\ \bibnamefont
  {Cui}}, \bibinfo {author} {\bibfnamefont {Wolfram}\ \bibnamefont {Helwig}}, \
  and\ \bibinfo {author} {\bibfnamefont {Hoi-Kwong}\ \bibnamefont {Lo}},\
  }\bibfield  {title} {\enquote {\bibinfo {title} {Bounds on probability of
  transformations between multipartite pure states},}\ }\href {\doibase
  10.1103/PhysRevA.81.012111} {\bibfield  {journal} {\bibinfo  {journal}
  {Physical Review A}\ }\textbf {\bibinfo {volume} {81}},\ \bibinfo {pages}
  {012111} (\bibinfo {year} {2010})}\BibitemShut {NoStop}%
\bibitem [{\citenamefont {Cabello}(2000)}]{Cabello2000}%
  \BibitemOpen
  \bibfield  {author} {\bibinfo {author} {\bibfnamefont {Adan}\ \bibnamefont
  {Cabello}},\ }\href@noop {} {\enquote {\bibinfo {title} {Multiparty key
  distribution and secret sharing based on entanglement swapping},}\ }
  (\bibinfo {year} {2000}),\ \Eprint {http://arxiv.org/abs/quant-ph/0009025}
  {arXiv:quant-ph/0009025 [quant-ph]} \BibitemShut {NoStop}%
\bibitem [{\citenamefont {Scarani}\ and\ \citenamefont
  {Gisin}(2001)}]{Scarani2001}%
  \BibitemOpen
  \bibfield  {author} {\bibinfo {author} {\bibfnamefont {Valerio}\ \bibnamefont
  {Scarani}}\ and\ \bibinfo {author} {\bibfnamefont {Nicolas}\ \bibnamefont
  {Gisin}},\ }\bibfield  {title} {\enquote {\bibinfo {title} {Quantum key
  distribution between n partners: Optimal eavesdropping and bell's
  inequalities},}\ }\href {\doibase 10.1103/PhysRevA.65.012311} {\bibfield
  {journal} {\bibinfo  {journal} {Physical Review A}\ }\textbf {\bibinfo
  {volume} {65}},\ \bibinfo {pages} {012311} (\bibinfo {year}
  {2001})}\BibitemShut {NoStop}%
\bibitem [{\citenamefont {Augusiak}\ and\ \citenamefont
  {Horodecki}(2009{\natexlab{b}})}]{Augusiak2009W}%
  \BibitemOpen
  \bibfield  {author} {\bibinfo {author} {\bibfnamefont {R.}~\bibnamefont
  {Augusiak}}\ and\ \bibinfo {author} {\bibfnamefont {P.}~\bibnamefont
  {Horodecki}},\ }\bibfield  {title} {\enquote {\bibinfo {title} {{W}-like
  bound entangled states and secure key distillation},}\ }\href {\doibase
  10.1209/0295-5075/85/50001} {\bibfield  {journal} {\bibinfo  {journal} {EPL
  (Europhysics Letters)}\ }\textbf {\bibinfo {volume} {85}},\ \bibinfo {pages}
  {50001} (\bibinfo {year} {2009}{\natexlab{b}})}\BibitemShut {NoStop}%
\bibitem [{\citenamefont {Grasselli}\ \emph {et~al.}(2019)\citenamefont
  {Grasselli}, \citenamefont {Kampermann},\ and\ \citenamefont
  {Bruß}}]{Grasselli2019}%
  \BibitemOpen
  \bibfield  {author} {\bibinfo {author} {\bibfnamefont {Federico}\
  \bibnamefont {Grasselli}}, \bibinfo {author} {\bibfnamefont {Hermann}\
  \bibnamefont {Kampermann}}, \ and\ \bibinfo {author} {\bibfnamefont {Dagmar}\
  \bibnamefont {Bruß}},\ }\href@noop {} {\enquote {\bibinfo {title}
  {Conference key agreement with single-photon interference},}\ } (\bibinfo
  {year} {2019}),\ \Eprint {http://arxiv.org/abs/1907.10288} {arXiv:1907.10288
  [quant-ph]} \BibitemShut {NoStop}%
\bibitem [{\citenamefont {{Horodecki}}\ \emph {et~al.}(2008)\citenamefont
  {{Horodecki}}, \citenamefont {{Pankowski}}, \citenamefont {{Horodecki}},\
  and\ \citenamefont {{Horodecki}}}]{Pankowski2008}%
  \BibitemOpen
  \bibfield  {author} {\bibinfo {author} {\bibfnamefont {K.}~\bibnamefont
  {{Horodecki}}}, \bibinfo {author} {\bibfnamefont {L.}~\bibnamefont
  {{Pankowski}}}, \bibinfo {author} {\bibfnamefont {M.}~\bibnamefont
  {{Horodecki}}}, \ and\ \bibinfo {author} {\bibfnamefont {P.}~\bibnamefont
  {{Horodecki}}},\ }\bibfield  {title} {\enquote {\bibinfo {title}
  {Low-dimensional bound entanglement with one-way distillable cryptographic
  key},}\ }\href {\doibase 10.1109/TIT.2008.921709} {\bibfield  {journal}
  {\bibinfo  {journal} {IEEE Transactions on Information Theory}\ }\textbf
  {\bibinfo {volume} {54}},\ \bibinfo {pages} {2621--2625} (\bibinfo {year}
  {2008})}\BibitemShut {NoStop}%
\bibitem [{\citenamefont {Pankowski}\ and\ \citenamefont
  {Horodecki}(2010)}]{Pankowski2010}%
  \BibitemOpen
  \bibfield  {author} {\bibinfo {author} {\bibfnamefont {{\L}ukasz}\
  \bibnamefont {Pankowski}}\ and\ \bibinfo {author} {\bibfnamefont {Micha{\l}}\
  \bibnamefont {Horodecki}},\ }\bibfield  {title} {\enquote {\bibinfo {title}
  {Low-dimensional quite noisy bound entanglement with a cryptographic key},}\
  }\href {\doibase 10.1088/1751-8113/44/3/035301} {\bibfield  {journal}
  {\bibinfo  {journal} {Journal of Physics A: Mathematical and Theoretical}\
  }\textbf {\bibinfo {volume} {44}},\ \bibinfo {pages} {035301} (\bibinfo
  {year} {2010})}\BibitemShut {NoStop}%
\bibitem [{\citenamefont {{Verstraete}}\ \emph {et~al.}(2002)\citenamefont
  {{Verstraete}}, \citenamefont {{Dehaene}},\ and\ \citenamefont {{de
  Moor}}}]{VDM-2002}%
  \BibitemOpen
  \bibfield  {author} {\bibinfo {author} {\bibfnamefont {Frank}\ \bibnamefont
  {{Verstraete}}}, \bibinfo {author} {\bibfnamefont {Jeroen}\ \bibnamefont
  {{Dehaene}}}, \ and\ \bibinfo {author} {\bibfnamefont {Bart}\ \bibnamefont
  {{de Moor}}},\ }\bibfield  {title} {\enquote {\bibinfo {title} {{On the
  geometry of entangled states}},}\ }\href {\doibase 10.1080/09500340110115488}
  {\bibfield  {journal} {\bibinfo  {journal} {Journal of Modern Optics}\
  }\textbf {\bibinfo {volume} {49}},\ \bibinfo {pages} {1277--1287} (\bibinfo
  {year} {2002})},\ \Eprint {http://arxiv.org/abs/quant-ph/0107155}
  {arXiv:quant-ph/0107155 [quant-ph]} \BibitemShut {NoStop}%
\bibitem [{\citenamefont {Horodecki}\ \emph {et~al.}(2000)\citenamefont
  {Horodecki}, \citenamefont {Horodecki},\ and\ \citenamefont
  {Horodecki}}]{Horodecki_2000}%
  \BibitemOpen
  \bibfield  {author} {\bibinfo {author} {\bibfnamefont {Michał}\ \bibnamefont
  {Horodecki}}, \bibinfo {author} {\bibfnamefont {Paweł}\ \bibnamefont
  {Horodecki}}, \ and\ \bibinfo {author} {\bibfnamefont {Ryszard}\ \bibnamefont
  {Horodecki}},\ }\bibfield  {title} {\enquote {\bibinfo {title} {Limits for
  entanglement measures},}\ }\href {\doibase 10.1103/physrevlett.84.2014}
  {\bibfield  {journal} {\bibinfo  {journal} {Physical Review Letters}\
  }\textbf {\bibinfo {volume} {84}},\ \bibinfo {pages} {2014–2017} (\bibinfo
  {year} {2000})}\BibitemShut {NoStop}%
\bibitem [{\citenamefont {Ishizaka}\ and\ \citenamefont
  {Plenio}(2005)}]{Ishizaka_2005}%
  \BibitemOpen
  \bibfield  {author} {\bibinfo {author} {\bibfnamefont {S.}~\bibnamefont
  {Ishizaka}}\ and\ \bibinfo {author} {\bibfnamefont {M.~B.}\ \bibnamefont
  {Plenio}},\ }\bibfield  {title} {\enquote {\bibinfo {title} {Publisher’s
  note: Multiparticle entanglement under asymptotic
  positive-partial-transpose-preserving operations [phys. rev.72, 042325
  (2005)]},}\ }\href {\doibase 10.1103/physreva.72.059907} {\bibfield
  {journal} {\bibinfo  {journal} {Physical Review A}\ }\textbf {\bibinfo
  {volume} {72}} (\bibinfo {year} {2005}),\
  10.1103/physreva.72.059907}\BibitemShut {NoStop}%
\bibitem [{\citenamefont {Plenio}\ and\ \citenamefont
  {Virmani}(2005)}]{plenio2005introduction}%
  \BibitemOpen
  \bibfield  {author} {\bibinfo {author} {\bibfnamefont {Martin~B.}\
  \bibnamefont {Plenio}}\ and\ \bibinfo {author} {\bibfnamefont
  {S.}~\bibnamefont {Virmani}},\ }\href@noop {} {\enquote {\bibinfo {title} {An
  introduction to entanglement measures},}\ } (\bibinfo {year} {2005}),\
  \Eprint {http://arxiv.org/abs/quant-ph/0504163} {arXiv:quant-ph/0504163
  [quant-ph]} \BibitemShut {NoStop}%
\bibitem [{Note1()}]{Note1}%
  \BibitemOpen
  \bibinfo {note} {See Supplemental Material at \protect \url
  {https://journals.aps.org/prx/abstract/10.1103/PhysRevX.11.041016##supplemental}
  for codes to get plots.}\BibitemShut {Stop}%
\bibitem [{\citenamefont {Pivoluska}\ \emph {et~al.}(2018)\citenamefont
  {Pivoluska}, \citenamefont {Huber},\ and\ \citenamefont {Malik}}]{PHM18}%
  \BibitemOpen
  \bibfield  {author} {\bibinfo {author} {\bibfnamefont {Matej}\ \bibnamefont
  {Pivoluska}}, \bibinfo {author} {\bibfnamefont {Marcus}\ \bibnamefont
  {Huber}}, \ and\ \bibinfo {author} {\bibfnamefont {Mehul}\ \bibnamefont
  {Malik}},\ }\bibfield  {title} {\enquote {\bibinfo {title} {Layered quantum
  key distribution},}\ }\href {\doibase 10.1103/PhysRevA.97.032312} {\bibfield
  {journal} {\bibinfo  {journal} {Physical Review A}\ }\textbf {\bibinfo
  {volume} {97}},\ \bibinfo {pages} {032312} (\bibinfo {year}
  {2018})}\BibitemShut {NoStop}%
\bibitem [{\citenamefont {Shannon}(1961)}]{Sha61}%
  \BibitemOpen
  \bibfield  {author} {\bibinfo {author} {\bibfnamefont {Claude~E.}\
  \bibnamefont {Shannon}},\ }\bibfield  {title} {\enquote {\bibinfo {title}
  {Two-way communication channels},}\ }in\ \href
  {https://projecteuclid.org/euclid.bsmsp/1200512185} {\emph {\bibinfo
  {booktitle} {Proceedings of the Fourth Berkeley Symposium on Mathematical
  Statistics and Probability, Volume 1: Contributions to the Theory of
  Statistics}}}\ (\bibinfo  {publisher} {University of California Press},\
  \bibinfo {address} {Berkeley, California},\ \bibinfo {year} {1961})\ pp.\
  \bibinfo {pages} {611--644}\BibitemShut {NoStop}%
\bibitem [{\citenamefont {Gamal}\ and\ \citenamefont {Kim}(2012)}]{GKbook}%
  \BibitemOpen
  \bibfield  {author} {\bibinfo {author} {\bibfnamefont {Abbas~El}\
  \bibnamefont {Gamal}}\ and\ \bibinfo {author} {\bibfnamefont {Young-Han}\
  \bibnamefont {Kim}},\ }\href@noop {} {\emph {\bibinfo {title} {Network
  Information Theory}}}\ (\bibinfo  {publisher} {Cambridge University Press},\
  \bibinfo {year} {2012})\ p.\ \bibinfo {pages} {709},\ \bibinfo {note}
  {arXiv:1001.3404}\BibitemShut {NoStop}%
\bibitem [{\citenamefont {Br\'adler}\ \emph {et~al.}(2010)\citenamefont
  {Br\'adler}, \citenamefont {Hayden}, \citenamefont {Touchette},\ and\
  \citenamefont {Wilde}}]{BHTW10}%
  \BibitemOpen
  \bibfield  {author} {\bibinfo {author} {\bibfnamefont {Kamil}\ \bibnamefont
  {Br\'adler}}, \bibinfo {author} {\bibfnamefont {Patrick}\ \bibnamefont
  {Hayden}}, \bibinfo {author} {\bibfnamefont {Dave}\ \bibnamefont
  {Touchette}}, \ and\ \bibinfo {author} {\bibfnamefont {Mark~M.}\ \bibnamefont
  {Wilde}},\ }\bibfield  {title} {\enquote {\bibinfo {title} {Trade-off
  capacities of the quantum {Hadamard} channels},}\ }\href {\doibase
  10.1103/PhysRevA.81.062312} {\bibfield  {journal} {\bibinfo  {journal}
  {Physical Review A}\ }\textbf {\bibinfo {volume} {81}},\ \bibinfo {pages}
  {062312} (\bibinfo {year} {2010})},\ \bibinfo {note}
  {arXiv:1001.1732}\BibitemShut {NoStop}%
\bibitem [{\citenamefont {Wang}\ \emph {et~al.}(2017)\citenamefont {Wang},
  \citenamefont {Das},\ and\ \citenamefont {Wilde}}]{WDW16}%
  \BibitemOpen
  \bibfield  {author} {\bibinfo {author} {\bibfnamefont {Qingle}\ \bibnamefont
  {Wang}}, \bibinfo {author} {\bibfnamefont {Siddhartha}\ \bibnamefont {Das}},
  \ and\ \bibinfo {author} {\bibfnamefont {Mark~M.}\ \bibnamefont {Wilde}},\
  }\bibfield  {title} {\enquote {\bibinfo {title} {Hadamard quantum broadcast
  channels},}\ }\href {\doibase 10.1007/s11128-017-1697-5} {\bibfield
  {journal} {\bibinfo  {journal} {Quantum Information Processing}\ }\textbf
  {\bibinfo {volume} {16}},\ \bibinfo {pages} {248} (\bibinfo {year} {2017})},\
  \bibinfo {note} {arXiv:1611.07651}\BibitemShut {NoStop}%
\bibitem [{\citenamefont {Leditzky}\ \emph {et~al.}(2019)\citenamefont
  {Leditzky}, \citenamefont {Alhejji}, \citenamefont {Levin},\ and\
  \citenamefont {Smith}}]{LALS19}%
  \BibitemOpen
  \bibfield  {author} {\bibinfo {author} {\bibfnamefont {Felix}\ \bibnamefont
  {Leditzky}}, \bibinfo {author} {\bibfnamefont {Mohammad~A}\ \bibnamefont
  {Alhejji}}, \bibinfo {author} {\bibfnamefont {Joshua}\ \bibnamefont {Levin}},
  \ and\ \bibinfo {author} {\bibfnamefont {Graeme}\ \bibnamefont {Smith}},\
  }\bibfield  {title} {\enquote {\bibinfo {title} {Playing games with multiple
  access channels},}\ }\href@noop {} {\  (\bibinfo {year} {2019})},\ \bibinfo
  {note} {arXiv:1909.02479}\BibitemShut {NoStop}%
\bibitem [{\citenamefont {Tan}\ and\ \citenamefont {Rohde}(2019)}]{TR19}%
  \BibitemOpen
  \bibfield  {author} {\bibinfo {author} {\bibfnamefont {Si-Hui}\ \bibnamefont
  {Tan}}\ and\ \bibinfo {author} {\bibfnamefont {Peter~P.}\ \bibnamefont
  {Rohde}},\ }\bibfield  {title} {\enquote {\bibinfo {title} {The resurgence of
  the linear optics quantum interferometer—recent advances \&
  applications},}\ }\href@noop {} {\bibfield  {journal} {\bibinfo  {journal}
  {Reviews in Physics}\ ,\ \bibinfo {pages} {100030}} (\bibinfo {year}
  {2019})}\BibitemShut {NoStop}%
\bibitem [{\citenamefont {Das}\ and\ \citenamefont {Wilde}(2019)}]{DW19}%
  \BibitemOpen
  \bibfield  {author} {\bibinfo {author} {\bibfnamefont {Siddhartha}\
  \bibnamefont {Das}}\ and\ \bibinfo {author} {\bibfnamefont {Mark~M.}\
  \bibnamefont {Wilde}},\ }\bibfield  {title} {\enquote {\bibinfo {title}
  {Quantum rebound capacity},}\ }\href {\doibase 10.1103/PhysRevA.100.030302}
  {\bibfield  {journal} {\bibinfo  {journal} {Physical Review A}\ }\textbf
  {\bibinfo {volume} {100}},\ \bibinfo {pages} {030302} (\bibinfo {year}
  {2019})},\ \bibinfo {note} {arXiv:1904.10344}\BibitemShut {NoStop}%
\bibitem [{\citenamefont {Frank}\ and\ \citenamefont {Lieb}(2013)}]{FL13}%
  \BibitemOpen
  \bibfield  {author} {\bibinfo {author} {\bibfnamefont {Rupert~L.}\
  \bibnamefont {Frank}}\ and\ \bibinfo {author} {\bibfnamefont {Elliott~H.}\
  \bibnamefont {Lieb}},\ }\bibfield  {title} {\enquote {\bibinfo {title}
  {Monotonicity of a relative {R\'enyi} entropy},}\ }\href@noop {} {\bibfield
  {journal} {\bibinfo  {journal} {Journal of Mathematical Physics}\ }\textbf
  {\bibinfo {volume} {54}},\ \bibinfo {pages} {122201} (\bibinfo {year}
  {2013})},\ \bibinfo {note} {arXiv:1306.5358}\BibitemShut {NoStop}%
\bibitem [{\citenamefont {Beigi}(2013)}]{Bei13}%
  \BibitemOpen
  \bibfield  {author} {\bibinfo {author} {\bibfnamefont {Salman}\ \bibnamefont
  {Beigi}},\ }\bibfield  {title} {\enquote {\bibinfo {title} {Sandwiched
  {R\'enyi} divergence satisfies data processing inequality},}\ }\href@noop {}
  {\bibfield  {journal} {\bibinfo  {journal} {Journal of Mathematical Physics}\
  }\textbf {\bibinfo {volume} {54}},\ \bibinfo {pages} {122202} (\bibinfo
  {year} {2013})},\ \bibinfo {note} {arXiv:1306.5920}\BibitemShut {NoStop}%
\bibitem [{\citenamefont {Hiai}\ and\ \citenamefont {Petz}(1991)}]{HP91}%
  \BibitemOpen
  \bibfield  {author} {\bibinfo {author} {\bibfnamefont {Fumio}\ \bibnamefont
  {Hiai}}\ and\ \bibinfo {author} {\bibfnamefont {D{\'e}nes}\ \bibnamefont
  {Petz}},\ }\bibfield  {title} {\enquote {\bibinfo {title} {The proper formula
  for relative entropy and its asymptotics in quantum probability},}\
  }\href@noop {} {\bibfield  {journal} {\bibinfo  {journal} {Communications in
  mathematical physics}\ }\textbf {\bibinfo {volume} {143}},\ \bibinfo {pages}
  {99--114} (\bibinfo {year} {1991})}\BibitemShut {NoStop}%
\bibitem [{\citenamefont {Nagaoka}(2001)}]{N01}%
  \BibitemOpen
  \bibfield  {author} {\bibinfo {author} {\bibfnamefont {Hiroshi}\ \bibnamefont
  {Nagaoka}},\ }\bibfield  {title} {\enquote {\bibinfo {title} {Strong converse
  theorems in quantum information theory},}\ }\href@noop {} {\bibfield
  {journal} {\bibinfo  {journal} {Proceedings of ERATO Workshop on Quantum
  Information Science}\ ,\ \bibinfo {pages} {33}} (\bibinfo {year} {2001})},\
  \bibinfo {note} {also appeared in Asymptotic Theory of Quantum Statistical
  Inference, ed. M. Hayashi, World Scientific, 2005}\BibitemShut {NoStop}%
\bibitem [{\citenamefont {Ogawa}\ and\ \citenamefont {Nagaoka}(2000)}]{ON00}%
  \BibitemOpen
  \bibfield  {author} {\bibinfo {author} {\bibfnamefont {T.}~\bibnamefont
  {Ogawa}}\ and\ \bibinfo {author} {\bibfnamefont {H.}~\bibnamefont
  {Nagaoka}},\ }\bibfield  {title} {\enquote {\bibinfo {title} {Strong converse
  and {S}tein's lemma in quantum hypothesis testing},}\ }\href {\doibase
  10.1109/18.887855} {\bibfield  {journal} {\bibinfo  {journal} {IEEE
  Transactions on Information Theory}\ }\textbf {\bibinfo {volume} {46}},\
  \bibinfo {pages} {2428--2433} (\bibinfo {year} {2000})},\ \bibinfo {note}
  {arXiv:quant-ph/9906090v1}\BibitemShut {NoStop}%
\bibitem [{\citenamefont {Cooney}\ \emph {et~al.}(2016)\citenamefont {Cooney},
  \citenamefont {Mosonyi},\ and\ \citenamefont {Wilde}}]{CMW14}%
  \BibitemOpen
  \bibfield  {author} {\bibinfo {author} {\bibfnamefont {Tom}\ \bibnamefont
  {Cooney}}, \bibinfo {author} {\bibfnamefont {Milan}\ \bibnamefont {Mosonyi}},
  \ and\ \bibinfo {author} {\bibfnamefont {Mark~M.}\ \bibnamefont {Wilde}},\
  }\bibfield  {title} {\enquote {\bibinfo {title} {Strong converse exponents
  for a quantum channel discrimination problem and quantum-feedback-assisted
  communication},}\ }\href@noop {} {\bibfield  {journal} {\bibinfo  {journal}
  {Communications in Mathematical Physics}\ }\textbf {\bibinfo {volume}
  {344}},\ \bibinfo {pages} {797--829} (\bibinfo {year} {2016})},\ \bibinfo
  {note} {arXiv:1408.3373}\BibitemShut {NoStop}%
\bibitem [{\citenamefont {Fawzi}\ \emph {et~al.}(2012)\citenamefont {Fawzi},
  \citenamefont {Hayden}, \citenamefont {Savov}, \citenamefont {Sen},\ and\
  \citenamefont {Wilde}}]{FHSSW11}%
  \BibitemOpen
  \bibfield  {author} {\bibinfo {author} {\bibfnamefont {Omar}\ \bibnamefont
  {Fawzi}}, \bibinfo {author} {\bibfnamefont {Patrick}\ \bibnamefont {Hayden}},
  \bibinfo {author} {\bibfnamefont {Ivan}\ \bibnamefont {Savov}}, \bibinfo
  {author} {\bibfnamefont {Pranab}\ \bibnamefont {Sen}}, \ and\ \bibinfo
  {author} {\bibfnamefont {Mark~M.}\ \bibnamefont {Wilde}},\ }\bibfield
  {title} {\enquote {\bibinfo {title} {Classical communication over a quantum
  interference channel},}\ }\href@noop {} {\bibfield  {journal} {\bibinfo
  {journal} {IEEE Transactions on Information Theory}\ }\textbf {\bibinfo
  {volume} {58}},\ \bibinfo {pages} {3670--3691} (\bibinfo {year} {2012})},\
  \bibinfo {note} {arXiv:1102.2624}\BibitemShut {NoStop}%
\bibitem [{\citenamefont {Devetak}(2005)}]{D05}%
  \BibitemOpen
  \bibfield  {author} {\bibinfo {author} {\bibfnamefont {Igor}\ \bibnamefont
  {Devetak}},\ }\bibfield  {title} {\enquote {\bibinfo {title} {The private
  classical capacity and quantum capacity of a quantum channel},}\ }\href
  {\doibase 10.1109/TIT.2004.839515} {\bibfield  {journal} {\bibinfo  {journal}
  {IEEE Transactions on Information Theory}\ }\textbf {\bibinfo {volume}
  {51}},\ \bibinfo {pages} {44--55} (\bibinfo {year} {2005})},\ \bibinfo {note}
  {arXiv:quant-ph/0304127}\BibitemShut {NoStop}%
\bibitem [{\citenamefont {Yard}\ \emph {et~al.}(2011)\citenamefont {Yard},
  \citenamefont {Hayden},\ and\ \citenamefont {Devetak}}]{YHD2006}%
  \BibitemOpen
  \bibfield  {author} {\bibinfo {author} {\bibfnamefont {Jon}\ \bibnamefont
  {Yard}}, \bibinfo {author} {\bibfnamefont {Patrick}\ \bibnamefont {Hayden}},
  \ and\ \bibinfo {author} {\bibfnamefont {Igor}\ \bibnamefont {Devetak}},\
  }\bibfield  {title} {\enquote {\bibinfo {title} {Quantum broadcast
  channels},}\ }\href@noop {} {\bibfield  {journal} {\bibinfo  {journal} {IEEE
  Transactions on Information Theory}\ }\textbf {\bibinfo {volume} {57}},\
  \bibinfo {pages} {7147--7162} (\bibinfo {year} {2011})},\ \bibinfo {note}
  {arXiv:quant-ph/0603098}\BibitemShut {NoStop}%
\bibitem [{\citenamefont {Yard}\ \emph {et~al.}(2008)\citenamefont {Yard},
  \citenamefont {Hayden},\ and\ \citenamefont {Devetak}}]{YHD05MQAC}%
  \BibitemOpen
  \bibfield  {author} {\bibinfo {author} {\bibfnamefont {Jon}\ \bibnamefont
  {Yard}}, \bibinfo {author} {\bibfnamefont {Patrick}\ \bibnamefont {Hayden}},
  \ and\ \bibinfo {author} {\bibfnamefont {Igor}\ \bibnamefont {Devetak}},\
  }\bibfield  {title} {\enquote {\bibinfo {title} {Capacity theorems for
  quantum multiple-access channels: Classical-quantum and quantum-quantum
  capacity regions},}\ }\href@noop {} {\bibfield  {journal} {\bibinfo
  {journal} {IEEE Transactions on Information Theory}\ }\textbf {\bibinfo
  {volume} {54}},\ \bibinfo {pages} {3091--3113} (\bibinfo {year} {2008})},\
  \bibinfo {note} {arXiv:quant-ph/0501045}\BibitemShut {NoStop}%
\bibitem [{\citenamefont {Ambainis}\ \emph {et~al.}(2000)\citenamefont
  {Ambainis}, \citenamefont {Mosca}, \citenamefont {Tapp},\ and\ \citenamefont
  {de~Wolf}}]{AMTW00}%
  \BibitemOpen
  \bibfield  {author} {\bibinfo {author} {\bibfnamefont {Andris}\ \bibnamefont
  {Ambainis}}, \bibinfo {author} {\bibfnamefont {Michele}\ \bibnamefont
  {Mosca}}, \bibinfo {author} {\bibfnamefont {Alain}\ \bibnamefont {Tapp}}, \
  and\ \bibinfo {author} {\bibfnamefont {Ronald}\ \bibnamefont {de~Wolf}},\
  }\bibfield  {title} {\enquote {\bibinfo {title} {Private quantum channels},}\
  }\href {\doibase http://doi.ieeecomputersociety.org/10.1109/SFCS.2000.892142}
  {\bibfield  {journal} {\bibinfo  {journal} {IEEE 41st Annual Symposium on
  Foundations of Computer Science}\ ,\ \bibinfo {pages} {547--553}} (\bibinfo
  {year} {2000})},\ \bibinfo {note} {arXiv:quant-ph/0003101}\BibitemShut
  {NoStop}%
\bibitem [{\citenamefont {K.C.~Toh}\ and\ \citenamefont
  {Tutuncu}(1999)}]{SDPT3}%
  \BibitemOpen
  \bibfield  {author} {\bibinfo {author} {\bibfnamefont {M.J.~Todd}\
  \bibnamefont {K.C.~Toh}}\ and\ \bibinfo {author} {\bibfnamefont {R.H.}\
  \bibnamefont {Tutuncu}},\ }\href {https://github.com/SQLP/SDPT3} {\enquote
  {\bibinfo {title} {Sdpt3 --- a matlab software package for semidefinite
  programming},}\ } (\bibinfo {year} {1999})\BibitemShut {NoStop}%
\end{thebibliography}%
\end{document}